%% file: main.tex
\documentclass[a4paper, fleqn]{./tex/modacmtrans2e}

\usepackage{./tex/mystyle}

\begin{document}

\title{%
  Definition and Implementation \\
  of a Points-To Analysis \\
    for C-like Languages
}

\begin{abstract}
\input{./tex/abstract_en.tex}
\end{abstract}

\author{%
  STEFANO SOFFIA \\
  Applied Formal Methods Laboratory \\
  Department of Mathematics, University of Parma, Italy
}

\markboth{%
S.~Soffia
}%
{Definition and Implementation of a Points-To Analysis}

\category{F3.1}{Logics and Meanings of Programs}{Specifying and Verifying and Reasoning about Programs.}
\terms{Languages, Static Analysis.}

\begin{bottomstuff}
This work has been developed in partial fulfillment of the requirements
for the first level degree in Computer Science.
\end{bottomstuff}

\maketitle

\tableofcontents

\input{./tex/introduction.tex}

\input{./tex/proof.tex}
\input{./tex/extensions.tex}
\input{./tex/conclusions.tex}

\begin{acknowledgments}
\input{./tex/acknowledgments.tex}
\end{acknowledgments}

\bibliographystyle{amsalpha}
\bibliography{biblio/mybib,biblio/ppl}

\end{document}

%% file: tex/abstract_en.tex
The points-to problem is the problem of determining the possible
run-time targets of pointer variables and is usually considered part of the more
general aliasing problem, which consists in establishing whether and
when different expressions can refer to the same memory address.  Aliasing
information is essential to every tool that needs to reason about the
semantics of programs.  However, due to well-known undecidability
results, for all interesting languages that admit aliasing, the exact
solution of nontrivial aliasing problems is not generally computable.  This
work focuses on approximated solutions to this problem by
presenting a store-based, flow-sensitive points-to analysis,
for applications in the field of automated software verification.
In contrast to software testing procedures, which
heuristically check the program against a finite set of executions, the
methods considered in this work are static analyses, where
the computed results are valid for all the possible executions of the
analyzed program.  We present a simplified programming language and its
execution model; then an approximated execution model is developed
using the ideas of abstract interpretation theory.  Finally, the
soundness of the approximation is formally proved.  The aim of developing a
realistic points-to analysis is pursued by presenting some extensions to
the initial simplified model and discussing the correctness of their
formulation.  This work contains original contributions to the issue of
points-to analysis, as it provides a formulation of a filter operation on
the points-to abstract domain and a formal proof of the soundness of the
defined abstract operations: these, as far as we now, are lacking from the
previous literature.

%% file: tex/introduction.tex
\section{Introduction}
\label{section:Introduction}
\subsection{The Aliasing Problem}
In imperative programming languages the concept of \emph{memory location}
is of main importance;  it
refers to an entity able to keep a finite quantity of information across
the subsequent steps of the computation. The concept of \emph{variable} is then
developed as a way to refer to memory locations.
In the different languages, different constructs allow for the
composition of variable names so as to form \emph{expressions}
(\refcodesnippet{composition of names}).  From the use of these constructs
comes the possibility to refer to the same memory location with
different expressions.
In the literature,
two expressions referring to the same memory
location are said to be \emph{aliases};
the set of pairs of alias expressions is
commonly referred to as \emph{alias information}
and the \emph{aliasing problem} is known as the problem of
analyzing the alias information of a program.
Due to the many mechanism that can lead to the
generation of aliases,
the aliasing problem is
complex even to characterize.
The following paragraphs
show how the different constructs of the C~language can
affect the alias information.

\codecaption
{different constructs of the
 C~language can be used to compose variables into
 expressions.  Note at \refline{7} the use of the \emph{dereference
 operator},
 of the \emph{index}
 and \emph{field selectors} in the same expression. Many are the available
 constructs and complex is the problem of analysing all their possible
 interactions.}
{composition of names}
\begin{codesnippet}
struct S {
  struct S *l, *r;
  int key;
} a[10];
int i;
...
a[i].l->key = ...
\end{codesnippet}

\subsubsection{Aliasing From the Use of Arrays}
The example presented in \refcodesnippet{aliasing from arrays} shows how,
through the use of the array's indexing mechanism, the aliasing problem is
influenced by the value of integer variables.  As shown by
\refcodesnippet{difference of pointers}, also the converse holds --- the
value of pointer variables, typically considered a alias-related issue, can
influence the value of integer variables.

\codecaption
{this example shows how the use of arrays may produce aliasing.
 At \refline{5} the variables \Variable{i} and \Variable{j} hold
 the same value;  then the expressions \QuotedCode{a[i]} and
 \QuotedCode{a[j]} denote the same memory location, i.e., they are
 \emph{aliases}.}
{aliasing from arrays}
\begin{codesnippet}
int a[10], i, j;
...
if (i == j) {
  ...
  a[i] = a[j];
  ...
}
\end{codesnippet}

\codecaption
{the value assigned to the variable \Variable{dist} at \refline{3} depends
on the distance between the elements referred to by the pointers
\Variable{p} and \Variable{q}.}
{difference of pointers}
\begin{codesnippet}
int a[10], *p, *q, dist;
...
dist = q - p;
...
\end{codesnippet}

\subsubsection{Aliasing From the Use of Pointers}
The simple example in \refcodesnippet{aliasing from pointers} shows how the
use of pointers can produce aliasing. In the C~language the support of
pointers is particularly flexible and powerful. For instance,
multiple levels of indirections are allowed (\refcodesnippet{multiple
levels of indirection}). These characteristics make the development of
alias analyses for the C language a challenging problem.  The study of the
aliasing problem requires also to cover recursive data structures; the use
of these can produce particularly complex alias relations
(\refcodesnippet{recursive data structures}).

\codecaption
{after the execution of \refline{2}, the pointer
 variable \Variable{p} contains the address of
 the variable \Variable{a};
 then the expressions \Variable{*p} and
 \Variable{a} are \emph{aliases}.}
{aliasing from pointers}
\begin{codesnippet}
int a, *p;
p = &a;
...
\end{codesnippet}

\codecaption
{at \refline{2} the address of \Variable{a} is assigned
to \Variable{p}; as a consequence, \QuotedCode{*p} and \QuotedCode{a}
become aliases.  At \refline{3} the address of \Variable{p}
is assigned to \Variable{pp}; as a consequence, \QuotedCode{*pp} and
\QuotedCode{p} become aliases.  Hence, also the expressions \QuotedCode{**pp}
and \QuotedCode{*p} are aliases.  Finally, by applying the transitive
property, it is possible to conclude
that \Variable{**pp} and \Variable{a} are aliases too.}
{multiple levels of indirection}
\begin{codesnippet}
int a, *p, **pp;
p = &a;
pp = &p;
...
\end{codesnippet}

\codecaption {
this example shows how recursive data structures can affect the
aliasing problem.  After the assignment at \refline{7}, the expressions
\QuotedCode{head.next->key}, \QuotedCode{head.next->next->key} ---and more
generally each expression of the form `\(
\Code{head.}(\Code{next->})^n\Code{key} \)' with \( n \in \naturals \)---
are all aliases of \QuotedCode{head.key}.  Even a simple
example can produce an infinite set of alias pairs.}
{recursive data structures}
\begin{codesnippet}
struct List {
  struct List *next;
  int key;
};
...
struct List head;
head.next = &head;
\end{codesnippet}

\subsubsection{Aliasing Subproblems}
Due to the many aspects that must be taken into account in order to provide
a complete coverage of the \emph{aliasing problem}, different area of
research have been developed; as a result, in the literature a
wide range of analyses is available, which encompasses all the alias subproblems
--- while a \emph{pointer analysis} attempts to
determine the possible run-time values of pointer variables,
a \emph{shape analysis} focuses on the precise approximation
of the aliasing relations produced by
recursive data structures; whereas a
\emph{numerical analysis} is required to track the value of array's indices.

\subsection{A Static Analysis}
The goal of this work is to present an automated method able to prove
certain alias properties of programs given in input.  In the following we
use the term alias \emph{analysis} to refer to the general and theoretical
ideas to approach the alias problem; whereas we use the term alias
\emph{analyzer} to stress the focus on the
implementation of an automated analysis. We are
interested in defining a \emph{static analysis}. Commonly, in the
context of \emph{software analysis}, the adjective \emph{static} referred
to the term \emph{analysis} designates a class of methods that
avoid the actual execution of the examined program.  In other words,
a \emph{static analysis} can be described as the process of extracting
semantic information about a program at compile time.  Static analysis
techniques are
necessary to any software tool that
requires compile-time information about the semantics of programs.
Consider indeed the following points.
\begin{itemize}
\item
  The termination problem is undecidable; as a consequence any method that
  requires the execution of the program is not guaranteed to terminate.
\item
  If the execution of the program is performed then the computational
  complexity of the analysis is bounded from below by the computational
  complexity of the analyzed program.
\item
  \emph{Testing} a program on some executions can prove the \emph{presence}
  of errors;  however, unless all of the possible executions are tried, it
  cannot prove the \emph{absence} of errors.  More generally, since a program
  can have an unbounded number of distinct executions, \emph{testing} can
  only prove that a \emph{property} holds on some executions, but it cannot
  prove that it holds always.
\end{itemize}
Hence, the existence of analysis methods that avoid the actual execution of
the program is motivated by the presence of constraints on the costs of the
analysis, the need of predictability of these or the need to verify a
property against all of the possible executions.
Usually, the results of an alias analysis are only an intermediate step of
the computation of a complete static analysis tool; this means that
an alias analysis
is commonly intended to answer to questions formulated by other automatic
analyses.
For instance, \emph{compilers} are the most common tools
that exploit the alias information --- almost all of the modern compilers
include some kind of alias analysis. From the practical perspective,
the kind of queries that are posed to the alias analyzer is greatly
influenced by the final application; whereas from the theoretical point of view it
is useful to assume that the questions posed to the alias analysis are always of
the form: does the property \(P\) hold on all/some executions of the
program?

\subsubsection{One Program, Many Executions}
\label{section:One program many executions}
Generally, the flow of the execution depends not only on the program's
source code but also on external sources of information, e.g., the user's
input or a random number generator;
when many
executions paths are possible, a property may hold on some but not on all
the possible executions
(\refcodesnippet{external_information_from_rand}).  In the following we
refer to a function declared as \QuotedCode{int rand()} as a source of
non-determinism;  we assume that this function always halts, that it can
return zero and not-zero values and that it has no side-effects on the
caller.

\codecaption
{at \refline{2} the address of \Variable{a} is assigned to
the variable \Variable{p};
then at \refline
3, during all executions, \QuotedCode{*p} is an alias of
\Variable{a}.  At \refline{4} the address of \Variable{b} is assigned to
\Variable{p}, but this statement is executed only when at \refline{3} the
call to \Code{rand()} returns a non-zero value.  Therefore, it is
possible to prove that there exists at least one execution path that
reaches \refline{5} in a state where \Variable{*p} is an alias of
\Variable{a} and also there exists at least one execution that reaches
\refline{5} in a state where the same property is false.}
{external_information_from_rand}
\begin{codesnippet}
int a, b, *p;
p = &a;
if (rand())
  p = &b;
...
\end{codesnippet}

\subsubsection{The Aliasing Problem Is Undecidable}
\label{section:aliasing_is_undecidable}
The problem of determining the alias properties
of a program is undecidable; it is indeed possible to reduce  a
problem that is well known to be undecidable, the halting problem,
to the aliasing problem.  In the sequel, we refer to a function
declared as \QuotedCode{int turing(int n)};  we assume that
(1)
  this function is defined somewhere in the source code and it emulates
  the execution on the input \( n \) of some
  Turing machine;
(2)
  the result of the execution of the emulated Turing machine is returned
  to the caller as the return value of the function;
(3)
  calling this function has no side effects on the caller environment.
\refcodesnippet{undecidability} highlights how the
aliasing problem is influenced by the halting problem.  For
this reason the aliasing problem is formulated assuming the
reachability as \emph{hypothesis}.  This assumption is not always valid but it is
\emph{safe}, or \emph{conservative}.  In \refcodesnippet{undecidability} it
is not possible to tell if \refline{5} will ever be reached; however,
\emph{in that case} what would happen?\footnote{The idea and the
motivations behind this approach are similar to those that drive the
development of Hoare's logic for \emph{partial correctness} specification, opposed
to the \emph{total correctness} specification, both introduced in \cite{Hoare03}.
The concept of \emph{Hoare's triple for partial correctness} is introduced
--- it is a triple \( \{ P \}\, \Code{C} \, \{ Q \}\) where \Code{C} is a command of
a given programming language and \(P\) and \(Q\) are two propositions
expressed in some fixed first order logic language. Informally, in Hoare's
logic the triple \( \{ P \}\, \Code{C} \, \{ Q \}\) is said to be true if
\emph{whenever}
\Code{C} is executed in a state satisfying \(P\) and the execution of
\Code{C} \emph{terminates} then the resulting output state satisfies
\(Q\).}
More generally the
question is --- if the execution reaches the program point \(p\) does the
property \(P\) hold at \(p\)?  The results of the analysis are
then expressed as an implication of the kind --- if \( p \) is
reached then \( P \) holds.  However, even in this weaker
form, the aliasing problem is still undecidable.
Consider for instance \refcodesnippet{undecidability_2},
where \refline{7} is reached if and only if
the call \QuotedCode{turing(K)} at \refline{3} halts;  in this case the
value of \Variable{p} is determined by the return value of
\QuotedCode{turing(K)}.  As a consequence of Rice's theorem
\cite{HopcroftMRU00}, also \emph{assuming} that \QuotedCode{turing(K)}
halts,
there exist no algorithms able to tell for every
\QuotedCode{K} if the execution reaches \refline{7} in a state
where \Variable{p} points to \Variable{a}.

\codecaption{%
at \refline{4} the call to the function
\QuotedCode{turing} starts the computation of the Turing machine.
Suppose that the call halts; in this case the
execution reaches \refline{5} causing an error due to a dereferenced null
pointer.  However, the problem of telling whether the execution of a Turing
machine will ever halt is undecidable --- there exists
no algorithm able to tell for each possible value of \Variable{K} if
\refline{5} will ever be reached by the execution; thus if there exists an
execution path where a null pointer is dereferenced.
}{undecidability}
\begin{codesnippet}
int K = ...;
int *p;
p = 0;
turing(K);
*p = 1;
\end{codesnippet}

\codecaption
{an example of the possible interactions of the aliasing problem and other
undecidable problems. There exist no algorithm able to tell for every
\QuotedCode{K} if there exist an execution that reaches \refline{7} in a
state such  that \QuotedCode{p} points to \QuotedCode{a}.}
{undecidability_2}
\begin{codesnippet}
int K = ...;
int *p, a, b;
if (turing(K))
  p = &a;
else
  p = &b;
...
\end{codesnippet}

\subsubsection{Summing Up}
\label{section:summing_up}
This section summarizes the various possibilities just presented.
Let  \( P \) be an alias property and \(\neg P\) its negation.
There exist four possible cases.
\begin{enumerate}
\item
  The property \(P\) holds on all of the possible executions or
  equivalently, \(\neg P\) never holds
  (\refcodesnippet{external_information_from_rand}).
\item
  The property \(P\) holds on some but not on all of the possible
  executions;  that is, there exists at least one execution in which \(P\)
  holds and also there exists at least one execution in which \(\neg P\)
  holds (\refcodesnippet{external_information_from_rand}).
\item
  The property \(P\) holds on some executions but it is not known if it holds
  always;  that is there exists at least one execution in which \(P\) holds
  but it is unknown whether there exists an execution in which \(\neg P\)
  holds (\refcodesnippet{nondeterminism_plus_undecidability}).
\item
  It is not known if there exists an execution in which \(P\) holds and also
  it is unknown whether there exists an execution in which \(\neg P\) holds
  (\refcodesnippet{undecidability_2}).
\end{enumerate}
For instance, suppose that \( P \) expresses the absence of some kind of
error. The first of the listed cases is the optimal case: it has been
proved that no errors are possible.  The second case is as much positive:
it has been proved that there exists at least one erroneous execution, that is
the program contains a bug.  In the third and the fourth case it is
\emph{unknown}, i.e., the absence of errors cannot be proved.  However,
assuming the reachability as hypothesis, alias analyses cannot prove the
result described in the second case. In other words, every static analysis
that assumes the reachability as hypothesis can only prove that \( P \)
holds always.
In this sense, testing procedures
are complementary to static analyses techniques.

\codecaption
{\refline{5} is reached only when the return value of the
call to \QuotedCode{rand()} evaluates to true. Thus, \refline{6} is reached
only when the execution reaches \refline{5} and the call
\QuotedCode{turing(K)} halts and the return value evaluates true.
Certainly there exists executions that reach \refline{7} in a state where
\Variable{p} points to \Variable{a}. However, also assuming that
\QuotedCode{turing(K)} halts, there exists no algorithm able to tell for
every \Code{K} if there exist an execution path that reaches \refline{7}
with \Variable{p} equal to null.}
{nondeterminism_plus_undecidability}
\begin{codesnippet}
int K = ...;
int *p, a;
p = &a;
if (rand())
  if (turing(K))
    p = 0;
*p = 1;
\end{codesnippet}

\subsection{Applications}
\codecaption
{analyzing this fragment of code without any aliasing information
would require a worst-case assumption about the locations pointed by
\QuotedCode{p}, that is, all the possible targets of an \QuotedCode{int*}
can be modified by the assignment at \refline{3}.}
{applications}
\begin{codesnippet}
void f(int *p) {
  ...
  *p = 0;
  ...
}
\end{codesnippet}
The alias information is required by many static analyses;
this is due to the following fact:
analyzing an indirect assignment, and generally an indirect memory
reference,
without knowing alias information requires to
assume that the assignment may modify almost anything
and, under these hypotheses, it is unlikely that the client analysis will
be able to deduce any useful result
(\refcodesnippet{applications}). For what concerns the final application,
there are two main areas where the
aliasing information is commonly used.
\begin{itemize}
\item Optimization and parallelization; used in compilers and interpreters.
\item Programs semantics understanding and verification;
      used in debugging/verifier tools.
\end{itemize}
These two uses have vastly different requirements on alias analyses. For
compiler oriented applications
there exist some upper bound on how much precision is useful.
There are various studies \cite{HindP00,HindP01} that state that
this upper bound is reached by
the current state of the art.  For the use in
program understanding/verification the picture is different; in this case
there is instead a lower bound on precision, below which, alias
information is pretty useless. It is commonly believed that the spectrum of
techniques currently available does not fully covers the requirements of
this kind of use: more research work is necessary.

\subsubsection{Client Analyses}
This section presents a brief list of the most common static analyses that
require the aliasing information.
\begin{description}
\item[Mod/Ref analysis]
This analysis determines what variables may be
modified/referenced\footnote{Here the term `referenced' means that the
value of the object is read.} at each program point. This information is
subsequently used by other analyses, such as \emph{reaching definitions}
and \emph{live variable analysis}. Each dereference in the program
generates a query of the alias information to determine the referenced
objects that are thus classified as \emph{modified} or \emph{referenced}
depending on the context in which the dereference operator occurs. For
example, in assignment statements, the objects referred by the last
dereference of the lhs are marked as \emph{modified}, all other objects
referred in the evaluation of the rhs and the lhs are instead marked as
\emph{read}.
\item[Live variable analysis]
It is common to many imperative languages that the life of a local variable
starts at the point of definition and ends at the end of the scope that
contains the definition.  At the extent of minimizing the memory usage of
the compiled program, while keeping unchanged its semantics,
it is possible to defer the creation to the point where the variable is
first assigned and anticipate its destruction to the last point where its value
is used.  The \emph{live variable analysis} tries to compute this information
that is useful to compilers for \emph{register allocation}, detecting
the use \emph{uninitialized variables} and finding \emph{dead assignments}.
\item[Reaching definitions analysis]
This analysis determines what variables may reach (in an execution sense) a
program point. This informations is useful in computing \emph{data
dependence} among
statements, which is an important step for the process of \emph{code-motion}
and parallelization.
\item[Interprocedural constant propagation]
This analysis tracks the value of constants all over the program and uses
this information to statically evaluate conditionals with the goal of
detecting if a branch is unreachable; thus allowing the detection of
unreachable code.
\end{description}

\subsection{Background}
Probably due to the different areas of application, historically this
field of research has treated as separate two fundamental
\emph{aliasing}-related problems: the \emph{may alias} and the \emph{must
alias} problem.  If the general interest of aliasing-related
static analyses is the study of how different expressions lead to the same
memory location, these two specializations can be characterized as follows.
\begin{description}
  \item[May alias]
    It tries to find the aliases that occur during \emph{some}
    execution of the program.
  \item[Must alias]
    Find the aliases that occur on \emph{all} the executions of the program.
\end{description}
Results exist that confirm that the former problem is not
\emph{recursive}\footnote{A problem \( P \) is said to be \emph{recursive}, or
\emph{decidable}, if there exists an algorithm that terminates after a finite
amount of time and correctly decides whether or not a given input belongs
to the set of the solutions of \( P \).}
while
the latter is not \emph{recursively enumerable}\footnote{A
\emph{recursively enumerable} problem \( P \) is a problem for which there exist
an algorithm \( A \) that halts on a given input \( n \)
if and only if \( n \) is a solution of \( P \).}
\cite{Landi92}.
In recent developments the same concepts are also expressed in terms of
\emph{possible} and \emph{definite} alias properties.  The term
\emph{definite alias property} is used to designate an alias property that
holds on every possible execution; whereas a \emph{possible alias property}
\( P \) is such that both \( P \) and \( \neg P \) cannot be proved to be
\emph{definite}.
Unfortunately,
the mismatch between the naming and the notation used in the published
works is not limited to the case just described.
For instance, in the literature the names
\emph{pointer} analysis, \emph{alias} analysis and \emph{points-to}
analysis are often uses interchangeably.  As suggested by \cite{Hind01},
we prefer to consider the \emph{points-to analyses} as a proper
subset of the \emph{alias analyses}.  An \emph{alias analysis} attempts to
determine when two expressions refer to the same memory location; whereas a
\emph{points-to analysis} \cite{Andersen94th,EmamiGH94,HindBCC99}
is focused in determining what memory locations a pointer can point to.
Points-to methods are also characterized by
the same \emph{representation} of the aliasing information.
As described in
\cite{Hind01}, the representation of the alias information is only one of
the several parameters that can be used to categorize
alias analyses.
\begin{description}
\item[Representation]
For the representation of alias information various options are possible.
\begin{description}
\item[Complete alias pairs] With this representation all the
alias pairs produced by the analysis are stored explicitly.
\item[Compact alias pairs] Only a subset of alias pairs is kept
explicitly. The complete relation can be derived applying the
\emph{dereference} operator, the transitivity and symmetry properties to
the pairs explicitly stored.
\item[Points-to pairs] This representation tracks only the
relations between the pointers and the pointed objects. The complete alias relation
can be derived from the points-to information in a way similar to what
done for the compact alias pair representation. This process is informally
described in \cite{Emami93th}.
\end{description}

\codecaption
{a program that exposes a simple alias relation.}
{simple alias relation}
\begin{codesnippet}
int i, *p, **q, *r;
p = &i;
q = &p;
r = p;
\end{codesnippet}

For instance, the alias relation generated by the sequence of assignments in
\refcodesnippet{simple alias relation} can be represented using the
points-to form as
\[
  \bigl\{
    \langle \Code{p}, \Code{i} \rangle,
    \langle \Code{q}, \Code{p} \rangle,
    \langle \Code{r}, \Code{i} \rangle
  \bigr\}.
\]
This corresponds to the complete alias pair set\footnote{In this case we
have omitted to explicitly write the alias pairs that can be obtained by
symmetrically closing this relation.}
\[
  \bigl\{
    \langle \Code{*p}, \Code{i} \rangle,
    \langle \Code{*q}, \Code{p} \rangle,
    \langle \Code{**q}, \Code{i} \rangle,
    \langle \Code{*r}, \Code{i} \rangle,
    \langle \Code{r}, \Code{p} \rangle,
    \langle \Code{*r}, \Code{*p} \rangle,
    \langle \Code{r}, \Code{*q} \rangle,
    \langle \Code{*r}, \Code{**q} \rangle
  \bigr\}.
\]
Note that these representations ---\emph{complete}, \emph{compact} and
\emph{points-to}--- are listed in order of decreasing expressive power ---
the rules of deduction used to infer the complete alias relation from the
compact and the points-to format impose a precise structure on the relation.
In the next we presents some examples to show how the points-to representation can
be less precise than the alias representation (\refsection{completeness}).
On the other hand
these deduction rules allow to reduce the set of pairs that have to be
explicitly represented thus decreasing the cost of the analysis.  Note also
that, due to recursive data structures, the complete
alias relation may contain an infinite number of
pairs.  If one of the possibilities to overcome this problem is to adopt a
compact or a points-to representation, other solutions, specialized in the
handling of recursive data structures, exist.
These methods use quite
different formalism from the ones presented here and they
have generated a quite independent field of research that is named
\emph{shape-analysis}.
An example of these alternative representations is briefly described in
\refsection{deutsch a notable example}.
\item[Flow-sensitivity]
The question is whether the \emph{control-flow}
information of the program is used by the analysis. By not considering
control-flow information ---therefore computing only a conservative
summary of it--- flow-\emph{insensitive} analyses compute one solution for
either the whole program or for each function
\cite{Andersen94th,Steensgaard96,HindBCC99}, whereas a flow-\emph{sensitive} analysis computes
a solution for each program point \cite{EmamiGH94,HindBCC99}. Therefore,
flow-insensitive methods are generally more efficient but less precise
than flow-sensitive ones.
\item[Context-sensitivity]
The point is if there is a distinction between the different
callers of a function, that is if the caller-context information is used when
analyzing a function.  If this is not the case, the information
can flow from one call site
(say caller A) through the called function (the callee) and then back to a
different call site (say caller B) thus generating a spurious data flow in the
computed solution on the code of the caller B.  Whenever a static analysis
combines information that reaches a particular program point via different
paths some accuracy may be lost. An analysis is \emph{context-sensitive} to the
extent that it separates information originating from different paths of
execution. Because programs generally have an unbounded number of potential
paths, a static analysis must combine information from different
paths --- in this sense, the \emph{context sensitivity}
is not a dichotomy but rather a matter of degree.
\item[Heap modeling]
The analysis of heap-allocated objects requires different
strategies from that of stack-allocated and global memory objects.  First
because heap objects have a different life-cycle with respect to
automatic and globals variables;
second, the term \emph{heap modelling}, is commonly but improperly used to
refer to the modelling of recursive data structures as these are usually
allocated on the heap.
Various trade-offs between the precision and the efficiency exist also for
this problem.
\begin{itemize}
\item
  The simpler solution consists in creating a single
  abstract memory location to model the whole heap \cite{EmamiGH94}.
\item
  Another solution distinguishes between heap allocated objects on the basis
  of the program point in which they are created, that is objects are
  \emph{named} by the creating statement (context-insensitive naming.)
\item
  A more precise solution names the objects not only by the program
  point of the creating statement but with the whole call path
  (context-sensitive naming.) For example, this means that if
  the program contains a user defined
  function for memory allocations (e.g., a wrapper of the
  \QuotedCode{malloc} function) then the analysis is able to discern
  objects created by different calls of the allocation routine.
\item
  Shape analysis methods adopts a quite different approach to the
  problem of naming locations, which is based on the expression used to
  refer to the memory location.
\end{itemize}
\item[Whole program]
Does the analysis method require the whole program or can a sound solution
be obtained by analyzing only its components?  In the current
panorama of software development, \emph{component programming} and the use
of libraries are becoming more and more popular. This trend requires
the capability to analyzing fragments of code as the whole program may not be
available \cite{LLVM}.
\item[Language type model]
In strongly typed languages, the type information ---that can be
easily extracted from the source code using common compiler techniques---
can be used by the alias analysis to deduce affordable informations
about the layout of pointers. This information, joined with other
assumptions on the memory model that usually accompany this kind of
languages, can greatly simplify the formulation of the alias analysis.
However, as noted in \cite{WilsonL95}, a pointer analysis algorithm cannot
safely rely on high-level type information for C programs. Because of
arbitrary type casts and union types, the defined types can always be
overridden. This means that type information cannot be used to determine
which memory locations may contain pointers.  To be safe, an analysis must
assume that any memory location could potentially contain a pointer to any
other location. Similarly, any assignment could modify pointers, even if
it is defined to operate on non-pointer types.
\item[Aggregate modeling]
This point regards how aggregate types
are treated: the main question is whether the subelements
are distinguished or collapsed into one object.
The choice of the analyzed language is of main relevance:
this task results particularly complex to address in
weakly-typed languages such as C/C++; in these languages the same memory
area can be read using different types. An analysis that aims to precisely
track pointers to fields must then consider the possible overlapping
between the memory layouts of the different types.  In strongly typed
languages like Java this difficulty does not exist, as these languages do
not allow for reading the memory with a type different from that used for
the allocation.
\end{description}

\subsection{The State of the Art}
Static analysis originally concentrated on Fortran and it was predominately
confined to a single procedure (\emph{intra-procedural} analysis).  Since
the emergence of the C language, static analysis of programs with dynamic
storage and recursive data structures has become a field of active
research producing methods of ever increasing sophistication.
In \cite{Hind01} it is noted that, during the past two decades, over
seventy-five papers and nine Ph.D.\ theses have been published on alias
analysis, leading the author to the question --- given the tomes of work on
this topic, haven't we solved this problem yet?  The answer is that though
many interesting results have been obtained, still many ``open
questions'' remain.
As shown in the introduction, also limited to the analysis of pointers, the
aliasing problem is still undecidable \cite{Landi92};
therefore, the main question that arise approaching it is about
the desired trade-off between the efficiency of the algorithm and the precision of the
approximated solution computed.
A wide range of \emph{worst-case} time complexities is available:
from almost linear \cite{Steensgaard96} to exponential \cite{Deutsch94}.
The current research effort is proceeding in at least two distinct
directions: improving the efficiency of the analyses while keeping the
actual precision and increasing the precision of the approximation while
keeping a reasonable computational costs.

\subsubsection{Improving the Efficiency}
Again in \cite{Hind01}, the problem of \emph{scalability} is
listed among the ``open questions''. About this topic two distinct efforts are currently
active and both proceed toward the goal analysing programs of
ever increasing size.  Today, flow-insensitive analyses
\cite{Steensgaard96,LLVM} can quickly analyze million-line programs.
It is commonly believed that the precision
provided by these fast methods is sufficient to satisfy
ordinary compiler-oriented client analyses \cite{Hind01}; but definitely they do not suffice
for verifier-oriented applications \cite{OrlovichR06,WangMD08}.  On the other side
various works \cite{HindBCC99} have increased the efficiency of the more
precise but slower flow-sensitive methods with respect to the initially
proposed methods \cite{EmamiGH94}.
It must be noted
that some studies \cite{EmamiGH94,HindP00,HindP01} show that client analyses
improved in efficiency as the pointer information was made more precise
because the input size to the client analysis becomes smaller; on average,
this reduction outweighed the initial cost of the pointer analysis.
However, these
studies focused on typical compiler oriented analyses --- no data
is available for the field of program understanding/verification.

\subsubsection{Improving Precision}
Another goal of the current research effort is to improve the precision
without sacrificing the scalability.  As for the scalability issue,
nowadays there are two main
directions in which researchers are investigating to improve the current
state of the art.
The first area of investigation tries to reconsider the notion of safety by
loosening the soundness constraints on the analysis.  The other
direction of investigation tries to recognize the areas of the source code
that needs to be analyzed with greater accuracy;  the idea is to
perform a quick alias analysis on the whole program and then refine the
first results only in those regions of the code where more precision is needed. In
other fields of the static analysis research this idea has yield to the
formalization of the concept of \emph{demand-driven} analysis
\cite{OrlovichR06,WangMD08}. Demand-driven methods can avoid the costly
computation of exhaustive solutions: given an initial query, the
analysis contains the logic to detect what other information are needed to
answer it and then it proceeds by recursively formulating a new
set of queries.  It is still an open question whether
the precise alias analyses currently available
--that is flow- and context-sensitive analyses and shape analyses--
can be reformulated in a demand-driven fashion \cite{Hind01}.

\subsubsection{Different Notions of Safety}
A reading of the literature available for
the field reveals that there exist two slightly different notions of
\emph{safety}, which are determined by the different areas of application.
Compiler targeted analyses
are required to produce a \emph{safe} approximation of the
alias information for \emph{every} standard-compliant program,
allowing thereby the analyzer to \emph{assume} that the analyzed program is
standard-compliant.\footnote{For some notion of standard-compliant;
there exists different possible language standards, hence
different notions of standard-compliance.} From
\cite{WilsonL95}
\begin{quote}
The possibility of non-pointer values [stored inside pointer variables]
is not always important. For example,
when a location is dereferenced, we can assume that it always contains a
pointer value, since otherwise the program would be erroneous.
\end{quote}
On the other hand, for software verification tools, the conformance
of the analyzed program to the standard is not an hypothesis but one of the theses
that need to be proved.  For example, a desirable feature for a verifier
tool would be to signal if a dereferenced pointer may hold an
undefined or a null value.
For analyses that cannot simply ignore the possibility of errors,
the approach called \emph{$\theta$-soundness} is usually applied
\cite{ConwayDNB08}:
when the analysis detects the possibility of an error, then the program
point is marked with a warning and the analysis proceeds assuming that the
condition that led to the error is not verified.
For example, if we have that to the pointer \QuotedCode{p}
corresponds the points-to set \( \{\, \Code{a}, \Code{null} \,\} \) ---i.e,
\QuotedCode{p} may point to the variable \QuotedCode{a} or be null--- then
the analysis of the statement \QuotedCode{*p} would produce a warning for a
possible dereferenced null pointer and
the execution will continue assuming that \QuotedCode{p} points only to \QuotedCode{a}.
Verifier targeted analyses are not allowed to assume
the absence of errors; in this sense,
the notion of safety required by compiler targeted analyses is weaker.
However, practical considerations softens the requirements on verifier's
analyses.  If compilers are required to expose a well-defined behaviour on
all conforming programs, verification tools often assume stricter rules
than those dictated by the standard of the programming language with the
result of restricting the class of analyzable programs to a set of
well-behaved ones.  For instance, assuming the absence of some kind of
\emph{casts} \cite{Acton06}, it is possible to simplify the analysis and
also improve its precision.  For those programs that do not belong
to this restricted set, the analysis produce some \emph{false
positives}\footnote{A \emph{false positive} is an error reported by the
analyzer which however cannot occur in any of the possible execution
paths.} and the process of $\theta$-soundness will erroneously remove from
the abstraction some of the possible executions yielding to a non-safe result.
As noted in \cite{Hind01} this can be acceptable in many areas:
\begin{quote}
  I was told the users actually liked the false-positives in my
  analysis because they claimed when my analysis got confused it was a good
  indication that the code was poorly written and likely to have other
  problems. This came as a complete surprise. While additional study is
  needed to claim these observations to be valid in a broader sense, they
  lead me to conclude that the notion of safety should be reconsidered for
  many applications of static analysis.
\end{quote}

\subsubsection{Measuring the Alias Analyses}
It is a quite accepted fact that in the alias analysis field, the
independent verification of the published results is a considerably
difficult task.  The first consequence of this is the absence of a clear
and complete comparison between the existing methods.  The difficulty of
reproducing the publicly available results can be explained by the
intrinsic difficulty of defining a valuable metric for the problem as a
great number of parameters must be taken into account: as the chosen
intermediate representation, the benchmark suite used for the testing phase
and, more generally, all the details of the infrastructure where the
analysis is put to work.
For instance, some analyses \cite{EmamiGH94,HindP00} work on an
intermediate representation of the code that results from a simplification
phase, which reduces all expressions to a normal form with
the goal of limiting the complexity of the implementation as less cases
need to be considered; however, it also introduces temporary variables and
intermediate assignments to emulate step by step the evaluation of the
original expressions.  Since many of the used metrics depend on the number
of variables, this transformation makes harder, if not impossible at all,
any comparison between these methods with other methods that do not perform
the simplification.
Moreover, alias information is not useful on its own, but it is needed by
other client analyses.  Thus, the definition of what is a good trade-off
between the cost of the analysis and the precision of the computed solution
inevitably depends on the client applications; it is indeed a common opinion
among the researchers that each area of application requires an ad-hoc
method or an adaptation of one described in the literature.  The result is
that a single metric that gives an absolute measure of the value of a
method does not exist.  However, to help implementors of aliasing analyses to
determine which pointer analysis is appropriate for their application and to
help researchers to identify which algorithms should be used as basis for
future advances, some partial metrics have been proposed \cite{HindP01}; the
idea is that since all these metrics have their strengths and
weaknesses, a combination should be used.  A first popular metric records
for each pointer variable the number  of pointed objects; the idea is that
a lower number of referenced objects would mean a more precise alias
information. Although this metric is quite simple to measure, it
presents some flaws.
\begin{itemize}
\item
Due to local variables
in recursive functions and the possibility of dynamically allocating memory
(heap-allocated objects), an alias analysis should be able to model
an unbounded number of objects.  To have a finite representation of the set
of the possible memory objects, each method defines a finitely representable
approximation.  For example in \cite{EmamiGH94} the whole heap is modeled
as a single object; in this case the metric will count only one for all
the referenced heap-allocated objects with the effect of incorrectly
suggesting a precise analysis.
\item
As anticipated, alias information is used by other client analyses, then
its real effectiveness can only be measured on the results of whole
process. But there are no straightforward relations between the results of
this metric and the precision of the client analyses; For example, the
removal of a single alias pair would allow for the client analysis to prove
the absence of a run-time error otherwise not provable.
\end{itemize}
The above metric is usually named \emph{direct} as it refers to a quantity
that is a direct result of the analysis.  To address the flaws just
highlighted, some \emph{indirect} metrics have been developed.
\begin{enumerate}
\item
A first kind of indirect metric measures
the relative improvement to the precision of the aliasing information with
respect to the worst-case assumption. This kind
of metric is reported to be particularly useful on
strongly-typed languages where the worst-case assumptions are not as bad as
in other weakly-typed languages like C \cite{Hind01}.
\item
A second kind of indirect metric requires to implement a client of the
alias information and then it measures the variation of the precision of
the results of the client analysis at the varying of the precision of the
supplied aliasing information.  The main weakness of this metric is that
its results cannot be generalized to other client analyses.
\end{enumerate}
Comparisons are difficult also for what concerns performances.
The careful
engineering of a points to analysis, particularly for flow-sensitive
analyses,\footnote{This is probably due to the greater complexity of
flow-sensitive analyses with respect to a flow-insensitive one. In a more
complex method there are more opportunities to improve.} can dramatically
improve its performance \cite{Hind01}.
The
worst-case complexities often do not reflect the mean cost of the
algorithm, which is greatly influenced by heuristics developed over the
default algorithm, which however require a great effort of fine tuning for the
specific target application.
However, as criticized in \cite{Hind01}, even today most published papers
about new analysis methods seldom present a complete quantitative
evaluation using these guidelines; also, for those works that provide
experimental data, too often the independent verification is missing and
the acceptance of the proposed results becomes
a matter of faith.

\subsubsection{Notes on the Analysis of the Java Language}
The Java language has emerged as a popular alternative to other mainstream
languages languages in many areas.  Java presents a clean and simple memory
model where conceptually all objects are allocated in a garbage-collected heap.
While useful to the programmer, this model comes with a cost. In many cases
it would be more efficient to allocate objects on the stack, eliminating
the dynamic memory management overhead for that object.  Aliasing analysis
allows to detect those cases in which it is possible to perform this
simplification.  Another characteristic of the Java language is the
availability of \emph{synchronized methods} that ensure that the body of
the function is executed
atomically by acquiring and releasing a lock in the receiver object. But
the lock overhead is wasted when only one thread can access the object; the
lock is required only when there is multiple threads
may attempt to access the same object simultaneously.  Also in this case, alias
analysis allows to detect which threads can access an object and thus
possibly allowing the removal of the code for the locking.  Studies have
shown \cite{WhaleyR99} that it is possible to eliminate a significant
number of heap allocations (in the tests between 22\% and 95\%) and
synchronization operations (in the tests between 24\% and 64\%).
For what concerns the realization of alias analyses, the Java language
---while adding new features like virtual functions and exception
handling--- may still be much easier to analyze than the C
language \cite{WilsonL95} because of its strong type system:\footnote{The
same consideration holds for all other strongly typed languages.} without type casts and
pointer arithmetics, the type information given by the static type system
of the language can be used to deduce affordable alias information.
Another feature of Java simplifies the analysis
algorithm: it does not support pointers into the middle of an object --- an
object reference in Java can point only to the beginning of an object.
This means that two pointers may either point to exactly the same location
or not; they cannot point to different offsets within one allocated block 
of memory as it is
possible in the C language.

\subsection{Organization}
\input{./tex/organization.tex}

\subsection{Purpose of the Work}
The presented method is targeted for application in the context of
\emph{software verification}.  Compiler-targeted applications require
relatively imprecise alias information, thus they can rely on \emph{fast}
algorithms for its computation.  However, as empirical studies have
evidenced \cite{HindP01,Hind01}, for software verification there is a
lower bound of precision below which the points-to information is pretty
useless.  For these reasons, our aim is to develop a points-to analysis that,
though less efficient than other methods based on the same representation,
computes a \emph{more precise}
approximation of stack-allocated objects and that is also suitable for
integration with the precise \emph{inter-procedural} techniques already
present in the literature \cite{Emami93th,WilsonL95}.

\subsection{Contributions}
The present work describes a \emph{store-based},
\emph{flow-sensitive} and intra-procedural points-to analysis
working on a relatively high-level intermediate representation of the
source code, which also makes no assumptions about the
inter-procedural analysis model.  In particular, beyond the
\emph{assignment} operation ---which is the most essential operation of a
points-to analysis and thus it is omnipresent in all the papers on the
topic--- we describe a \emph{filter} operation that enables the analysis to
increase the precision of the computed solution by exploiting the
expressions used in branching statements.  Moreover, a
formal proof of the soundness of the presented operations is developed.

\section{Preliminaries}
\label{section:Preliminaries}
This section presents informally the approach used for the
definition of the points-to analysis.

\subsection{Notation}
Before proceeding, some clarifications about the used notation are
necessary.
Let \( A \) and \( B \) be two sets. We write `\( A \to B \)' to denote a
\emph{total} function from the set \( A \) to the set \( B \); we write `\(
A \rightarrowtail B \)' to denote a \emph{partial} function from \( A \) to
\( B \).  We use `\( S \defeq A \to B \)' to denote the \emph{set} of all
(total) functions from \( A \) to \( B \); whereas we write `\(
\functiondef{ f }{ A }{ B } \)' to mean that \( f \) \emph{is} a (total)
function from \( A \) to \( B \).

We denote as `\( \Bool \)' the set \( \{ 0, 1 \} \)
and as `\( \absBool \)' the set \( \partsof\bigl( \{ 0, 1\} \bigr) \);
for convenience of notation we use `\( \bot \)' to refer to the empty element and
`\( \top \)' for the \( \{ 0, 1 \} \) element. We refer to the \emph{complete
lattice
associated to the set \( \absBool \)} as the structure \( \bigl\langle  \absBool,
\subseteq, \union, \intersection, \{ 0, 1 \}, \emptyset \bigr\rangle \).

Let \( n, m \in \naturals \), where \( n < m \), we write `\( \{ n, \cdots
, m \} \)' to denote the set of the naturals from \( n \) to \( m \), i.e.,
\( \{ i \in \naturals \mid n \leq i \leq m \} \).

Let \( A, B, C \) be \emph{finite} sets. We write `\( \cardinality A \)' to mean the
\emph{cardinality} of the set \( A \).
Let \( \functiondef{ f }{ A }{ \partsof(B) } \) and \( \functiondef{ g }{ B
}{ C } \) and \( a \in A \) be such that \( \cardinality f(a) = 1 \);
for convenience of notation we write `\( g\bigl(f(a)\bigr) \)'
to mean \( g(b) \) where \( f(a) = \{ b \} \).

\subsection{The Execution Model and Its Operations}
\label{section:the execution model and its operations}
Though our work is ideally targeted for the C language, we need to
introduce some kind of formal \emph{execution model}.
The standard of the C language has indeed
many \emph{implementation defined} issues that every execution model is
required to specify in order to provide a working environment for the
execution of programs.  The literature provides several of such
formalizations \cite{BagnaraHZ08}; however, for this presentation many of
the details would be useless. With the aim of keeping a simple notation, we
introduce the following concepts.  We denote with `\( \expressions
\)' the set of the expressions of the language.  With \emph{execution
model} we mean a formally specified computing device able to execute
programs written in the analyzed language.
With \emph{memory description}, or simply \emph{memory}, we mean a
description of the state of the execution model at some step of the
computation.\footnote{In the formalization of Turing
machines \cite{HopcroftMRU00}, an \emph{instantaneous description} is a
complete description of the computing device at one of the steps of the
computation; here, with \emph{memory description} we mean an instantaneous
description of the chosen execution model.} Fixed the execution model, we
denote with `\( \memories \)' the set of the memory descriptions.
We make few assumptions about the
structure of the memory model;
we assume that a memory is composed by a set of \emph{memory
locations},\footnote{Now we use the term \emph{memory location} a synonym of
\emph{memory address}. Basically, with location we mean a tag that can be
used to identify the information stored in the memory description.}
denoted as `\( \Locations \)'. Given a memory description \( m \in
\memories \) and a location \( l \in \Locations \), we denote with `\( m[l]
\)' the information that \( m \) stores at the location \( l \).  We also
assume the existence of a partial \emph{evaluation function}
\[
  \parfunctiondef
    { \eval }
    { \memories \times \expressions }
    { \Locations }.
\]
In the real world, the execution of a program acts in different ways on the
memory structure of the computing machine.  With the aim of formalizing these
interactions, we introduce the concept of \emph{operation};
an \emph{operation} is defined as a partial function
\[
  \parfunctiondef
    { \operation }
    { \memories \times \externalDomain }
    { \memories }
\]
where `\( \externalDomain \)' is an unspecified set
that formalizes the use of external information. Note that we have specified \(
\operation \) as a partial function --- this is needed to model the fact
that the possible actions that can be performed on the memory structure are
not defined on all of the possible states. For example, to process
the return statement of a function, the stack of the memory must contain at
least one activation frame. By aiming to perform a \emph{static} analysis, we are
interested in determining all the possible memory descriptions that can
be generated at a specified program point. To express the transition from a
set of memory descriptions to another
as a consequence of an operation, we extend the definition of
the operation \( \operation \) to sets. Let
\[
  \functiondef
    { \operation }
    { \partsof(\memories) \times \externalDomain }
    { \partsof(\memories) }
\]
be defined as follows.  Let \( M \subseteq \memories \)
and \(e \in \externalDomain\), then
\[
  \operation(M, e)
  \defeq
    \bigl\{\, \operation(m, e)
    \bigm|
        m \in M \land
        \operation(m, e) \text{ is defined}
    \,\bigr\}.
\]
\begin{example}
Consider the modifications to the memory triggered by the declaration of a
local variable. To formalize this event we introduce an operation \(
\acsnewvar \), which takes the memory description of the execution prior to
the declaration, plus some information about the declaration.  In this
case, the set `\( \externalDomain \)' represents the type of the declared
variable and, if present, the expression used as initializer.  The returned
memory describes the properly updated execution state.  Now suppose that
the set \( M \subseteq \memories \) represents the possible memory
configurations at a given program point \( p \), which is immediately
followed by a local variable declaration. Let \( e \in \externalDomain \)
be the information associated to the declaration; then we express
the set of all possible memory configurations resulting from the
declaration as \( \acsnewvar(M, e) \).
\end{example}

\subsection{The Abstract Interpretation Approach}
\label{section:abstract interpretation approach}
As shown in the introduction, the aliasing problem
is undecidable. Following the approach proposed by the \emph{abstract
interpretation theory} \cite{CousotC77,CousotC79,CousotC92},
to overcome this limitation we proceed by developing
a computable approximation of the execution model and its operations.

\begin{definition}
\summary{Concrete domain of the aliasing problem}
We define the \emph{concrete domain of the aliasing problem} as the complete
lattice generated by the powerset of \( \memories \)
\[
  \bigl<
    \partsof(\memories), \subseteq, \union, \intersection, \emptyset,
    \memories
  \bigr>
\]
\end{definition}
Then we need to develop an abstract counterpart of the chosen execution
model --- an \emph{abstract domain} \( \absmemories \) that provides an approximation of the
concrete domain \( \partsof(\memories) \). We formalize \( \absmemories \) as a
complete lattice
\[
  \bigl<
    \absmemories, \sqsubseteq, \sqcup, \sqcap, \bot, \top
  \bigr>.
\]
To formally express the \emph{semantics} of the approximation we provide a
\emph{concretization function}
\[
  \functiondef
    { \concretization }
    { \absmemories }
    { \partsof(\memories) }.
\]
We say that a memory description \( m \in \memories \) is
approximated, or \emph{abstracted}, by an element of the abstract domain \(
m^\sharp \in \absmemories \) when \( m \in \concretization(m^\sharp) \).
The formalism also requires the definition of an abstract counterpart \(
\absoperation \) of the concrete operations \( \operation \)
\[
  \functiondef
    { \absoperation }
    { \absmemories \times \externalDomain }
    { \absmemories }.
\]
To prove the \emph{soundness} of the proposed abstract model by
it is necessary to show that for all \( m^\sharp \in \absmemories \) holds that
\[
  \operation \bigl( \concretization(m^\sharp), e \bigr) \subseteq
  \concretization \bigl( \absoperation(m^\sharp, e) \bigr);
\]
that is, the approximation provided by the abstract operation \(
\absoperation \) is \emph{safe} with respect to the concrete operation \(
\operation \).
Beyond the operations already defined on the concrete execution
model\footnote{ Such as the \( \acsnewvar \), the assignment and all other
operations required to define the behaviour of the concrete execution model.} \(
\memories \), the formalization of the abstraction requires the definition
of other operations
that can be described as
\[
  \functiondef
    { \absoperation }
    { (\absmemories)^n \times \externalDomain }
    { \absmemories };
\]
along with the corresponding concrete counterpart,
\[
  \functiondef
    { \operation }
    { \partsof(\memories)^n \times \externalDomain }
    { \partsof(\memories) }.
\]
The soundness of these operations is expressed in the same way, that is
for all \( m^\sharp_i \in \absmemories \)
\[
  \operation \bigl(
    \concretization(m^\sharp_1),
    \ldots,
    \concretization(m^\sharp_n),
    e \bigr) \subseteq
  \concretization \bigl( \absoperation(
    m^\sharp_1,
    \ldots,
    m^\sharp_n,
  e) \bigr).
\]
These additional operations include
for instance, the `meet' and `join' operations of the domain.
With a slight change of notation, this definition can be accommodated to
describe also the requirement of correctness on the partial order `\(
\sqsubseteq \)', i.e.,
\[
  \concretization(m^\sharp_0) \subseteq \concretization(m^\sharp_1)
  \iff
  m^\sharp_0 \sqsubseteq m^\sharp_1.
\]

\subsection{Queries}
\label{section:using queries}
This section introduces the concept of \emph{query} on a domain.  A
query defines an interface on the domain, it helps to isolating the relevant
information from other uninteresting details. When the analysis
process is composed by more abstract domains,
the use of queries is useful to
formalize the interactions between them.
More details on this approach can be found in \cite{CortesiLCVH94}.  In the
following we show how queries can also be used to formalize the semantics
of the abstraction, that is how the concretization function \(
\concretization \) can be expressed in terms of queries.
Fixed the number of arguments \( n \), we denote with `\( \queryDomain \)'
the space of the concrete query functions and with `\( \absQueryDomain \)'
the space of the abstract query functions,
\begin{gather*}
  \queryDomain    \defeq ( \expressions )^n \to \Bool; \\
  \absQueryDomain \defeq ( \expressions )^n \to \absBool.
\end{gather*}
The \emph{concrete query domain} is then defined as the complete lattice
generated by the powerset of `\( \queryDomain \)'
\[
  \bigl<
    \partsof(\queryDomain), \subseteq, \cap, \cup, \emptyset, \queryDomain
  \bigr>;
\]
whereas the \emph{abstract query domain} is defined as a complete lattice on the set \(
\absQueryDomain \),
\[
  \bigl<
    \absQueryDomain, \sqsubseteq, \sqcap, \sqcup, \bot, \top
  \bigr>,
\]
where `\(\mathord{\sqsubseteq}\)' is the point-wise extension of the
ordering of \(\absBool\); \(\bot\) and \(\top\) are the minimum and maximum
elements of \(\absQueryDomain\) with respect to this ordering,
respectively; `\(\mathord{\sqcap}\)' and `\(\mathord{\sqcup}\)' are the
obvious point-wise extensions of \(\absBool\)'s operations.
Note that `\( \queryDomain \)' can be seen as
subset of `\( \absQueryDomain \)';\footnote{Consider indeed the injection
\( \functiondef{ f }{ \queryDomain }{ \absQueryDomain } \) that maps every
\( \query \in \queryDomain \) to a \( \absquery \in \absQueryDomain \) such
that, for all \( e \in \expressions^n \),
\( \absquery(e) = \bigl\{ \query(e) \bigr\}.  \)}
from this fact, the concretization function
\[
  \functiondef
    { \concretization }
    { \absQueryDomain }
    { \partsof(\queryDomain) },
\]
is defined as, for all \( \query \in \queryDomain \) and
\( \absquery \in \absQueryDomain \),
\[
  { \mathord{\query} \in \concretization( \mathord{\absquery} ) }
  \quad \defiff \quad
  { \mathord{\query} \sqsubseteq \mathord{\absquery} }.
\]
In order to define the semantics of
\( \absmemories \) in terms of queries,
it is necessary to describe other two
steps of the concretization.
First we have to define how the query has to be performed on the concrete
domain, that is how to extract the relevant information from a concrete
memory. In symbol,
\[
  \functiondef
    { \concretization }
    { \queryDomain }
    { \partsof(\memories) }.
\]
Also, we have to define how the query has to be performed on the abstract
domain, i.e,
\[
  \functiondef
    { \concretization }
    { \absmemories }
    { \partsof(\absQueryDomain) }.
\]
The semantics of the abstraction \( \absmemories \) is then
defined as the composition of these three steps
(\reffigure{semantics with queries}.)

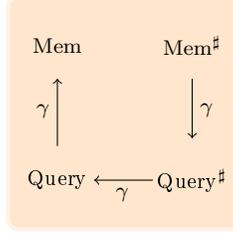
\begin{figure}
\begin{center}
\input{figures/07}
\caption
  {A representation of the three steps required to define the semantics of an
  abstract domain using queries.}
\label{figure:semantics with queries}
\end{center}
\end{figure}

\subsubsection{The Alias Query}
\label{section:concrete aliasing query}
The following definitions present the formal meaning of the statement --- \( e_0 \) and \( e_1
\) are aliases in \( m \in \memories \).  Basically, two expressions are
considered aliases in a memory description when they evaluate to the same
memory location.

\begin{definition}
\definitionsummary{Concrete alias query domain}
Let
\[
  \aliasQueryDomain \defeq
    ( \expressions \times \expressions ) \to { \Bool }.
\]
We define the \emph{concrete alias query domain} as the complete lattice
generated by the powerset of \( \aliasQueryDomain \)
\[
  \bigl<
    \partsof(\aliasQueryDomain), \subseteq, \union, \intersection, \emptyset,
    \aliasQueryDomain
  \bigr>.
\]
\end{definition}
\begin{definition}
\definitionsummary{Concrete alias query semantics}
Let
\[
  \functiondef
    { \concretization }
    { \aliasQueryDomain }
    { \partsof(\memories) }
\]
be defined as follows. Let \( \mathord{\aliasquery} \in \aliasQueryDomain \)
and \( m \in \memories \); then we define
\( m \in \concretization(\mathord{\aliasquery}) \) when, for all \( e, f
\in \expressions \) holds that
\[
  \aliasquery(e, f) =
  \begin{cases}
    1, & \text{if } \eval(m, e) = \eval(m, f); \\
    0, & \text{otherwise}.
  \end{cases}
\]
\end{definition}
Given a concrete memory description \( m \in \memories \), we denote as \(
\mathord{\aliasquery_m} \) the concrete alias relation that abstracts \( m \); also
we call \( \mathord{\aliasquery_m} \) the \emph{alias
information} of the memory \( m \).
As anticipated in \refsection{using queries}, the alias query \(
\mathord{\aliasquery} \)
acts as an interface onto \( m \in \memories \) selecting the interesting
details;
this idea is shown in \refcodesnippet{induced alias relation}.

\codecaption
{in this example the execution can reach \refline{4} in many different
states due to the different values that the variable \QuotedCode{a} can
assume. However, considering a \emph{type based} alias analysis ---that is assuming that
the analysis tracks only the value of pointer variables---
each of the possible \( m \in \memories \) carries
the same aliasing information.}
{induced alias relation}
\begin{codesnippet}
int a, *p;
p = &a;
a = rand();
...
\end{codesnippet}

\begin{definition}
\definitionsummary{Abstract alias query domain}
Let
\[
    \absAliasQueryDomain \defeq
    ( \expressions \times \expressions ) \to { \absBool }.
\]
We define the \emph{abstract alias query domain} as the complete lattice generated
by the powerset of \( \absAliasQueryDomain \)
\[
  \bigl<
    \absAliasQueryDomain, \sqsubseteq, \sqcup, \sqcap, \bot, \top
  \bigr>.
\]
\end{definition}
The semantics of the abstract alias query domain
\[
  \functiondef
    { \concretization }
    { \absAliasQueryDomain }
    { \partsof( \aliasQueryDomain ) }
\]
as already specified in \refsection{using queries}, is defined as
\[
  { \aliasquery \in \concretization(\absaliasquery) }
  \quad \defiff \quad
  { \aliasquery \sqsubseteq \absaliasquery }.
\]
The last step required in order to complete the definition of the semantics
of the abstraction, that is from \( \absmemories \) to \( \partsof(\absaliasquery)
\) (\reffigure{semantics with queries}), depends on the details of
the chosen approximation method \( \absmemories \). The next
section presents some of the available approaches.

\subsection{Representation of the Abstract Alias Domain}
By looking forward to the realization of an alias analyzer, another problem
arises. A realistic implementation cannot aim to \emph{directly} represent
abstract alias queries (\refdefinition{Abstract alias query domain}).  As
demonstrated in \refcodesnippet{recursive data structures}, there can be an
infinite number of aliasing pairs making impossible a direct
representation. In this sense, the domain \( \absmemories \) introduces
an additional layer of
abstraction providing a representation suitable for the implementation.

\subsubsection{Techniques For Approximating the Alias Information}
In the literature, different classes of methods exist.  One of these is the
class of \emph{access-path} based methods. A brief description of an access-path
based method is reported below.
Another class is identified by the name of \emph{store based methods};
more details on these are presented in \refsection{general store based
methods}.

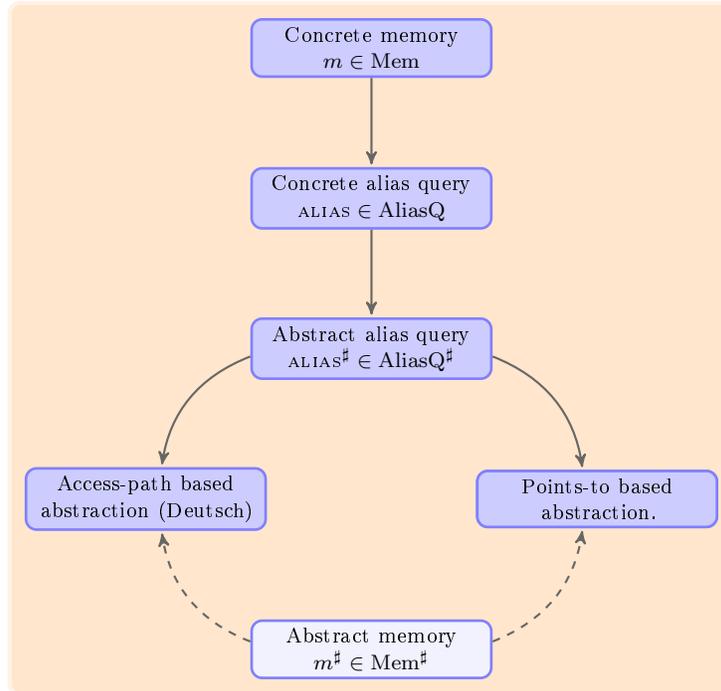
\begin{figure}
\begin{center}
\input{figures/03}
\caption
  {A representation of the abstraction relations under discussion.  Arrows should
  be read as `is abstracted by.'}
\label{figure:abstraction schema}
\end{center}
\end{figure}

\codecaption
{in this code two recursive structures, \Code{List} and
\Code{Tree}, are defined.}
{deutsch example}
\begin{codesnippet}
struct List {
  struct List *n;
  int key;
} *x;

struct Tree {
  struct Tree *l, *r;
  int key;
} *y;

...
\end{codesnippet}

\subsubsection{A Notable Example of Access-Path Based Approximation}
\label{section:deutsch a notable example}
In the literature, the term \emph{access-path} is used to design a
simplified form of language expressions.
A notable example of access-path based method for the approximation
of
the abstract alias query domain (\refdefinition{Abstract alias query
domain}) is presented in the Ph.D.\ dissertation of A.\ Deutsch
\cite{Deutsch94}.  In this proposal the elements of the abstract domain \(
\absmemories \) are formalized as pairs \( m^\sharp = \langle P, C \rangle
\) where \( P \) is a set of pairs of \emph{symbolic access paths} and \( C
\) is a set of constraints on \( P \).  A symbolic access path
is an approximation of a set of expressions;\footnote{The term \emph{symbolic
access path} comes from the original paper \cite{Deutsch94} and it
actually means an abstraction of the concept of \emph{expression}. With
our notation, the term \emph{abstract expression} would be probably used
instead.}
the concretization of a symbolic access path is defined using
a mechanism similar to regular expressions.
Consider for instance the code presented in \refcodesnippet{deutsch example},
and let \( \langle P, C \rangle \in \absmemories \)
be an abstract memory description of the program at
\refline{11}, such that \( C = \{ i = j \} \).
The set of constraints \( C \) has the
set of solutions
\[
   \bigl\{ \langle n, n \rangle \bigr\}_{ n \in \naturals }
\]
in the variables \( \langle i, j \rangle \).
Let \( p \in P \) be the pair of symbolic access paths
\[
   p \defeq
    \bigl<
      \Code{x}(\Code{->n})^i\Code{->key},
      \Code{y}(\Code{->l}, \Code{->r})^j\Code{->key}
    \bigr>.
\]
The semantics of \( p \) is a set of pairs of concrete expressions and it
can be computed by replacing the occurrences of the variables \( i \) and \( j
\) found in the symbolic access paths of \( p \), with the values given by the
solutions of \( C \).  For instance, by replacing
the occurrences of the index `\(j\)' with the integer \(2\)
in the symbolic access path
\[
  \Code{y}(\Code{->l}, \Code{->r})^j\Code{->key},
\]
we obtain the regular expression
\[
  \Code{y}(\Code{->l},\Code{->r})^2\Code{->key},
\]
that can be finally translated into the following set of expressions
\[
  \{
    \Code{y->l->l->key}, \Code{y->r->l->key},
    \Code{y->l->r->key}, \Code{y->r->r->key}
  \}.
\]
Depending on the considered solution of \( C \),
the pair \( p \) approximates different sets of alias pairs.
For instance, using the solution \( \langle 0, 0 \rangle \) we have
\[
  \bigl<
    \Code{x}(\Code{->n})^0\Code{->key},
    \Code{y}(\Code{->l},\Code{->r})^0\Code{->key}
   \bigr>
  = \bigl\{ \langle \Code{x->key}, \Code{y->key} \rangle \bigr\}.
\]
With the solution \( \langle 1, 1 \rangle \) we have
\begin{gather*}
  \bigl<
    \Code{x}(\Code{->n})^1\Code{->key},
    \Code{y}(\Code{->l},\Code{->r})^1\Code{->key}
  \bigr> \\
  \qquad =
  \bigl\{
    \langle \Code{x->n->key}, \Code{y->l->key} \rangle,
    \langle \Code{x->n->key}, \Code{y->r->key} \rangle
  \bigr\}.
\end{gather*}
Using the solution \( \langle 2, 2 \rangle \) we obtain
\begin{gather*}
  \bigl<
    \Code{x}(\Code{->n})^2\Code{->key},
    \Code{y}(\Code{->l},\Code{->r})^2\Code{->key}
  \bigr> \\
    \begin{aligned}
      \qquad =
      \bigl\{
        &
         \langle \Code{x->n->n->key}, \Code{y->l->l->key} \rangle,
         \langle \Code{x->n->n->key}, \Code{y->r->l->key} \rangle, \\
        &
         \langle \Code{x->n->n->key}, \Code{y->l->r->key} \rangle,
         \langle \Code{x->n->n->key}, \Code{y->r->r->key} \rangle
      \bigr\}.
    \end{aligned}
\end{gather*}
Generally, \( C \) is a set of constraints on a tuple of indices
\(I = \{i_1, \ldots, i_n\}\).  The indices of \(I\) also occur in the
symbolic access paths of \( P \).  To each solution \( \functiondef{
S }{ I }{ \naturals } \) of \( C \) corresponds a different alias query
expressed as a set of pairs of (concrete) expressions.  Given a solution \(
S \) to \( C \), the corresponding alias query, say \( P(S)
\), can be obtained from \( P \) by replacing every occurrence of the
index \( i_k \) in \( P \) with the solution \( S(i_k) \), for each index
\( i_k \) of \( I \).  As shown above, this replacement yields a set of
pairs of no-longer-symbolic access paths. Seen as regular expressions,
these no-longer-symbolic access paths are transformed in a set of pairs of
expressions.
The semantics of \( \absmemories \) can be finally expressed in
terms of queries as follows. Let \( \absaliasquery \in
\absAliasQueryDomain \) and \( \langle P, C \rangle \in \absmemories \). We
say that \( \absaliasquery \in \concretization\bigl( \langle P, S \rangle
\bigr) \) when
\[
  \exists S \text{ solution of } C \suchthat
  \forall e, f \in \expressions \itc
    \absaliasquery(e, f) \in \{ \top, 1 \} \implies
    \langle e, f \rangle \in P(S).
\]
Note that this has two main consequences.
\begin{itemize}
\item
  This formulation is
  unable to represent \emph{definite alias properties}, that in terms of
  abstract alias queries correspond to the answer `\( 1 \)';  the
  approximation provided by this method is indeed also called \emph{may-alias
  information}.
  For example, at \refline 3 of
  \refcodesnippet{external_information_from_rand}, in all of the possible
  executions, the expression \QuotedCode{*p}
  is an alias of \QuotedCode{a}. However, this method is only able to
  tell that \QuotedCode{*p} is \emph{possibly} an alias of \QuotedCode{a},
  that in terms of abstract alias query corresponds to the outcome \(
  \top \).
\item
  Every solution of \( C \) corresponds to a different abstract alias
  query, whereas, as we will show in \refsection{completeness},
  the concretization of
  a \emph{points-to abstraction} consists of only one abstract alias query.
  As a consequence, this representation of the alias information is able
  capture relational information, whereas points-to methods cannot.
\end{itemize}
To represent the set of
integer constraints \(C\) different options exist.  The literature
on this field provides a wide choice of numeric lattices offering
different trade-off between accuracy and efficiency; from non
relational domains ---like arithmetic intervals and arithmetic
congruences---
up to relational domains \cite{BagnaraHZ08}.  The alias analysis
just described is completely parametric with respect to the chosen numeric
domain and ---due to the large availability of numeric domains--- this is a
point of strength of the method.

\subsection{The Store Based Approach}
\label{section:general store based methods}
This section introduces some concepts that are useful to understand
the approach of \emph{store based methods}.  Points-to analyses are special
cases of stored-based methods. The idea common to all store
based methods is the explicit introduction of formal entities to represent
memory locations.  As in the concrete situation we use the notation `\( \Locations
\)' to represent the set of the memory locations; now we introduce
the notation \( \absLocations \) to denote the set of the \emph{abstract
locations}.  Store based information usually consist of some sort of
compact representation of a binary relation `\( P \)' on the set of the
abstract locations.
To bind the concept of location to the concept of expression an
\emph{environment function} is provided. Basically, the environment function
is needed to resolve \emph{identifiers} into abstract locations.  Denoting
with `\( \identifiers \)' the set of identifiers, an environment function
can be described as
\[
  \functiondef
    { \textsc{id} }
    { \identifiers }
    { \absLocations }
\]
Since \emph{identifiers} are the base case for the definition of the \(
\expressions \) set, from the elements \( \langle P,
\textsc{id} \rangle \) it is possible to build the \emph{abstract evaluation function}
\[
  \functiondef
    { \eval }
    { ( \absmemories \times \expressions ) }
    { \partsof( \absLocations ) }.
\]
The `\(\mathord{\eval}\)' function is defined inductively following the
inductive definition of the `\( \expressions \)' set. The details
depend on the chosen language and intermediate representation;
a complete definition is presented in \refsection{our method}.
The `\( \eval \)' function is then used to define the semantics of
\( \absmemories \) in terms of abstract alias queries;
for instance, a possible definition is the following. Let \( m^\sharp \in
\absmemories \) and \( \absaliasquery \in \absAliasQueryDomain \), we say
that \( \absaliasquery \in \concretization(m^\sharp) \)  when, for all \(
e, f \in \expressions \) holds that
\[
  \absaliasquery(e, f) =
  \begin{cases}
    0, & \text{if }
      \eval(m^\sharp, e) \intersection \eval(m^\sharp, f) = \emptyset; \\
    \top, & \text{otherwise.}
  \end{cases}
\]
This is an oversimplified definition, presented only to give an idea of how
a store based approximation can answer to alias queries; note indeed that
we have omitted to consider \emph{definite} alias informations.
Due to the introduction of the set of abstract locations \( \absLocations
\), the semantics of the abstract domain \( \absmemories \) can also be
expressed in terms of the value of locations:
we have an abstraction function
\[
  \functiondef
    { \abstraction }
    { \Locations }
    { \absLocations };
\]
where, for each \( l \in \Locations \), \( \abstraction(l) \)
denotes the abstract location that approximates \( l \).
Given \( l \in \absLocations \) and an abstraction \( m^\sharp \in
\absmemories \), we denote as \( m^\sharp[l] \) the value
of the abstract location \( l \) in the abstract memory description \(
m^\sharp \). Now, let \( m \in \memories \) and \( m^\sharp \in
\absmemories \); then we have
\[
  m \in \concretization(m^\sharp)
  \quad \defiff \quad
  \biggl(
  \forall l \in \locations \itc
  m[l] \text{ is defined} \implies m[l]
  \in \concretization\Bigl( m^\sharp\bigl[\abstraction(l)\bigr] \Bigr)
  \biggl).
\]
This formulation of the semantics of \( \absmemories \) can be applied to
points-to methods only, but it has the advantage that it can be generalized
to the case where the points-to domain is coupled with some other abstract
domain, provided that its semantics can be expressed in the same way.
Moreover, the concretization function expressed in terms of locations is
more similar to the algorithms actually implemented as client analyses are
more likely to reason in terms of ``pointed locations'' than in terms of
``aliased expressions''.

\codecaption
{in this code the assignments at \refline{7} and \refline{8} contain
expressions the dereference operator occurs more than once.}
{pristine_code}
\begin{codesnippet}
int **pp, *p, a;
struct List {
  struct List *n;
  int key;
} *h;
...
**pp = *p;
h->key = h->n->n->key;
\end{codesnippet}

\codecaption
{this is the simplified version of \refline{7} and \refline{8} from
\refcodesnippet{pristine_code}. Note the use of the additional
variables \Code{tmp0}, \Code{tmp1} and \Code{tmp2}. Note that
all expressions contain at most one occurrence of the dereference operator.}
{simplified_code}
\begin{codesnippet}
...
int *tmp0 = *pp;
*tmp = *p;
struct List *tmp1 = h->n;
struct List *tmp2 = tmp1->n;
h->key = tmp2->key;
\end{codesnippet}

\subsubsection{Practical Considerations on Store Based Methods}
Despite the
commonalities of store based methods described in the previous section,
from the
implementation perspective many different options exist.  For example Emami
et al.\ \cite{Emami93th,EmamiGH94} and also \cite{Ghiya95}
do not define a complete
abstract evaluation function \(\evaluation\).  Instead, they prefer to
work on a simplified version of the code. To accomplish this they
introduce a \emph{simplification phase} to be performed before
the actual analysis.  Basically, this phase breaks the occurrences of
``complex'' expressions into a simpler form by means of the introduction of
auxiliar variables and assignments.  For example, in the simplified code all the expressions
contain at most one occurrence of the dereference operator.
\refcodesnippet{simplified_code} presents the result of the simplification
phase applied to the code in \refcodesnippet{pristine_code}.
Having reduced all the expressions to a base form, the definition of the
evaluation function \(\evaluation\) is greatly simplified.  However, the
simplification phase has also other side effects.  First, assuming to have
already proved the correctness of the analysis, its results
are valid on the code resulting from the simplification phase; to obtain
any formal result on the original code it must be proved that the applied
simplification does not change the semantics of the code.
From the point of view of
the efficiency, it is unclear whether or not a simpler evaluation function
\( \evaluation \) allows a more efficient analysis.  In both cases the same steps
of evaluation must be made; the difference is that in one case temporaries
are made explicit.  In our approach we have chosen to avoid the
simplification phase as we believe that enabling the analyzer to see
complete expressions can improve the precision.

\codecaption
{an example of `complex' expression occurring in the condition of an
\Code{if} statement.}
{simplification and filter}
\begin{codesnippet}
int a, b, *p, *q, **pp;
...
if (**pp == &a) {
  ...
}
\end{codesnippet}

\codecaption
{the result of the simplification  of \refcodesnippet{simplification and
filter}.}
{simplification and filter 2}
\begin{codesnippet}
int a, b, *p, *q, **pp;
...
int *temp = *pp;
if (*temp == &a) {
  ...
}
\end{codesnippet}

\begin{example}
Assume that at \refline 3 of \refcodesnippet{simplification and filter}
holds the following points-to information:
\begin{gather*}
  P(\Code{pp}) = \{ \Code q, \Code p \},           \\
  P(\Code{p})  = \{ \Code a \},           \\
  P(\Code{q})  = \{ \Code b \}.
\end{gather*}
Looking at the condition of  the \Code{if} statement at \refline 3,
it is possible to refine
the points to information of \refline 4; that is, inside the `then'
branch, \QuotedCode{pp} points only to \QuotedCode{p}.
However, on the simplified code (\refcodesnippet{simplification and filter
2}), looking only at the simplified condition of the \Code{if} statement, it
is not possible to infer any useful information about \QuotedCode{pp}, as
it occurs no more in the expression.
It is possible to prove that \QuotedCode{temp} points only to
\QuotedCode{a}, but this information is useless as \QuotedCode{temp} is a
auxiliar variable introduced by the simplification phase and thus it is not used
elsewhere.
\end{example}

\subsection{Precision Limits of the Alias Query Representation}
\label{section:Precision limits of the alias query representation}
This section presents an example that highlights the limitations of
the alias query representation;
alias queries (\refsection{concrete aliasing query})
fail to represent relational information.
For instance, the code presented in Listings~\ref{codesnippet:alias query non relational 1}
and \ref{codesnippet:alias query non relational 2} induce the same abstract
alias query;  in particular in \refcodesnippet{alias query non relational 1}
the alias representation is unable to express that, at \refline 4, if
\QuotedCode{p} points to \QuotedCode{a} then \QuotedCode{q} points
to \QuotedCode{c}. This situation is illustrated in
Figures~\ref{figure:alias query non relational} and
\ref{figure:alias query non relational b}.

\codecaption
{in this code only two possible executions exist.  At \refline{4}, knowing
the value of one of the two pointers \Variable{p} and \Variable{q}, it is
possible to determine the value of the other.}
{alias query non relational 1}
\begin{codesnippet}
int a, b, c, d, *p, *q;
if (rand()) { p = &a; q = &c; }
else        { p = &b; q = &d; }
...
\end{codesnippet}

\codecaption
{in this code four executions are possible;
at \refline{7}, also knowing
the value of one of the two pointers \Variable{p} and \Variable{q}, it is
\emph{not} possible to determine the value of the other.}
{alias query non relational 2}
\begin{codesnippet}
int a, b, c, d, *p, *q;
if (rand()) p = &a;
else        p = &b;

if (rand()) q = &c;
else        q = &d;
...
\end{codesnippet}

\begin{figure}
\begin{center}
\input{figures/05}
\caption
  {a representation of the alias query induced by the code in
  \refcodesnippet{alias query non relational 1}.}
\label{figure:alias query non relational}
\end{center}
\end{figure}

\begin{figure}
\begin{center}
\input{figures/05b}
\caption
  {continuation of \reffigure{alias query non relational}.}
\label{figure:alias query non relational b}
\end{center}
\end{figure}

%% file: tex/organization.tex
Starting from \refsection{Preliminaries}, the paper provides a general
description of the instruments commonly used to approach the points-to and
the alias problems.

Starting from \refsection{our method}, a simplified language and a
simplified execution model are introduced; the execution model comprehends
the \emph{memory model} and the operations that acts on it.  Subsequently,
an \emph{approximated memory model} and the approximated operations are
presented. Following the methodology of the abstract interpretation theory,
the soundness of the approximated execution model is proved. Finally, some
informal considerations about the precision of the abstraction are
presented.

Starting from \refsection{extensions}, in order to present a realistic
points-to analysis, some extensions to the model introduced in the previous
sections are presented and a possible implementation of the approximated
memory model is described.

Finally, \refsection{conclusion} draws the conclusions of the work
and it discusses some of the possible future developments of the present
work.

%% file: figures/07.tex
\input{figures/common_style}
\tikzstyle{domain style}=[
    node distance=6em,
    circle,
    join=round,
    line width=0.3pt,
    inner sep=0pt,
    outer sep=1pt,
    minimum size=2.5em]

\tikzstyle{abstract arrow style}=[
    shorten >=1pt,
    ->,
    auto,
    node distance=0.5cm]

\begin{tikzpicture}
  \def \domainNode{\node[domain style]}
  \def \abstracts(#1,#2)
    {\path (#2) edge [abstract arrow style] node { \( \gamma \) } (#1);}

  \domainNode (exec)                          {\( \memories \)};
  \domainNode (query)     [below of=exec]     {\( \queryDomain \)};
  \domainNode (absQuery)  [right of=query]    {\( \absQueryDomain \)};
  \domainNode (absExec)   [above of=absQuery] {\( \absmemories \)};
  \abstracts(exec, query)
  \abstracts(query, absQuery)
  \abstracts(absQuery, absExec)
  \figureBackground
      (query.south -| query.west)
      (absExec.north -| absExec.east)
\end{tikzpicture}

%% file: figures/03.tex
\input{figures/common_style}
\tikzstyle{textBlock}=[
    rectangle,
    join=round,
    rounded corners,
    text badly centered,
    draw=LocationBorderColor,
    line width=1pt,
    fill=LocationFillColor]
\begin{tikzpicture}
  \node[textBlock, text width=10em] (execution context) at (3, 8)
  {
    Concrete memory \( m \in \memories \)
  };
  \node[textBlock, text width=10em] (concrete alias query) at (3, 6)
  {
    Concrete alias query
    \( \aliasquery \in \aliasQueryDomain \)
  };
  \node[textBlock, text width=10em] (abstract alias query) at (3, 4)
  {
    Abstract alias query
    \( \absaliasquery \in \absAliasQueryDomain \)
  };
  \node[textBlock, text width=10em] (access path abstraction) at (0, 2)
  {
    Access-path based abstraction (Deutsch)
  };
  \node[textBlock, text width=10em] (pointsTo based abstraction) at (6, 2)
  {
    Points-to based abstraction.
  };
  \node[textBlock, text width=10em] (pointsTo based abstraction) at (6, 2)
  {
    Points-to based abstraction.
  };
  \node[textBlock, text width=10em, fill=blue!5] (abs mem desc) at (3, 0)
  {
    Abstract memory \( m^\sharp \in \absmemories \)
  };
  \path
    (execution context) edge [pointsTo] (concrete alias query)
    (concrete alias query) edge [pointsTo] (abstract alias query)
    (abstract alias query) edge [pointsTo, bend right] (access path abstraction)
    (abstract alias query) edge [pointsTo, bend left] (pointsTo based abstraction)
    (abs mem desc) edge [pointsTo, bend left, dashed] (access path abstraction)
    (abs mem desc) edge [pointsTo, bend right, dashed] (pointsTo based abstraction);
  \figureBackground
    (abs mem desc.south -| access path abstraction.west)
    (execution context.north -| pointsTo based abstraction.east)
\end{tikzpicture}

%% file: figures/05.tex
\input{figures/common_style}
\begin{tikzpicture}
  \def\drawThePicture[#1,#2,#3,#4]{
    \abstractLocation a (0, 1)
    \abstractLocation b (1, 1)
    \abstractLocation c (2.5, 1)
    \abstractLocation d (3.5, 1)
    \abstractLocation p (0.5, 0)
    \abstractLocation q (3, 0)
    \node[anchor=north] at (1.75, -0.75) (table) {
      \begin{tabular}{*{5}{>{\(}c<{\)}}}
        \hline \hline
        #1 & \Code{a} & \Code{b} & \Code{c} & \Code{d} \\
        \hline
        \Code{*p}  & #2 \\
        \Code{*q}  & #3 \\
        \hline \hline
      \end{tabular}
    };
    \figureBackground
       (table.south -| table.west)
       (d.north -| table.east)
    \node [description label] at (table.north east) {#4};
  }
  \begin{scope}
    \drawThePicture[
      \aliasquery_{\executionEnvironment_0},
      1 & 0 & 0 & 0,
      0 & 0 & 1 & 0,
      {
        Below an extract of the concrete alias query \(
        \aliasquery_{m_0} \) induced by the concrete memory description \(
        m_0 \in \memories \).  Above a graphical representation of the
        points-to information associated to the same memory.
      }]
    \path \edgePointsTo(p,a)
          \edgePointsTo(q,c);
  \end{scope}
  \begin{scope}[yshift=-5cm]
    \drawThePicture[
      \aliasquery_{\executionEnvironment_1},
      0 & 1 & 0 & 0,
      0 & 0 & 0 & 1,
      {
        As above, on the concrete memory description \(
        \executionEnvironment_1 \).
      }]
    \path \edgePointsTo(p,b)
          \edgePointsTo(q,d);
  \end{scope}
  \begin{scope}[yshift=-10cm]
    \drawThePicture[
      \absaliasquery{},
      \top  & \top  & 0     & 0,
      0     & 0     & \top  & \top,
      {
        The abstract alias query
        \[
           \absaliasquery \in \absAliasQueryDomain
        \] is defined as
        \( \aliasquery_{\executionEnvironment_0}
            \sqcup \aliasquery_{\executionEnvironment_1}
        \) and it is the most precise abstract alias query that abstracts the
        set of concrete alias queries
        \[
          \{ \aliasquery_{\executionEnvironment_0},
            \aliasquery_{\executionEnvironment_1} \}.
        \]
      }]
    \path \edgePointsTo(p,a)
          \edgePointsTo(p,b)
          \edgePointsTo(q,c)
          \edgePointsTo(q,d);
  \end{scope}
\end{tikzpicture}

%% file: figures/05b.tex
\input{figures/common_style}
\begin{tikzpicture}
  \def\drawThePicture[#1,#2,#3,#4]{
    \abstractLocation a (0, 1)
    \abstractLocation b (1, 1)
    \abstractLocation c (2.5, 1)
    \abstractLocation d (3.5, 1)
    \abstractLocation p (0.5, 0)
    \abstractLocation q (3, 0)
    \node[anchor=north] at (1.75, -0.75) (table) {
      \begin{tabular}{*{5}{>{\(}c<{\)}}}
        \hline \hline
        #1 & \Code{a} & \Code{b} & \Code{c} & \Code{d} \\
        \hline
        \Code{*p}  & #2 \\
        \Code{*q}  & #3 \\
        \hline \hline
      \end{tabular}
    };
    \figureBackground
       (table.south -| table.west)
       (d.north -| table.east)
    \node [description label] at (table.north east) {#4};
  }
  \begin{scope}
    \drawThePicture[
      \aliasquery_{ \executionEnvironment_2 },
      1 & 0 & 0 & 0,
      0 & 0 & 0 & 1,
      {
        An example of a spurious element of the concretization of \(
        \absaliasquery \). We have that
        \[
          \aliasquery_{ \executionEnvironment_2 } \in
          \concretization( \absaliasquery ).
        \]
        However, \( \aliasquery_{ \executionEnvironment_2
        } \) can not be generated by the program.
      }]
    \path \edgePointsTo(p,a)
          \edgePointsTo(q,d);
  \end{scope}
  \begin{scope}[yshift=-5cm]
    \drawThePicture[
      \aliasquery_{ \executionEnvironment_3 },
      0 & 1 & 0 & 0,
      0 & 0 & 1 & 0,
      {
        Another spurious element of the concretization of
        \( \absaliasquery \). Note that:
        \begin{align*}
          \concretization(\absaliasquery )
             & = \concretization(
                  \aliasquery_{ \executionEnvironment_0 }
                  \sqcup
                  \aliasquery_{ \executionEnvironment_1 }) \\
             &
              \begin{aligned}
                = \{
                & \aliasquery_{\executionEnvironment_0 },
                \aliasquery_{\executionEnvironment_1 }, \\
                & \aliasquery_{\executionEnvironment_2 },
                  \aliasquery_{\executionEnvironment_3 }
               \}.
               \end{aligned}
        \end{align*}
      }]
    \path \edgePointsTo(p,b)
          \edgePointsTo(q,c);
  \end{scope}
\end{tikzpicture}

%% file: tex/proof.tex
\section{The Analysis Method}
\label{section:our method}
This part of the work is meant to be as much self-contained as possible.
The aim of this sections is to present few simple but formal definitions of
a simplified but general memory model and, on these, build the algorithms
and prove their correctness.

\subsection{The Domain}
Let \( \locations \) be a given set that we call the \emph{locations set}
and whose elements are called \emph{locations}.

\begin{definition}
\definitionsummary{Abstract and concrete domains}
We call \emph{support set of the concrete domain} the set \(
\concretedomain \) of the total functions from \( \locations \) to \(
\locations \); we call \emph{support set of the abstract domain} the set
\( \abstractdomain \) of the binary relations on the set \( \locations
\)
\begin{gather*}
  \concretedomain \defeq \locations \to \locations; \\
  \abstractdomain \defeq \partsof(\locations \times \locations).
\end{gather*}
We define the \emph{concrete domain} as the complete lattice generated
by the powerset of \( \concretedomain \)
\[
  \big< \partsof(\concretedomain), \subseteq, \union, \intersection,
        \emptyset, \concretedomain \big>.
\]
We define the \emph{abstract domain} as the complete lattice
\[
  \big< \abstractdomain, \subseteq, \union, \intersection,
        \emptyset, \locations \times \locations \big>.
\]
\end{definition}

Note that from the above definition we have that \( \concretedomain
\subseteq \abstractdomain \).  Though we use the same notation for the
operations of the two lattices they obviously have different definitions.
For the abstract domain the partial order `\( \subseteq \)',
the operations `\( \union \)' and `\( \intersection \)' are referred to
sets of pairs of locations; whereas for the concrete domain they are
referred to sets of functions \( \locations \to \locations \).
The semantics of the abstract domain is defined using the fact that \(
\concretedomain \subseteq \abstractdomain \) and the partial order
`\(\subseteq\)' on sets of pairs of locations.

\begin{definition}
\summary{Concretization function}
\label{definition:concretization function}
Let
\[
  \functiondef
    { \concretization }
    { \abstractdomain }
    { \partsof(\concretedomain) }
\]
be defined, for all \( A \in \abstractdomain \), as
\[
    \concretization(A) \defeq
      \bigl\{\, C \in \concretedomain \bigm| C \subseteq A \,\bigr\}.
\]
\end{definition}

Now we present some definitions useful to define how we navigate the
poinst-to graph.

\begin{definition}
\summary{The prev and post functions}
\label{definition:prev function}
\label{definition:post function}
Let
\[
  \functiondef
    { \prev, \post }
    { \abstractdomain \times \locations }
    { \partsof(\locations) }
\]
be defined, for all \( A \in \abstractdomain \) and \( l \in \locations \), as
\begin{gather*}
  \prev(A, l)
    \defeq \bigl\{\, m \in \locations \bigm| (m, l) \in A \,\bigr\};
    \\
  \post(A, l)
    \defeq \bigl\{\, m \in \locations \bigm| (l, m) \in A \,\bigr\}.
\end{gather*}
\end{definition}

For convenience we generalize the definition of the \( \post \) and
\( \prev \) functions to sets of locations.

\begin{definition}
\summary{Extended prev and post functions}
\label{definition:extended prev and post functions}
Let
\[
  \functiondef
    { \prev, \post }
    { \abstractdomain \times \partsof(\locations) }
    { \partsof(\locations) }
\]
be defined, for all \( A \in \abstractdomain \) and \( L \subseteq
\locations \), as
\begin{gather*}
  \prev(A, L)
    \defeq \bigunion \bigl\{\, \prev(A, l) \bigm| l \in A \,\bigr\}; \\
  \post(A, L)
    \defeq \bigunion \bigl\{\, \post(A, l) \bigm| l \in A \,\bigr\}.
\end{gather*}
\end{definition}

\subsection{The Language}
In this section we present a simple language to model the points-to
problem.

\begin{definition}
\summary{Expressions}
\label{definition:expressions}
We define the set \( \expressions \) as the language generated by the
grammar
\[
  e ::= l \vbar \indirection e
\]
where \( l \in \locations \) and \( \mathord{\indirection} \not \in
\locations \) is a terminal symbol.
\end{definition}

\begin{definition}
\summary{Evaluation of expressions}
\label{definition:eval function}
Let
\[
  \functiondef
    { \eval }
    { \abstractdomain \times \expressions }
    { \partsof(\locations) }
\]
be defined inductively on \( \expressions \)
(\refdefinition{expressions}). Let \( A \in \abstractdomain \), \( l \in
\locations \) and \( e \in \expressions \); then we define
\begin{gather*}
    \eval(A, l) \defeq \{ l \}; \\
    \eval(A, \indirection e) \defeq \post\bigl(A, \eval(A, e)\bigr).
\end{gather*}
\end{definition}

Not necessary for the goal of this section, for completeness we report the
concretization of the points-to abstract domain in terms of \emph{abstract
alias queries}.

\begin{definition}
\definitionsummary{Induced alias relation}
We define
\[
  \functiondef
    { \concretization }
    { \abstractdomain }
    { \absAliasQueryDomain }
\]
as follows. Let \( A \in \abstractdomain \), then let
\(
  \concretization(A) \defeq \mathord{\absaliasquery},
\)
where, forall \( e, f \in \expressions \), we have
\begin{gather*}
  E \defeq \eval(A, e); \\
  F \defeq \eval(A, f); \\
    \absaliasquery(e, f) \defeq
      \begin{cases}
        0,  & \text{if } E \intersection F = \emptyset; \\
        1,  & \text{if } E = F \land \cardinality E = 1; \\
        \top, & \text{ otherwise.}
      \end{cases}
\end{gather*}
\end{definition}

\begin{definition}
\definitionsummary{Conditions}
We define the set of \emph{conditions} as the set
\[
  \conditions \defeq
  \{ \equality, \inequality \} \times \expressions \times \expressions
\]
\end{definition}

\begin{definition}
\definitionsummary{Value of conditions}
Let
\[
  \trueconditions \subseteq \concretedomain \times \conditions
\]
be a set defined, for all \( C \in \concretedomain \) and \( e, f \in
\expressions \), as
\begin{gather*}
    \bigl(C, (\equality, e, f)\bigr) \in \trueconditions
    \quad \defiff \quad
    \eval(C, e) = \eval(C, f); \\
    \bigl(C, (\inequality, e, f)\bigr) \in \trueconditions
    \quad \defiff \quad
    \bigl(C, (\equality, e, f)\bigr) \not \in \trueconditions.
\end{gather*}
Let \( C \in \concretedomain \) and let \( c \in \conditions \), for
convenience of notation we write \( C \models c \) when \( (C, c) \in
\trueconditions \).  We also introduce the function
\[
  \functiondef
    { \modelset }
    { \conditions }
    { \partsof(\concretedomain) },
\]
defined, for all \( c \in \conditions \), as
\[
  \modelset(c) \defeq
    \bigl\{\, C \in \concretedomain \bigm| C \models c \,\bigr\}.
\]
In other words, `\( \modelset(c) \)' is the set of the concrete
memory descriptions where the condition \( c \) is true.
\end{definition}

\subsubsection{Assignment}
\begin{definition}
\definitionsummary{Assignment evaluation}
We define the set of \emph{assignments} as
\[
  \assignments \defeq \expressions \times \expressions
\]
Let
\[
  \functiondef
    { \assign }
    { \abstractdomain \times \assignments }
    { \abstractdomain }
\]
be defined as follows. For all \( A \in \abstractdomain \) and \( e, f \in
\expressions \),
let
\[
    \assign\bigl(A, (e, f)\bigr) \defeq
    \eval(A, e) \times \eval(A, f) \union ( A \setminus K )
\]
where, the \emph{kill set} \( K \) is defined as
\[
  K \defeq
    \begin{cases}
      \eval(A, e) \times \locations,
      & \text{if } \cardinality \eval(A, e) = 1; \\
      \emptyset,
      & \text{otherwise}.
    \end{cases}
\]
\end{definition}

The following lemma shows that
the `\( \mathord{\assign} \)' function just described,
defines also the concrete semantics of the assign operation, i.e.'
performing an assignment on an element of
the concrete domain yields another element of the concrete domain.

\begin{lemma}
\lemmasummary{Restriction of the assignment to the concrete domain}
The set \( \concretedomain \) is closed with respect to the function \(
\assign \), that is, for all \( C \in \concretedomain \) and \( a \in
\assignments \) holds that
\(
    \assign(C, a) \in \concretedomain.
\)
\end{lemma}

Therefore, the function \( \assign \) restricted to \( \concretedomain \) can be
written as
\[
  \functiondef
    { \assign }
    { \concretedomain \times \assignments }
    { \concretedomain }.
\]
In other words, the formalization of the assignment operation given in
\refdefinition{Assignment evaluation} is a generalization of the concrete
assignment behaviour.  At this point, we define the concrete semantics of the
assignment.

\begin{definition}
\definitionsummary{Concrete assignment operation}
Let
\[
  \functiondef
    { \assign }
    { \partsof( \concretedomain ) \times \assignments }
    { \partsof( \concretedomain ) }
\]
defined, forall \( D \subseteq \concretedomain \) and
\( a \in \assignments \), as
\[
    \assign(D, a) \defeq \bigr\{\, \assign(C, a) \bigm| C \in D  \,\bigl\}.
\]
\end{definition}

\subsubsection{Filter}
\begin{definition}
\definitionsummary{Concrete filter semantics}
Let
\[
  \functiondef
    { \filter }
    { \partsof(\concretedomain) \times \conditions }
    { \partsof(\concretedomain) }
\]
be defined, for all \( D \subseteq \concretedomain \) and \( c \in
\conditions \), as
\[
    \filter(D, c) \defeq \modelset(c) \intersection D.
\]
\end{definition}

In other words, given a set \( D \) of concrete memory descriptions and a
boolean condition \( c \) we denote with \( \filter(D, c) \) the subset of
\( D \) of the elements in which the condition \( c \) is true.  We
proceed in the definition of the abstract filter operation.  Since we want
to track step by step the evaluation of expressions, we extend the
definition of the \( \eval \) function to allow this.

\begin{definition}
\summary{Extended eval function}
\label{definition:extended eval function}
Let
\[
  \functiondef
    { \eval }
    { \abstractdomain \times \expressions \times \naturals }
    { \partsof(\locations) }
\]
be inductively defined as follows.
Let \( A \in \abstractdomain \), \( l \in \locations \), \( e \in
\expressions \) and \( i \in \naturals \); then we define
\begin{gather*}
    \eval(A, e, 0)                  \defeq \eval(A, e); \\
    \eval(A, l, i + 1)              \defeq \emptyset; \\
    \eval(A, \indirection e, i + 1) \defeq \eval(A, e, i).
\end{gather*}
\end{definition}

\begin{definition}
\definitionsummary{Target function}
We define the function
\[
  \functiondef
    { \target }
    { \abstractdomain
        \times \partsof(\locations)
        \times \expressions
        \times \naturals }
    { \partsof(\locations) }
\]
inductively as follows. Let \( A \in \abstractdomain \), \(M \subseteq
\locations \), \( e \in \expressions \) and \( i \in \naturals \); then we
define
\begin{gather*}
  \target(A, M, e, 0)
    \defeq \eval(A, e) \intersection M; \\
  \target(A, M, e, i + 1)
    \defeq
    \eval(A, e, i + 1) \intersection \prev\bigl(A, \target(A, e, i)\bigr).
\end{gather*}
\end{definition}

\begin{definition}
\definitionsummary{Filter 1}
Let
\[
  \functiondef
    { \filter }
    { \abstractdomain
        \times \partsof(\locations)
        \times \expressions
        \times \naturals }
    { \abstractdomain }
\]
be defined as follows. Let \( A \in \abstractdomain \),
\( M \subseteq \locations \), \( e \in \expressions \) and
\( i \in \naturals \). For convenience of notation let \( x = \langle A, M,
e \rangle \); then we define
\begin{gather*}
    \filter(x, 0) \defeq A; \\
    T \defeq \target(x, i + 1); \\
    \filter(x, i + 1) \defeq
    \filter(x, i)
    \setminus
    \begin{cases}
      T \times
        \bigl( \locations \setminus \target(x, i) \bigr),
        & \text{if } \cardinality T = 1; \\
      \emptyset,
        & \text{otherwise}.
    \end{cases}
\end{gather*}
\end{definition}

\begin{definition}
\definitionsummary{Filter 2}
Let
\[
  \functiondef
    { \filter }
    { \abstractdomain \times \partsof(\locations) \times \expressions }
    { \abstractdomain }
\]
be defined, for all \( A \in \abstractdomain \), \( M \subseteq \locations
\) and \( l \in \locations \), as
\begin{gather*}
  \filter(A, M, l) \defeq
    \begin{cases}
      A,  & \text{if } l \in M; \\
      \bot, & \text{otherwise;}
    \end{cases}
  \\
  \filter(A, M, \indirection e) \defeq
    \bigintersection_{ i \in \naturals } \filter(A, M, \indirection e, i).
\end{gather*}
\end{definition}

\begin{definition}
\definitionsummary{Filter 3}
Let
\[
  \functiondef
    { \filter }
    { \abstractdomain \times \conditions }
    { \abstractdomain }
\]
be defined as follows. Let \( e, f \in \expressions \), and let
\begin{gather*}
  I \defeq \eval(A, e) \intersection \eval(A, f); \\
  E \defeq \eval(A, e) \setminus \eval(A, f); \\
  F \defeq \eval(A, f) \setminus \eval(A, e).
\end{gather*}
Then, for all \( A \in \abstractdomain \), we define
\begin{gather*}
  \filter\bigl(A, (\equality, e, f)\bigr) \defeq
    \filter(A, I, e) \intersection \filter(A, I, f),
  \\
  \filter\bigl( A, (\inequality, e, f) \bigr) \defeq
    \begin{cases}
      \filter(A, E, e) \union \filter(A, F, f),
        & \text{if } \cardinality I = 1; \\
      A,
        & \text{otherwise}.
      \end{cases}
\end{gather*}
\end{definition}

\subsection{Examples}
This section presents some examples to illustrated how the model just
presented works.
\input{./tex/examples.tex}

\subsection{Results}
\subsubsection{Notation}
The proofs are organized as sequences of deductions, for convenience of
notation presented inside tables.  Each table is organized in three
columns: the first column contains the \emph{tag} used to name the step;
the second column contains the statement and the third column contains a
list of tags that represents the list of statements used to infer the
current row.  There are three kind of tags.  The first kind of tag, denoted
as `TS', is used to mark the \emph{thesis}, which, if explicitly presented,
occurs always in the top row. The second kind of tag is used to
describe the \emph{hypotheses}, marked as `H0', \ldots, `H\(n\)'. Among the
hypotheses we improperly list the lemmas used in the proof.
The third kind of tag is used to describe
\emph{deductions}, displayed as `D0', \ldots, `D\(n\)', with the exception
of the last deduction step, which is tagged with the symbol `\(\proofleaf\)'.
Within the table, the hypotheses are displayed below the thesis and
deductions below the hypotheses.  To stress the separation of the thesis
from the hypotheses and of the hypotheses from the deductions horizontal
line are used.  Deductions, between themselves, are sorted in topological
order, such that, if the deduction `D\(m\)' \emph{requires} the deduction
`D\(n\)', then \( m > n\).  When the proof consists of more \emph{cases},
then multiple tables are used; in this case, an initial table containing
the hypotheses common to all cases may be present.  Cases are marked as
`C1', \ldots, `C\(n\)'; if an hypothesis comes from considering the case
`C\(n\)', then the tag  `C\(n\)' is also reported in the third column of
the corresponding row. In \emph{inductive} proofs, the \emph{inductive}
hypothesis is marked with a `(ind.\ hyp.)'.

\subsubsection{Concrete Assignment}
We start by showing that the assignment operation is
closed with respect to the set of the concrete memory descriptions \(
\concretedomain \).

\begin{lemma}
\lemmasummary{Eval cardinality on the concrete domain}
Let \( C \in \concretedomain \) and \( e \in \expressions \), then
\( \cardinality \eval(C, e) = 1.  \)
\end{lemma}

\begin{proof}
Let \( C \in \concretedomain \).
We proceed by induction on the definition of the set \( \expressions \)
(\refdefinition{expressions}).
\begin{prooftable}
TS  & \( \cardinality \eval(C, e) = 1 \) \\
\hline
H0  & \refdefinition{eval function}, the eval function. \\
\end{prooftable}
For the base case let \( e = l \in \locations \).
\begin{prooftable}
TS  & \( \cardinality \eval(C, l) = 1 \) \\
\hline
D0  & \( \eval(C, l) = \{ l \} \) & (H0) \\
\proofleaf
    & \( \cardinality \eval(C, l) = 1 \) & (D0) \\
\end{prooftable}
For the inductive case let \( e \in \expressions \).
\begin{prooftable}
TS  & \( \cardinality \eval(C, \indirection e) = 1 \) \\
\hline
H1  & \( \cardinality \eval(C, e) = 1 \) &  (ind.\ hyp.) \\
H2  & \refdefinition{Abstract and concrete domains}, the concrete domain. \\
H3  & \refdefinition{post function}, the post function. \\
\hline
D0  & \( \eval(C, \indirection e) = \post\bigl( C, \eval(C, e) \bigr) \) & (H0) \\
D1  & \( \forall l \in \locations \itc
         \cardinality \bigl\{ (l, m) \in C \bigr\} = 1 \)
    & (H2) \\
D2  & \( \cardinality \bigl\{\, (l, m) \in C \bigm| l \in \eval(C, e) \,\bigr\} = 1 \)
    & (H1, D1) \\
D3  & \( \cardinality \post\bigl( C, \eval(C, e) \bigr) = 1 \) & (D2, H3) \\
\proofleaf
    & \( \cardinality \eval(C, \indirection e) = 1 \) & (D3, D0) \\
\end{prooftable}
\end{proof}

\begin{lemma}
\lemmasummary{Assignment on the concrete domain}
Let \( C \in \concretedomain \) and \( e, f \in \expressions \).
For convenience of notation, let \( a \in \assignments \) such that \( a = (e, f) \).
Let
\begin{gather*}
  \eval(C, e) = \{ l \}; \\
  \post(C, l) = \{ n \}; \\
  \eval(C, f) = \{ m \}; \\
\intertext{then}
  \assign(C, a) =
  \Bigr( C \setminus \bigl\{ (l, n) \bigr\} \Bigr)
  \union \bigl\{ (l, m) \bigr\}.
\end{gather*}
\end{lemma}

\begin{proof}
Let \( C \in \concretedomain \).
First note that from the definition of the concrete domain
(\refdefinition{Abstract and concrete domains}) and
\reflemma{Eval cardinality on the concrete domain},
\( \cardinality \eval(C, e) = \cardinality \eval(C, f) = 1 \).
From the definion of the post function (\refdefinition{post function})
also \( \cardinality \post(C, l) = 1 \). Thus the above statement is well
formed.
\begin{prooftable}
TS  & \( \assign(C, a) = \Bigr( A \setminus \bigl\{ (l, n) \bigr\} \Bigr)
         \union \bigl\{ (l, m) \bigr\} \) \\
\hline
H0  & \( \{ l \} = \eval(C, e) \) \\
H1  & \( \{ n \} = \post(C, l) \) \\
H2  & \( \{ m \} = \eval(C, f) \) \\
H3  & \refdefinition{Assignment evaluation},
      the assignment evaluation. \\
\hline
D0  & \( \cardinality \eval(C, e) = 1 \) & (H0) \\
D1  & \( \assign(C, a) = \eval(C, e) \times \eval(C, f) \) \\
    & \( \qquad \union \, \bigl( C \setminus \eval(C, e) \times \locations \bigr) \)
    & (D0, H3) \\
D2  & \( \eval(C, e) \times \eval(C, f) = \bigl\{ (l, m) \bigr\} \)
    & (H0, H2) \\
D3  & \( C \intersection \eval(C, e) \times \locations =
         \bigl\{ (l, n) \bigr\} \)
    & (H1) \\
D4  & \( C \setminus \eval(C, e) \times \locations =
         C \setminus \bigl\{ (l, n) \bigr\} \)
    & (D3) \\
\proofleaf
    & \( \assign(C, a) = \Bigr( C \setminus \bigl\{ (l, n) \bigr\} \Bigr)
         \union \bigl\{ (l, m) \bigr\} \)
    & (D4, D2, D1) \\
\end{prooftable}
\end{proof}

\begin{proof}
\summary{Restriction of the assignment to the concrete, \reflemma{Restriction of the assignment to the concrete domain}}
This result is a simple corollary of \reflemma{Assignment on the concrete
domain}.
\end{proof}

\subsubsection{Observations on the Domain}
First we present the following simple result about the monotonicity of the
concretization function.

\begin{lemma}
\lemmasummary{Monotonicity of the concretization function}
Let \( A, B \in \abstractdomain \), then
\[
  A \subseteq B \implies \concretization(A) \subseteq \concretization(B).
\]
\end{lemma}

\begin{proof}
Let \( A, B \in \abstractdomain \). If
\( \concretization(A) = \emptyset \) then the thesis is trivially verified.
Otherwise let \( C \in \concretization(A) \), we have to show that \( C \in
\concretization(B) \) too.
\begin{prooftable}
TH  & \( C \in \concretization(B) \) \\
\hline
H0  & \( C \in \concretization(A) \) \\
H1  & \refdefinition{concretization function}, the concretization function.  \\
H2  & \( A \subseteq B \) \\
\hline
D0  & \( C \subseteq A \) & (H0, H1) \\
D1  & \( C \in \concretedomain \) & (H0, H1) \\
D2  & \( C \subseteq B \) & (D0, H2) \\
\proofleaf
    & \( C \in \concretization(B) \) & (D2, D1, H1) \\
\end{prooftable}
\end{proof}

From the definition of the concrete and of the abstract domain
(\refdefinition{Abstract and concrete domains}) and the definition of the
concretization
function (\refdefinition{concretization function}) we complete the
description of the abstraction by presenting
the \emph{abstraction function}.

\begin{definition}
\definitionsummary{Abstraction function}
Let
\[
  \functiondef
    { \abstraction }
    { \partsof(\concretedomain) }
    { \abstractdomain }
\]
be defined, for all \( C \subseteq \concretedomain \), as
\[
  \abstraction(C) \defeq \bigunion_{ D \in C } D.
\]
\end{definition}

It is possible to show that
\(
  \bigl< \partsof(\concretedomain), \abstraction, \abstractdomain,
  \concretization \bigr>
\)
is a \emph{Galois connection}, that is, for all \( C \in \concretedomain \)
and \( A \in \abstractdomain \), holds that:
\[
  \abstraction(C) \subseteq A \iff C \subseteq \concretization(A).
\]
Indeed, given \( C \subseteq \concretedomain \)
and \( A \in \abstractdomain \) the following steps are
all equivalent
\begin{gather*}
  \abstraction(C) \subseteq A, \\
  \bigunion_{ D \in C } D \subseteq A, \\
  \forall D \in C \itc D \subseteq A, \\
  \forall D \in C \itc D \in \concretization(A), \\
  C \subseteq \concretization(A).
\end{gather*}
On the presented abstraction holds also the following result.
The following lemma shows that given a non-bottom abstraction \( a \in
\abstractdomain \), then for each arc \( (l, m) \in A \) there is a
concrete memory \( C \) abstracted by \( A \) that contains the arc \( (l,
m) \).

\begin{lemma}
\lemmasummary{Concrete coverage}
Let \( A \in \abstractdomain \), then
\[
    \concretization(A) \neq \emptyset
    \implies
    \forall (l, m) \in A \itc
    \exists C \in \concretization(A) \suchthat
    (l, m) \in C.
\]
\end{lemma}

\begin{proof}
Let \( A \in \abstractdomain \) such that
\( \concretization(A) \neq \emptyset \) and let
\( (l, m) \in A \).
Let \( C \in \concretization(A) \), let \( n \in \locations \).
\begin{prooftable}
TS  & \( \exists D \in \concretization(A) \suchthat (l, m) \in D \) \\
\hline
H0  & \( (l, m) \in A \) \\
H1  & \( C \in \concretization(A) \) \\
H2  & \( \{ n \} \in \post(C, l) \) \\
H3  & \refdefinition{concretization function},
      concretization function. \\
H4  & \refdefinition{post function}, the post function. \\
\hline
D0  & \( C \subseteq A \) & (H1, H3) \\
D1  & \( C \setminus \bigl\{ (l, n) \bigr\} \subseteq A \) & (D0) \\
D2  & \( (l, n) \in C \) & (H2, H4) \\
D3  & \( \Bigr( C \setminus \bigl\{ (l, n) \bigr\} \Bigr)
         \union \bigl\{ (l, m) \bigr\} \in \concretedomain \)
    & (D2) \\
D4  & \( \Bigr( C \setminus \bigl\{ (l, n) \bigr\} \Bigr)
         \union \bigl\{ (l, m) \bigr\} \subseteq A \)
    & (D1, H0) \\
D5  & \( \Bigr( C \setminus \bigl\{ (l, n) \bigr\} \Bigr)
         \union \bigl\{ (l, m) \bigr\} \in \concretization(A) \)
    & (D3, D4, H3) \\
\proofleaf
    & \( \exists D \in \concretization(A) \suchthat (l, m) \in C \)
    & (D5) \\
\end{prooftable}
\end{proof}

\begin{lemma}
\lemmasummary{Abstraction effect}
Let \( A \in \abstractdomain \), then
\[
    \abstraction\bigl( \concretization( A ) \bigr) \subseteq A;
\]
moreover
\[
    \concretization(A) \neq \emptyset \implies
    \abstraction\bigl( \concretization(A) \bigr) = A.
\]
\end{lemma}

\begin{proof}
Let \( A \in \abstractdomain \).
Consider that
\begin{prooftable}
H0  & \refdefinition{Abstraction function}, the abstraction function. \\
\hline
D0  & \( \abstraction\bigl( \concretization( A ) \bigr) =
         \bigunion_{C \in \concretization(A) } C \) & (H0) \\
\end{prooftable}
We proceed by showing the two inclusions separately.
For the first inclusion let \( (l, m) \in \abstraction\bigl(
\concretization( A ) \bigr) \); then we have
\begin{prooftable}
TS  & \( (l, m) \in A \) \\
\hline
H1  & \( (l, m) \in \abstraction\bigl( \concretization( A ) \bigr) \) \\
H2  & \refdefinition{concretization function}, the concretization function. \\
\hline
D1  & \( \exists C \in \concretization(A) \suchthat (l, m) \in C \)
    & (D0, H1) \\
D2  & \( \forall C \in \concretization(A) \itc C \subseteq A \) & (H2) \\
\proofleaf
    & \( (l, m) \in A \) & (D1, D2) \\
\end{prooftable}
For the second inclusion assume that \( \concretization(A) \neq
\emptyset \) and let \( (l, m) \in A \); then we have
\begin{prooftable}
TS  & \( (l, m) \in \abstraction\bigl( \concretization( A ) \bigr) \) \\
\hline
H1  & \( \concretization(A) \neq \emptyset \) \\
H2  & \( (l, m) \in A \) \\
H3  & \reflemma{Concrete coverage}, concrete coverage. \\
\hline
D1  & \( \exists C \in \concretization(A) \suchthat
         (l, m) \in C \) & (H2, H3, H1) \\
\proofleaf
    & \( (l, m) \in \abstraction\bigl( \concretization( A ) \bigr) \) 
    & (D0, D1) \\
\end{prooftable}
\end{proof}

\subsubsection{Results of Correctness}
We formalize the requirement of correctness of the abstract operations
presented ---the expression evaluation, the assignment and the filter
operations--- with the following theorems.

\begin{theorem}
\theoremsummary{Correctness of expression evaluation}
Let \( A \in \abstractdomain \) and \( e \in \expressions \); then
\[
    \bigunion_{ C \in \concretization(A) } \eval(C, e)
    \subseteq \eval(A, e).
\]
\end{theorem}

\begin{theorem}
\theoremsummary{Correctness of the assignment}
Let \( A \in \abstractdomain \) and \( a \in \assignments \); then
\[
    \assign\bigl( \concretization(C), a\bigr) \subseteq
    \concretization\bigl( \assign(A, a) \bigr).
\]
\end{theorem}

\begin{theorem}
\theoremsummary{Correctness of the filter}
Let \( A \in \abstractdomain \) and \( c \in \conditions \); then
\[
    \filter\bigl( \concretization(A), c \bigr)
    \subseteq
    \concretization\bigl( \filter(A, c) \bigr).
\]
\end{theorem}

\subsubsection{Proofs}
We present some technical lemmas that will lead to the proof of the
correctness theorems.

\begin{lemma}
\summary{Monotonicity of post}
\label{lemma:monotonicity of post}
Let \( A, B \in \abstractdomain \) and \( l \in \locations \); then
\[
    A \subseteq B \implies \post(A, l) \subseteq \post(B, l)
\]
\end{lemma}

\begin{proof}
Let \( A, B \in \abstractdomain \) such that \( A \subseteq B \). Let \(
l, m \in \locations \).
\begin{prooftable}
TS  & \( m \in \post(B, l) \) \\
\hline
H0  & \( m \in \post(A, l) \) \\
H1  & \( A \subseteq B \) \\
H2  & \refdefinition{post function}, the post function. \\
\hline
D0  & \( (l, m) \in A \) & (H0, H2) \\
D1  & \( (l, m) \in B \) & (D0, H1) \\
\proofleaf
    & \( m \in \post(B, l) \) & (D1, H2) \\
\end{prooftable}
\end{proof}

\begin{lemma}
\summary{Monotonicity of the extended post function 1}
\label{lemma:monotonicity of the extended post 1} \\
Let \( A, B \in \abstractdomain \) and \( L \subseteq \locations \); then
\[
    A \subseteq B \implies \post(A, L) \subseteq \post(B, L).
\]
\end{lemma}

\begin{proof}
Let \( A, B \in \abstractdomain \) such that \( A \subseteq B \) and let
\( L \subseteq \locations \). Let \( m \in \locations \).
\begin{prooftable}
TS  & \( m \in \post(B, L) \) \\
\hline
H0  & \( m \in \post(A, L) \) \\
H1  & \( A \subseteq B \) \\
H2  & \refdefinition{extended prev and post functions},
      the extended post function. \\
H3  & \reflemma{monotonicity of post},
      monotonicity of the post function.  \\
\hline
D0  & \( \exists l \in L \suchthat m \in \post(A, l) \) & (H0, H2) \\
D1  & \( \exists l \in L \suchthat m \in \post(B, l) \) & (D0, H1, H3) \\
\proofleaf
    & \( m \in \post(B, L) \) & (D1, H2) \\
\end{prooftable}
\end{proof}

\begin{lemma}
\summary{Monotonicity of the extended post function 2}
\label{lemma:monotonicity of the extended post 2}
Let \( A \in \abstractdomain \) and \( L, M \subseteq \locations \); then
\[
    L \subseteq M \implies \post(A, L) \subseteq \post(A, M).
\]
\end{lemma}

\begin{proof}
Let \( A \in \abstractdomain \) an let \( L \subseteq M \subseteq
\locations \). Let \( m \in \locations \).
\begin{prooftable}
TS  & \( m \in \post(A, M) \) \\
\hline
H0  & \( m \in \post(A, L) \) \\
H1  & \( L \subseteq M \) \\
H2  & \refdefinition{extended prev and post functions},
      the extended post function. \\
\hline
D0  & \( \exists l \in L \suchthat m \in \post(A, l) \) & (H0, H2) \\
D1  & \( \exists l \in M \suchthat m \in \post(A, l) \) & (D0, H1) \\
\proofleaf
    & \( m \in \post(A, M) \) & (D1, H2) \\
\end{prooftable}
\end{proof}

\begin{lemma}
\lemmasummary{Monotonicity of the eval function}
Let \( A, B \in \abstractdomain \) and \( e \in \expressions \); then
\[
    A \subseteq B \implies \eval(A, e) \subseteq \eval(B, e).
\]
\end{lemma}

\begin{proof}
Let \( A, B \in \abstractdomain \) such that \( A \subseteq B \).
We proceed inductively on the definition of the \( \eval \)
function.
For the base case let \( l \in \locations \).
\begin{prooftable}
TS  & \( \eval(A, l) \subseteq \eval(B, l) \) \\
\hline
H0  & \refdefinition{eval function}, the eval function. \\
\hline
D0  & \( \eval(A, l) = \{ l \} \) & (H0) \\
D1  & \( \eval(B, l) = \{ l \} \) & (H0) \\
\proofleaf
    & \( \eval(A, l) \subseteq \eval(B, l) \) & (D1, D0) \\
\end{prooftable}
For the inductive case let \( e \in \expressions \).
\begin{prooftable}
TS  & \( \eval(A, \indirection e) \subseteq \eval(B, \indirection e) \) \\
\hline
H0  & \( A \subseteq B \) \\
H1  & \refdefinition{eval function}, the eval function. \\
H2  & \reflemma{monotonicity of the extended post 1},
      monotonicity of the ext.\ post 1. \\
H3  & \reflemma{monotonicity of the extended post 2},
      monotonicity of the ext.\ post 2. \\
H4  & \( \eval(A, e) \subseteq \eval(B, e) \) & (ind.\ hyp.) \\
\hline
D0  & \( \eval(A, \indirection e) = \post\bigl(A, \eval(A, e) \bigr) \) & (H1) \\
D1  & \( \eval(B, \indirection e) = \post\bigl(B, \eval(B, e) \bigr) \) & (H1) \\
D2  & \( \post\bigl(A, \eval(A, e) \bigr) \subseteq
         \post\bigl(A, \eval(B, e) \bigr) \) & (H4, H3) \\
D3  & \( \post\bigl(A, \eval(B, e) \bigr) \subseteq
         \post\bigl(B, \eval(B, e) \bigr) \) & (H0, H2) \\
D4  & \( \post\bigl(A, \eval(A, e) \bigr) \subseteq
         \post\bigl(B, \eval(B, e) \bigr) \) & (D2, D3) \\
\proofleaf
    & \( \eval(A, \indirection e) \subseteq \eval(B, \indirection e) \) & (D4, D1, D0) \\
\end{prooftable}
\end{proof}

\begin{lemma}
\lemmasummary{Monotonicity of the extended eval 1}
Let \( A, B \in \abstractdomain \), \( e \in \expressions \) and \( i \in
\naturals \); then
\[
    A \subseteq B \implies \eval(A, e, i) \subseteq \eval(B, e, i).
\]
\end{lemma}

\begin{proof}
Let \( A, B \in \abstractdomain \) such that \( A \subseteq B \).
We proceed inductively on the definition of the extended eval function.
\begin{prooftable}
H0  & \( A \subseteq B \) \\
H1  & \refdefinition{extended eval function},
      the extended eval function. \\
\end{prooftable}
For the first base case let \( l \in \locations \) and let
\( i \in \naturals \).
\begin{prooftable}
TS  & \( \eval(A, l, i + 1) \subseteq \eval(B, l, i + 1) \) \\
\hline
D0  & \( \eval(A, l, i + 1) = \emptyset \) & (H1) \\
D1  & \( \eval(B, l, i + 1) = \emptyset \) & (H1) \\
\proofleaf
    & \( \eval(A, l, i + 1) \subseteq \eval(B, l, i + 1) \) & (D0, D1) \\
\end{prooftable}
For the second base case let \( e \in expressions \).
\begin{prooftable}
TS  & \( \eval(A, e, 0) \subseteq \eval(B, e, 0) \) \\
\hline
H2  & \reflemma{Monotonicity of the eval function}, the monotonicity of eval. \\
\hline
D0  & \( \eval(A, e, 0) = \eval(A, e) \) & (H1) \\
D1  & \( \eval(B, e, 0) = \eval(B, e) \) & (H1) \\
D2  & \( \eval(A, e) \subseteq \eval(B, e) \) & (H0, H2) \\
\proofleaf
    & \( \eval(A, e, 0) \subseteq \eval(B, e, 0) \) & (D2, D1, D0) \\
\end{prooftable}
For the inductive step let \( e \in \expressions \) and let
\( i \in \naturals \).
\begin{prooftable}
TS  & \( \eval(A, \indirection e, i + 1) \subseteq \eval(B, \indirection e, i + 1) \) \\
\hline
H2  & \( \eval(A, e, i) \subseteq \eval(B, e, i) \)
    & (ind.\ hyp.) \\
\hline
D0  & \( \eval(A, \indirection e, i + 1) = \eval(A, e, i) \) & (H1) \\
D1  & \( \eval(B, \indirection e, i + 1) = \eval(B, e, i) \) & (H1) \\
D2  & \( \eval(A, e, i) \subseteq \eval(B, e, i) \) & (H2) \\
\proofleaf
    & \( \eval(A, \indirection e, i + 1) \subseteq \eval(B, \indirection e, i + 1) \)
    & (D2, D1, D0) \\
\end{prooftable}
\end{proof}

\begin{lemma}
\lemmasummary{Eval cardinality on the abstract domain}
Let \( A \in \abstractdomain \) and \( e \in \expressions \); then
\[
  \concretization(A) \neq \emptyset \implies \cardinality \eval(A, e) > 0.
\]
\end{lemma}

\begin{proof}
Let \( A \in \abstractdomain \) such that \( \concretization(A) \neq
\emptyset \). Let \( e \in \expressions \).
\begin{prooftable}
TS  & \( \cardinality \eval(A, e) > 0 \) \\
\hline
H0  & \( \concretization(A) \neq \emptyset \) \\
H1  & \refdefinition{concretization function},
      the concretization function. \\
H2  & \reflemma{Eval cardinality on the concrete domain},
      eval cardinality on the concrete domain. \\
H3  & \reflemma{Monotonicity of the eval function},
      monotonicity of the eval function. \\
\hline
D0  & \( \exists C \in \concretedomain \suchthat
          C \in \concretization(A) \) & (H0) \\
D1  & \( \exists C \in \concretedomain \suchthat
          C \subseteq A \) & (D0, H1) \\
D2  & \( \exists C \in \concretedomain \suchthat
         \eval(C, e) \subseteq \eval(A, e) \) & (D1, H3) \\
D3  & \( \exists C \in \concretedomain \suchthat
         \cardinality \eval(C, e) \leq \cardinality \eval(A, e) \) & (D2) \\
D4  & \( 1 \leq \cardinality \eval(A, e) \) & (D3, H2) \\
\proofleaf
    & \( \cardinality \eval(A, e) > 0 \) & (D4) \\
\end{prooftable}
\end{proof}

\begin{lemma}
\lemmasummary{Extended eval cardinality on the abstract domain}
Let \( A, B \in \abstractdomain \), \( e \in \expressions \) and \( i \in
\naturals \); then
\[
    \bigl(
      \concretization(A) \neq \emptyset \land
      A \subseteq B \land
      \cardinality \eval(B, e, i) > 0
    \bigr)
    \implies
      \cardinality \eval(A, e, i) > 0.
\]
\end{lemma}

\begin{proof}
Let \( A, B \in \abstractdomain \), let \( e \in \expressions \) and let
\( i \in \naturals \).
\begin{prooftable}
TS  & \( \cardinality \eval(B, e, i) > 0 \implies
         \cardinality \eval(A, e, i) > 0 \) \\
\hline
H0  & \( A \subseteq B \) \\
H1  & \( \concretization(A) \neq \emptyset \) \\
H2  & \refdefinition{extended eval function},
      the extended eval function. \\
\end{prooftable}
We proceed by induction on \( i \) and on \( e \)
(\refdefinition{expressions}).
For the base case let \( i = 0 \).
\begin{prooftable}
TS  & \( \cardinality \eval(B, e, 0) > 0 \implies
         \cardinality \eval(A, e, 0) > 0 \) \\
\hline
H3  & \reflemma{Eval cardinality on the abstract domain},
      eval cardinality on the abstract domain. \\
\hline
D0  & \( \eval(A, e, 0) = \eval(A, e) \) & (H2) \\
D1  & \( \cardinality \eval(A, e) > 0 \) & (H1, H3) \\
\proofleaf
    & \( \cardinality \eval(B, e, 0) > 0 \implies
         \cardinality \eval(A, e, 0) > 0 \) & (D1, D0) \\
\end{prooftable}
Let \( i > 0 \).
For the second base case let \( e = l \in \locations \).
\begin{prooftable}
TS  & \( \cardinality \eval(B, l, i + 1) > 0 \implies
         \cardinality \eval(A, l, i + 1) > 0 \) \\
\hline
D0  & \( \eval(B, l, i + 1) = \emptyset \) & (H2) \\
D1  & \( \cardinality \eval(B, l, i + 1) = 0 \) & (D0) \\
\proofleaf
    & \( \cardinality \eval(B, l, i + 1) > 0 \implies
         \cardinality \eval(A, l, i + 1) > 0 \) \\
\end{prooftable}
For the inductive case let \( e = \indirection f \) where \( f \in \expressions \).
\begin{prooftable}
TS  & \( \cardinality \eval(B, \indirection f, i + 1) > 0 \implies
         \cardinality \eval(A, \indirection f, i + 1) > 0 \) \\
\hline
H3  & \( \cardinality \eval(B, f, i) > 0 \implies
         \cardinality \eval(A, f, i) > 0 \) 
    & (ind.\ hyp.) \\
\hline
D0  & \( \eval(B, \indirection f, i + 1) = \eval(B, f, i) \) & (H2) \\
D1  & \( \eval(A, \indirection f, i + 1) = \eval(A, f, i) \) & (H2) \\
\proofleaf
    & \( \cardinality \eval(B, \indirection f, i + 1) > 0 \implies
         \cardinality \eval(A, \indirection f, i + 1) > 0 \)
    & (H3, D0, D1) \\
\end{prooftable}
\end{proof}

\begin{proof}
\summary{Correctness of the expression evaluation, \reftheorem{Correctness
of expression evaluation}}
Let \( A \in \abstractdomain \) and let \( e \in \expressions \).
We distinguish two cases.
First case: \( \concretization(A) = \emptyset \) then the thesis is
trivially verified.
For the second case, \( \concretization(A) \neq \emptyset \), let
\( C \in \concretization(A) \). Then we have
\begin{prooftable}
TS  & \( \eval(C, e) \subseteq \eval(A, e) \) \\
\hline
H0  & \( C \in \concretization(A) \) \\
H1  & \refdefinition{concretization function},
      the concretization function. \\
H2  & \reflemma{Monotonicity of the eval function},
      monotonicity of the eval function. \\
\hline
D0  & \( C \subseteq A \) & (H0, H1) \\
\proofleaf
    & \( \eval(C, e) \subseteq \eval(A, e) \)
    & (D0, H2) \\
\end{prooftable}
\end{proof}

\begin{lemma}
\lemmasummary{Effects of the assignment}
Let \( A \in \abstractdomain \), \( l \in \locations \) and \( e, f \in
\expressions \); for convenience of notation let \( a \in \assignments \)
be such that \( a = (e, f) \) and \( E = \eval(A, e) \).
Then
\[
    \post\bigl( \assign(A, a), l \bigr) \defeq
    \begin{cases}
      \post(A, l), & \text{if } l \not \in E; \\
      \eval(A, f), & \text{if } E = \{ l \}; \\
      \post(A, l) \union \eval(A, f),
        & \text{if } l \in E \land \cardinality E > 1.
    \end{cases}
\]
\end{lemma}

\begin{proof}
Let \( A \in \abstractdomain \), let \( (e, f) = a \in \assignments \) and
let \( l \in \locations \).
We proceed case by case.
These are our initial hypotheses.
\begin{prooftable}
H0  & \refdefinition{Assignment evaluation}, assignment definition. \\
H1  & \refdefinition{post function}, the post function. \\
\end{prooftable}
We consider separately the three cases of the lemma
\begin{gather*}
  l \not \in \eval(A, e);
    \tag{C1} \\
  \{ l \} = \eval(A, e);
    \tag{C2} \\
  l \in \eval(A, e) \land \cardinality \eval(A, e) > 1.
    \tag{C3}
\end{gather*}
First case.
\begin{prooftable}
TS  & \( \post\bigl( \assign(A, a), l \bigr) = \post(A, l) \) \\
\hline
H2  & \( l \not \in \eval(A, e) \) &  (C1) \\
\end{prooftable}
To prove TS we prove the two inclusions.
\begin{gather*}
  \post(A, l) \subseteq \post\bigl( \assign(A, a), l \bigr);
    \tag{C1.1} \\
  \post\bigl( \assign(A, a), l \bigr) \subseteq \post(A, l).
    \tag{C1.2}
\end{gather*}
Let \( m \in \locations \). For the first sub-case we have
\begin{prooftable}
TS  & \( m \in \post\bigl( \assign(A, a), l \bigr) \) \\
\hline
H3  & \( m \in \post(A, l)  \) & (C1.1) \\
\hline
D0  & \( (l, m) \in A \) & (H3, H1) \\
D1  & \( (l, m) \not \in \eval(A, e) \times \locations \) & (H2) \\
D3  & \( (l, m) \in \assign(A, a) \) & (D0, D1, H0) \\
\proofleaf
    & \( m \in \post\bigl( \assign(A, a), l \bigr) \) & (D3, H1) \\
\end{prooftable}
For the second sub-case:
\begin{prooftable}
TS  & \( m \in \post(A, l)  \) \\
\hline
H3  & \( m \in \post\bigl( \assign(A, a), l \bigr) \) & (C1.2) \\
\hline
D0  & \( (l, m) \in \assign(A, a)  \) & (H3, H1) \\
D1  & \( (l, m) \not \in \eval(A, e) \times \eval(A, f) \) & (H2) \\
D2  & \( (l, m) \in A \) & (D0, D1, H0) \\
\proofleaf
    & \( m \in \post(A, l) \) & (D2, H1) \\
\end{prooftable}
For the second and third cases we have to prove an intermediate result.
\[
  l \in \eval(A, e)
  \implies
  \eval(A, f) \subseteq \post\bigl( \assign(A, a), l \bigr).
\]
Let \( m \in \locations \).
\begin{prooftable}
TS  & \( m \in \post\bigl( \assign(A, a), l \bigr) \) \\
\hline
H2  & \( l \in \eval(A, e) \) & (C2) \\
H3  & \( m \in \eval(A, f) \) & (C2) \\
\hline
D0  & \( \eval(A, e) \times \eval(A, f) \subseteq \assign(A, a) \) & (H0) \\
D1  & \( (l, m) \in \eval(A, e) \times \eval(A, f) \) & (H2, H3) \\
D2  & \( (l, m) \in \assign(A, a) \) & (D1, D0) \\
\proofleaf
    & \( m \in \post\bigl( \assign(A, a), l \bigr) \) & (D2, H1) \\
\end{prooftable}
Note that for both the second and third case we assume that \( l \in
\eval(A, e) \) thus \( \cardinality \eval(A, e) > 0 \) so we will check
only the cases \( \cardinality \eval(A, e) = 1 \) (2nd case) and \(
\cardinality \eval(A, e) > 1 \) (3rd case).
Now the second case.
\begin{prooftable}
TS  & \( \post\bigl( \assign(A, a), l \bigr) = \eval(A, f) \) \\
\hline
H2  & \( l \in \eval(A, e) \) \\
H3  & \( \cardinality \eval(A, e) = 1 \) \\
\hline
D0  & \( \assign(A, a) = \{ l \} \times \eval(A, f)
         \union \Bigl( A \setminus
          \bigl( \{ l \} \times \locations \bigr) \Bigr) \)
    & (H3, H2, H0) \\
\end{prooftable}
Also in this case, to prove TS we prove the two inclusions
\begin{gather*}
  \post\bigl( \assign(A, a), l \bigr) \subseteq \eval(A, f);
    \tag{C2.1} \\
  \post\bigl( \assign(A, a), l \bigr) \supseteq \eval(A, f).
    \tag{C2.2}
\end{gather*}
One inclusion (C2.2) comes by modus ponens by applying the hypothesis H2 to
the previous intermediate result.
For the other inclusion (C2.1), let \( m \in \locations \); then we have
\begin{prooftable}
TS  & \( m \in \eval(A, f) \) \\
\hline
H4  & \( m \in \post\bigl( \assign(A, a), l \bigr) \) & (C2.1) \\
\hline
D1  & \( (l, m) \in \assign(A, a) \) & (H4, H1) \\
D2  & \( (l, m) \not \in A \setminus
      \bigl( \{ l \} \times \locations \bigr) \) \\
D3  & \( (l, m) \in \{ l \} \times \eval(A, f) \)
    & (D2, D1, D0) \\
\proofleaf
    & \( m \in \eval(A, f) \) & (D3) \\
\end{prooftable}
Now the third case.
\begin{prooftable}
TS  & \( \post\bigl( \assign(A, a), l \bigr) = \post(A, l) \union \eval(A, f) \) \\
\hline
H2  & \( l \in \eval(A, e) \) & (C3) \\
H3  & \( \cardinality \eval(A, e) > 1 \) & (C3) \\
\hline
D0  & \( \assign(A, a) = \eval(A, e) \times \eval(A, f) \union A \)
    & (H3, H2, H0) \\
\end{prooftable}
Again, we prove separately the two inclusions.
\begin{gather*}
  \post\bigl( \assign(A, a), l \bigr) \subseteq \post(A, l) \union \eval(A, f);
    \tag{C3.1} \\
  \post\bigl( \assign(A, a), l \bigr) \supseteq \post(A, l) \union \eval(A, f).
    \tag{C3.2}
\end{gather*}
For the inclusion (C3.2),
applying the modus ponens to the hypothesis H2 and to the above intermediate
result we have that
\[ \eval(A, f) \subseteq \post\bigl( \assign(A, a), l \bigr). \]
For the other part
\begin{prooftable}
TS  & \( \post\bigl( \assign(A, a), l \bigr) \supseteq \post(A, l) \) \\
\hline
H4  & \reflemma{Monotonicity of the eval function},
    monotonicity of eval.  \\
\hline
D1  & \( \assign(A, a) \supseteq A \) & (D0) \\
\proofleaf
    & \( \post\bigl( \assign(A, a), l \bigr) \supseteq \post(A, l) \)
    & (D1, H4) \\
\end{prooftable}
For the remaining inclusion (C3.1),
let \( m \in \locations \) so that
\( m \in \post\bigl( \assign(A, a), l \bigr) \).
We need to show that \( m \in \eval(A, f) \union \post(A, l) \) too:
to do this we show that
\( m \not \in \eval(A, f) \implies m \in \post(A, l) \).
\begin{prooftable}
TS  & \( m \in \post(A, l) \) \\
\hline
H4  & \( m \not \in \eval(A, f) \) & (C3.1) \\
H5  & \( m \in \post\bigl( \assign(A, a), l \bigr) \) & (C3.1) \\
\hline
D1  & \( (l, m) \not \in \eval(A, e) \times \eval(A, f) \) & (H4) \\
D2  & \( (l, m) \in \assign(A, a) \) & (H5, H1) \\
D3  & \( (l, m) \in A \) & (D2, D1, D0) \\
\proofleaf
    & \( m \in \post(A, l) \) & (D3, H1) \\
\end{prooftable}
\end{proof}

\begin{lemma}
\summary{Monotonicity of the assignment}
\label{lemma:monotonicity of the assignment}
Let \( A, B \in \abstractdomain \) and \( a \in \assignments \); then
\[
    \bigl( A \subseteq B \land \concretization(A) \neq \emptyset \bigr)
    \implies \assign(A, a) \subseteq \assign(B, a).
\]
\end{lemma}

\begin{proof}
Let \( A, B \in \abstractdomain \) be such that \( A \subseteq B \) and
\( \concretization(A) \neq \emptyset \). Let
\( (e, f) = a \in \assignments \).
We have to prove that \( \assign(A, a) \subseteq \assign(B, a) \).
Then let \( l, m \in \locations \) be such that \( (l, m) \in \assign(A, a) \).
To prove this lemma we have to prove that \( (l, m) \in \assign(B, a) \)
too.
Thus we have
\begin{prooftable}
TS  & \( (l, m) \in \assign(B, a) \) \\
\hline
H0  & \( A \subseteq B \) \\
H1  & \( \concretization(A) \neq \emptyset \) \\
H2  & \( (l, m) \in \assign(A, a) \) \\
H3  & \reflemma{Effects of the assignment}, effects of the assignment. \\
H4  & \reflemma{Monotonicity of the eval function}, monotonicity of eval. \\
H5  & \refdefinition{post function}, the post function. \\
H6  & \reflemma{monotonicity of the extended post 1},
      monotonicity of the extended post 1. \\
\hline
D0  & \( \eval(A, e) \subseteq \eval(B, e) \) & (H0, H4) \\
D1  & \( \eval(A, f) \subseteq \eval(B, f) \) & (H0, H4) \\
D2  & \( \post(A, l) \subseteq \post(B, l) \) & (H0, H6) \\
\end{prooftable}
We distinguish two cases.
\begin{gather*}
  l \in \eval(A, e);
    \tag{C1} \\
  l \not \in \eval(A, e).
    \tag{C2}
\end{gather*}
For the first case.
\begin{prooftable}
H7  & \( l \in \eval(A, e) \) & (C1) \\
\hline
D3  & \( l \in \eval(B, e) \) & (H7, D0) \\
D4  & \( m \in \post\bigl( \assign(A, a), l \bigr)\) & (H2, H5) \\
D5  & \( \eval(B, f) \subseteq \post\bigl( \assign(B, a), l \bigr)\)
    & (D3, H3) \\
\end{prooftable}
Note that from H7 follows that
\(
  \cardinality \eval(A, e) \geq 1
\)
and now we distinguish the two sub-cases
\begin{gather*}
  \cardinality \eval(A, e) = 1;
    \tag{C1.1} \\
  \cardinality \eval(A, e) > 1;
    \tag{C1.2}
\end{gather*}
which cover all the possibilities.  Now the first sub-case.
\begin{prooftable}
H8  & \( \cardinality \eval(A, e) = 1 \) & (C1.1) \\
\hline
D6  & \( \post\bigl( \assign(A, a), l \bigr) = \eval(A, f) \) & (H7, H8, H3) \\
D7  & \( m \in \eval(A, f) \) & (D4, D6) \\
D8  & \( m \in \eval(B, f) \) & (D7, D1) \\
D9  & \( m \in \post\bigl( \assign(B, a), l \bigr) \) & (D8, D5) \\
\proofleaf & \( (l, m) \in \assign(B, a) \) & (D9, H5) \\
\end{prooftable}
For the other sub-case
\begin{prooftable}
H8  & \( \cardinality \eval(A, e) > 1 \) & (C1.2) \\
\hline
D6  & \( \post\bigl( \assign(A, a), l \bigr) = \eval(A, f) \cup \post(A, l) \)
    & (H7, H8, H3) \\
D7  & \( \cardinality \eval(B, a) > 1 \) & (H8, D0) \\
D8  & \( \post\bigl( \assign(B, a), l \bigr) = \eval(B, f) \cup \post(B, l) \)
    & (D3, D7, H3) \\
D9  & \( \post\bigl( \assign(A, a), l \bigr) \subseteq
         \post\bigl( \assign(B, a), l \bigr) \)
    & (D6, D8, D0, D2) \\
D10  & \( m \in \post\bigl( \assign(B, e), l \bigr) \) & (D4, D9) \\
\proofleaf & \( (l, m) \in \assign(B, a) \) & (D10, H5) \\
\end{prooftable}
This completes the first case (C1).  Now the second case (C2).
\begin{prooftable}
H7  & \( l \not \in \eval(A, e) \) & (C2) \\
\hline
D3  & \( \post\bigl( \assign(A, a), l \bigr) = \post(A, l) \) & (H7, H3) \\
D4  & \( m \in \post\bigl( \assign(A, a), l \bigr) \) & (H2, H5) \\
D5  & \( m \in \post(A, l) \) & (D4, D3) \\
D6  & \( m \in \post(B, l) \) & (D5, D2) \\
\end{prooftable}
Also in the second case we distinguish two sub-cases.
\begin{gather*}
  l \not \in \eval(B, e);
    \tag{C2.1} \\
  l \in \eval(B, e).
    \tag{C2.2}
\end{gather*}
Now the first sub-case.
\begin{prooftable}
H8  & \( l \not \in \eval(B, e) \) & (C2.1) \\
\hline
D7  & \( \post\bigl( \assign(B, a), l \bigr) = \post(B, l) \) & (H8, H3) \\
D8  & \( m \in \post\bigl( \assign(B, a), l \bigr) \) & (D6, D7) \\
\proofleaf & \( (l, m) \in \assign(B, a) \) & (D8, H5) \\
\end{prooftable}
Now the other second sub-case
\begin{prooftable}
H8  & \( l \in \eval(B, e) \) & (C2.2) \\
H9  & \reflemma{Eval cardinality on the abstract domain},
      eval cardinality on the abstract domain. \\
\hline
D7  & \( \cardinality \eval(A, e) > 0 \) & (H1, H9) \\
D8  & \( \cardinality \eval(A, e) \leq \cardinality \eval(B, e) \) & (D0) \\
D9  & \( \cardinality \eval(A, e) < \cardinality \eval(B, e) \)
    & (D8, H7, H8) \\
D10 & \( \cardinality \eval(B, e) > 1 \) & (D9, D7) \\
D11 & \( B \subseteq \assign(B, a) \) & (D10, H3) \\
D12 & \( (l, m) \in B \) & (D6, H5) \\
\proofleaf & \( (l, m) \in \assign(B, a) \) & (D12, D11) \\
\end{prooftable}
\end{proof}

It is worth stressing that \( \concretization(A) \neq \emptyset \) is a
necessary hypothesis of \reflemma{monotonicity of the assignment}.
Consider indeed the following example: \( \locations = \{ l, m, n \}
\) and \( A, B \in \abstractdomain \) such that
\( A = \bigl\{ (m, n) \bigr\} \) and
\( B = \bigl\{ (l, m), (m, n) \bigr\} \). We have obviously that \( A
\subseteq B \). Consider what happens to the arc \( (m, n) \) when the
assignment \( a = (\indirection l, l) \) is performed: \( \eval(A, \indirection l) = \emptyset \)
while \( \eval(B, \indirection l) = \{ m \} \) resulting in \( \assign(A, a) = A \) and
\( \assign(B, a) = \bigl\{ (l, m), (m, l) \bigr\} \). Thus
\( \assign(A, a) \not \subseteq \assign(B, a) \).
\begin{center}
\input{figures/17}
\end{center}

\begin{proof}
\summary{Correctness of the assignment, \reftheorem{Correctness of the
assignment}}
Let \( A \in \abstractdomain \), let \( C \in \concretization(A) \) and
let \( a \in \assignments \).
\begin{prooftable}
H0  & \( C \in \concretization(A) \) \\
H1  & \refdefinition{concretization function},
      the concretization function. \\
H2  & \reflemma{monotonicity of the assignment},
      monotonicity of the assignment. \\
\hline
D0  & \( C \subseteq A \) & (H0, H1) \\
D1  & \( C \in \concretization(C) \) & (H1) \\
D2  & \( \concretization(C) \neq \emptyset \) & (D1) \\
D3  & \( \assign(C, a) \subseteq \assign(A, a) \) & (D0, D2, H2) \\
D4  & \( \assign(C, a) \in \concretization\bigl( \assign(A, a) \bigr) \)
    & (D3) \\
\end{prooftable}
\end{proof}

To proceed in the proof af the correctness of the filter abstract
operation,
now we reformulate all the previous lemmas on the post function on the prev
function.

\begin{definition}
\definitionsummary{Transposed abstract domain}
Let
\[
  \functiondef
    { \transpose }
    { \abstractdomain }
    { \abstractdomain }
\]
be defined, for all \( A \in \abstractdomain \), as
\[
    \transpose(A) = \bigl\{\, (m, l) \bigm| (l, m) \in A \,\bigr\}.
\]
\end{definition}

\begin{lemma}
\lemmasummary{Transpose is idempotent}
Let \( A \in \abstractdomain \), then
\[
    \transpose\bigl(\transpose(A)\bigr) = A.
\]
\end{lemma}

\begin{proof}
This result can be easily derived from the definition of the transpose
function (\refdefinition{Transposed abstract domain}).
\end{proof}

\begin{lemma}
\lemmasummary{Duality of prev and post}
Let \( A \in \abstractdomain \) and \( l \in \locations \); then
\begin{gather*}
    \post(A, l) = \prev\bigl( \transpose(A), l \bigr); \\
    \post\bigl( \transpose(A), l \bigr) = \prev(A, l).
\end{gather*}
\end{lemma}

\begin{proof}
Let \( A \in \abstractdomain \) and let \( l \in \locations \).
\begin{prooftable}
TS  & \( \post(A, l) = \prev\bigl( \transpose(A), l \bigr) \) \\
\hline
H0  & \refdefinition{prev function}, the prev function. \\
H1  & \refdefinition{post function}, the post function. \\
H2  & \refdefinition{Transposed abstract domain},
      transposed abstract domain. \\
\end{prooftable}
We proceed by prooving separately the two inclusions
\begin{gather*}
  \post(A, l) \subseteq \prev\bigl( \transpose(A), l \bigr);
    \tag{C1} \\
  \post(A, l) \supseteq \prev\bigl( \transpose(A), l \bigr).
    \tag{C2}
\end{gather*}
Let \( m \in \locations \).
For the first inclusion (C1)
\begin{prooftable}
TS  & \( m \in \prev\bigl( \transpose(A), l \bigr) \) \\
\hline
H3  & \( m \in \post(A, l) \) & (C1) \\
\hline
D0  & \( (l, m) \in A \) & (H3, H1) \\
D1  & \( (m, l) \in \transpose(A) \) & (D0, H2) \\
\proofleaf
    & \( m \in \prev\bigl( \transpose(A), l \bigr) \) & (D1, H0) \\
\end{prooftable}
For the second inclusion (C2)
\begin{prooftable}
TS  & \( m \in \post(A, l) \) \\
\hline
H3  & \( m \in \prev\bigl( \transpose(A), l \bigr) \) & (C2) \\
\hline
D0  & \( (m, l) \in \transpose(A) \) & (H3, H0) \\
D1  & \( (l, m) \in A \) & (D0, H2) \\
\proofleaf
    & \( m \in \post(A, l) \) & (D1, H1) \\
\end{prooftable}
The other half of this lemma can be proved observing that the transpose
function is idempotent and applying this result to the first part of the
lemma.
\begin{prooftable}
TS  & \( \post\bigl( \transpose(A), l \bigr) = \prev(A, l) \) \\
\hline
H0  & \( \post(A, l) = \prev\bigl( \transpose(A), l \bigr) \) \\
H1  & \reflemma{Transpose is idempotent}, transpose is idempotent. \\
\hline
D0  & \( \prev(A, l) =
         \prev\Bigr( \transpose\bigl( \transpose(A) \bigr), l \Bigr) \)
    & (H1) \\
D1  & \( \prev\Bigr( \transpose\bigl( \transpose(A) \bigr), l \Bigr) =
         \post\bigl( \transpose(A), l \bigr) \) & (H0) \\
\proofleaf
    & \( \post\bigl( \transpose(A), l \bigr) = \prev(A, l) \) & (D1, D0) \\
\end{prooftable}
\end{proof}

\begin{lemma}
\lemmasummary{Monotonicity of prev}
Let \( A, B \in \abstractdomain \) and \( l \in \locations \); then
\[
    A \subseteq B \implies \prev(A, l) \subseteq \prev(B, l).
\]
\end{lemma}

\begin{proof}
Let \( A, B \in \abstractdomain \) such that
\( A \subseteq B \) and let \( l \in \locations \).
\begin{prooftable}
TS  & \( \prev(A, l) \subseteq \prev(B, l) \) \\
\hline
H0  & \( A \subseteq B \) \\
H1  & \reflemma{Duality of prev and post}, the duality of prev and post. \\
H2  & \refdefinition{Transposed abstract domain}, transposed abstract domain. \\
H3  & \reflemma{monotonicity of post}, the monotonicity of post. \\
\hline
D0  & \( \prev(A, l) = \post\bigl( \transpose(A), l \bigr) \) & (H1) \\
D1  & \( \prev(B, l) = \post\bigl( \transpose(B), l \bigr) \) & (H1) \\
D2  & \( \transpose(A) \subseteq \transpose(B) \) & (H2, H0) \\
D3  & \( \post\bigl( \transpose(A), l \bigr) \subseteq
         \post\bigl( \transpose(B), l \bigr) \) & (D2, H3) \\
\proofleaf
    & \( \prev(A, l) \subseteq \prev(B, l) \) & (D3, D1, D0) \\
\end{prooftable}
\end{proof}

\begin{lemma}
\lemmasummary{Duality of extended prev and post}
Let \( A \in \abstractdomain \) and \( L \subseteq \locations \); then
\begin{gather*}
    \post(A, L) = \prev\bigl( \transpose(A), L \bigr); \\
    \post\bigl( \transpose(A), L \bigr) = \prev(A, L).
\end{gather*}
\end{lemma}

\begin{proof}
This result comes easily from the definition of the extended prev and post
functions (\refdefinition{extended prev and post functions})
applying the result of duality of prev and post
(\reflemma{Duality of prev and post}).
\end{proof}

\begin{lemma}
\lemmasummary{Monotonicity of the extended prev 1}
Let \( A, B \in \abstractdomain \) and \( L \subseteq \locations \); then
\[
    A \subseteq B \implies \prev(A, L) \subseteq \prev(B, L).
\]
\end{lemma}

\begin{proof}
Let \( A, B \in \abstractdomain \) such that
\( A \subseteq B \) and let \( L \subseteq \locations \).
\begin{prooftable}
TS  & \( \prev(A, L) \subseteq \prev(B, L) \) \\
\hline
H0  & \( A \subseteq B \) \\
H1  & \reflemma{Duality of extended prev and post},
      the duality of extended prev and post. \\
H2  & \refdefinition{Transposed abstract domain}, transposed abstract domain. \\
H3  & \reflemma{monotonicity of the extended post 1},
      the monotonicity of extended post 1. \\
\hline
D0  & \( \prev(A, L) = \post\bigl( \transpose(A), L \bigr) \) & (H1) \\
D1  & \( \prev(B, l) = \post\bigl( \transpose(B), L \bigr) \) & (H1) \\
D2  & \( \transpose(A) \subseteq \transpose(B) \) & (H2, H0) \\
D3  & \( \post\bigl( \transpose(A), L \bigr) \subseteq
         \post\bigl( \transpose(B), L \bigr) \) & (D2, H3) \\
\proofleaf
    & \( \prev(A, L) \subseteq \prev(B, L) \) & (D3, D1, D0) \\
\end{prooftable}
\end{proof}

\begin{lemma}
\lemmasummary{Monotonicity of the extended prev 2}
Let \( A \in \abstractdomain \) and \( L, M \subseteq \locations \); then
\[
    L \subseteq M \implies \prev(A, L) \subseteq \prev(A, M).
\]
\end{lemma}

\begin{proof}
Let \( A \in \abstractdomain \) and let
\( L \subseteq M \subseteq \locations \).
\begin{prooftable}
TS  & \( \prev(A, L) \subseteq \prev(A, M) \) \\
\hline
H0  & \( L \subseteq M \) \\
H1  & \reflemma{Duality of extended prev and post},
      the duality of extended prev and post. \\
H2  & \reflemma{monotonicity of the extended post 2},
      the monotonicity of extended post 2. \\
\hline
D0  & \( \prev(A, L) = \post\bigl( \transpose(A), L \bigr) \) & (H1) \\
D1  & \( \prev(A, l) = \post\bigl( \transpose(A), M \bigr) \) & (H1) \\
D2  & \( \post\bigl( \transpose(A), L \bigr) \subseteq
         \post\bigl( \transpose(A), M \bigr) \) & (H0, H2) \\
\proofleaf
    & \( \prev(A, L) \subseteq \prev(B, M) \) & (D2, D1, D0) \\
\end{prooftable}
\end{proof}

\begin{lemma}
\lemmasummary{Location closure}
Let \( A \in \abstractdomain \) and \( l \in \locations \); then
\[
    \post(A, l) \neq \emptyset \implies
    l \in \prev\bigl( A, \post(A, l) \bigr).
\]
\end{lemma}

\begin{proof}
Let \( A \in \abstractdomain \) and let \( l \in \locations \) such that
\( \post(A, l) \neq \emptyset \). Let then \( m \in \post(A, l) \).
\begin{prooftable}
TS  & \( l \in \prev\bigl( A, \post(A, l) \bigr) \) \\
\hline
H0  & \( m \in \post(A, l) \) \\
H1  & \reflemma{Monotonicity of the extended prev 2},
      monotonicity of extended prev 2. \\
H2  & \refdefinition{prev function}, the prev function. \\
H3  & \refdefinition{post function}, the post function. \\
\hline
D0  & \( \prev\bigl( A, \{ m \} \bigr) \subseteq
         \prev\bigl( A, \post(A, l) \bigr) \) & (H0, H1) \\
D1  & \( (l, m) \in A \)  & (H0, H3) \\
D2  & \( l \in \prev\bigl( A, \{ m \} \bigr) \) & (D1, H2) \\
\proofleaf
    & \( l \in \prev\bigl( A, \post(A, l) \bigr)  \) & (D2, D0) \\
\end{prooftable}
\end{proof}

\begin{lemma}
\lemmasummary{Extended location closure}
Let \( A \in \abstractdomain \) and \( L \subseteq \locations \); then
\[
    \concretization(A) \neq \emptyset \implies
    L \subseteq \prev\bigl( A, \post(A, L) \bigr).
\]
\end{lemma}

\begin{proof}
Let \( A \in \abstractdomain \)
such that \( \concretization(L) \neq \emptyset \)
and let \( L \subseteq \locations \).
If \( L = \emptyset \) then the thesis is trivially verified.
Otherwise, let \( l \in L \).
\begin{prooftable}
TS  & \( l \in \prev\bigl( A, \post(A, L) \bigr)  \) \\
\hline
H0  & \( l \in L \) \\
H1  & \( \concretization(A) \neq \emptyset \) \\
H2  & \refdefinition{concretization function},
      concretization function. \\
H3  & \reflemma{Location closure}, location closure. \\
H4  & \reflemma{monotonicity of the extended post 2},
      monotonicity of post 2. \\
H5  & \reflemma{Monotonicity of the extended prev 2},
      monotonicity of prev 2. \\
\hline
D0  & \( \post(A, l) \neq \emptyset \) & (H1, H2) \\
D1  & \( l \in \prev\Bigl( A, \post\bigr(A, \{ l \} \bigl) \Bigr) \)
    & (D0, H3) \\
D2  & \( \post\bigl( A, \{ l \} \bigr) \subseteq \post(A, L) \)
    & (H0, H4) \\
D3  & \( \prev\Bigl( A, \post\bigr(A, \{ l \} \bigl) \Bigr) \subseteq
         \prev\bigl( A, \post(A, L) \bigr) \)
    & (D2, H5) \\
\proofleaf
    & \( l \in \prev\bigl( A, \post(A, L) \bigr) \) & (D1, D3) \\
\end{prooftable}
\end{proof}

\begin{lemma}
\lemmasummary{Monotonicity of extended eval 3}
Let \( A \in \abstractdomain \), \( e \in \expressions \) and \( i \in
\naturals \); then
\[
    \concretization(A) \neq \emptyset
    \implies
    \eval(A, e, i + 1) \subseteq
    \prev\bigl( A, \eval(A, e, i) \bigr).
\]
\end{lemma}

\begin{proof}
Let \( A \in \abstractdomain \) such that
\( \concretization(A) \neq \emptyset \),
let \( e \in \expressions \) and let \( i \in \naturals \).
We proceed by induction on the definition of the extended eval function.
\begin{prooftable}
TS  & \( \eval(A, e, i + 1) \subseteq \prev\bigl( A, \eval(A, e, i) \bigr) \) \\
\hline
H0  & \refdefinition{extended eval function}, the extended eval function. \\
H1  & \( \concretization(A) \neq \emptyset \) \\
\end{prooftable}
For every \( i \) let \( e = l \in \locations \).
\begin{prooftable}
TS  & \( \eval(A, l, i + 1) \subseteq \prev\bigl( A, \eval(A, l, i) \bigr) \) \\
\hline
D0  & \( \eval(A, l, i + 1) = \emptyset \) & (H0) \\
\proofleaf
    & \( \eval(A, l, i + 1) \subseteq \prev\bigl( A, \eval(A, l, i) \bigr) \)
    & (D0) \\
\end{prooftable}
Let \( e = \indirection f \) with \( f \in \expressions \).
For \( i = 0 \)
\begin{prooftable}
TS  & \( \eval(A, \indirection f, 1) \subseteq \prev\bigl( A, \eval(A, \indirection f, 0) \bigr) \) \\
\hline
H2  & \refdefinition{eval function}, the eval function. \\
H3  & \reflemma{Extended location closure},
      the extended location closure.\\
\hline
D0  & \( \eval(A, \indirection f, 1) = \eval(A, f) \) & (H0) \\
D1  & \( \eval(A, \indirection f, 0) = \eval(A, \indirection f) \) & (H0) \\
D2  & \( \eval(A, \indirection f) = \post\bigl( A, \eval(A, f) \bigr) \)
    & (H2) \\
D3  & \( \eval(A, \indirection f, 0) = \post\bigl( A, \eval(A, f) \bigr) \)
    & (D1, D2) \\
D4  & \( \eval(A, f) \subseteq
         \prev\Bigr( A, \post\bigl( A, \eval(A, f) \bigr) \Bigr)\)
    & (H3, H1) \\
D5  & \( \eval(A, f) \subseteq
         \prev\bigl( A, \eval(A, \indirection f, 0) \bigr)\)
    & (D3, D4) \\
\proofleaf
    & \( \eval(A, \indirection f, 1) \subseteq \prev\bigl( A, \eval(A, \indirection f, 0) \bigr) \)
    & (D5, D0) \\
\end{prooftable}
For \( i > 0 \) for convenience of notation let \( i = j + 1 \).
\begin{prooftable}
TS  & \( \eval(A, \indirection f, j + 2) \subseteq
         \prev\bigl( A, \eval(A, \indirection f, j + 1) \bigr) \) \\
\hline
H2  & \( \eval(A, f, j + 1) \subseteq
         \prev\bigl( A, \eval(A, f, j) \bigr) \)
    & (ind.\ hyp.) \\
\hline
D0  & \( \eval(A, \indirection f, j + 2) = \eval(A, f, j + 1) \) & (H0) \\
D1  & \( \eval(A, \indirection f, j + 1) = \eval(A, f, j) \) & (H0) \\
\proofleaf
    & \( \eval(A, \indirection f, j + 2) \subseteq
         \prev\bigl( A, \eval(A, \indirection f, j + 1) \bigr) \)
    & (D1, D0, H2) \\
\end{prooftable}
\end{proof}

\begin{lemma}
\lemmasummary{Monotonicity of extended eval 3b}
Let \( A \in \abstractdomain \), \( e \in \expressions \) and \( i \in
\naturals \); then
\[
    \concretization(A) \neq \emptyset
    \implies
    \post\bigl( A, \eval(A, e, i + 1) \bigr) \subseteq
    \eval(A, e, i).
\]
\end{lemma}

\begin{proof}
Let \( A \in \abstractdomain \) such that
\( \concretization(A) \neq \emptyset \),
let \( e \in \expressions \) and let \( i \in \naturals \).
We proceed again by induction on the definition of the extended eval
function.
\begin{prooftable}
TS  & \( \post\bigl( A, \eval(A, e, i + 1) \bigr) \subseteq \eval(A, e, i) \) \\
\hline
H0  & \refdefinition{extended eval function}, the extended eval function. \\
H1  & \( \concretization(A) \neq \emptyset \) \\
\end{prooftable}
For every \( i \) let \( e = l \in \locations \).
\begin{prooftable}
TS  & \( \post\bigl( A, \eval(A, l, i + 1) \bigr) \subseteq \eval(A, l, i) \) \\
\hline
H2  & \refdefinition{post function}, post function. \\
\hline
D0  & \( \eval(A, l, i + 1) = \emptyset \) & (H0) \\
D1  & \( \post\bigl( A, \eval(A, l, i + 1) \bigr) = \emptyset \) & (D0, H2) \\
\proofleaf
    & \( \post\bigl( A, \eval(A, l, i + 1) \bigr) \subseteq \eval(A, l, i) \)
    & (D1) \\
\end{prooftable}
Let \( e = \indirection f \) with \( f \in \expressions \).
For \( i = 0 \)
\begin{prooftable}
TS  & \( \post\bigl( A, \eval(A, \indirection f, 1) \bigr) \subseteq \eval(A, \indirection f, 0) \) \\
\hline
H2  & \refdefinition{eval function}, the eval function. \\
H3  & \reflemma{Extended location closure},
      the extended location closure.\\
\hline
D0  & \( \eval(A, \indirection f, 1) = \eval(A, f) \) & (H0) \\
D1  & \( \post\bigl( A, \eval(A, \indirection f, 1) \bigr) =
         \post\bigl( A, \eval(A, f) \bigr) \)
    & (D0) \\
D2  & \( \post\bigl( A, \eval(A, \indirection f, 1) \bigr) = \eval(A, \indirection f) \)
    & (D1, H2) \\
\proofleaf
    & \( \post\bigl( A, \eval(A, \indirection f, 1) \bigr) = \eval(A, \indirection f, 0) \)
    & (D2, H0) \\
\end{prooftable}
For \( i > 0 \) for convenience of notation let \( i = j + 1 \).
\begin{prooftable}
TS  & \( \post\bigl( A, \eval(A, \indirection f, j + 2) \bigr) \subseteq
         \eval(A, \indirection f, j + 1) \) \\
\hline
H2  & \( \post\bigl( A, \eval(A, f, j + 1) \bigr) \subseteq
         \eval(A, f, j) \)
    & (ind.\ hyp.) \\
\hline
D0  & \( \eval(A, \indirection f, j + 2) = \eval(A, f, j + 1) \) & (H0) \\
D1  & \( \eval(A, \indirection f, j + 1) = \eval(A, f, j) \) & (H0) \\
\proofleaf
    & \( \post\bigl( A, \eval(A, \indirection f, j + 2) \bigr) \subseteq
         \eval(A, \indirection f, j + 1) \)
    & (D1, D0, H2) \\
\end{prooftable}
\end{proof}

\begin{lemma}
\lemmasummary{Monotonicity of target}
Let \( A, B \in \abstractdomain \), \( e \in \expressions \), \( L
\subseteq \locations \) and \( i, j \in \naturals \); then
\begin{multline*}
    \bigl(
      A \subseteq B \land i \leq j \land \concretization(A) \neq \emptyset
    \bigr)
    \implies
    \\
    \bigl(
      \eval(A, e, i) \subseteq \target(B, L, e, i)
      \implies
      \eval(A, e, j) \subseteq \target(B, L, e, j)
    \bigr)
\end{multline*}
\end{lemma}

\begin{proof}
Note that if \( i = j \) then the consequent of the implication in the
above statement is always true thus the thesis is trivially verified.
For the case \( i < j \) we will prove that
\begin{multline*}
    \bigl(
      A \subseteq B \land \concretization(A) \neq \emptyset
    \bigr) \implies
    \\
    \bigl(
      \eval(A, e, i) \subseteq \target(B, L, e, i)
      \implies
      \eval(A, e, i + 1) \subseteq \target(B, L, e, i + 1)
    \bigr)
\end{multline*}
as this implies, by a trivial induction on \( i \), the original result.
Let \( A, B \in \abstractdomain \) such that \( A \subseteq B \) and
\( \concretization(A) \neq \emptyset \), let \( i \in \naturals \),
\( e \in \expressions \) and \( L \subseteq \locations \).
\begin{prooftable}
TS  & \( \eval(A, e, i + 1) \subseteq \target(B, L, e, i + 1) \) \\
\hline
H0  & \( A \subseteq B \) \\
H1  & \( \eval(A, e, i) \subseteq \target(B, L, e, i) \) \\
H2  & \( \concretization(A) \neq \emptyset \) \\
H3  & \refdefinition{Target function}, the target function. \\
H4  & \reflemma{Monotonicity of extended eval 3},
      monotonicity of the extended eval 3. \\
H5  & \reflemma{Monotonicity of the extended prev 1},
      monotonicity of the extended prev 1. \\
H6  & \reflemma{Monotonicity of the extended prev 2},
      monotonicity of the extended prev 2. \\
H7  & \reflemma{Monotonicity of the extended eval 1},
      monotonicity of the extended eval 1. \\
\hline
D0  & \( \target(B, L, e, i + 1) \) \\
    & \( \quad = \eval(B, e, i + 1)
         \intersection \prev\bigl( B, \target(B, L, e, i) \bigr) \)
    & (H3) \\
D1  & \( \eval(A, e, i + 1) \subseteq \eval(B, e, i + 1) \) & (H7) \\
D2  & \( \eval(A, e, i + 1) \subseteq
         \prev\bigl( A, \eval(A, e, i) \bigr) \)  & (H2, H4) \\
D3  & \( \eval(A, e, i + 1) \subseteq
         \prev\bigl( B, \eval(A, e, i) \bigr) \)  & (D2, H0, H5) \\
D4  & \( \eval(A, e, i + 1) \subseteq
         \prev\bigl( B, \target(B, L, e, i) \bigr) \)  & (D3, H1, H6) \\
D5  & \( \eval(A, e, i + 1) \) \\
    & \( \quad \subseteq \eval(B, e, i + 1)
         \intersection \prev\bigl( B, \target(B, L, e, i) \bigr) \)
    & (D1, D4) \\
\proofleaf
    & \( \eval(A, e, i + 1) \subseteq \target(B, L, e, i + 1) \)
    & (D5, D0) \\
\end{prooftable}
\end{proof}

\begin{lemma}
\lemmasummary{Generalized correctness of filter 2}
Let \( A, B \in \abstractdomain \), \( L \subseteq \locations \) and
\( e \in \expressions \); then
\[
    \bigl(
      A \subseteq B \land
      \concretization(A) \neq \emptyset \land
      \eval(A, e) \subseteq L
    \bigr)
    \implies
    A \subseteq \filter(B, L, e).
\]
\end{lemma}

\begin{proof}
Let \( A, B \in \abstractdomain \),
let \( L \subseteq \locations \) and let \( e \in \expressions \).
We distinguish two cases
\begin{gather*}
  e = l \in \locations;
    \tag{C1} \\
  e = \indirection f \in \expressions \setminus \locations.
    \tag{C2}
\end{gather*}
For the first case (C1) let \( l \in \locations \). We have
\begin{prooftable}
TS & \( A \subseteq \filter(B, L, l) \) \\
\hline
H0  & \refdefinition{Filter 2}, filter 2. \\
H1  & \( \eval(A, l) \subseteq L \) \\
H2  & \refdefinition{eval function}, the eval function. \\
H3  & \( A \subseteq B \) \\
\hline
D0  & \( \eval(A, l) = \{ l \} \) & (H2) \\
D1  & \( l \in L \) & (D0, H1) \\
D2  & \( \eval(B, l) = \{ l \} \) & (H2) \\
D3  & \( \eval(B, l) \subseteq L \) & (D1, D2) \\
D4  & \( \filter(B, L, l) = B \) & (H0, D3) \\
\proofleaf
   & \( A \subseteq \filter(B, L, l) \) & (D4, H3) \\
\end{prooftable}
Now the second case (C2).
\begin{prooftable}
TS  & \( A \subseteq \filter(B, L, e) \) \\
\hline
H0  & \( \eval(A, e) \subseteq L \) \\
H1  & \( A \subseteq B \) \\
H2  & \( \concretization(A) \neq \emptyset \) \\
H3  & \refdefinition{Filter 2}, filter 2. \\
H4  & \refdefinition{Filter 1}, filter 1. \\
\hline
T0  & \( A \subseteq \bigintersection_{ i \in \naturals }
         \filter(B, L, e, i) \)
    & (TS, H3) \\
T1  & \( \forall i \in \naturals \itc A \subseteq \filter(B, L, e, i) \)
    & (T0) \\
T2  & \( \forall (l, m) \in A \itc
         \forall i \in \naturals \itc
         (l, m) \in \filter(B, L, e, i) \)
    & (T1) \\
\end{prooftable}
To prove TS we will prove the equivalent result T2.
Let \( (l, m) \in A \) and let \( i \in \naturals \).
\begin{prooftable}
TS  & \( (l, m) \in \filter(B, L, e, i) \) \\
\hline
H5  & \( (l, m) \in A \) \\
\end{prooftable}
We proceed by induction on \( i \).
For the base case let \( i = 0 \).
\begin{prooftable}
TS  & \( (l, m) \in \filter(B, L, e, 0) \) \\
\hline
D0  & \( \filter(B, L, e, 0) = B \) & (H4) \\
D1  & \( (l, m) \in B \)  & (H5, H1) \\
\proofleaf
    & \( (l, m) \in \filter(B, L, e, 0) \) & (D1, D0) \\
\end{prooftable}
For the inductive case let \( i > 0 \). For convenience of notation
let \( j \in \naturals \) such that \( i = j + 1 \)
\begin{prooftable}
TS  & \( (l, m) \in \filter(B, L, e, j + 1) \) \\
\hline
H6  & \( (l, m) \in \filter(B, L, e, j) \) & (hyp.\ ind.) \\
\end{prooftable}
We distinguish two cases depending on the cardinality of the target set
\begin{gather*}
  \cardinality \target(B, L, e, j+1) \neq 1; \tag{C2.1} \\
  \cardinality \target(B, L, e, j+1) = 1.    \tag{C2.2}
\end{gather*}
For the first case (C2.1), we have
\begin{prooftable}
H7  & \( \cardinality \target(B, L, e, j + 1) \neq 1 \) & (C2.1) \\
\hline
D0  & \( \filter(B, L, e, j + 1) = \filter(B, L, e, j) \)
    & (H7, H4) \\
\proofleaf
    & \( (l, m) \in \filter(B, L, e, j + 1) \) & (D0, H6) \\
\end{prooftable}
For the second case (C2.2), assume that
\begin{prooftable}
H7  & \( \cardinality \target(B, L, e, j + 1) = 1 \) & (C2.2) \\
\hline
D0  & \( \filter(B, L, e, j + 1) = \filter(B, L, e, j) \) \\
    & \( \quad \setminus \Bigl( \target(B, L, e, j + 1) \times
         \bigl( \locations \setminus \target(B, L, e, j) \bigr) \Bigr) \)
    & (H7, H4) \\
\end{prooftable}
We distinguish two sub-cases
\begin{gather*}
  l \not \in \target(B, L, e, j + 1); \tag{C2.2.1} \\
  l \in \target(B, L, e, j + 1).      \tag{C2.2.2}
\end{gather*}
For the first sub-case (C2.2.1) assume that
\begin{prooftable}
H8  & \( l \not \in \target(B, L, e, j + 1) \) & (C2.2.1) \\
\hline
D1  & \( (l, m) \not \in \target(B, L, e, j + 1) \times
         \bigl( \locations \setminus \target(B, L, e, j) \bigr) \)
    & (H8) \\
\proofleaf
    & \( (l, m) \in \filter(B, L, e, j + 1) \) \\
    & (D1, D0, H6) \\
\end{prooftable}
For the second sub-case (C2.2.2) we have
\begin{prooftable}
H8  & \( l \in \target(B, L, e, j + 1) \) & (C2.2.2) \\
\hline
H9  & \reflemma{Monotonicity of the eval function}, monotonicity of eval. \\
H10  & \refdefinition{Target function}, the target function. \\
H11 & \refdefinition{extended eval function}, the extended eval function. \\
H12 & \reflemma{Monotonicity of target}, monotonicity of target. \\
H13 & \reflemma{Extended eval cardinality on the abstract domain},
      ext.\ eval cardinality. \\
H14 & \refdefinition{post function}, the post function. \\
H15 & \reflemma{Monotonicity of extended eval 3b},
      monotonicity of ext.\ eval 3b. \\
\hline
D1  & \( \eval(A, e) \subseteq \eval(B, e) \) & (H1, H9) \\
D2  & \( \eval(A, e) \subseteq L \intersection \eval(B, e) \) & (D1, H0) \\
D3  & \( \target(B, L, e, 0) = L \intersection \eval(B, e) \) & (H10) \\
D4  & \( \eval(A, e) \subseteq \target(B, L, e, 0) \) & (D2, D3) \\
D5  & \( \eval(A, e, 0) \subseteq \target(B, L, e, 0) \) & (D4, H11) \\
D6  & \( \eval(A, e, j + 1) \subseteq \target(B, L, e, j + 1) \)
    & (D5, H1, H2, H12) \\
D7  & \( \eval(A, e, j) \subseteq \target(B, L, e, j) \)
    & (D5, H1, H2, H12) \\
D8  & \( \target(B, L, e, j + 1) = \eval(B, L, e, j + 1) \) \\
    & \( \qquad \intersection \prev\bigl( B, \target(B, L, e, j) \bigr) \)
    & (H10) \\
D9  & \( \target(B, L, e, j + 1) \subseteq \eval(B, L, e, j + 1) \)
    & (D8) \\
D10 & \( l \in \eval(B, L, e, j + 1) \) & (H8, D9) \\
D11 & \( \cardinality \eval(B, L, e, j + 1) > 0 \) & (D10) \\
D12 & \( \cardinality \eval(A, L, e, j + 1) > 0 \) & (H13, H1, H2, D11) \\
D13 & \( \{ l \} = \target(B, L, e, j + 1) \) & (H7, H8) \\
D14 & \( \{ l \} = \eval(A, e, j + 1) \) & (D13, D12, D6) \\
D15 & \( m \in \post(A, l) \) & (H5, H14) \\
D16 & \( m \in \post\bigl( A, \eval(A, e, j + 1) \bigr) \) & (D15, D14) \\
D17 & \( \post\bigl( A, \eval(A, e, j + 1) \bigr) \subseteq \eval(A, e, j) \)
    & (H2, H15) \\
D18 & \( m \in \eval(A, e, j) \) & (D16, D17) \\
D19 & \( m \in \target(B, L, e, j) \) & (D7, D18) \\
D20 & \( m \not \in \locations \setminus \target(B, L, e, j) \) & (D19) \\
D21 & \( (l, m) \not \in \target(B, L, e, j + 1) \) \\
    & \( \qquad \times \bigl( \locations \setminus \target(B, L, e, j) \bigr) \)
    & (D20) \\
\proofleaf
    & \( (l, m) \in \filter(B, L, e, j + 1) \)
    & (D21, H6, D0) \\
\end{prooftable}
\end{proof}

\begin{lemma}
\lemmasummary{Correctness of filter 2}
Let \( A \in \abstractdomain \), \( L \subseteq \locations \) and \( e \in
\expressions \); then
\[
  \forall C \in \concretization(A) \itc
    \eval(C, e) \subseteq L
    \implies
    C \in \concretization\bigl( \filter(A, L, e) \bigr).
\]
\end{lemma}

\begin{proof}
This is a simple corollary of
\reflemma{Generalized correctness of filter 2}.
Let \( A \in \abstractdomain \) and let \( C \in \concretization(A) \).
Let \( e \in \expressions \) and let \( L \subseteq \locations \) such
that \( \eval(C, e) \subseteq L \).
\begin{prooftable}
TS  & \( C \in \concretization\bigl( \filter(A, L, e) \bigr) \) \\
\hline
H0  & \( \eval(C, e) \subseteq L \) \\
H1  & \( C \in \concretization(A) \) \\
H2  & \reflemma{Generalized correctness of filter 2},
      generalized correctness of filter. \\
H3  & \refdefinition{concretization function},
      concretization function. \\
\hline
D0  & \( C \in \concretization(C) \) & (H3) \\
D1  & \( \concretization(C) \neq \emptyset \) & (D0) \\
D2  & \( C \subseteq A \) & (H1, H3) \\
D3  & \( C \subseteq \filter(A, L, e) \) & (D2, D1, H0, H2) \\
\proofleaf
    & \( C \in \concretization\bigl( \filter(A, L, e) \bigr) \)
    & (D3, H3) \\
\end{prooftable}
\end{proof}

\begin{lemma}
\lemmasummary{Equality target}
Let \( A \in \abstractdomain \) and \( e, f \in \expressions \). For
convenience of notation let \( c \in \conditions \) be such that \( c =
(\equality, e, f) \). Finally, let \( C \in \filter\bigl( \concretization(A), c
\bigr) \). Then
\[
    \eval(C, e) \subseteq \eval(A, e) \intersection \eval(A, f).
\]
\end{lemma}

\begin{proof}
Let \( A \in \abstractdomain \), let \( (\equality, e, f) = c \in
\conditions \) and let \( C \in \concretedomain \) such that
\( C \in \filter\bigl( \concretization(A), c \bigr) \).
Note that from the definition of the concrete semantics of the filter
operation (\refdefinition{Concrete filter semantics}) we have
\[
  \filter\bigl( \concretization(A), c \bigr)
  = \concretization(A) \intersection \modelset(c).
\]
Thus, \( C \in \modelset(c) \) and \( C \in \concretization(A) \).
\begin{prooftable}
TS  & \( \eval(C, e) \subseteq \eval(A, e) \intersection \eval(A, f) \) \\
\hline
H0  & \( C \in \modelset(c) \) \\
H1  & \( C \in \concretization(A) \) \\
H2  & \refdefinition{Value of conditions},
      value of conditions. \\
H3  & \refdefinition{concretization function},
      concretization function. \\
H4  & \reflemma{Monotonicity of the eval function},
      monotonicity of the eval function. \\
\hline
D0  & \( C \models c \) & (H0) \\
D1  & \( \eval(C, e) = \eval(C, f) \) & (H2, D0) \\
D2  & \( C \subseteq A \) & (H1, H3) \\
D3  & \( \eval(C, e) \subseteq \eval(A, e) \) & (D2, H4) \\
D4  & \( \eval(C, f) \subseteq \eval(A, f) \) & (D2, H4) \\
D5  & \( \eval(C, e) \subseteq \eval(A, f) \) & (D4, D1) \\
\proofleaf
    & \( \eval(C, e) \subseteq
         \eval(A, e) \intersection \eval(A, f) \) & (D4, D5) \\
\end{prooftable}
\end{proof}

\begin{lemma}
\lemmasummary{Inequality target}
Let \( A \in \abstractdomain \) and \( e, f \in \expressions \).
For convenience of notation let \( c \in \conditions \) be such that \( c =
( \inequality, e, f) \). Let \( C \in \filter\bigl(\concretization(A),
c\bigr) \) and let
\begin{gather*}
  I = \eval(A, e) \intersection \eval(A, f), \\
  E = \eval(A, e) \setminus \eval(A, f), \\
  F = \eval(A, f) \setminus \eval(A, e).
\end{gather*}
Then
\[
  \cardinality I = 1 \implies
    \eval(C, e) \subseteq E \lor \eval(C, f) \subseteq F.
\]
\end{lemma}

\begin{proof}
Let \( A \in \abstractdomain \), let \( (\inequality, e, f) = c \in
\conditions \) and let \( C \in \concretedomain \).
To show the thesis we assume that \( \cardinality I = 1\) and
\( \eval(C e) \not \subseteq E \) and then we show that \( \eval(C, f)
\subseteq F \).
\begin{prooftable}
TS  & \( \eval(C, f) \subseteq \eval(A, f) \setminus \eval(A, e) \) \\
\hline
H0  & \( C \in \concretization(A) \) \\
H1  & \( C \models c \) \\
H2  & \( \cardinality
         \bigl( \eval(A, e) \intersection \eval(A, f) \bigr) = 1 \) \\
H3  & \( \neg \bigl( \eval(C, e) \subseteq \eval(A, e) \setminus
         \eval(A, f) \bigr) \) \\
H4  & \refdefinition{concretization function}, the concretization function.\\
H5  & \reflemma{Monotonicity of the eval function}, monotonicity of eval.\\
H6  & \refdefinition{Value of conditions}, value of conditions.\\
H7  & \reflemma{Eval cardinality on the concrete domain},
      eval cardinality on \( \concretedomain \). \\
\hline
D0  & \( \neg \bigl( \eval(C, e) \subseteq \eval(A, e) \land
         \eval(C, e) \not \subseteq \eval(A, f) \bigr) \) & (H3) \\
D1  & \( \eval(C, e) \not \subseteq \eval(A, e) \lor
         \eval(C, e) \subseteq \eval(A, f) \) & (D0) \\
D2  & \( C \subseteq A \) & (H0, H4) \\
D3  & \( \eval(C, e) \subseteq \eval(A, e) \) & (D2, H5) \\
D4  & \( \eval(C, e) \subseteq \eval(A, e) \land
         \eval(C, e) \subseteq \eval(A, f) \) & (D3, D1) \\
D5  & \( \eval(C, e) \subseteq
         \eval(A, e) \intersection \eval(A, f) \) & (D4) \\
D6  & \( \cardinality \eval(C, e) = 1 \)  & (H7) \\
D7  & \( \eval(C, e) = \eval(A, e) \intersection \eval(A, f) \)
    & (H2, D6, D5) \\
D8  & \( \eval(C, e) \neq \eval(C, f) \) & (H1, H6) \\
D9  & \( \eval(C, f) \neq \eval(A, e) \intersection \eval(A, f) \)
    & (D6, D5, H2) \\
D10 & \( \cardinality \eval(C, f) = 1 \)  & (H7) \\
D11 & \( \eval(C, f) \not \subseteq
         \eval(A, e) \intersection \eval(A, f) \) & (D10, D9) \\
D12 & \( \eval(C, f) \not \subseteq \eval(A, e) \lor
         \eval(C, f) \not \subseteq \eval(A, f) \) & (D11) \\
D13 & \( \eval(C, f) \subseteq \eval(A, f) \) & (D2, H5) \\
D14 & \( \eval(C, f) \subseteq \eval(A, f) \land
         \eval(C, f) \not \subseteq \eval(A, e) \)
    & (D13, D12) \\
\proofleaf
    & \( \eval(C, f) \subseteq \eval(A, f) \setminus \eval(A, e) \)
    & (D14) \\
\end{prooftable}
\end{proof}

\begin{proof}
\summary{Correctness of the filter, \reftheorem{Correctness of the filter}}
Let \( A \in \abstractdomain \), let \( c \in \conditions \) and let
\( C \in \concretedomain \).
For convenience of notation let
\( I = \eval(A, e) \intersection \eval(A, f) \).
\begin{prooftable}
TS  & \( C \in \concretization\bigl( \filter(A, c) \bigr) \) \\
\hline
H0  & \( C \models c \) \\
H1  & \( C \in \concretization(A) \) \\
H2  & \reflemma{Correctness of filter 2}, correctness of filter 2. \\
H3  & \refdefinition{Filter 3}, filter 3. \\
H4  & \refdefinition{concretization function}, concretization function. \\
\end{prooftable}
We distinguish two cases
\begin{gather*}
  c = (\equality, e, f);   \tag{C1} \\
  c = (\inequality, e, f). \tag{C2}
\end{gather*}
For the first case (C1) let \( c = ( \equality, e, f ) \).
\begin{prooftable}
H5  & \reflemma{Equality target}, the equality target. \\
\hline
D0  & \( \eval(C, e) \subseteq I \) & (H5, H1, H0) \\
D1  & \( \eval(C, f) \subseteq I \) & (H5, H1, H0) \\
D2  & \( C \in \concretization\bigl( \filter(A, I, e) \bigr) \)
    & (D0, H1, H2) \\
D3  & \( C \in \concretization\bigl( \filter(A, I, f) \bigr) \)
    & (D1, H1, H2) \\
D4  & \( C \subseteq \filter(A, I, e) \)  & (D2, H4) \\
D5  & \( C \subseteq \filter(A, I, f) \)  & (D3, H4) \\
D6  & \( \filter(A, c) = \filter(A, I, e) \intersection \filter(A, I, f) \)
    & (H3) \\
D7  & \( C \subseteq \filter(A, I, e) \intersection \filter(A, I, f) \)
    & (D5, D4) \\
D8  & \( C \subseteq \filter(A, c) \)
    & (D7, D6) \\
\proofleaf
    & \( C \in \concretization\bigl( \filter(A, c) \bigr) \)
    & (D8, H4) \\
\end{prooftable}
For the second case (C2) let \( c = ( \inequality, e, f ) \).
We distinguish two sub-cases.
\begin{gather*}
  \cardinality I \neq 1; \tag{C2.1} \\
  \cardinality I = 1.    \tag{C2.2}
\end{gather*}
For the first sub-case (C2.1)
\begin{prooftable}
H5  & \( \cardinality I \neq 1 \) & (C2.1) \\
\hline
D0  & \( \filter(A, c) = A \) & (H3, H5) \\
\proofleaf
    & \( C \in \concretization\bigl( \filter(A, c) \bigr) \)
    & (D0, H1) \\
\end{prooftable}
In the second sub-case (C2.2) for convenience of notation let
\( E, F \subseteq \locations \) be defined as
\begin{gather*}
  E = \eval(A, e) \setminus \eval(A, f), \\
  F = \eval(A, f) \setminus \eval(A, e).
\end{gather*}
\begin{prooftable}
H5  & \( \cardinality I = 1 \)  & (C2.1) \\
H6  & \reflemma{Inequality target}, the inequality target. \\
\hline
D0  & \( \filter(A, c) = \filter(A, E, e) \union \filter(A, F, f) \)
    & (H3, H5) \\
D1  & \( \eval(C, e) \subseteq E \lor \eval(C, f) \subseteq F \)
    & (H6, H5, H1, H0) \\
D2  & \( \eval(C, e) \subseteq E \implies
         C \in \concretization\bigl( \filter(A, E, e) \bigr) \)
    & (H2, H1) \\
D3  & \( \eval(C, f) \subseteq F \implies
         C \in \concretization\bigl( \filter(A, F, f) \bigr) \)
    & (H2, H1) \\
D4  & \( C \in \concretization\bigl( \filter(A, E, e) \bigr) \lor
         C \in \concretization\bigl( \filter(A, F, f) \bigr) \)
    & (D1, D2, D3) \\
D5  & \( C \subseteq \filter(A, E, e) \lor C \subseteq \filter(A, F, f) \)
    & (D4, H4) \\
D6  & \( C \subseteq \filter(A, E, e) \union \filter(A, F, f) \)
    & (D5) \\
D7  & \( C \subseteq \filter(A, c) \)
    & (D6, D0) \\
\proofleaf
    & \( C \in \concretization\bigl( \filter(A, c) \bigr) \)
    & (D7, H4) \\
\end{prooftable}
\end{proof}

\section{Precision Limits}
\label{section:completeness}
This section presents some considerations about the precision of the
analysis; starting from questions that regard the points-to
\emph{representation}, that is, common to all points-to methods; to
questions about the specific method presented.

\subsection{Precision of the Points-To Representation}
\label{section:Precision limits points-to representation}
Reconsider now the correctness results presented in
\reftheorem{Correctness of the assignment} and \ref{theorem:Correctness of
the filter}.
Let \( A \in \abstractdomain \)
and \( a \in \assignments \), the correctness of the assignment
\[
  \concretization\bigl( \assign(A, a) \bigr) \supseteq
  \assign\bigl( \concretization(A), a \bigr),
\]
using the definition of the abstraction function
(\refdefinition{concretization function})
and \reflemma{Abstraction effect}, implies that
\begin{align*}
  \assign(A, a)
  & \supseteq
    \abstraction\Bigl( \assign\bigl( \concretization(A), c \bigr) \Bigr) \\
  & =
    \bigunion
      \Bigl\{\,
        C \Bigm|
        C \in \assign\bigl( \concretization(A), c \bigr)
      \,\Bigr\} \\
  & =
    \bigunion
      \bigl\{\,
        \assign(C,a)
      \bigm|
        C \in \concretization(A)
      \,\bigr\}.
\end{align*}
Moreover, given \( c \in \conditions \), the correctness of the filter
\[
  \concretization\bigl( \filter(A), c \bigr) \supseteq
  \filter\bigl( \concretization(A), c \bigr),
\]
using \reflemma{Abstraction effect}, implies that
\begin{align*}
  \filter(A, c)
  & \supseteq
    \abstraction\Bigl( \filter\bigl( \concretization(A), c \bigr) \Bigr) \\
  & = \bigunion
      \Bigl\{\,
        C
      \Bigm|
        C \in \filter\bigl( \concretization(A), c \bigr)
      \,\Bigr\} \\
  & = \bigunion
    \bigl\{\,
      C
    \bigm|
      C \in \concretization(A) \intersection \modelset(c)
    \,\bigr\}.
\end{align*}
Expressed in this form, the correctness results highlight the attribute
independent nature of the points-to abstract domain;
in this sense these results
provide a limit to the precision attainable.  Note that these limits,
\begin{gather*}
  \bigunion
    \bigl\{\,
      \assign(C,a)
    \bigm|
      C \in \concretization(A)
    \,\bigr\},
  \\
  \bigunion
    \bigl\{\,
      C
    \bigm|
      C \in \concretization(A) \intersection \modelset(c)
    \,\bigr\},
\end{gather*}
do not depend in any way on the definition of the abstract
operations but only on the characteristics of the abstract and concrete
domains (\refdefinition{Abstract and concrete domains}), their semantics
(\refdefinition{concretization function}) and on the concrete semantics of
the operations (\refdefinition{Concrete assignment operation} and
\ref{definition:Concrete filter semantics}).  In other words, these
are limitations of the \emph{points-to representation} and are thus
common to any method based on it. In \refsection{Precision limits of the alias query representation}
we have presented an example of the limitations of the \emph{alias query
representation}; now we show some examples of the limitations of the
\emph{points-to} representation, which is strictly less powerful.

\codecaption
{in this example two executions are possible. However, in both of them at
\refline 7 the
pointers \Variable{p} and \Variable{q} point to the same location.}
{store based abstraction limitations 1}
\begin{codesnippet}
int a, b, *p, *q;
if (...) p = &a;
else     p = &b;
// \( \eval(*p) = \{ a, b \} \)
q = p;
// \( \eval(*p) = \eval(*q) = \{ a, b \} \)
...
\end{codesnippet}

\begin{figure}
\begin{center}
\input{figures/06}
\caption
  {a representation of the points-to and alias information associated to the code in
  \refcodesnippet{store based abstraction limitations 1}.}
\label{figure:store based abstraction limitations 1}
\end{center}
\end{figure}

\begin{figure}
\begin{center}
\input{figures/06b}
\caption
  {continuation of \reffigure{store based abstraction limitations 1}.}
\label{figure:store based abstraction limitations 1b}
\end{center}
\end{figure}

\begin{example}
\label{example:points to limitation 1}
An abstract alias query is able to correctly represent when two pointers
point to the same location, also when the pointed location is not known.
The points-to representation is unable to do it,
as illustrated in \refcodesnippet{store based abstraction limitations 1}:
at \refline 7 the abstract alias query that approximates the program
is able express that in all of the possible executions, the expressions
\QuotedCode{*p} and
\QuotedCode{*q} are aliases, that is the variables \Variable{p} and \Variable{q}
point to the same location.  On the other hand, the most precise
points-to approximation cannot capture this fact.
Let \( A \in \abstractdomain \) where
\begin{gather*}
  \locations = \{ p, q, a, b \}, \\
  A = \bigl\{ (p, a), (p, b), (q, a), (q, b) \bigr\}.
\end{gather*}
\begin{center}
\input{./figures/filter_limitation_1.tex}
\end{center}
We have \( \concretization(A) = \{ C_0, C_1, C_2, C_3 \} \)
where
\begin{gather*}
  C_0 = \bigl\{ (p, a), (q, a) \bigr\}, \\
  C_1 = \bigl\{ (p, b), (q, b) \bigr\}, \\
  C_2 = \bigl\{ (p, b), (q, a) \bigr\}, \\
  C_3 = \bigl\{ (p, a), (q, b) \bigr\}.
\end{gather*}
Consider the condition
\( (\equality, \indirection p, \indirection q) = c \in \conditions \);
we have
\[
  \filter\bigl(\concretization(A), c\bigr)
    = \concretization(A) \intersection \modelset(c)
    = \{ C_0, C_1 \};
\]
but the abstraction yields
\[
  \abstraction\Bigl( \filter\bigl( \concretization(A), c \bigr) \Bigr) =
  \abstraction\bigl( \{ C_0, C_1 \} \bigr) =
  C_0 \union C_1 =
  \bigl\{ (p, a), (q, a) , (p, b), (q, b) \bigr\} = A.
\]
The \( \abstraction\bigl( \{ C_0, C_1 \} \bigr) \) is the most precise
points-to abstraction that approximates both \( C_0 \) and \( C_1 \);
however, it also approximates \( C_2 \) and \( C_3 \), which are not models
of the condition \( c \).  Again, this is due
to the fact that the points-to representation is attribute independent: in the
above example we are unable to record that when a concrete element \( C \)
is such that \( C \models c \) and \( (p, a) \in C \) then also \( (q, a)
\in C \).
This situation is also is illustrated
in Figures~\ref{figure:store based abstraction limitations 1} and
\ref{figure:store based abstraction limitations 1b}.
\end{example}

In other words, this example shows that it is not possible to define the
filter operation such that it always filters away all the
concrete points-to descriptions that are not model of the supplied condition \( c
\). In symbols:
\[
  \neg \Bigl(
  \forall A \in \abstractdomain \itc
  \forall c \in \conditions \itc
  \concretization\bigl( \filter(A, c) \bigr)
    \subseteq \modelset(c)
  \Bigr).
\]

\codecaption
{in this example two executions are possible. However, in both of them at
\refline{10} the expression
\Variable{**r} is an alias of \Variable{b}.  }
{store based abstraction limitations 2}
\begin{codesnippet}
int a, b, c, *p, *q, **r;
p = &a;
// \( \eval(*p) = \{ a \} \)
q = &c;
// \( \eval(*q) = \{ c \} \)
if (...) r = &p;
else     r = &q;
// \( \eval(*r) = \{ p, q \} \)
*r = &b;
// \( \eval(*p) = \{ a, b \} \)
// \( \eval(*q) = \{ b, c \} \)
\end{codesnippet}

\begin{figure}
\begin{center}
\input{figures/01}
\caption
  {a representation of the points-to information before and after the
  execution of \refline{9} in \refcodesnippet{store based abstraction
  limitations 2}.}
\label{figure:store based abstraction limitations 2}
\end{center}
\end{figure}

\begin{figure}
\begin{center}
\input{figures/04}
\caption
  {this example shows how, in the code of \refcodesnippet{store based
  abstraction limitations 2}, the points-to representation fails to
  describe
  that \Variable{**r} is alias of \Variable{b} on all of the possible
  executions.}
\label{figure:store based abstraction limitations 2b}
\end{center}
\end{figure}

\begin{figure}
\begin{center}
\input{figures/04b}
\caption
  {continuation of \reffigure{store based abstraction limitations 2b}}
\label{figure:store based abstraction limitations 2bb}
\end{center}
\end{figure}

\begin{figure}
\begin{center}

\input{figures/04c}

\caption
  {continuation of \reffigure{store based abstraction limitations 2bb}}
\label{figure:store based abstraction limitations 2bc}
\end{center}
\end{figure}

\begin{example}
The points-to representation keeps track only of the relations between
pointers and pointed objects that span exactly one level of indirection.
For example, in \refcodesnippet{store based
abstraction limitations 2}, the points-to representation is unable to
\emph{natively} express that \Variable{**r} is an alias of \Variable{b},
this information ---though present in the \emph{complete} alias relation---
is \emph{inferred} from the points-to pairs explicitly memorized
by applying the transitive
property:  it is known that \Variable{r} points to
\Variable{p} and that \Variable{p} points to \Variable{b};
then it can be deduced
that \Variable{*r} points to \Variable{b}.  But this step causes a loss of
accuracy when there are more intermediate variables (\reffigure{store
based abstraction limitations 2}). The alias query representation
is able describe that after the execution of \refline{9}, the expression
\QuotedCode{**r} is \emph{definitely} an alias of \QuotedCode{b}, whereas
the points-to representation fails to do it.
Let \( A \in \abstractdomain \) such that
\begin{gather*}
  \locations = \{ p, q, r, a, b, c \}, \\
  A = \bigl\{ (r, p), (r, q), (p, a), (q, c) \bigr\}.
\end{gather*}
We have that \( \{ C_0, C_1 \} = \concretization(A) \)
where
\begin{gather*}
  C_0 = \bigl\{ (r, p), (p, a), (q, c) \bigr\}, \\
  C_1 = \bigl\{ (r, q), (p, a), (q, c) \bigr\}.
\end{gather*}
Let \( (\indirection r, b) = x \in \assignments \).
Performing the assignment on the elements found in the concretization of
\( A \) we obtain
\begin{gather*}
  \assign(C_0, x) = \bigl\{ (r, p), (p, b), (q, c) \bigr\}, \\
  \assign(C_1, x) = \bigl\{ (r, q), (p, a), (q, b) \bigr\}.
\end{gather*}
Computing the abstraction of the result of the concrete operation we find
\begin{align*}
  \abstraction\Bigl( \assign\bigl( \concretization(A), x \bigr) \Bigr)
  & = \abstraction\Bigl( \bigl\{ \assign(C_0, x), \assign(C_1, x) \bigr\} \Bigr) \\
  & = \assign(C_0, x) \union \assign(C_1, x) \\
  & = A \union \bigl\{ (p, b), (q, b) \bigr\}.
\end{align*}
Let
\[
  C_3 = \bigl\{ (r, p), (p, a), (q, b) \bigr\}
  \subseteq
  \abstraction\Bigl( \assign\bigl( \concretization(A), x \bigr) \Bigr);
\]
note that
\begin{prooftable}
H0  & \reftheorem{Correctness of the assignment}, correctness of the assignment. \\
H1  & \reflemma{Monotonicity of the concretization function}, monotonicity of the concretization function. \\
H2  & \reflemma{Abstraction effect}, the abstraction effect. \\
H3  & \refdefinition{Abstraction function}, the abstraction function. \\
\hline
D0  & \( \assign\bigl(\concretization(A), x\bigr) \subseteq
         \concretization\bigl(\assign(A, x) \bigr) \) & (H0) \\
D1  & \( \abstraction\Bigl( \assign\bigl(\concretization(A), x\bigr) \Bigr) \subseteq
         \abstraction\Bigl( \concretization\bigl(\assign(A, x) \bigr) \Bigr) \) & (D0, H3) \\
D2  & \( \abstraction\Bigl( \concretization\bigl(\assign(A, x) \bigr) \Bigr)
         \subseteq \assign(A, x) \) & (H2) \\
D3  & \( \abstraction\Bigl( \assign\bigl(\concretization(A), x\bigr) \Bigr)
         \subseteq \assign(A, x) \) & (D1, D2) \\
\proofleaf
    & \( \concretization\biggl(\abstraction\Bigl(\assign\bigl(\concretization(A),a\bigr) \Bigr) \biggr)
         \subseteq \concretization\bigl(\assign(A, x)\bigr) \) & (D3, H1) \\
\end{prooftable}
then
\begin{gather*}
  C_3 \in \concretization\bigl( \assign(A, a) \bigr); \\
\intertext{but}
  C_3 \neq \assign(C_0, a); \\
  C_3 \neq \assign(C_1, a);
\end{gather*}
that is, there exist no concrete elements \( C \in
\concretization(A) \) such that \( C_3 = \assign(C, a) \). Again this
inaccuracy is due to the lack of relational information in the points-to
representation: in this example, given a concrete element \( C \in
\assign(A, a) \), we are unable to tell that if \( (r, p) \in C \) then
\( (p, b) \in C \) and \( (q, b) \not \in C \).
The situation just described is illustrated in
Figures~\ref{figure:store based abstraction limitations 2b},
\ref{figure:store based abstraction limitations 2bb} and
\ref{figure:store based abstraction limitations 2bc}.
\end{example}

In other words, this example shows that it is not possible to formulate the
assignment operation in such a way that each concrete element approximated
by \( \assign(A, a) \) can be expressed as the result of the
concrete assignment performed on one
of the elements of \( \concretization(A) \).
In symbols
\begin{multline*}
  \neg \Bigl(
  \forall A \in \abstractdomain \itc
  \forall a \in \assignments \itc \\
    \forall C \in \concretization\bigl( \assign(A, a) \bigr) \itc
    \exists D \in \concretization(A) \suchthat
      C = \assign(D, a)
  \Bigr).
\end{multline*}

\subsection{Precision of the Presented Method}
The two examples introduced above present a limitation of the form --- all points-to
based methods are \emph{not enough precise} to capture \emph{this} fact.
In terms of the partial order of the domain
this can be seen as a \emph{lower limit} to the precision
attainable with points-to based methods.
On the other hand it is also interesting to find out what are the
precision \emph{upper limits} of the proposed method, i.e., statements of
the form --- the given points-to based method \emph{is enough
precise} to capture \emph{that} fact.
In particular, we want to analyze the situation of the presented method with
respect to the limitations of the points-to representation, that is
whether or not the inclusions in
\reftheorem{Correctness of the assignment} and \ref{theorem:Correctness of
the filter} are also equalities, i.e., if it holds that, for all \( A \in
\abstractdomain \), \( e \in \expressions \), \( a \in \assignments \) and
\( c \in \conditions \)
\begin{gather*}
    \bigunion
      \bigl\{\,
        \eval(C, e)
      \bigm|
        C \in \concretization(A)
      \,\bigr\}
    \supseteq \eval(A, e);
  \\
    \assign\bigl( \concretization(A), a \bigr)
    \supseteq
    \concretization\bigl(\assign(A, a)\bigr);
  \\
    \filter\bigl( \concretization(A), c \bigr)
    \supseteq \concretization\bigl(\filter(A, c)\bigr).
\end{gather*}
From the characterization presented in \refsection{Precision limits
points-to representation} these can be rewritten to stress the
attribute independent nature of the points-to representation, i.e., by focusing on
the single \emph{arcs} instead of the whole points-to relation. Let \( A \in
\abstractdomain \) such that
\( \concretization(A) \neq \emptyset \),\footnote{Note that the additional hypothesis, \(
\concretization(A) \neq \emptyset \), is required by \reflemma{Abstraction
effect} to prove the opposite of the inclusions used for the correctness
results.} then we have
\begin{gather*}
  \forall l \in \eval(A, e) \itc
  \exists C \in \concretization(A) \suchthat
    \eval(C, e) = \{ l \},
  \\
    \forall (l, m) \in \assign(A, a) \itc
    \exists C \in \concretization(A) \suchthat
      (l, m) \in \assign(C, a),
  \\
  \forall (l, m) \in \filter(A, c) \itc
  \exists C \in \concretization(A) \suchthat
    C \in \modelset(c) \land (l, m) \in C,
\end{gather*}
respectively.
Unfortunately, for all these cases there exists a counterexample.

\subsubsection{The Abstract Evaluation Is Not Optimal}
The following example highlights that the abstract evaluation function
(\refdefinition{eval function}) is not optimal with respect to the
points-to representation, i.e., there exists \( A \in \abstractdomain \),
\( \concretization(A) \neq \emptyset \) and \( e \in \expressions \)
such that
\[
    \eval(A, e) \setminus
    \bigunion
      \bigl\{\,
        \eval(C, e)
      \bigm|
        C \in \concretization(A)
      \,\bigr\} \neq \emptyset.
\]

\begin{example}
\label{example:eval is not optimal}
Let \( A \in \abstractdomain \) such that
\begin{gather*}
  \locations = \{ a, b, c \}, \\
  A = \bigl\{ (a, a), (a, b), (b, c) \bigr\}.
\end{gather*}
We have that \( \{ C_1, C_2 \} = \concretization(A) \) where
\begin{gather*}
  C_1 = \bigl\{ (a, a), (b, c) \bigr\}, \\
  C_2 = \bigl\{ (a, b), (b, c) \bigr\}.
\end{gather*}
Consider the expression \( e = \indirection \indirection a \).
Performing the evaluation of \( e \)
as described in \refdefinition{eval function} we obtain
\begin{filtertable}
\( i \)
  & \( \eval(C_1, e, i) \)
  & \( \eval(C_2, e, i) \)
  & \( \eval(A, e, i) \) \\
\hline
2 & \( \{ a \} \) & \( \{ a \} \) & \( \{ a \} \) \\
1 & \( \{ a \} \) & \( \{ b \} \) & \( \{ a, b \} \) \\
0 & \( \{ a \} \) & \( \{ c \} \) & \( \{ a, b, c \} \) \\
\end{filtertable}
Note that \( b \in \eval(A, e) \) but there exist no
\( C \in \concretization(A) \) such that \( \{ b \} = \eval(C, e) \),
indeed
\[
  \bigunion_{ C \in \concretization(A) } \eval(C, e) = \{ a, c \}.
\]
The spurious location \( b \) in the result of the evaluation of the
expression \( e \) in \( A \) is due to the fact that the formulation of
the abstract evaluation does not exploit that in a concrete points-to
description a location can point to only one location; in this case there
exists no \( C \in \concretization(A) \) such that \( \bigl\{ (a, a), (a,
b) \bigr\} \subseteq C \).
A graphical representation of this example is reported in Figures~\ref{figure:eval
is not optimal} and \ref{figure:eval is not optimal 2}.
\end{example}

\begin{figure}
\begin{center}
\input{figures/08}
\caption
  {the abstract evaluation function is not optimal.}
\label{figure:eval is not optimal}
\end{center}
\end{figure}

\begin{figure}
\begin{center}
\input{figures/16}
\caption
  {the evaluation process of the expression \( \indirection \indirection a
  \) on the memory \( A \) of \refexample{eval is not optimal}. An optimal
  evaluation function would not follow the arc \( (a, b) \) between \( i = 1 \)
  and \( i = 0 \) (the dashed arc in the figure).}
\label{figure:eval is not optimal 2}
\end{center}
\end{figure}

\subsubsection{The Abstract Assignment Is Not Optimal}
We present another example that highlights how the abstract assignment
operation formulated in \refdefinition{Assignment evaluation} is not
optimal for the points-to representation, i.e.,
there exists \( A \in \abstractdomain \),
\( \concretization(A) \neq \emptyset \), \( a \in \assignments \) such that
\[
    \assign(A, a) \setminus
    \assign\bigl( \concretization(A), a\bigr) \neq \emptyset.
\]
Note that this limitation is still true also assuming to have an optimal
abstract evaluation function.

\begin{example}
\label{example:filter incompleteness}
Let \( A \in \abstractdomain \) such that
\begin{gather*}
  \locations = \{ a, b, c \}, \\
  A = \bigl\{ (a, b), (a, c) \bigr\}.
\end{gather*}
We have that \( \{ C_1, C_2 \} = \concretization(A) \) where
\begin{gather*}
  C_1 = \bigl\{ (a, b) \bigr\}, \\
  C_2 = \bigl\{ (a, c) \bigr\}.
\end{gather*}
Let \( (*a, *a) = x \in \assignments \).
Performing the assignment \( x \) on the elements
of \( \concretization(A) \) we obtain
\begin{gather*}
  \assign(C_1, x) = \bigl\{ (a, b), (b, b) \bigr\}, \\
  \assign(C_2, x) = \bigl\{ (a, c), (c, c) \bigr\}.
\end{gather*}
Computing the abstraction of the result of the concrete operation we find
\begin{align*}
  \abstraction\Bigl( \assign\bigl( \concretization(A), x \bigr) \Bigr)
  & = \abstraction\Bigl( \bigl\{ \assign(C_1, x), \assign(C_2, x) \bigr\} \Bigr) \\
  & = \assign(C_1, x) \union \assign(C_2, x) \\
  & = A \union \bigl\{ (b, b), (c, c) \bigr\}.
\end{align*}
Note that
performing the abstract evaluation of the lhs and the rhs of the
assignment as described in \refdefinition{eval function} yields
\(
  \eval(A, \indirection a) = \{ b, c \},
\)
which is the most precise result possible for the abstract
evaluation of the expression \( \indirection a \), indeed
\begin{align*}
  \bigunion
    \bigl\{\,
      \eval(C, a)
    \bigm|
      C \in \concretization(A)
    \,\bigr\} & =
  \eval(C_1, \indirection a) \union \eval(C_2, \indirection a) \\
    & = \{ b \} \union \{ c \} = \{ b, c \} \\
    & = \eval(A, \indirection a).
\end{align*}
In this case, the abstract assignment (\refdefinition{Assignment evaluation}) yields
\begin{align*}
  \assign(A, x)
  & = A \union \eval(A, \indirection a) \times \eval(A, \indirection a) \\
  & = A \union \{ b, c \} \times \{ b, c \} \\
  & = A \union \bigl\{ (b, b), (b, c), (c, b), (c, c) \bigr\}.
\end{align*}
Note that
\[
  \assign(A, x) \setminus \assign\bigl( \concretization(A), x \bigr)
  = \bigl\{ (b, c), (c, b) \bigr\}.
\]
The arcs \( \bigl\{ (b, c), (c, b) \bigr\} \) do not correspond to
any concrete assignment: they are artifacts of this abstraction.
But note
that in this case the inaccuracy cannot be ascribed to the abstract
evaluation of the expressions that, in this case, exposes an optimal
behaviour.
The problem is that the evaluation of the rhs and the lhs for the
assignment are not related each other: this way it becomes possible that
the lhs evaluates to \QuotedCode{b} and the rhs evaluates to \QuotedCode{c}
---thus generating the spurious arc \( (b,c) \)--- also when the rhs and
the lhs are the \emph{same} expression.
This example is illustrated in
\reffigure{assignment is not optimal 2}.
\end{example}

\begin{figure}
\begin{center}
\input{figures/09}
\caption
  {the abstract assignment formulation is not optimal.}
\label{figure:assignment is not optimal 2}
\end{center}
\end{figure}

\subsubsection{The Abstract Filter Is Not Optimal}
Finally, we report an example that shows the same inaccuracy in the filter operation,
i.e., there exists \( A \in \abstractdomain \), \( \concretization(A) \neq
\emptyset \), and \( c \in \conditions \) such that
\[
  \filter(A, c) \setminus
  \filter\bigl( \concretization(A), c\bigr) \neq \emptyset.
\]

\begin{example}
\label{example:filter is not optimal}
Let \( A \in \abstractdomain \) such that
\begin{gather*}
  \locations = \{ a, b \}, \\
  A = \bigl\{ (a, a), (a, b), (b, b) \bigr\}.
\end{gather*}
We have that \( \{ C_1, C_2 \} = \concretization(A) \) where
\begin{gather*}
  C_1 = \bigl\{ (a, a), (b, b) \bigr\}, \\
  C_2 = \bigl\{ (a, b), (b, b) \bigr\}.
\end{gather*}
Consider now the condition \( (**a, b) = c \in \conditions \).  Since
\begin{gather*}
  \eval(C_1, \indirection \indirection a) = \{ a \}, \\
  \eval(C_2, \indirection \indirection a) = \{ b \},
\end{gather*}
only \( C_2 \) satisfies \( c \), i.e.
\[
  \filter\bigl( \concretization(A), c \bigl) =
  \concretization(A) \intersection \modelset(c) = \{ C_2 \}.
\]
Performing the filter operation as described in \refdefinition{Filter 3} on \( A \)
we do not improve the precision, that is \( \filter(A, c) = A \).
\begin{filtertable}
\( i \)
  & \( \eval(A, \indirection \indirection a, i) \),
  & \( \target(A, \indirection \indirection a, i) \),
  & Removed arcs \\
\hline
2 & \( \{ a \} \) & \( \{ a \} \)  & \( \emptyset \) \\
1 & \( \{ a, b \} \) & \( \{ a, b \} \) & \( \emptyset \) \\
0 & \( \{ a, b \} \) & \( \{ b \} \)  & \( \emptyset \) \\
\end{filtertable}
Then note that
\[
  \filter(A, c) \setminus
  \abstraction\Bigl( \filter\bigl( \concretization(A), c\bigr) \Bigr)
  = A \setminus C_2
  = \bigl\{ (a, a) \bigr\}.
\]
This means that the filter is unable to remove the spurious arc \( (a, a)
\).
A graphical representation of this situation is presented in
\reffigure{filter is not optimal}, while \reffigure{filter is not optimal
b} present a graphical representation of the filter computation.
\end{example}

\begin{figure}
\begin{center}
\input{figures/10}
\caption
  {the abstract filter formulation is not optimal.}
\label{figure:filter is not optimal}
\end{center}
\end{figure}

\begin{figure}
\begin{center}
\input{figures/15}
\caption
  {in \refexample{filter is not optimal} the filter is unable to remove the spurious
  arc \( (a, a) \).}
\label{figure:filter is not optimal b}
\end{center}
\end{figure}

Though the current formulation of the filter
operation is not optimal, in the next example we show that
iterating the application of the filter on the same condition it is
possible to refine the points-to approximation.

\begin{example}
\label{example:iterating filter 1}
Let \( A \in \abstractdomain \) such that
\begin{gather*}
  \locations = \{ a, b, c \}, \\
  A = \bigl\{ (a, a), (a, b), (a, c), (b, c), (c, a) \bigr\}.
\end{gather*}
\begin{center}

\input{figures/filter_iteration}
\end{center}
Consider the condition \( x = ( \equality, \indirection \indirection a, c )
\in \conditions \).
From the definition of the evaluation funtion (\refdefinition{eval
function}), we have
\begin{gather*}
  \eval(A, \indirection \indirection a) = \{ a, b, c \}, \\
  \eval(A, c) = \{ c \}.
\end{gather*}
From the filter definition (\refdefinition{Filter 3}) we have
\begin{gather*}
  I = \eval(A, \indirection \indirection a) \intersection \eval(A, c)
    = \{ c \}, \\
  \filter(A, x)
  = \filter(A, I, \indirection \indirection a) \intersection
    \filter(A, I, c)
  = \filter\bigl( A, \{ c \}, \indirection \indirection a \bigr)
    \intersection
    \filter\bigl(A, \{ c \}, c \bigr).
\end{gather*}
We consider only the lhs
\( \filter\bigl( A, \{ c \}, \indirection \indirection a \bigr) \)
as, from the definition of the filter 2, it is clear that filtering on the
rhs does not improve the precision of the
approximation, that is, \( \filter\bigl(A, \{ c \}, c \bigr) = A \).
We have
\begin{filtertable}
\( i \) & \( \eval(A, \indirection \indirection a) \)
        & \( \target\bigl(A, \{ c \}, \indirection \indirection a, i \bigr)\)
        & \( \filter\bigl(A, \{ c\}, \indirection \indirection a, i\bigr) \) \\
\hline
2 & \( \{ a \} \)       & \( \{ a \} \)
  & \( A \setminus \bigl\{ (a, c) \bigr\} \) \\
1 & \( \{ a, b, c \} \) & \( \{ a, b \} \) & \( A \) \\
0 & \( \{ a, b, c \} \) & \( \{ c \} \)    & \( A \) \\
\end{filtertable}
That is, from the first application of the filter we can remove the spurious
arc \( (a, c) \).
Now we proceed applying the filter again.
Let \( B = \filter\bigl(A, \{ c\}, \indirection \indirection a, i\bigr)
= A \setminus \bigl\{ (a, c) \bigr\} \).
We have
\begin{filtertable}
\( i \) & \( \eval(B, \indirection \indirection a) \)
        & \( \target\bigl(B, \{ c \}, \indirection \indirection a, i \bigr)\)
        & \( \filter\bigl(B, \{ c\}, \indirection \indirection a, i\bigr) \) \\
\hline
2 & \( \{ a \} \)       & \( \{ a \} \)
  & \( B \setminus \bigl\{ (a, a) \bigr\} \) \\
1 & \( \{ a, b \} \) & \( \{ b \} \) & \( B \) \\
0 & \( \{ a, b, c \} \) & \( \{ c \} \)    & \( B \) \\
\end{filtertable}
Note that in the second application of the filter we are able to remove
another arc, \( (a, a) \), that it was not removed during the first iteration.
\begin{center}
\input{figures/filter_iteration_2}
\end{center}
\end{example}

\subsubsection{Another Consideration on the Precision of the Filter
Operation}
It is possible to show that the formulation of the abstract filter
operation does not generate spurious memory descriptions not already
present in the initial approximation, i.e., for all \( A \in
\abstractdomain \) and \( c \in \conditions \)
\[
  \concretization\bigl( \filter(A, c) \bigr)
    \subseteq \concretization(A).
\]
Note that
by composing this result with the result of correctness for the filter
(\reftheorem{Correctness of the filter}) it is possible to write
\[
  \concretization(A) \intersection \modelset(c)
    \subseteq \concretization\bigl( \filter(A, c) \bigr)
    \subseteq \concretization(A).
\]
Basically, the filter never adds new arcs then
it is not possible to obtain a worse approximation of that given in input.
Though the idea is quite simple, for completeness we report a formal proof.

\begin{lemma}
\lemmasummary{Filter upper bound 1}
Let \( A \in \abstractdomain \), \( M \subseteq \locations \), \( e \in
\expressions \) and \( n \in \naturals \); then
\[
    \filter(A, M, e, n) \subseteq A.
\]
\end{lemma}

\begin{proof}
Let \( A \in \abstractdomain \), \( M \subseteq \locations \), \( e \in
\expressions \) and \( n \in \naturals \).
\begin{prooftable}
TS  & \( \filter(A, M, e, n) \subseteq A \) \\
\hline
H0  & \refdefinition{Filter 1}, filter 1. \\
\end{prooftable}
We proceed inductively on \( n \).
For the first case we assume \( n = 0 \).
\begin{prooftable}
D0  & \( \filter(A, M, e, 0) = A \) & (H0) \\
\proofleaf
    & \( \filter(A, M, e, 0) \subseteq A \) & (D0) \\
\end{prooftable}
Now the inductive case.
\begin{prooftable}
H1  & \( \filter(A, M, e, n) \subseteq A \) & (ind.\ hyp.) \\
\hline
D0  & \( \filter(A, M, e, n + 1) = \filter(A, M, e, n) \setminus \ldots \)
    & (H0) \\
D1  & \( \filter(A, M, e, n + 1) \subseteq \filter(A, M, e, n) \) & (D0) \\
\proofleaf
    & \( \filter(A, M, e, n + 1) \subseteq A \) & (D1, H1) \\
\end{prooftable}
\end{proof}

\begin{lemma}
\lemmasummary{Filter upper bound 2}
Let \( A \in \abstractdomain \), \( M \subseteq \locations \), \( e \in
\expressions \); then
\[
    \filter(A, M, e) \subseteq A.
\]
\end{lemma}

\begin{proof}
Let \( A \in \abstractdomain \), \( M \subseteq \locations \), \( e \in
\expressions \).
Following the definition of the filter 2 (\refdefinition{Filter 2}) we
consider separately two cases.
For the first case let \( e = l \in \locations \); if \( \eval(A, l) \in M
\) then we have \( \filter(A, M, l) = A \), otherwise \( \filter(A, M, l) =
\bot \). In both the cases we have the thesis.  For the second case let \(
e \in \expressions \setminus \locations \).
\begin{prooftable}
TS  & \( \filter(A, M, e) \subseteq A \) \\
\hline
H0  & \refdefinition{Filter 2}, filter 2. \\
H1  & \reflemma{Filter upper bound 1}, filter upper bound 1. \\
\hline
D0  & \( \filter(A, M, e) = \bigintersection_{ n \in \naturals }
         \filter(A, M, e, n) \) & (H0) \\
D1  & \( \forall n \in \naturals \itc \filter(A, M, e, n) \subseteq A \)
    & (H1) \\
D2  & \( \bigintersection_{ n \in \naturals } \filter(A, M, e, n)
         \subseteq A \) & (D1) \\
\proofleaf
    & \( \filter(A, M, e) \subseteq A \) & (D2, D0) \\
\end{prooftable}
\end{proof}

\begin{lemma}
\lemmasummary{Filter upper bound 3}
Let \( A \in \abstractdomain \) and \( c \in \conditions \); then
\[
    \concretization\bigl( \filter(A, c) \bigr)
    \subseteq \concretization(A).
\]
\end{lemma}

\begin{proof}
Let \( A \in \abstractdomain \) and let \( c \in \conditions \).
\begin{prooftable}
TS  & \( \filter(A, c) \subseteq A \) \\
\hline
H0  & \refdefinition{Filter 3}, filter 3. \\
H1  & \refdefinition{concretization function}, concretization function. \\
H2  & \reflemma{Filter upper bound 2}, filter upper bound 2. \\
\end{prooftable}
As in the definition of the filter (\refdefinition{Filter 3}) we
distinguish two cases
\begin{gather*}
  c = ( \equality, e, f);   \tag{C1} \\
  c = ( \inequality, e, f). \tag{C2}
\end{gather*}
For the first case (C1) we have
\begin{prooftable}
D0  & \( \filter\bigl(A, (\equality, e, f) \bigr) =
         \filter\bigl( A, \eval(A, e) \intersection \eval(A, f), e \bigr) \) \\
    & \( \qquad \intersection
         \filter\bigl(A, \eval(A, e) \intersection \eval(A, f), f \bigr) \)
    & (H0) \\
D1  & \( \filter\bigl( A, \eval(A, e) \intersection \eval(A, f), e \bigr)
         \subseteq A \)
    & (H2) \\
D2  & \( \filter\bigl(A, \eval(A, e) \intersection \eval(A, f), f \bigr)
         \subseteq A \)
    & (H2) \\
D3  & \( \filter\bigl( A, \eval(A, e) \intersection \eval(A, f), e \bigr) \) \\
    & \( \qquad \intersection
         \filter\bigl(A, \eval(A, e) \intersection \eval(A, f), f \bigr)
         \subseteq A \)
    & (D1, D2) \\
\proofleaf
    & \( \filter\bigl(A, (\equality, e, f) \bigr) \subseteq A \)
    & (D3, D0) \\
\end{prooftable}
Now the second case (C2).
If \( \cardinality \bigl( \eval(A, e) \intersection \eval(A, f) \bigr) \neq
1 \) from H0 we have that \( \filter\bigl(A, (\inequality, e, f) \bigr) = A
\) then the thesis is trivially verified.
Otherwise assume
\( \cardinality \bigl( \eval(A, e) \intersection \eval(A, f) \bigr) = 1 \).
Then we have
\begin{prooftable}
H3  & \( \cardinality \bigl( \eval(A, e) \intersection \eval(A, f) \bigr) =
          1 \) \\
\hline
D0  & \( \filter\bigl(A, (\inequality, e, f) \bigr) =
        \filter\bigl( A, \eval(A, e) \setminus \eval(A, f), e \bigr) \) \\
    & \( \qquad \union
        \filter\bigl( A, \eval(A, f) \setminus \eval(A, e), f \bigr) \)
    & (H3, H0) \\
D1  & \( \filter\bigl( A, \eval(A, e) \setminus \eval(A, f), e \bigr)
         \subseteq A \)
    & (H2) \\
D2  & \( \filter\bigl( A, \eval(A, f) \setminus \eval(A, e), f \bigr)
         \subseteq A \)
    & (H2) \\
D3  & \( \filter\bigl( A, \eval(A, e) \setminus \eval(A, f), e \bigr) \) \\
    & \( \qquad \union
          \filter\bigl( A, \eval(A, f) \setminus \eval(A, e), f \bigr)
         \subseteq A \)
    & (D1, D2) \\
\proofleaf
    & \( \filter\bigl(A, (\inequality, e, f) \bigr) \subseteq A \)
    & (D3, D0) \\
\end{prooftable}
From the definition of the concretization function H1 we have that
\[
  \filter(A, c) \subseteq A \implies
  \concretization\bigl(\filter(A, c)\bigr) \subseteq \concretization(A),
\]
Since we have just proved the antecedent of this implication, we have the
truth of the consequent, which is the thesis.
\end{proof}

\subsection{A Final Consideration}
As stated in the first few lines of this section, the presented model is
intentionally simplified to ease the presentation and the proofs. However, these concepts can
be generalized to treat more complex environments and languages. In
\refcodesnippet{filter incompleteness} we present an example\footnote{This
example comes from the test suite of our implementation of the algorithms.}
that shows a more realistic implementation of the situation presented in
\refexample{filter incompleteness}.
This example shows how using recursive data structures it is possible to
generate the points-to relations presented in the previous examples: in
particular loops and locations pointing to themselves, which are quite
uncommon to see using only basic types.
\codecaption
{An example of code that shows the incompleteness of the filter algorithm
  using a recursive data structure.}
{filter incompleteness}
\begin{codesnippet}
struct L {
  struct L * next;
  int value;
};

int main() {
  struct L a, b, c;

  if (...)    a.next = &a;
  else
    if (...)  a.next = &b;
    else      a.next = &c;

  b.next = &c;
  c.next = &a;

  if (a.next->next == &c) {
    ...
  }
}
\end{codesnippet}

%% file: tex/examples.tex
%% EXAMPLE 1 %%%%%%%%%%%%%%%%%%%%%%%%%%%%%%%%%%%%%%%%%%%%%%%%%%%%%%%%

\codecaption
{an example of application of the assignment operation.}
{assignment example 01}
\begin{codesnippet}
int **pp, *q, *p, *r, a, b, c;

if (...) pp = &p;
else     pp = &q;
// \( \eval(\indirection pp) = \{ p, q \} \)

if (...) r = &a;
else     r = &c;
// \( \eval(\indirection r) = \{ a, c \} \)

p = &a;
// \( \eval(\indirection p) = \{ a \} \)
q = &b;
// \( \eval(\indirection q) = \{ b \} \)

*pp = r;
// \( \eval(\indirection \indirection pp) =  \eval(\indirection q) = \{ a, b, c \} \)
// \( \eval(\indirection p) = \{ a, c \} \)
\end{codesnippet}

\begin{figure}
\begin{center}
\input{figures/assignment_example_01}
\caption
  {a representation of the points-to information of the program
  in \refcodesnippet{assignment example 01} before and after the assignment
  at \refline{16}.}
\label{figure:assignment example 01}
\end{center}
\end{figure}

\begin{example}
This example is about the abstract assignment
operation. Consider the code in \refcodesnippet{assignment example 01}.
Note that the C assignment \QuotedCode{*pp = r} in our simplified language
is expressed as the pair \( (\indirection pp, \indirection r) \).
Assume to reach \refline{15} with the approximated points-to information
\( A \in \abstractdomain \)
\begin{gather*}
  \eval(A, \indirection pp) = \{ p, q \}, \\
  \eval(A, \indirection p)  = \{ a \}, \\
  \eval(A, \indirection q)  = \{ b \}, \\
  \eval(A, \indirection r)  = \{ a, c \};
\end{gather*}
then
\[
  \eval(A, \indirection pp) \times \eval(A, \indirection r) =
    \bigl\{ (p, a), (p, c), (q, a), (q, c) \bigr\}.
\]
The result of the evaluation of the rhs of the
assignment, \( \indirection pp \), contains more that one locations, \( p \) and \( q \); then
from the definition of the assignment operation (\refdefinition{Assignment
evaluation}) we have that the \emph{kill set} \( K \) is empty, then the
result of the assignment can be expressed as
\begin{align*}
  \assign\bigl(A, (\indirection pp, \indirection r) \bigr)
    & = A \union \eval(A, \indirection pp) \times \eval(A, \indirection r) \\
    & = A \union \bigl\{ (p, a), (p, b), (q, a), (q, b) \bigr\}.
\end{align*}
Note that after the execution of the assignment
(\reffigure{assignment example 01}), the old values of the variables
\QuotedCode{p} and \QuotedCode{q} are not overwritten, i.e.,
\[
  \bigl\{ (p, a), (q, b) \bigr\} \subseteq
    \assign\bigl(A, (\indirection pp, \indirection r) \bigl).
\]
\end{example}

%% EXAMPLE 2 %%%%%%%%%%%%%%%%%%%%%%%%%%%%%%%%%%%%%%%%%%%%%%%%%%%%%%%%

\codecaption
{another example of application of the assignment operation.}
{assignment example 02}
\begin{codesnippet}
int **pp, *p, *r, a, b, c;

pp = &p;
p = &c;
// \( \eval(\indirection \indirection pp) = \{ c \} \)
if (...) r = &a;
else     r = &b;
// \( \eval(\indirection r) = \{ a, b \} \)

*pp = r;
// \( \eval(\indirection pp) = \{ a, b \} \)
\end{codesnippet}

\begin{figure}
\begin{center}
\input{figures/assignment_example_02}
\caption
  {a representation of points-to information before and after the
   execution of the assignment operation at \refline{10}
  of \refcodesnippet{assignment example 02}.}
\label{figure:assignment example 02}
\end{center}
\end{figure}

\begin{example}
This is another example of the application of the abstract assignment
operation. Consider the code in \refcodesnippet{assignment example 02}.
Again, the C assignment \QuotedCode{*pp = r} in our simplified language
is expressed as the pair \( (\indirection pp, \indirection r) \).
Assume to reach \refline{9} with the approximated points-to information
\( A \in \abstractdomain \) such that
\begin{gather*}
  \eval(A, \indirection pp) = \{ p \}, \\
  \eval(A, \indirection p)  = \{ c \}, \\
  \eval(A, \indirection r)  = \{ a, b \};
\end{gather*}
then
\[
  \eval(A, \indirection pp) \times \eval(A, \indirection r) =
    \bigl\{ (p, a), (p, b) \bigr\}.
\]
But this time
the evaluation of the rhs of the
assignment, \( \indirection pp \),  contains only one location, \( p \). From
(\refdefinition{Assignment evaluation}) we have
\[
  K = \eval(A, \indirection pp) \times \locations = \bigl\{ (p, c) \bigr\},
\]
and then (\reffigure{assignment example 01})
\begin{align*}
  \assign\bigl(A, (\indirection pp, \indirection r) \bigr)
    & = (A \setminus K) \union
        \eval(A, \indirection pp) \times \eval(A, \indirection r) \\
    & = \Bigl(A \setminus \bigl\{ (p, c) \bigr\} \Bigr)
     \union \bigl\{ (p, a), (p, b) \bigr\}.
\end{align*}
Note that, in this case, the assignment deletes the old value of the
variable \QuotedCode{p}, i.e.,
\[
  (p, c) \not \in \assign\bigl(A, (\indirection pp, \indirection r) \bigr).
\]
\end{example}

%% EXAMPLE 3 %%%%%%%%%%%%%%%%%%%%%%%%%%%%%%%%%%%%%%%%%%%%%%%%%%%%%%%%

\codecaption
{an example of application of the filter operation.}
{filter example 01}
\begin{codesnippet}
int *p, *q, a, b, c;

if (...) p = &a;
else     p = &b;
// \( \eval(\indirection p) = \{ a, b \} \)

if (...) q = &b;
else     q = &c;
// \( \eval(\indirection q) = \{ b, c \} \)

if (p == q) {
  // \( \eval(\indirection p) = \eval(\indirection q) = \{ b \} \)
}
\end{codesnippet}

\begin{figure}
\begin{center}
\input{figures/filter_example_01}
\caption
  {a representation of the points-to information before and after the
  execution of the filter operation on the condition of the \Code{if}
  statement at \refline{11} of \refcodesnippet{filter example 01}.}
\label{figure:filter example 01}
\end{center}
\end{figure}

\begin{figure}
\begin{center}
\input{figures/filter_example_01_b}
\caption
  {a representation of computation of the filter operation for the example
  in \refcodesnippet{filter example 01}.}
\label{figure:filter example 01 b}
\end{center}
\end{figure}

\begin{example}
Consider the example program in \refcodesnippet{filter example 01}.
As anticipated in the annotations of the
presented code, the filter operation, acting on the condition \QuotedCode{p
== q}, is able to detect that inside the body of the \Code{if} statement
at \refline{12} both \QuotedCode{p} and \QuotedCode{q} definitely point to
\QuotedCode{b}.
Now we want to show step by step how this result is obtained from the given
definitions. Since the situation for \QuotedCode{p} and \QuotedCode{q} is
symmetrical, we show only how it can be derived that \QuotedCode{p}
definitely points to \QuotedCode{b}.
Recall that the boolean expression of the C
language \QuotedCode{q == p} corresponds, in our simplified language, to the triple
\( (\equality, \indirection p, \indirection q) \).
Assume now that \refline{10} is reached with the following approximated
points-to information \( A \in \abstractdomain \) (\reffigure{filter
example 01})
\begin{gather*}
  \eval(A, \indirection p) = \{ a, b \}, \\
  \eval(A, \indirection q) = \{ b, c \}.
\end{gather*}
From the definition of the abstract filter operation (\refdefinition{Filter
3}) we have
\begin{gather*}
  I = \eval(A, \indirection p) \intersection \eval(A, \indirection q),
  \\
  \filter\bigl(A, (\equality, \indirection p, \indirection q)\bigr) =
    \filter(A, I, \indirection p) \intersection
    \filter(A, I, \indirection q).
\end{gather*}
The evaluation of the expressions is illustrated by the following table.
\begin{filtertable}
  \( i \) &
  \( \eval(A, \indirection p, i) \) &
  \( \eval(A, \indirection q, i) \) \\
\hline
  1 & \( \{ p \} \) & \( \{ q \} \) \\
  0 & \( \{ a, b \} \) & \( \{ b, c \} \) \\
\end{filtertable}
For \( i = 0 \), the target set of the filter (\refdefinition{Filter 2}) is
then defined as
\begin{align*}
  I
  & = \eval(A, \indirection p) \intersection \eval(A, \indirection q) \\
  & = \eval(A, \indirection p, 0) \intersection \eval(A, \indirection q, 0) \\
  & = \{ a, b \} \intersection \{ b, c \} = \{ b \}.
\end{align*}
Then, recalling from \refdefinition{Filter 1} that
\[
  \target(A, I, e, i + 1) = \eval(A, e, i+1) \intersection
    \prev\bigl(A, \target(A, I, e, i) \bigr),
\]
we compute backward the sequence of target sets
for the expression \( \indirection p \) as
\begin{filtertable}
  \( i \) &
  \( \target\bigl(A, \{ b \}, \indirection p, i\bigr) \) &
  Removed arcs \\
\hline
  1 & \( \{ p \} \) & \( \bigl\{ (p, a) \bigr\} \)  \\
  0 & \( \{ b \} \) & \( \emptyset \) \\
\end{filtertable}
Since the target set for \( i = 1 \) consists of the only element \( p \) and the
node \( a \) is not part of the target set for \( i = 0 \), then the filter
removes the arc \( (p, a) \) from the points-to information.
See \reffigure{filter example 01 b} for a graphical representation of the
described situation.
\end{example}

%% EXAMPLE 4 %%%%%%%%%%%%%%%%%%%%%%%%%%%%%%%%%%%%%%%%%%%%%%%%%%%%%%%%

\codecaption
{an example of application of the filter operation.}
{filter example 02}
\begin{codesnippet}
int *p, *q, a, b, c, d;

if (...)
  if (...) p = &a;
  else     p = &b;
else
  if (...) p = &c;
  else     p = &d;
// \( \eval(\indirection p) = \{ a, b, c, d \} \)

if (...) q = &b;
else     q = &c;
// \( \eval(\indirection q) = \{ b, c \} \)

if (p == q) {
  // \( \eval(\indirection p) = \eval(\indirection q) = \{ b, c \} \)
}
\end{codesnippet}

\begin{figure}
\begin{center}
\input{figures/filter_example_02}
\caption
  {a representation of the points-to information before and after the
  execution of the filter operation on the condition of the \Code{if}
  statement at \refline{15} of \refcodesnippet{filter example 02}.}
\label{figure:filter example 02}
\end{center}
\end{figure}

\begin{figure}
\begin{center}
\input{figures/filter_example_02_b}
\caption
  {a representation of the computation of the filter operation for the example
  in \refcodesnippet{filter example 02}.}
\label{figure:filter example 02 b}
\end{center}
\end{figure}

\begin{example}
Now we present a similar situation to show that when the abstract
filter operation cuts some arcs (\refdefinition{Filter 1}) what matters is
the cardinality of the set of the ``pointers'' and not carditality of
the set of the ``pointed'' objects.
Consider the code in \refcodesnippet{filter example 02};
the points-to information \( A \in \abstractdomain \) at \refline{14},
is presented in \reffigure{filter example 02}.
In this case the evaluation of the two expressions \( \indirection p \) and
\( \indirection q \) proceeds as follows
\begin{filtertable}
  \( i \) &
  \( \eval(A, \indirection p, i) \) &
  \( \eval(A, \indirection q, i) \) \\
\hline
  1 & \( \{ p \} \) & \( \{ q \} \) \\
  0 & \( \{ a, b, c, d \} \) & \( \{ b, c \} \) \\
\end{filtertable}
For \( i = 0 \), we have the target set
\begin{align*}
  \eval(A, \indirection p)  \intersection \eval(A, \indirection q)
    & = \eval(A, \indirection p, 0) \intersection \eval(A, \indirection q, 0) \\
    & = \{ a, b, c, d \} \intersection \{ b, c \} = \{ b, c \}.
\end{align*}
The computation of the filter on the expression \QuotedCode{p} proceeds as
follows (\reffigure{filter example 02 b})
\begin{filtertable}
  \( i \) &
  \( \target\bigl(A, \{ b, c \}, \indirection p, i\bigr) \) &
  Removed arcs \\
\hline
  1 & \( \{ p \} \)     & \( \bigl\{ (p, a), (p, d) \bigr\} \) \\
  0 & \( \{ b, c \} \)  & \( \emptyset \) \\
\end{filtertable}
\end{example}

%% EXAMPLE 5 %%%%%%%%%%%%%%%%%%%%%%%%%%%%%%%%%%%%%%%%%%%%%%%%%%%%%%%%

\begin{figure}
\begin{center}
\input{figures/filter_example_03}
\caption
  {a representation of the points-to information
  \( A \in \abstractdomain \) before and after the
  execution of the filter operation on the condition \( (\equality,
  \indirection \indirection pp, a) \).}
\label{figure:filter example 03}
\end{center}
\end{figure}

\begin{figure}
\begin{center}
\input{figures/filter_example_03_b}
\caption
  {a representation of computation of the filter operation for the example
  in \reffigure{filter example 03}.}
\label{figure:filter example 03 b}
\end{center}
\end{figure}

\begin{example}
Consider the points-to approximation \( A \in \abstractdomain \)
described in \reffigure{filter example 03}.
In this case there are two levels of indirection.
Assume to filter the points-to approximation \( A \)
with respect to
the condition \( (\equality, \indirection \indirection pp, a) \).
The evaluation of the lhs and the rhs of the condition
proceeds as follows
\begin{filtertable}
  \( i \)
  & \( \eval(A, \indirection \indirection pp, i) \)
  & \( \eval(A, a, i) \) \\
\hline
  2 & \( \{ pp \} \)      & \( \emptyset \) \\
  1 & \( \{ p, q \} \)    & \( \emptyset \) \\
  0 & \( \{ a, b, c \} \) & \( \{ a \} \) \\
\end{filtertable}
Then, for \( i = 0 \), we have the target set
\begin{align*}
  \eval(A, \indirection \indirection p) \intersection \eval(A, a)
    & = \eval(A, \indirection \indirection p, 0) \intersection \eval(A, a, 0) \\
    & = \{ a, b, c \} \intersection \{ a \} = \{ a \}.
\end{align*}
The computation of the filter on the lhs proceeds as
\begin{filtertable}
  \( i \) &
  \( \target\bigl(A, \{ a \}, \indirection \indirection p, i\bigr) \) &
  Removed arcs \\
\hline
  2 & \( \{ pp \} \)  & \( \bigl\{ (pp, q) \bigr\} \) \\
  1 & \( \{ p \} \)   & \( \bigl\{ (p, b) \bigr\} \) \\
  0 & \( \{ a \} \)  & \( \emptyset \) \\
\end{filtertable}
\reffigure{filter example 03 b} depicts the computation just
described.
\end{example}

%% EXAMPLE 6 %%%%%%%%%%%%%%%%%%%%%%%%%%%%%%%%%%%%%%%%%%%%%%%%%%%%%%%%

\begin{figure}
\begin{center}
\input{figures/filter_example_04}
\caption
  {on the left a representation of the initial points-to information \( A
  \in \abstractdomain \), in the middle
  the information resulting by filtering the initial
  information \( A \) on the condition \( (\equality, \indirection \indirection pp,
  a) \); finally, on the right, the points-to information resulting from
  filtering the approximation \( A \) on the condition \( (\inequality, \indirection
  \indirection pp, a) \).}
\label{figure:filter example 04}
\end{center}
\end{figure}

\begin{figure}
\begin{center}
\input{figures/filter_example_04_b}
\caption
  {a representation of computation of the filter operation for the example
  in \reffigure{filter example 04} on the condition
    \( (\equality, \indirection \indirection pp, a) \).}
\label{figure:filter example 04 b}
\end{center}
\end{figure}

\begin{figure}
\begin{center}
\input{figures/filter_example_04_c}
\caption
  {a representation of computation of the filter operation for the example
  in \reffigure{filter example 04} on the condition
    \( (\inequality, \indirection \indirection pp, a) \).}
\label{figure:filter example 04 c}
\end{center}
\end{figure}

\begin{example}
Consider the points-to approximation \( A \in \abstractdomain \)
described in \reffigure{filter example 04}.
The evaluation of the the expression
\( \indirection \indirection pp \) follows the steps
\begin{filtertable}
  \( i \) &
  \( \eval(A, \indirection \indirection p, i) \) \\
\hline
  2 & \( \{ pp \} \) \\
  1 & \( \{ p, q, r\} \) \\
  0 & \( \{ a, b \} \) \\
\end{filtertable}
Assume to filter the points-to approximation \( A \) on
the condition \( (\equality, \indirection \indirection pp, a) \)
and also on the opposite condition
\( (\inequality, \indirection \indirection pp, a) \).
For \( i = 0 \), for the equality and the inequality conditions we have the
target sets
\begin{gather*}
  \eval(A, \indirection \indirection pp) \intersection
    \eval(A, a) = \{ a, b \} \intersection \{ a \} = \{ a \}, \\
  \eval(A, \indirection \indirection pp) \setminus
    \eval(A, a) = \{ a, b \} \setminus \{ a \} = \{ b \},
\end{gather*}
respectively.  The computation of the filter on the lhs proceeds as
\begin{filtertable}
  \( i \)
  & \( \target\bigl(A, \{ a \}, \indirection \indirection pp, i\bigr) \)
  & Removed arcs \\
\hline
  2 & \( \{ pp \} \)  & \( \bigl\{ (pp, r) \bigr\} \) \\
  1 & \( \{ p, q \} \)   & \( \emptyset \) \\
  0 & \( \{ a \} \)  & \( \emptyset \) \\
\hline \hline
  \( i \)
  & \( \target\bigl(A, \{ b \}, \indirection \indirection pp, i\bigr) \)
  & Removed arcs \\
\hline
  2 & \( \{ pp \} \)  & \( \bigl\{ (pp, p), (pp, q) \bigr\} \) \\
  1 & \( \{ r \} \)   & \( \emptyset \) \\
  0 & \( \{ b \} \)  & \( \emptyset \) \\
\end{filtertable}
\reffigure{filter example 04 b} and \reffigure{filter example 04 c}
depict the filter computation just described.
\end{example}

%% EXAMPLE 7 %%%%%%%%%%%%%%%%%%%%%%%%%%%%%%%%%%%%%%%%%%%%%%%%%%%%%%%%

\begin{figure}
\begin{center}
\input{figures/filter_example_05}
\caption
  {on the left a representation of the points-to approximation \( A \in
  \abstractdomain \),
  on the right a representation of the approximation resulting from the
  application of the filter on the condition
  \( (\equality, \indirection \indirection \indirection ppp, a) \). The
  arcs \( \bigl\{ (ppp, rr), (p, b) \bigr\} \) have
  been removed.}
\label{figure:filter example 05}
\end{center}
\end{figure}

\begin{figure}
\begin{center}
\input{figures/filter_example_05_b}
\caption
  {a representation of computation of the filter operation for the example
  in \reffigure{filter example 05} on the condition
    \( (\equality, \indirection \indirection \indirection ppp, a) \).}
\label{figure:filter example 05 b}
\end{center}
\end{figure}

\begin{example}
Now consider the points-to approximation \( A \in \abstractdomain \)
described in \reffigure{filter example 05}.
The evaluation of the the expression
\( \indirection \indirection \indirection ppp \) follows the steps
\begin{filtertable}
  \( i \) &
  \( \eval(A, \indirection \indirection p, i) \) \\
\hline
  3 & \( \{ ppp \} \) \\
  2 & \( \{ pp, qq, rr \} \) \\
  1 & \( \{ p, r\} \) \\
  0 & \( \{ a, b, c \} \) \\
\end{filtertable}
Assume to filter the points-to approximation \( A \) on
the condition \( (\equality, \indirection \indirection \indirection ppp, a) \).
For \( i = 0 \), for the equality condition we have the
target set
\[
  \eval(A, \indirection \indirection \indirection ppp) \intersection
    \eval(A, a) = \{ a, b, c \} \intersection \{ a \} = \{ a \}.
\]
The computation of the filter on the lhs proceeds as follows
\begin{filtertable}
  \( i \)
  & \( \target\bigl(A, \{ a \}, \indirection \indirection \indirection ppp, i\bigr) \)
  & Removed arcs \\
\hline
  3 & \( \{ ppp \} \)  & \( \bigl\{ (ppp, rr) \bigr\} \) \\
  2 & \( \{ pp, qq \} \)  & \( \emptyset \) \\
  1 & \( \{ p \} \)   & \( \bigl\{ (p, b) \bigr\} \) \\
  0 & \( \{ a \} \)  & \( \emptyset \) \\
\end{filtertable}
\reffigure{filter example 05 b}
depicts the filter computation just described.
\end{example}

%% file: figures/assignment_example_01.tex
\input{figures/common_style}
\def \drawThePicture[#1]{
  \abstractLocation pp (0, 1)
  \abstractLocation p (1, 0)
  \abstractLocation q (1, 2)
  \abstractLocation a (3, 0)
  \abstractLocation b (3, 2)
  \abstractLocation c (3, 1)
  \abstractLocation r (4, 1)
  \figureBackground
    (a.south -| pp.west)
    (b.north -| r.east)

  \node at (2,-1) {#1};
  \path
    \edgePointsToR(pp,p)
    \edgePointsToL(pp,q)
    \edgePointsTo(p,a)
    \edgePointsTo(q,b)
    \edgePointsToL(r,a)
    \edgePointsTo(r,c)
    ;
}
\begin{tikzpicture}
  \begin{scope}
    \drawThePicture[Before.]
  \end{scope}
  \begin{scope}[xshift=5.5cm]
    \drawThePicture[After.]
    \path
      \edgePointsTo(p,c)
      \edgePointsTo(q,a)
      \edgePointsTo(q,c)
      ;
  \end{scope}
\end{tikzpicture}

%% file: figures/assignment_example_02.tex
\input{figures/common_style}
\def \drawThePicture[#1]{
  \abstractLocation pp (0, 1)
  \abstractLocation p (1, 1)
  \abstractLocation a (2, 2)
  \abstractLocation c (2, 1)
  \abstractLocation b (2, 0)
  \abstractLocation r (3, 1)
  \figureBackground
    (b.south -| pp.west)
    (a.north -| r.east)

  \node at (1.5,-1) {#1};
  \path
    \edgePointsTo(pp,p)
    \edgePointsToR(r,a)
    \edgePointsToL(r,b)
    ;
}
\begin{tikzpicture}
  \begin{scope}
    \drawThePicture[Before.]
    \path
      \edgePointsTo(p,c)
      ;
  \end{scope}
  \begin{scope}[xshift=5.5cm]
    \drawThePicture[After.]
    \path
      \edgePointsToL(p,a)
      \edgePointsToR(p,b)
      ;
  \end{scope}
\end{tikzpicture}

%% file: figures/filter_example_01.tex
\input{figures/common_style}
\def \drawThePicture[#1]{
  \abstractLocation a (1, 2)
  \abstractLocation b (1, 1)
  \abstractLocation c (1, 0)
  \abstractLocation p (0, 1)
  \abstractLocation q (2, 1)
  \figureBackground
    (c.south -| p.west)
    (a.north -| q.east)

  \node at (1,-1) {#1};
  \path
    \edgePointsTo(p,b)
    \edgePointsTo(q,b);
}
\begin{tikzpicture}
  \begin{scope}
    \drawThePicture[Before.]
    \path
      \edgePointsToL(p,a)
      \edgePointsToL(q,c);
  \end{scope}
  \begin{scope}[xshift=4cm]
    \drawThePicture[After.]
  \end{scope}
\end{tikzpicture}

%% file: figures/filter_example_01_b.tex
\input{figures/common_style}
\def \drawThePicture{
  \abstractLocation a (2, 3)
  \abstractLocation b (2, 1)
  \abstractLocation p (0, 1)
  \node (ev1) at (0, 4) { \( \eval(1) \) };
  \node (ev0) at (2, 4) { \( \eval(0) \) };
  \node (tr1) at (0, 0) { \( \target(1) \) };
  \node (tr0) at (2, 0) { \( \target(0) \) };
  \figureBackground
    (tr0.south -| tr1.west)
    (ev0.north east)

  \targetBackground
    (tr1.south west)
    (p.north -| tr1.east)

  \targetBackground
    (tr0.south west)
    (b.north -| tr0.east)
  \path
    (p) edge [points--to style, bend left, red, dashed] (a)
    \edgePointsTo(p,b);
}
\begin{tikzpicture}
  \begin{scope}
    \drawThePicture
  \end{scope}
\end{tikzpicture}

%% file: figures/filter_example_02.tex
\input{figures/common_style}
\def \drawThePicture[#1]{
  \abstractLocation a (1, 3)
  \abstractLocation b (1, 2)
  \abstractLocation c (1, 1)
  \abstractLocation d (1, 0)
  \abstractLocation p (0, 1.5)
  \abstractLocation q (2, 1.5)
  \figureBackground
    (d.south -| p.west)
    (a.north -| q.east)

  \node at (1,-1) {#1};
  \path
    \edgePointsToL(p,b)
    \edgePointsToR(p,c)
    \edgePointsToL(q,c)
    \edgePointsToR(q,b);
}
\begin{tikzpicture}
  \begin{scope}
    \drawThePicture[Before.]
    \path
      \edgePointsToL(p,a)
      \edgePointsToR(p,d);
  \end{scope}
  \begin{scope}[xshift=4cm]
    \drawThePicture[After.]
  \end{scope}
\end{tikzpicture}

%% file: figures/filter_example_02_b.tex
\input{figures/common_style}
\def \drawThePicture{
  \abstractLocation a (2, 3)
  \abstractLocation b (2, 2)
  \abstractLocation c (2, 1)
  \abstractLocation d (2, 0)
  \abstractLocation p (0, 1.5)
  \node (ev1) at (0, 4) { \( \eval(1) \) };
  \node (ev0) at (2, 4) { \( \eval(0) \) };
  \node (tr1) at (0, 0.5) { \( \target(1) \) };
  \node (tr0) at (3, 1.5) { \( \target(0) \) };
  \figureBackground
    (d.south -| tr1.west)
    (ev0.north -| tr0.east)

  \targetBackground
    (tr1.south west)
    (p.north -| tr1.east)

  \path (c.south west)+(-0.1,0) node (zzzz) {};
  \targetBackground
    (zzzz)
    (b.north -| tr0.east)
  \path
    (p) edge [points--to style, bend left, red, dashed] (a)
    (p) edge [points--to style, bend right, red, dashed] (d)
    \edgePointsToL(p,b)
    \edgePointsToR(p,c)
    ;
}
\begin{tikzpicture}
  \begin{scope}
    \drawThePicture
  \end{scope}
\end{tikzpicture}

%% file: figures/filter_example_03.tex
\input{figures/common_style}
\def \drawThePicture[#1]{
  \abstractLocation pp (0, 1)
  \abstractLocation p  (1, 1.5)
  \abstractLocation q  (1, 0.5)
  \abstractLocation a  (2, 2)
  \abstractLocation b  (2, 1)
  \abstractLocation c  (2, 0)
  \figureBackground
    (c.south -| pp.west)
    (a.north east)

  \node at (1,-1) {#1};
  \path
    \edgePointsTo(pp,p)
    \edgePointsTo(p,a)
    \edgePointsTo(q,c)
    ;
}
\begin{tikzpicture}
  \begin{scope}
    \drawThePicture[Before.]
    \path
      \edgePointsTo(p,b)
      \edgePointsTo(pp,q)
      ;
  \end{scope}
  \begin{scope}[xshift=4cm]
    \drawThePicture[After.]
  \end{scope}
\end{tikzpicture}

%% file: figures/filter_example_03_b.tex
\input{figures/common_style}
\def \drawThePicture{
  \abstractLocation pp (0, 2)
  \abstractLocation p  (2, 3)
  \abstractLocation q  (2, 1)
  \abstractLocation a  (4, 4)
  \abstractLocation b  (4, 2)
  \abstractLocation c  (4, 0)
  \node (ev2) at (0, 5) { \( \eval(2) \) };
  \node (ev1) at (2, 5) { \( \eval(1) \) };
  \node (ev0) at (4, 5) { \( \eval(0) \) };
  \node (tr2) at (0, 1) { \( \target(2) \) };
  \node (tr1) at (2, 2) { \( \target(1) \) };
  \node (tr0) at (4, 3) { \( \target(0) \) };
  \figureBackground
    (c.south -| ev2.west)
    (ev0.north east)

  \targetBackground
    (tr2.south west)
    (pp.north -| tr2.east)

  \targetBackground
    (tr1.south west)
    (p.north -| tr1.east)

  \targetBackground
    (tr0.south west)
    (a.north -| tr0.east)
  \path
    \edgePointsTo(pp,p)
    \edgePointsTo(p,a)
    \edgePointsTo(q,c)
    (p) edge [points--to style, red, dashed] (b)
    (pp) edge [points--to style, red, dashed] (q)
    ;
}
\begin{tikzpicture}
  \begin{scope}
    \drawThePicture
  \end{scope}
\end{tikzpicture}

%% file: figures/filter_example_04.tex
\input{figures/common_style}
\def \drawThePicture[#1]{
  \abstractLocation pp (0, 1)
  \abstractLocation p  (1, 2)
  \abstractLocation q  (1, 1)
  \abstractLocation r  (1, 0)
  \abstractLocation a  (2, 1)
  \abstractLocation b  (2, 0)
  \figureBackground
    (r.south -| pp.west)
    (p.north -| a.east)

  \node at (1,-1) {#1};
  \path
    \edgePointsToL(p,a)
    \edgePointsTo(q,a)
    \edgePointsTo(r,b)
    ;
}
\begin{tikzpicture}
  \begin{scope}
    \drawThePicture[Original.]
    \path
      \edgePointsToL(pp,p)
      \edgePointsTo(pp,q)
      \edgePointsToR(pp,r)
      ;
  \end{scope}
  \begin{scope}[xshift=4cm]
    \drawThePicture[Equality.]
    \path
      \edgePointsToL(pp,p)
      \edgePointsTo(pp,q)
      ;
  \end{scope}
  \begin{scope}[xshift=8cm]
    \drawThePicture[Inequality.]
    \path
      \edgePointsToR(pp,r)
      ;
  \end{scope}
\end{tikzpicture}

%% file: figures/filter_example_04_b.tex
\input{figures/common_style}
\def \drawThePicture{
  \abstractLocation pp (0, 2)
  \abstractLocation p  (2, 4)
  \abstractLocation q  (2, 2)
  \abstractLocation r  (2, 0)
  \abstractLocation a  (4, 2)
  \abstractLocation b  (4, 0)

  \node (ev2) at (0, 5) { \( \eval(2) \) };
  \node (ev1) at (2, 5) { \( \eval(1) \) };
  \node (ev0) at (4, 5) { \( \eval(0) \) };
  \node (tr2) at (0, 1) { \( \target(2) \) };
  \node (tr1) at (2, 3) { \( \target(1) \) };
  \node (tr0) at (4, 1) { \( \target(0) \) };
  \figureBackground
    (r.south -| ev2.west)
    (ev0.north east)

  \targetBackground
    (tr2.south west)
    (pp.north -| tr2.east)

  \targetBackground
    (tr1.west |- q.south)
    (p.north -| tr1.east)

  \targetBackground
    (tr0.south west)
    (a.north -| tr0.east)
  \path
    \edgePointsToL(p,a)
    \edgePointsTo(q,a)
    \edgePointsTo(r,b)
    \edgePointsToL(pp,p)
    \edgePointsTo(pp,q)
    (pp) edge [points--to style, bend right, red, dashed] (r)
    ;
}
\begin{tikzpicture}
  \begin{scope}
    \drawThePicture
  \end{scope}
\end{tikzpicture}

%% file: figures/filter_example_04_c.tex
\input{figures/common_style}
\def \drawThePicture{
  \abstractLocation pp (0, 2)
  \abstractLocation p  (2, 4)
  \abstractLocation q  (2, 2)
  \abstractLocation r  (2, 0)
  \abstractLocation a  (4, 2)
  \abstractLocation b  (4, 0)

  \node (ev2) at (0, 5) { \( \eval(2) \) };
  \node (ev1) at (2, 5) { \( \eval(1) \) };
  \node (ev0) at (4, 5) { \( \eval(0) \) };
  \node (tr2) at (0, 1) { \( \target(2) \) };
  \node (tr1) at (2, 1) { \( \target(1) \) };
  \node (tr0) at (4, 1) { \( \target(0) \) };
  \figureBackground
    (r.south -| ev2.west)
    (ev0.north east)

  \targetBackground
    (tr2.south west)
    (pp.north -| tr2.east)

  \targetBackground
    (tr1.west |- r.south)
    (tr1.north east)

  \targetBackground
    (tr0.west |- b.south)
    (tr0.north east)
  \path
    \edgePointsToL(p,a)
    \edgePointsTo(q,a)
    \edgePointsTo(r,b)
    \edgePointsToR(pp,r)
    (pp) edge [points--to style, bend left, red, dashed] (p)
    (pp) edge [points--to style, red, dashed] (q)
    ;
}
\begin{tikzpicture}
  \begin{scope}
    \drawThePicture
  \end{scope}
\end{tikzpicture}

%% file: figures/filter_example_05.tex
\input{figures/common_style}
\def \drawThePicture[#1]{
  \abstractLocation ppp (0, 2)
  \abstractLocation pp  (1, 3)
  \abstractLocation qq  (1, 1)
  \abstractLocation rr  (1, 0)
  \abstractLocation p   (2, 2)
  \abstractLocation r   (2, 0)
  \abstractLocation a   (3, 3)
  \abstractLocation b   (3, 1)
  \abstractLocation c   (3, 0)
  \figureBackground
    (ppp.west |- rr.south)
    (a.north east)

  \node at (1.5,-1) {#1};
  \path
    \edgePointsToL(ppp,pp)
    \edgePointsToR(ppp,qq)
    \edgePointsTo(pp,p)
    \edgePointsTo(qq,p)
    \edgePointsTo(qq,r)
    \edgePointsTo(rr,r)
    \edgePointsTo(p,a)
    \edgePointsTo(r,c)
    ;
}
\begin{tikzpicture}
  \begin{scope}
    \drawThePicture[Before.]
    \path
      \edgePointsToR(ppp,rr)
      \edgePointsTo(p,b)
      ;
  \end{scope}
  \begin{scope}[xshift=5cm]
    \drawThePicture[After.]
  \end{scope}
\end{tikzpicture}

%% file: figures/filter_example_05_b.tex
\input{figures/common_style}
\def \drawThePicture{
  \abstractLocation ppp (0, 4)
  \abstractLocation pp  (2, 6)
  \abstractLocation qq  (2, 2)
  \abstractLocation rr  (2, 0)
  \abstractLocation p   (4, 4)
  \abstractLocation r   (4, 0)
  \abstractLocation a   (6, 6)
  \abstractLocation b   (6, 2)
  \abstractLocation c   (6, 0)

  \node (ev3) at (0, 7) { \( \eval(3) \) };
  \node (ev2) at (2, 7) { \( \eval(2) \) };
  \node (ev1) at (4, 7) { \( \eval(1) \) };
  \node (ev0) at (6, 7) { \( \eval(0) \) };
  \node (tr3) at (0, 5) { \( \target(3) \) };
  \node (tr2) at (2, 5) { \( \target(2) \) };
  \node (tr1) at (4, 5) { \( \target(1) \) };
  \node (tr0) at (6, 5) { \( \target(0) \) };
  \figureBackground
    (rr.south -| ev3.west)
    (ev0.north east)

  \targetBackground
    (tr3.west |- ppp.south)
    (tr3.north east)

  \targetBackground
    (tr2.west |- qq.south)
    (pp.north -| tr2.east)

  \targetBackground
    (tr1.west |- p.south)
    (tr1.north east)

  \targetBackground
    (tr0.south west)
    (a.north -| tr0.east)
  \path
    \edgePointsToL(ppp,pp)
    \edgePointsToR(ppp,qq)
    \edgePointsTo(pp,p)
    \edgePointsTo(qq,p)
    \edgePointsTo(qq,r)
    \edgePointsTo(rr,r)
    \edgePointsTo(p,a)
    \edgePointsTo(r,c)
    (ppp) edge [points--to style, bend right, red, dashed] (rr)
    (p) edge [points--to style, red, dashed] (b)
    ;
}
\begin{tikzpicture}
  \begin{scope}
    \drawThePicture
  \end{scope}
\end{tikzpicture}

%% file: figures/17.tex
\input{figures/common_style}
\def\drawThePicture[#1]{
  \abstractLocation l (0,0)
  \abstractLocation m (1,0)
  \abstractLocation n (2,0)
  \figureBackground
      (l.south west)
      (n.north east)
  \node [description label] at (2, 0) {#1};
}
\begin{tikzpicture}
  \begin{scope}
    \drawThePicture[{ The abstraction \( A = \assign(A, a) \). }]
    \path
      \edgePointsTo(m,n)
      ;
  \end{scope}
  \begin{scope}[yshift=-1.5cm]
    \drawThePicture[{ The abstraction \( B \). }]
    \path
      \edgePointsTo(l,m)
      \edgePointsTo(m,n)
      ;
  \end{scope}
  \begin{scope}[yshift=-3cm]
    \drawThePicture[{ The abstraction \( \assign(B, a) \). }]
    \path
      \edgePointsTo(l,m)
      \edgePointsToR(m,l)
      ;
  \end{scope}
\end{tikzpicture}

%% file: figures/06.tex
\input{figures/common_style}
\def\drawThePicture[#1,#2,#3,#4]{
  \abstractLocation a (1,1)
  \abstractLocation b (1,0)
  \abstractLocation p (0,0.5)
  \abstractLocation q (2,0.5)
  \node[anchor=north] at (1, -0.75) (table) {
    \begin{tabular}{*{4}{>{\(}c<{\)}}}
      \hline \hline
      #1   & \Code{a} & \Code{b} & \Code{*q} \\
      \hline
      \Code{*p}   & #2 \\
      \Code{*q}   & #3 \\
      \hline \hline
    \end{tabular}
  };
  \figureBackground
      (table.south -| table.west)
      (a.north -| table.east)
  \node [description label] at (table.north east) {#4};
}
\begin{tikzpicture}
  \begin{scope}
    \drawThePicture[
      \aliasquery_{ \executionEnvironment_0 },
      1 & 0 & \cellMark{1},
      1 & 0 &  ,
      {
        Below, an extract of the concrete alias query \(
        \aliasquery_{ m_0 } \) induced by the concrete memory description
        \( m_0 \in \memories \).  Above, a graphical representation
        of the points-to information
        \( C_0 \in \concretedomain \)
        associated to the same memory.
      }]
    \path \edgePointsToL(p,a)
          \edgePointsToR(q,a);
  \end{scope}
  \begin{scope}[yshift=-5cm]
    \drawThePicture[
      \aliasquery_{ \executionEnvironment_1 },
      0 & 1 & \cellMark{1},
      0 & 1 &  ,
      {
        As above, on the concrete memory description \( m_1 \in \memories \).
      }]
    \path \edgePointsToR(p,b)
          \edgePointsToL(q,b);
  \end{scope}
  \begin{scope}[yshift=-10cm]
    \drawThePicture[
      \absaliasquery,
      \top  & \top  & \cellMark{1},
      \top  & \top  &     ,
      {
        Above, a graphical representation of the points-to abstraction
        \( A \)
        \[
          \abstraction\bigl( \{ C_0, C_1 \} \bigr) =
            C_0 \union C_1 = A \in \abstractdomain.
        \]
        Below, an extract of the abstract alias query
        \begin{gather*}
          \abstraction
            \bigl( \{ \aliasquery_{m_0}, \aliasquery_{m_1} \} \bigr) \\
          \quad = \aliasquery_{ m_0 } \sqcup \aliasquery_{ m_1 } \\
          \quad = \absaliasquery \in \absAliasQueryDomain.
        \end{gather*}
        Note that \( \absaliasquery(\Code{*p}, \Code{*q}) = 1 \), that is
        the abstract alias query is able to represent that \Variable{p} and
        \Variable{q} always point to the same location.
      }]
    \path \edgePointsToL(p,a)
          \edgePointsToR(p,b)
          \edgePointsToR(q,a)
          \edgePointsToL(q,b);
  \end{scope}
\end{tikzpicture}

%% file: figures/06b.tex
\input{figures/common_style}
\def\drawThePicture[#1,#2,#3,#4]{
  \abstractLocation a (1,1)
  \abstractLocation b (1,0)
  \abstractLocation p (0,0.5)
  \abstractLocation q (2,0.5)
  \node[anchor=north] at (1, -0.75) (table) {
    \begin{tabular}{*{4}{>{\(}c<{\)}}}
      \hline \hline
      #1   & \Code{a} & \Code{b} & \Code{*q} \\
      \hline
      \Code{*p}   & #2 \\
      \Code{*q}   & #3 \\
      \hline \hline
    \end{tabular}
  };
  \figureBackground
      (table.south -| table.west)
      (a.north -| table.east)
  \node [description label] at (table.north east) {#4};
}
\begin{tikzpicture}
  \begin{scope}
    \drawThePicture[
      A,
      \top  & \top  & \cellMark{\top},
      \top  & \top  &,
      {
        Here both the table and the graph represent the points-to
        information \( A \).
        Note that
        \[
          \cardinality \bigl( \eval(A, \indirection p) > 1;
        \]
        that is (\refdefinition{Induced alias relation})
        \( A(\indirection p, \indirection q) = \top \), i.e.,
        the points-to representation is unable to express that
        \Variable{p} and \Variable{q} will definitely point to the same
        location.
      }]
    \path \edgePointsToL(p,a)
          \edgePointsToR(p,b)
          \edgePointsToR(q,a)
          \edgePointsToL(q,b);
  \end{scope}
  \begin{scope}[yshift=-5cm]
    \drawThePicture[
      C_2,
      1 & 0 & \cellMark{0},
      0 & 1 &  ,
      {
        The concrete points-to information \( C_2 \in \concretedomain \)
        is a spurious element of \( \concretization(A) \).  Note that
        \[
            \aliasquery_{ C_2 }
          \notin
            \concretization( \absaliasquery );
        \]
        that is the concrete alias relation \( \aliasquery_{ C_2 } \)
        induced by \( C_2 \) would not be generated using the alias
        representation.
      }]
    \path \edgePointsToL(p,a)
          \edgePointsToL(q,b);
  \end{scope}
  \begin{scope}[yshift=-10cm]
    \drawThePicture[
      C_3,
      0 & 1 & \cellMark{0},
      1 & 0 &  ,
      {
        The concrete points-to information \( C_3 \in \concretedomain \) is
        another spurious element of \(\concretization(A) \).
        Note that:
        \begin{gather*}
            \concretization(A) = \{ C_0, C_1, C_2, C_3 \}; \\
            \concretization(\absaliasquery) =
            \{ \aliasquery_{C_0}, \aliasquery_{C_1} \}.
        \end{gather*}
      }]
    \path \edgePointsToR(p,b)
          \edgePointsToR(q,a);
  \end{scope}
\end{tikzpicture}

%% file: figures/filter_limitation_1.tex
\input{figures/common_style}
\def\drawThePicture{
  \abstractLocation a (1,1)
  \abstractLocation b (1,0)
  \abstractLocation p (0,0.5)
  \abstractLocation q (2,0.5)
  \figureBackground
      (b.south -| p.west)
      (a.north -| q.east)
  \path
    \edgePointsToL(p,a)
    \edgePointsToR(p,b)
    \edgePointsToR(q,a)
    \edgePointsToL(q,b)
    ;
}
\begin{tikzpicture}
  \begin{scope}
    \drawThePicture
  \end{scope}
\end{tikzpicture}

%% file: figures/common_style.tex
\tikzstyle{absLoc}=[
    rectangle,
    join=round,
    draw=LocationBorderColor,
    line width=1pt,
    fill=LocationFillColor,
    inner sep=0pt,
    minimum size=5mm]
\tikzstyle{pointsTo}=[
    ->,
    shorten >=1pt,
    >=stealth',
    thick,
    black!60]
\tikzstyle{alias table style}=[
  fill=FloatingBackgroundColor,
  rounded corners]
\tikzstyle{abstract location style}=[
    rectangle,
    join=round,
    draw=LocationBorderColor,
    line width=1pt,
    rounded corners,
    fill=blue!15,
    inner sep=2pt,
    outer sep=0pt,
    minimum size=5mm]
\tikzstyle{selected abstract location style}=[
    circle,
    join=round,
    draw=orange,
    line width=1pt,
    fill=yellow,
    inner sep=2pt,
    outer sep=0pt,
    minimum size=7mm]
\tikzstyle{points--to style}=[
    ->,
    shorten >=1pt,
    shorten <=2pt,
    >=stealth',
    thick,
    black!60]
\tikzstyle{picture background style}=[
    line width=5mm,
    join=round,
    FloatingBackgroundColor]
\tikzstyle{target background style}=[
    line width=5mm,
    join=round,
    yellow,
    draw=orange]
\tikzstyle{floatingBackgroundStyle}=[
    line width=4mm,
    join=round,
    FloatingBackgroundColor]
\tikzstyle{description label style}=[
  rounded corners,
  fill=DescriptionLabelColor,
  ultra thick,
  draw=black!5]
\tikzstyle{description label}=[
  description label style,
  text width=6.5cm,
  right=1cm, anchor=west]
%
% Some macros.
%
\newcommand*{\cellMark}[1]
  {\fcolorbox{LocationBorderColor}{LocationFillColor}{\(#1\)}}
\def \abstractLocation #1 (#2,#3){
  \node[abstract location style] (#1) at (#2, #3) {\Code{#1}};
}
\def \selAbstractLocation #1 (#2,#3){
  \node[selected abstract location style] (#1) at (#2, #3) {\Code{#1}};
}
\def \edgePointsTo(#1,#2){
  (#1) edge [points--to style] (#2)
}
\def \edgePointsToWithStyle(#1,#2, #3){
  (#1) edge [points--to style, #3] (#2)
}
\def \edgePointsToL(#1,#2){
  (#1) edge [points--to style, bend left] (#2)
}
\def \edgePointsToR(#1,#2){
  (#1) edge [points--to style, bend right] (#2)
}
\def \figureBackground (#1) (#2){
  \begin{pgfonlayer}{background}
    \path (#1)+(-0.2, -0.2) node (xxxx) {};
    \path (#2)+(0.2, 0.2) node (yyyy) {};
    \path[rounded corners,
          fill=orange!20,
          ultra thick,
          draw=orange!10]
      (xxxx) rectangle (yyyy);
  \end{pgfonlayer}
}
\def \targetBackground (#1) (#2){
  \begin{pgfonlayer}{background}
    \path (#1)+(0,-0.1) node (xxxx) {};
    \path (#2)+(0,+0.1) node (yyyy) {};
    \path[fill=blue!3,rounded corners, draw=black!20, ultra thick]
      (xxxx) rectangle (yyyy);
  \end{pgfonlayer}
}

%% file: figures/01.tex
\input{figures/common_style}
\def \drawThePicture[#1]{
  \abstractLocation a (2, 0)
  \abstractLocation b (2, 1)
  \abstractLocation c (2, 2)
  \abstractLocation r (0, 1)
  \abstractLocation p (1, 0)
  \abstractLocation q (1, 2)
  \figureBackground
    (p.south -| r.west)
    (c.north -| c.east)
  \node at (1,-1) {#1};
  \path \edgePointsToR(r,p)
        \edgePointsToL(r,q)
        \edgePointsTo(p,a)
        \edgePointsTo(q,c);
}
\begin{tikzpicture}
  \begin{scope}
    \drawThePicture[Before.]
  \end{scope}
  \begin{scope}[xshift=4cm]
    \drawThePicture[After.]
    \path \edgePointsTo(p,b)
          \edgePointsTo(q,b);
  \end{scope}
\end{tikzpicture}

%% file: figures/04.tex
\input{figures/common_style}
\def \drawThePicture[#1,#2,#3,#4,#5,#6] {
  \abstractLocation a (2, 0)
  \abstractLocation b (2, 1)
  \abstractLocation c (2, 2)
  \abstractLocation r (0, 1)
  \abstractLocation p (1, 0)
  \abstractLocation q (1, 2)
  \node[anchor=north] at (1, -0.5) (table) {
    \begin{tabular}{*{4}{>{\(}c<{\)}}}
      \hline \hline
             #1   &  \Code{a} & \Code{b} & \Code{c} \\
      \hline
      \Code{*p}   & #2 \\
      \Code{*q}   & #3 \\
      \Code{**r}  & #4 \\
      \hline \hline
                  &  \Code{p} & \Code{q} & \\
      \hline
      \Code{*r}   & #5 \\
      \hline \hline
    \end{tabular}
  };
  \figureBackground
    (table.south -| table.west)
    (c.north -| table.east)
  \node [description label] at (table.north east) {#6};
}
\begin{tikzpicture}
  \begin{scope}
    \drawThePicture[
      \aliasquery_{ \executionEnvironment_0 },
      1 & 0 & 0,
      0 & 1 & 0,
      0 & \cellMark{1} & 0,
      0 & 1 &,
      {
        Below, an extract of the concrete alias query \(
        \aliasquery_{ m_0 } \) induced by the concrete memory description
        \( m_0 \in \memories \).  Above, a graphical representation
        of the points-to information
        \( C_0 \in \concretedomain \)
        associated to the same memory.
      }]
    \path \edgePointsTo(p,a)
          \edgePointsTo(q,b)
          \edgePointsToL(r,q);
  \end{scope}
  \begin{scope}[yshift=-7.5cm]
    \drawThePicture[
      \aliasquery_{ \executionEnvironment_1 },
      0 & 1 & 0,
      0 & 0 & 1,
      0 & \cellMark{1} & 0,
      1 & 0 &  ,
      {
        As above, on the concrete memory description \(
        \executionEnvironment_1 \).
      }]
    \path \edgePointsTo(p,b)
          \edgePointsTo(q,c)
          \edgePointsToR(r,p);
  \end{scope}
\end{tikzpicture}

%% file: figures/04b.tex
\input{figures/common_style}
\def \drawThePicture[#1,#2,#3,#4,#5,#6] {
  \abstractLocation a (2, 0)
  \abstractLocation b (2, 1)
  \abstractLocation c (2, 2)
  \abstractLocation r (0, 1)
  \abstractLocation p (1, 0)
  \abstractLocation q (1, 2)
  \node[anchor=north] at (1, -0.5) (table) {
    \begin{tabular}{*{4}{>{\(}c<{\)}}}
      \hline \hline
             #1   &  \Code{a} & \Code{b} & \Code{c} \\
      \hline
      \Code{*p}   & #2 \\
      \Code{*q}   & #3 \\
      \Code{**r}  & #4 \\
      \hline \hline
                  &  \Code{p} & \Code{q} & \\
      \hline
      \Code{*r}   & #5 \\
      \hline \hline
    \end{tabular}
  };
  \figureBackground
    (table.south -| table.west)
    (c.north -| table.east)
  \node [description label] at (table.north east) {#6};
}
\begin{tikzpicture}
  \begin{scope}
    \drawThePicture[
      \absaliasquery,
      \top & \top & 0,
      0 & \top & \top,
      0 & \cellMark{1} & 0,
      \top & \top &  ,
      {
        Above, a graphical representation of the most precise points-to
        abstraction \( A \)
        \[
          \abstraction\bigl( \{ C_0, C_1 \} \bigr) =
            C_0 \union C_1 = A \in \abstractdomain.
        \]
        Below, an extract of the abstract alias query
        \begin{gather*}
          \abstraction
            \bigl( \{ \aliasquery_{m_0}, \aliasquery_{m_1} \} \bigr) \\
          \quad = \aliasquery_{ m_0 } \sqcup \aliasquery_{ m_1 } \\
          \quad = \absaliasquery \in \absAliasQueryDomain.
        \end{gather*}
        Note that \( \absaliasquery(\indirection \indirection r, b) = 1 \),
        that is
        the abstract alias query is able to represent that the expressions
        \(\indirection \indirection r\) and \( b \)
        are definitely aliases.
      }]
    \path \edgePointsTo(p,a)
          \edgePointsTo(p,b)
          \edgePointsTo(q,b)
          \edgePointsTo(q,c)
          \edgePointsToL(r,q)
          \edgePointsToR(r,p);
  \end{scope}
  \begin{scope}[yshift=-7.5cm]
    \drawThePicture[
      A,
      \top & \top & 0,
      0 & \top & \top,
      \top & \cellMark{\top} & \top,
      \top & \top &  ,
      {
        Here both the table and the graph represent the points-to
        information \( A \).
        Note that
        \[
          \eval(A, \indirection \indirection r) \neq \eval(A, b),
        \]
        that is (\refdefinition{Induced alias relation})
        \( A(\indirection \indirection r, b) = \top \), i.e.,
        the points-to representation is unable to express that
        \( \indirection \indirection r \) and \( a \) will be definitely aliases.
      }]
    \path \edgePointsTo(p,a)
          \edgePointsTo(p,b)
          \edgePointsTo(q,b)
          \edgePointsTo(q,c)
          \edgePointsToL(r,q)
          \edgePointsToR(r,p);
  \end{scope}
\end{tikzpicture}

%% file: figures/04c.tex
\input{figures/common_style}
\def \drawThePicture[#1,#2,#3,#4,#5,#6] {
  \abstractLocation a (2, 0)
  \abstractLocation b (2, 1)
  \abstractLocation c (2, 2)
  \abstractLocation r (0, 1)
  \abstractLocation p (1, 0)
  \abstractLocation q (1, 2)
  \node[anchor=north] at (1, -0.5) (table) {
    \begin{tabular}{*{4}{>{\(}c<{\)}}}
      \hline \hline
             #1   &  \Code{a} & \Code{b} & \Code{c} \\
      \hline
      \Code{*p}   & #2 \\
      \Code{*q}   & #3 \\
      \Code{**r}  & #4 \\
      \hline \hline
                  &  \Code{p} & \Code{q} & \\
      \hline
      \Code{*r}   & #5 \\
      \hline \hline
    \end{tabular}
  };
  \figureBackground
    (table.south -| table.west)
    (c.north -| table.east)
  \node [description label] at (table.north east) {#6};
}
\begin{tikzpicture}
  \begin{scope}
    \drawThePicture[
      C_2,
      1 & 0 & 0,
      0 & 0 & 1,
      1 & \cellMark{0} & 0,
      1 & 0 &  ,
      {
        The concrete points-to information \( C_2 \in \concretedomain \)
        is a spurious element of \( \concretization(A) \).  Note that
        \[
            \aliasquery_{ C_2 }
          \notin
            \concretization( \absaliasquery );
        \]
        that is, the concrete alias relation \( \aliasquery_{ C_2 } \)
        induced by \( C_2 \) would not be generated using the alias
        representation.
      }]
    \path \edgePointsTo(p,a)
          \edgePointsTo(q,c)
          \edgePointsToR(r,p);
  \end{scope}
\end{tikzpicture}

%% file: figures/08.tex
\input{figures/common_style}
\begin{tikzpicture}
  \def\drawThePicture[#1]{
    \abstractLocation a (0, 1)
    \abstractLocation b (1, 1)
    \abstractLocation c (1, 0)
    \figureBackground
        (c.south -| a.west) (b.north -| b.east)
    \node [description label] at (1, 0.5) {#1};
  }
  \begin{scope}
    \drawThePicture[
      {
        The abstract memory \( A \),
        \[ \concretization(A) = \{ C_1, C_2 \}. \]
      }]
    \path \edgePointsToWithStyle(a,a, loop left)
          \edgePointsTo(a,b)
          \edgePointsTo(b,c);
  \end{scope}
  \begin{scope}[yshift=-2.5cm]
    \drawThePicture[
      {
        The concrete memory description \( C_1 \).
      }]
    \path \edgePointsTo(a,b)
          \edgePointsTo(b,c);
  \end{scope}
  \begin{scope}[yshift=-5cm]
    \drawThePicture[
      {
        The concrete memory description \( C_2 \).
      }]
    \path (a) edge [points--to style, loop left] (a)
          \edgePointsTo(b,c);
  \end{scope}
\end{tikzpicture}

%% file: figures/16.tex
\input{figures/common_style}
\def \drawThePicture{
  \node[abstract location style] (a) at (0, 0) {\Code{a}};
  \node[abstract location style] (ai) at (2, 0) {\Code{a}};
  \node[abstract location style] (bi) at (2, 1) {\Code{b}};
  \node[abstract location style] (aii) at (4, 0) {\Code{a}};
  \node[abstract location style] (bii) at (4, 1) {\Code{b}};
  \node[abstract location style] (cii) at (4, 2) {\Code{c}};
  \node (ev2) at (0, 3) { \( \eval(2) \) };
  \node (ev1) at (2, 3) { \( \eval(1) \) };
  \node (ev0) at (4, 3) { \( \eval(0) \) };
  \figureBackground
    (ev2.west |- a.south)
    (ev0.north east)
  \path
    \edgePointsTo(a,ai)
    \edgePointsTo(a,bi)
    \edgePointsTo(ai,aii)
    (ai) edge [points--to style, red, dashed] (bii)
    \edgePointsTo(bi,cii)
    ;
}
\begin{tikzpicture}
  \begin{scope}
    \drawThePicture
  \end{scope}
\end{tikzpicture}

%% file: figures/09.tex
\input{figures/common_style}
\begin{tikzpicture}
  \def\drawThePicture[#1]{
    \abstractLocation a (1, 1)
    \abstractLocation b (0, 1)
    \abstractLocation c (1, 0)
    \figureBackground
        (b.west |- c.south) (a.north east)
    \node [description label] at (1, 0.5) {#1};
  }
  \begin{scope}
    \drawThePicture[
      {
        The abstraction \( A \) before the execution of the assignment
        \( x = (\indirection a, \indirection a) \).
        \[ \concretization(A) = \{ C_1, C_2 \}. \]
      }]
    \path \edgePointsTo(a,b)
          \edgePointsTo(a,c);
  \end{scope}
  \begin{scope}[yshift=-2.5cm]
    \drawThePicture[
      {
        The concrete memory description \( C_1 \).
      }]
    \path \edgePointsTo(a,b);
  \end{scope}
  \begin{scope}[yshift=-5cm]
    \drawThePicture[
      {
        The concrete memory description
        \[
          \assign(C_1, x).
        \]
      }]
    \path \edgePointsTo(a,b)
          (b) edge [points--to style, loop left] (b);
  \end{scope}
  \begin{scope}[yshift=-7.5cm]
    \drawThePicture[
      {
        The concrete memory description \( C_2 \).
      }]
    \path \edgePointsTo(a, c);
  \end{scope}
  \begin{scope}[yshift=-10cm]
    \drawThePicture[
      {
        The concrete memory description
        \[
          \assign(C_2, x).
        \]
      }]
    \path \edgePointsTo(a,c)
          (c) edge [points--to style, loop right] (c);
  \end{scope}
  \begin{scope}[yshift=-12.5cm]
    \drawThePicture[
      {
        The abstract memory
        \begin{multline*}
          \abstraction\Bigl( \bigl\{ \assign(C_1, x), \assign(C_2, x) \bigr\}
          \Bigr) \\
            = \assign(C_1, x) \union \assign(C_2, x).
        \end{multline*}
      }]
    \path \edgePointsTo(a,b)
          \edgePointsTo(a,c)
          (b) edge [points--to style, loop left] (b)
          (c) edge [points--to style, loop right] (c);
  \end{scope}
  \begin{scope}[yshift=-15cm]
    \drawThePicture[
      {
        The abstraction \( \assign(A, x) \) resulting from the
        execution of the assignment. The spurious arcs are
        \( \bigr\{ (b, c), (c, b) \bigl\} \).
      }]
    \path \edgePointsTo(a,b)
          \edgePointsTo(a,c)
          (b) edge [points--to style, loop left] (b)
          (c) edge [points--to style, loop right] (c)
          (b) edge [points--to style, red, dashed] (c)
          (c) edge [points--to style, red, dashed, bend left] (b);
  \end{scope}
\end{tikzpicture}

%% file: figures/10.tex
\input{figures/common_style}
\begin{tikzpicture}
  \def\drawThePicture[#1]{
    \abstractLocation a (0, 0)
    \abstractLocation b (1, 0)
    \figureBackground
        (a.west |- a.south) (b.north -| b.east)
    \node [description label] at (b.east) {#1};
  }
  \begin{scope}
    \drawThePicture[
      {
        The abstraction \( A \).
        \[ \concretization(A) = \{ C_1, C_2 \}. \]
      }]
    \path \edgePointsTo(a,b)
          (a) edge [points--to style, loop below] (a)
          (b) edge [points--to style, loop below] (b);
  \end{scope}
  \begin{scope}[yshift=-2.5cm]
    \drawThePicture[
      {
        The concrete memory description \( C_1 \).
        This is not a model of
        \[
          c = (\indirection \indirection a, b).
        \]
      }]
    \path (a) edge [points--to style, loop below] (a)
          (b) edge [points--to style, loop below] (b);
  \end{scope}
  \begin{scope}[yshift=-5cm]
    \drawThePicture[
      {
        The concrete memory description \( C_2 \).
        This is a model of \( c \);
      }]
    \path \edgePointsTo(a, b)
          (b) edge [points--to style, loop below] (b);
  \end{scope}
  \begin{scope}[yshift=-7.5cm]
    \drawThePicture[
      {
        The abstraction \( \filter(A, c) = A\).
        The spurious arc is \( \bigr\{ (a, a) \bigl\} \).
      }]
    \path \edgePointsTo(a,b)
          (a) edge [points--to style, loop below, red, dashed] (a)
          (b) edge [points--to style, loop below] (b);
  \end{scope}
\end{tikzpicture}

%% file: figures/15.tex
\input{figures/common_style}
\def \drawThePicture{
  \node[abstract location style] (a) at (0, 3) {\Code{a}};
  \node[abstract location style] (ai) at (2, 3) {\Code{a}};
  \node[abstract location style] (bi) at (2, 1) {\Code{b}};
  \node[abstract location style] (aii) at (4, 3) {\Code{a}};
  \node[abstract location style] (bii) at (4, 1) {\Code{b}};
  \node (ev2) at (0, 4) { \( \eval(2) \) };
  \node (ev1) at (2, 4) { \( \eval(1) \) };
  \node (ev0) at (4, 4) { \( \eval(0) \) };
  \node (tr2) at (0, 2) { \( \target(2) \) };
  \node (tr1) at (2, 0) { \( \target(1) \) };
  \node (tr0) at (4, 0) { \( \target(0) \) };
  \figureBackground
    (tr2.west |- tr1.south)
    (ev0.north east)

  \targetBackground
    (tr2.south west)
    (a.north -| tr2.east)

  \targetBackground
    (tr1.south west)
    (ai.north -| tr1.east)

  \targetBackground
    (tr0.south west)
    (bii.north -| tr0.east)
  \path
    \edgePointsTo(a,ai)
    \edgePointsTo(ai,aii)
    \edgePointsTo(a,bi)
    \edgePointsTo(bi,bii)
    \edgePointsTo(ai,bii)
    ;
}
\begin{tikzpicture}
  \begin{scope}
    \drawThePicture
  \end{scope}
\end{tikzpicture}

%% file: figures/filter_iteration.tex
\input{figures/common_style}
\def \drawThePicture[#1]{
  \abstractLocation a  (0, 0)
  \abstractLocation c  (2, 0)
  \abstractLocation b  (1, 1.5)
  \figureBackground
    (a.south west)
    (b.north -| c.east)

  \node at (1,-1) {#1};
  \path
    \edgePointsToL(a,b)
    \edgePointsToL(b,c)
    (c) edge [points--to style, bend left] (a)
    ;
}
\begin{tikzpicture}
  \begin{scope}
    \drawThePicture[Initial.]
    \path
      (a) edge [points--to style, loop] (a)
      (a) edge [points--to style, bend left] (c)
      ;
  \end{scope}
  \begin{scope}[xshift=4cm]
    \drawThePicture[Iteration One.]
    \path
      (a) edge [points--to style, loop] (a)
      ;
  \end{scope}
  \begin{scope}[xshift=8cm]
    \drawThePicture[Iteration Two.]
  \end{scope}
\end{tikzpicture}

%% file: figures/filter_iteration_2.tex
\input{figures/common_style}
\begin{tikzpicture}
  \begin{scope}
    \node[abstract location style] (a) at (0, 3) {\Code{a}};
    \node[abstract location style] (ai) at (2, 3) {\Code{a}};
    \node[abstract location style] (bi) at (2, 2) {\Code{b}};
    \node[abstract location style] (ci) at (2, 1) {\Code{c}};
    \node[abstract location style] (aii) at (4, 3) {\Code{a}};
    \node[abstract location style] (bii) at (4, 2) {\Code{b}};
    \node[abstract location style] (cii) at (4, 1) {\Code{c}};
    \node (ev2) at (0, 4) { \( \eval(2) \) };
    \node (ev1) at (2, 4) { \( \eval(1) \) };
    \node (ev0) at (4, 4) { \( \eval(0) \) };
    \node (tr2) at (0, 2.5) { \( \target(2) \) };
    \node (tr1) at (2, 2.5) { \( \target(1) \) };
    \node (tr0) at (4, 0.5) { \( \target(0) \) };
    \figureBackground
      (tr2.west |- tr0.south)
      (ev0.north east)

    \targetBackground
      (tr2.south west)
      (tr2.east |- a.north)

    \targetBackground
      (bi.south -| tr1.west)
      (ai.north -| tr1.east)

    \targetBackground
      (tr0.south west)
      (tr0.east |- cii.north)
    \path
      \edgePointsTo(a,ai)
      \edgePointsTo(a,bi)
      (a) edge [points--to style, red, dashed] (ci)
      \edgePointsTo(ai,aii)
      \edgePointsTo(ai,bii)
      \edgePointsTo(ai,cii)
      \edgePointsTo(bi,cii)
      \edgePointsTo(ci,aii)
      ;
    \node at (2,-0.5) {First iteration.};
  \end{scope}

  \begin{scope}[xshift=6.5cm]
    \node[abstract location style] (a) at (0, 3) {\Code{a}};
    \node[abstract location style] (ai) at (2, 3) {\Code{a}};
    \node[abstract location style] (bi) at (2, 2) {\Code{b}};
    \node[abstract location style] (aii) at (4, 3) {\Code{a}};
    \node[abstract location style] (bii) at (4, 2) {\Code{b}};
    \node[abstract location style] (cii) at (4, 1) {\Code{c}};
    \node (ev2) at (0, 4) { \( \eval(2) \) };
    \node (ev1) at (2, 4) { \( \eval(1) \) };
    \node (ev0) at (4, 4) { \( \eval(0) \) };
    \node (tr2) at (0, 2.5) { \( \target(2) \) };
    \node (tr1) at (2, 1.5) { \( \target(1) \) };
    \node (tr0) at (4, 0.5) { \( \target(0) \) };
    \figureBackground
      (tr2.west |- tr0.south)
      (ev0.north east)

    \targetBackground
      (tr2.south west)
      (tr2.east |- a.north)

    \targetBackground
      (tr1.south west)
      (bi.north -| tr1.east)

    \targetBackground
      (tr0.south west)
      (tr0.east |- cii.north)
    \path
      (a) edge [points--to style, red, dashed] (ai)
      \edgePointsTo(a,bi)
      \edgePointsTo(ai,aii)
      \edgePointsTo(ai,bii)
      \edgePointsTo(bi,cii)
      ;
    \node at (2,-0.5) {Second iteration.};
  \end{scope}
\end{tikzpicture}

%% file: tex/extensions.tex
\section{The Extended Abstract Memory Model}
\label{section:extensions}
With the aim of presenting a realistic points-to analysis, this section
discusses some extensions to the simplified model previously introduced.
More precisely, this section describes a more realistic \emph{memory model} by augmenting
the previously described domains with some details not directly related to
the points-to problem, which are however necessary for
the definition of a working memory.

\subsection{Abstract and Concrete Locations}
\label{section:Locations approximation}
One of the main limitations of the formal model presented in
\refsection{our method}
is due to the assumption that both the concrete and
the abstract domains share the same set of locations \( \locations \). Any
abstract domain that aims to be practically applicable cannot rely on this
assumption.  From the definitions in
\refsection{our method}
we have that
for every variable created in a concrete execution there must be a distinct
location in the abstract memory description. This is obviously a problem since,
with the use of recursion and dynamic allocation, the number of variables
created during a  concrete execution can be unbounded.  But also when the
number of variables is known statically it is usually unfeasible to use a
one-to-one approximation; consider for instance the case of arrays:
under this assumption an abstract memory would be required to represent every
element of an array with a distinct location.  Typically, real
implementations use one abstract location to approximate a \emph{set} of
concrete locations. For instance, a simple strategy is to approximate all
the elements of an array, independently from their number, with the same
abstract location.  Previously we have used the symbol \( \locations \) to
denote the set of \emph{locations}. From now on we denote with \(
\locations \) the set of the \emph{concrete locations} and with \(
\abslocations \) a set that we call the \emph{abstract location set}.  We
still formalize the concrete domain as the complete lattice generated by
the powerset of the total functions \( \locations \to \locations \).
However, we have to adapt the definition of the abstract domain as follows.

\begin{definition}
\definitionsummary{Extended abstract domain}
Let \( \abstractdomain \) the \emph{support set of the abstract domain} be
defined as
\begin{gather*}
  \funlocations \defeq
  \concretedomain \times \locations \rightarrowtail \abslocations;
  \\
  \abstractdomain \defeq
    \funlocations \times \abslocations \times \abslocations.
\end{gather*}
In words, an element \( A \in \abstractdomain \) is a pair \( \langle
f, P \rangle \) where \( f \in
\funlocations \) represents the abstraction function from the concrete to
the abstract locations and \( P \subseteq \abslocations \times
\abslocations \) is an abstract points-to relation.
We call \emph{abstract domain} the complete lattice
\[
  \big< \abstractdomain, \sqsubseteq, \sqcup, \sqcap, \bot, \top \big>,
\]
where, for all \( \langle f, P \rangle,  \langle g, Q \rangle \in
\abstractdomain \), holds that
\begin{gather*}
  \langle f, P \rangle \sqsubseteq \langle g, Q \rangle
    \quad \defiff \quad
    f = g \land P \subseteq Q;
  \\
  \langle f, P \rangle \sqcap \langle g, Q \rangle
  \defeq
  \begin{cases}
    \langle f, P \intersection Q \rangle, & \text{if } f = g; \\
    \bot,                                 & \text{otherwise;}
  \end{cases}
  \\
  \langle f, P \rangle \sqcup \langle g, Q \rangle
  \defeq
  \begin{cases}
    \langle f, P \union Q \rangle, & \text{if } f = g; \\
    \top,                          & \text{otherwise.}
  \end{cases}
\end{gather*}
and the bottom (\(\bot\)) and top \((\top)\) elements are defined ad-hoc to satisfy the
properties of the complete lattice.
\end{definition}

Informally, given an abstract element \( \langle f, P \rangle = A \in
\abstractdomain \), for every concrete element \( C \in \concretedomain \)
and every concrete location \( l \in \locations \), \( f(C, l) \)
is the abstract location that in \( C \) abstracts \( l \).  The semantics
of the abstract domain can thus be defined as follows.

\begin{definition}
\definitionsummary{Extended abstract domain semantics}
Let \( C \in \concretedomain \) and
\( \langle f, P \rangle = A \in \abstractdomain \).
We define
\[
  C \in \concretization(A)
  \quad \defiff \quad
  \Bigl\{\,
    \bigl( f(C, l), f(C, m) \bigr)
  \Bigm|
    (l, m) \in C
  \,\Bigr\} \subseteq P.
\]
\end{definition}

The initial definition of the concretization function
(\refdefinition{concretization function}) simply checks if all the pairs of
\( C \) are also in \( A \); now, to handle the concept of
\emph{abstract locations}, every concrete points-to pair \( (l, m) \in C
\) is abstracted, obtaining the pair \( \bigl( f(C, l), f(C, m) \bigr) \),
and then we check in this ``abstract pair'' is in \( A \).
But the distinction between concrete and abstract
locations introduces a new problem in the formalization of the abstract
analysis.

\codecaption{
the annotations resulting from the use of strong updates.}
{weak strong update 2}
\begin{codesnippet}
int a, b, c, d, *p;
p = &a;       // \( \eval(\Code{*p}) = \{ \Code{a} \} \)
p = &b;       // \( \eval(\Code{*p}) = \{ \Code{b} \} \)
p = &c;       // \( \eval(\Code{*p}) = \{ \Code{c} \} \)
p = &d;       // \( \eval(\Code{*p}) = \{ \Code{d} \} \)
\end{codesnippet}
\codecaption{
the annotations resulting from the use of weak updates.}
{weak strong update 3}
\begin{codesnippet}
int a, b, c, d, *p;
p = &a;       // \( \eval(\Code{*p}) = \{ \Code{a} \} \)
p = &b;       // \( \eval(\Code{*p}) = \{ \Code{a}, \Code{b} \} \)
p = &c;       // \( \eval(\Code{*p}) = \{ \Code{a}, \Code{b}, \Code{c} \} \)
p = &d;       // \( \eval(\Code{*p}) = \{ \Code{a}, \Code{b}, \Code{c}, \Code{d} \} \)
\end{codesnippet}

\codecaption{
an example where it is necessary to apply weak updates to obtain a safe
approximation.}
{weak strong update 4}
\begin{codesnippet}
int **pp, *p1, *p2, a, b, c;
if (...)  pp = &p1;
else      pp = &p2;
p1 = &a;
p2 = &c;
*pp = &b;
\end{codesnippet}

\subsection{Weak Updates and Strong Updates}
This section gives an insight of the distinction between
\emph{weak} and \emph{strong} updates. In the literature,
the term \emph{update} usually means an operation that acts on a memory,
concrete or abstract, modifying its state.  An update can be triggered by
any the of usual operations, e.g., as the assignment
(\refdefinition{Assignment evaluation}).  However, the distinction between
\emph{strong} and \emph{weak} updates pertains only to the formalization of
the \emph{abstract} domain.  A \emph{strong} update has the effect of
\emph{overwriting} the previous information with new data; instead, a
\emph{weak} update acts by \emph{merging} the original with the new data.
Listings~\ref{codesnippet:weak strong update 2} and \ref{codesnippet:weak strong
update 3} present the different results of the analysis performed on the
same program: in the first case using strong updates, whereas in the second
case weak updates are applied.
By using weak updates it is not possible to increase the precision of the
approximation --- each weak update yields
a new abstraction that \emph{subsumes}
the original information.
Note that in \refcodesnippet{weak strong update 3}, to illustrate the
difference between the two options, we have \emph{forced} the analysis to
use weak updates.
However, there are situations where the use of weak updates is necessary to
obtain a safe approximation. Consider the example in \refcodesnippet{weak
strong update 4}.
The abstract execution reaches the last line with the approximation
\begin{gather*}
  \eval(\Code{*p1}) = \{ \Code{a} \}, \\
  \eval(\Code{*p2}) = \{ \Code{c} \}, \\
  \eval(\Code{*pp}) = \{ \Code{p1}, \Code{p2} \}.
\end{gather*}
By applying the assignment as presented in \refdefinition{Assignment
evaluation} we obtain the description
\begin{gather*}
  \eval(\Code{*p1}) = \{ \Code{a}, \Code{b} \}, \\
  \eval(\Code{*p2}) = \{ \Code{b}, \Code{c} \}, \\
  \eval(\Code{*pp}) = \{ \Code{p1}, \Code{p2} \}.
\end{gather*}
In this case the abstract assignment algorithm has performed a weak update:
the old values of the variables \QuotedCode{p1} and \QuotedCode{p2}
are not overwritten.  By forcing a strong update we would obtain instead
\begin{gather*}
  \eval(\Code{*p1}) = \{ \Code{b} \}, \\
  \eval(\Code{*p2}) = \{ \Code{b} \}, \\
  \eval(\Code{*pp}) = \{ \Code{p1}, \Code{p2} \},
\end{gather*}
which is clearly a wrong approximation because there exists at least a
concrete execution such that, after the execution of the assignment
\QuotedCode{*p = \&b}, \( \eval(\Code{*pp}) = \{ \Code{p2} \} \) holds and
then \( \eval(\Code{*p1}) = \{ \Code{a} \} \).  Note that in the definition
of the abstract assignment (\refdefinition{Assignment
evaluation}), given \( (e, f) \in \assignments \), what triggers the use
of a \emph{strong} instead of a \emph{weak} update is the fact that the
lhs \( e \) evaluates to a single location:
\[
  K \defeq
    \begin{cases}
      \cdots,      & \text{if } \cardinality \eval(A, e) = 1; \\
      \emptyset,   & \text{otherwise}.
    \end{cases}
\]
where \( K \) denotes the set of the \emph{killed points-to pairs}.
The basic idea behind this approach is that when we have to update a set of
\emph{more than one} location
it is possible that there exists a concrete memory description approximated
by the current abstraction in which
only one of the locations of this
set will be modified while the others will retain their original value.
In the above example when \QuotedCode{pp} points to \QuotedCode{p1} then
\QuotedCode{p2} is left unchanged by the assignment \QuotedCode{*pp = \&b}.
Otherwise, when we are sure that the there is \emph{only one} possible
modified location we can afford that in none of the concrete memories \( C
\in \concretization\bigl( \operation(A, \cdots) \bigr) \) that location
will still have the old value.
However, by distinguishing between concrete and abstract locations, we
are no more able to discern when a strong update can be used.
Now, also when the lhs evaluates to a single
location, \( \eval(A, e) = \{ l^\sharp \} \),
we cannot safely apply a strong update as
it is possible that \( l^\sharp \) abstracts \emph{more that one}
concrete locations.  To overcome this problem we introduce the following
definition.

\begin{definition}
\definitionsummary{Singular locations}
Let
\[
  \singularloc \subseteq \abstractdomain \times \abslocations
\]
be defined as follows.
Let \( \langle f, P \rangle = A \in \abstractdomain \) and \( l^\sharp \in
\abslocations \). We say that the location \( l^\sharp \) is
\emph{singular} in the memory abstraction \( A \) when
\[
  (A, l^\sharp) \in \singularloc
  \quad \defiff \quad
  \forall C \in \concretization(A) \itc
    \cardinality
      \bigl\{\, l \in \locations \bigm| f(C, l) = l^\sharp \,\bigr\}
    \leq 1.
\]
\end{definition}

The above definition can be read as follows. We say that an abstract
location \( l^\sharp \) is singular with respect to the abstract memory
description \( A \in \abstractdomain \) if it does not exist any concrete
memory description  \( C \in \concretization(A) \) such that \( l^\sharp \) approximates more than
one of the locations of \( C \).
For convenience of notation we write \( \singularloc(A) \) to denote the
set of the singular locations of the memory \( A \), i.e,
\[
  \singularloc(A) \defeq
    \bigl\{\,
      l^\sharp
    \bigm|
      (A, l^\sharp) \in \singularloc
    \,\bigr\}
\]
The abstract assignment operation (\refdefinition{Assignment evaluation})
must be adapted in order to provide a safe approximation. In particular,
the definition of the \emph{kill} set needs to be rewritten as
\begin{gather*}
  E \defeq \eval(A, e); \\
  K \defeq
    \begin{cases}
      E \times \locations,
      & \text{if } \cardinality E = 1 \land E \subseteq \singularloc(A); \\
      \emptyset,
      & \text{otherwise}.
    \end{cases}
\end{gather*}
Also the definition of the filter operation (\refdefinition{Filter 1})
must be updated accordingly. Given \( x \in \abstractdomain \times
\partsof(\locations) \times \expressions \) and \( i \in \naturals
\); we have
\begin{gather*}
  K \defeq \locations \setminus \target(x, i);
  \\
  T \defeq \target(x, i + 1);
  \\
  \filter(x, i + 1) \defeq
  \filter(x, i)
  \setminus
  \begin{cases}
      K,
      & \text{if } \cardinality T = 1 \land
      T \subseteq \singularloc(A); \\
    \emptyset, & \text{otherwise}.
  \end{cases}
\end{gather*}
Also the definition of the filter for the `\(\inequality\)' operator
(\refdefinition{Filter 3}) needs to be updated accordingly. Given \( A \in
\abstractdomain \) and \( e, f \in \expressions \); let
\begin{gather*}
  I \defeq \eval(A, e) \intersection \eval(A, f); \\
  E \defeq \eval(A, e) \setminus \eval(A, f); \\
  F \defeq \eval(A, f) \setminus \eval(A, e);
\end{gather*}
then
\[
  \filter\bigl( A, (\inequality, e, f) \bigr) \defeq
      \begin{cases}
        \filter( A, E, e ) \union \filter( A, F, f ),
          & \text{if }
          \cardinality I = 1 \land I \subseteq \singularloc(A); \\
        A, & \text{otherwise}.
      \end{cases}
\]
Finally, also the definition of the alias relation induced by a points-to
abstraction must be adapted in the same way. From \refdefinition{Induced
alias relation}, for all \( A \in \abstractdomain \), we define
\( \concretization(A) \defeq \mathord{\absaliasquery} \) as follows.
For every \( e, f \in expressions \)
\begin{gather*}
  E \defeq \eval(A, e); \\
  F \defeq \eval(A, f); \\
  \absaliasquery(e, f) \defeq
    \begin{cases}
      0,  & \text{if } E \intersection F = \emptyset; \\
      1,  & \text{if } \cardinality E = 1 \land E = F \land E \subseteq \singularloc(A); \\
      \top, & \text{otherwise.}
    \end{cases}
\end{gather*}
With these modified definitions, assuming for instance to approximate all
the elements of an array with only one (non-singular) abstract location,
the analysis applied to the code in \refcodesnippet{singular locations}
produces the indicated annotations.

\codecaption
{this code shows the difference between singular and non-singular
locations. Remember that in the C language global variables are
zero-initialized. Assume that all the indices left unspecified are valid.}
{singular locations}
\begin{codesnippet}
int *p[10], *q;

int main() {
  int x, y, z;
                // \( \eval(\Code{*p}) = \{ \Code{null} \} \)
  p[...] = &x;  // \( \eval(\Code{*p}) = \{ \Code{null}, \Code{x} \} \)
  p[...] = &y;  // \( \eval(\Code{*p}) = \{ \Code{null}, \Code{x}, \Code{y} \} \)
  p[...] = &x;  // \( \eval(\Code{*p}) = \{ \Code{null}, \Code{x}, \Code{y} \} \)
  p[...] = &z;  // \( \eval(\Code{*p}) = \{ \Code{null}, \Code{x}, \Code{y}, \Code{z} \} \)

                // \( \eval(\Code{*q}) = \{ \Code{null} \} \)
  q = &x;       // \( \eval(\Code{*q}) = \{ \Code{x} \} \)
  q = &y;       // \( \eval(\Code{*q}) = \{ \Code{y} \} \)
  q = &x;       // \( \eval(\Code{*q}) = \{ \Code{x} \} \)
  q = &z;       // \( \eval(\Code{*q}) = \{ \Code{z} \} \)
}
\end{codesnippet}

\subsection{Notation}
\label{section:extended memory model}
In the following description we use more than once the concept of
\emph{sequence}. With \emph{sequence} we mean a set \( S \) whose elements
are enumerated, thus they can be identified and compared against their
\emph{position} inside the sequence. With \emph{position} we mean an
index ranging from \( 0 \) up to \( n \) where \( n + 1 \) is the number of
elements\footnote{The concept of \emph{position} is not defined for the
empty sequence.} of \( S \).  For convenience of notation we write `\(
S.\size \)' to denote the number of elements of the sequence \( S \); we
write \( S_i \) or \( S(i) \) to denote the element of \( S \) with index
\(i\) and \( \domain(S) \) as an abbreviation of the set of the indices
of \( S \), i.e.,
\(
  \domain(S) =
    \{\, n \in \naturals \mid 0 \leq n < S.\size \,\}.
\)
To explicitly represents the elements of the sequence we write
\( S = [ S_0, \cdots, S_n ] \).
When we are not interested
in the definition of any particular order among the elements of \( S \),
we use the concept of \emph{labelled set}.
A labelled set can be defined as the triple
\( \langle F, L, S \rangle, \)
where \( S \) is the set of the \emph{labelled} elements, \( L \) is a
set of \emph{labels} and
\( \parfunctiondef{ F }{ L }{ S } \)
is a partial surjective \emph{labelling}
function that gives a unique name, or label, to all the elements of \( S \).
For convenience of notation,
when \( F \) and \( L \) are clear from the context,
we write only \( S \) to refer to the labelled set \(
\langle F, L, S \rangle \); we write \( S_l \) or \( S(l) \) as an
abbreviation of \( F(l) \) and \( \domain(S) \) as a shortcut for \(
\domain(F) \).
To explicitly represent the elements of \( S \) we write \( S = \{ S_0,
\cdots, S_n \} \).\footnote{That is, at the only extent of denoting the
elements of \( S \), we \emph{enumerate} it.}
Note that this definition of labelled set is a
generalization of the concept of sequence where \( L = \naturals \) ---
hence the following definitions given for labelled sets can be applied also
to sequences.  We use also the concept of \emph{attribute}.
Given two labelled sets \( S \) and \( A \),
we say that the pair
\( \langle S, A \rangle \)
is a \emph{labelled set with attributes set \( A \)}.
Again, when the attribute set \( A \) is clear from the
context, we write \( S \) to mean the pair \( \langle S, A \rangle \);
we say that \( S \) has the \emph{attribute} \( X \) to mean that \( X
\in \domain(A) \) and we write `\( S.X \)' as a shortcut for `\( A(X) \)'.

\subsection{The Concept of Memory Shape}
The abstract memory model that we want to describe is parametric with
respect to the underlying abstract domain, e.g., the points-to
domain or some numerical domain.  In other words, the analysis can be seen
as the coupling of a chosen abstract domain and some additional
`structural' information, concerning for instance the memory model of the
target language/machine.  With the concept of \emph{shape} we want to formalize this
`structural' information.  Recalling
the definition of the extended abstract domain \refdefinition{Extended abstract
domain}, this information is needed to identify the function \( f \in
\funlocations \), that is, how concrete locations are mapped to abstract
locations.

\begin{definition}
\definitionsummary{Shape of a labelled set}
Let \( \langle F, L, S \rangle \) be a given labelled set.  We define the
shape of \( \langle F, L, S \rangle \) as the other labelled set
\[
  \shape\bigl( \langle F, L, S \rangle \bigr) \defeq
  \langle G, L, T \rangle,
\]
where
\[
  T \defeq \bigl\{\, \shape(e) \bigm| e \in S \,\bigr\},
\]
and \( \parfunctiondef{ G }{ L }{ T } \)
is such that \( \domain(G) = \domain(F) \) and is defined, for all \( l \in
\locations \), as
\[
    G(l) \defeq \shape\bigl( F(l) \bigr).
\]
Now let \( \langle S, A \rangle \) be a labelled set with the attribute
set \( A \).
We define its shape as
\[
  \shape\bigl( \langle S, A \rangle \bigr) \defeq
    \bigl\langle \shape(S), A \bigl\rangle.
\]
\end{definition}

Note that, as a consequence of this definition, the shape of a sequence
\(
  S = [S_0, \cdots, S_n]
\)
is the sequence of the shapes
\[
  \shape\bigl([S_0, \cdots, S_n]\bigr) \defeq
    \bigl[ \shape(S_0), \cdots, \shape(S_n) \bigr].
\]
Note also that the \(\shape\) function does not change
the attributes of a labelled set.

\subsection{Common Concepts}
The following sections will describe the structure of both the concrete
and abstract memory models.
Before proceeding we need to introduce some common concepts.  We refer to
\cite{BagnaraHPZ07TR}
for a rigorous formalization of some of the ideas that
we present only informally.
\begin{description}
\item[Location.]
  The basic unit for describing the structure of the memory is the concept of
  \emph{location}. Each location has a `\( \type \)' attribute.
\item[Allocation.]
  We use the concept of \emph{allocation block} to describe the unit of
  allocation of the memory. An allocation block is a \emph{sequence of
  locations}, it has a `\( \type \)' attribute and it is
  the base case of the inductive definition of the
  concept of shape (\refdefinition{Shape of a labelled set}).  We define
  the shape of an allocation block \( A \) as its `type'
  attribute\footnote{That is, the shape of an allocation \emph{is} its
  `type' attribute and \emph{not} the sequence of the shapes of its
  locations.}
  \[
    \shape(A) \defeq A.\type.
  \]
  The type attribute of an allocation block uniquely determines the
  shape of the sequence of its locations; the details of this aspect will be clarified
  later.  Informally, we say that each variable definition in the analyzed
  program has the effect of creating an allocation block in the memory, or,
  if speaking of an abstract memory, updating an already existing allocation
  block.  In the next, when clear from the context, we call an
  allocation block simply \emph{allocation}.
\end{description}
To describe the structure of the concrete \emph{stack} and its
abstraction we introduce these definitions.
\begin{description}
\item[Block.]
  We use the term \emph{deallocation block} to mean a sequence of
  allocations.  The deallocation block is the unit for the deallocation of
  stack allocated memory.  Ideally, the deallocation block
  is intended to represent the concept of \emph{block} of
  declarations as it is defined by the C language.  The order of the
  allocations inside a deallocation block reflects the order of creation
  of the variables.  For conformance with the C Standard, when clear from
  the context we will refer to a deallocation block simply as a
  \emph{block}.  A block can also be described as the portion of the
  stack between two subsequent block marks \cite{BagnaraHPZ07TR}.
\item[Frame.]
  With \emph{frame} we mean a sequence of blocks. In the concrete memory
  model, a frame can also be characterized as the portion of the stack
  segment between two subsequent link marks \cite{BagnaraHPZ07TR}.  Each
  link mark uniquely identifies the call statement   that has generated
  the link mark.  To identify the call statements of the program under
  analysis we use the concept of \emph{call site} --- each call statement
  in the program is uniquely identified by a call site.\footnote{A
  reasonable choice to implement the \emph{call site} concept is to use the
  \emph{program point}
  associated to the call statement. However, for clarity we want to
  keep separate the concept of call site and program point.} Each frame
  has a `\( \callsite \)' attribute. The value of this attribute is equal
  to the call site of the link mark that closes the frame --- with this
  definition, from the program source code, the call site of a frame
  uniquely determines the shape of the whole frame.
\end{description}

\begin{example}
\codecaption
{the call site completely identifies the shape of the frame.}
{frame shape}
\begin{codesnippet}
void g();

void f(int p) {
  int a, b;
  if (...) {
    int c;
    ...
    g();      // Call site 1
  } else {
    int d, e;
    ...
    g();      // Call site 2
  }
}
\end{codesnippet}
Consider the code in \refcodesnippet{frame shape}. The frame
identified by the call site 1, that corresponds to the call statement at
\refline 8, can be described as
\[
  \bigl[ [\Code{int p}], [\Code{int a}, \Code{int b}], [\Code{int c}] \bigr].
\]
Instead, the frame
identified by the call site 2, that corresponds to the call statement at
\refline{12}, can be described as
\[
  \bigl[ [\Code{int p}], [\Code{int a}, \Code{int b}], [\Code{int d},
  \Code{int e}] \bigr].
\]
\end{example}
In the next we apply to the concepts just introduced the qualifiers
\emph{concrete} and \emph{abstract}. If \( X \) is a labelled set of \( Y
\) objects, then with \emph{concrete} \( X \) we mean a labelled set of
\emph{concrete} \( Y \) objects; with \emph{abstract} \( X \) we mean a
labelled set of \emph{abstract} \( Y \) objects.  For example we call
`abstract frame' a sequence of abstract blocks; with `concrete allocation'
we mean a sequence of concrete locations.  When the qualifier
abstract/concrete is not specified, the context will clarify the intended
one or if the statement is applicable to both cases.

\subsection{The Concrete Memory Model}
The concrete memory is organized as a labelled set of \emph{segments}.
\begin{description}
\item[Text.]
  The \emph{text} segment is a labelled set of allocations used to
  represent the set of the possible targets of function pointers:
  basically there is one allocation for each function declared in the
  analyzed program.  Each allocation is identified by the program point
  associated to the function declaration.\footnote{In case the same
  function is defined multiple times, then obvious disambiguation methods
  are necessary; for example, as considering only the first occurrence of the
  declaration.}
\item[Heap.]
  The \emph{heap} segment is a labelled set of allocations used to
  represent the objects created using the functions of the
  \QuotedCode{malloc} family. In this segment each allocation is labelled
  by an address\footnote{At this level we are not interested in the details
  of the addressing schema of the concrete execution model. We simply
  require that each heap allocation can be identified inside the segment
  by a tag or \emph{address}.}
  and has the attribute `\( \programpoint \)' (program point) that
  uniquely identifies the statement that has caused the allocation.  Note
  that once fixed the analyzed program, the program point of the
  allocating statement identifies the type attribute of the allocation,
  that is, the shape of the allocation.  As a consequence, given two heap
  allocations with the same program point attribute we know that these
  allocations have also the same shape.
\item[Global.]
  The \emph{global} segment is a sequence of allocations that represent
  the global variables of the analyzed program. Note that the order of the
  allocations inside the global segment is not specified by the C
  Standard; thus, this detail is left to the particular execution model
  implemented; for instance, this order may be influenced by
  the particular combination of
  architecture/compiler chosen as target the for the analysis.
\item[Stack frames.]
  The \emph{stack frames} segment is a sequence of frames. The sequence is
  organized such that the frame of index 0 represents the topmost
  frame\footnote{With \emph{topmost frame} we mean the most recent frame on
  the stack, that is the frame below the topmost link mark.}  on the stack
  and the frame of index \( n \) ---where \( n+1 \) is the size of this
  segment--- represents the oldest frame.\footnote{For instance, in
  the analysis of a complete program, the oldest frame, if present, is
  generated by one of the call statements contained in the
  \QuotedCode{main()} function.}
\item[Stack top.]
  The \emph{stack top} segment represents the locations above the topmost
  link mark.  It is a sequence of blocks (not contained in any frame),
  followed by a sequence of allocations (not contained in any block.).
\end{description}
Then we define a concrete memory \( a \in \memories \) as a labelled set of
the form
\[
  a = \{ \textseg, \heap, \globalseg, \stackframes, \stacktop \}.
\]
For convenience of notation we use the notation `\( a.X \)' to refer to the
\( X \) segment of the memory \( a \). For example, we write `\( a.\textseg
\)' to denote the text segment of \( a \).  Before describing how the
\emph{type} attribute of a concrete allocation determines the shape of the
sequence of its locations, we need to introduce some notation.

\begin{definition}
\definitionsummary{Concatenation of sequences}
Let \( A = [ A_0, \cdots, A_n ] \) and \( B = [ B_0, \cdots, B_m ] \)
be two sequences;
then we define \( A \concat B \) as the concatenation of the two sequences
\[
  A \concat B \defeq [ A_0, \cdots, A_n, B_0, \cdots, B_m ].
\]
\end{definition}

\begin{definition}
\definitionsummary{Concrete allocations}
We define the `\( \allocation \)' function by structural induction on the
set of types `\( \types \)'.
Let \( t \in \types \).  If \( t \) is a scalar type or a function type
then\footnote{For the definition of the concept of \emph{type} we refer to
the C Standard \cite[6.2.5.21]{ISO-C-1999}: \emph{arithmetic} types and
\emph{pointer} types are collectively called \emph{scalar} types. Array and
structure types are collectively called \emph{aggregate} types.}
\[
  \allocation(t) \defeq [ t ].
\]
If \( t \) is an
array of \( n \in \naturals \) elements of type \( t_0 \) then
\[
  \allocation(t) \defeq
    \underbrace{\allocation(t_0) \concat \cdots \concat \allocation(t_0)}_{
    n + 1 \text{ times}}.
\]
If \( t \) is a structure type with fields:
\( t_0 \text{ field}_0; \cdots; t_n \text{ field}_n; \)
we define
\[
  \allocation(t) \defeq
    \allocation(t_0) \concat \cdots \concat \allocation(t_n).
\]
\end{definition}

\codecaption
{the definition of an aggregate type.}
{aggregate type definition}
\begin{codesnippet}
struct A {
  int a[4];
  float b;
};

struct B {
  double x;
  struct A a;
  char y;
};
\end{codesnippet}

\begin{example}
Consider \refcodesnippet{aggregate type definition}; then we have
\begin{gather*}
  \allocation(\Code{int[4]})
  = [\underbrace{\Code{int}, \Code{int}, \Code{int},
  \Code{int}, \Code{int}}_{ 5 \Code{ times}}];
  \\
  \allocation(\Code{struct A})
  = [\Code{int}, \Code{int}, \Code{int},
  \Code{int}, \Code{int}, \Code{float}]; \\
  \allocation(\Code{struct B}) =
    [ \Code{double},
      \Code{int}, \Code{int}, \Code{int}, \Code{int},
      \Code{int}, \Code{int},
      \Code{float}, \Code{char}].
\end{gather*}
\end{example}

\subsection{The Abstract Memory Model}
Now, having introduced these basic ingredients, we can describe the
organization of the abstract memory that, as the concrete memory
model, is composed by different \emph{segments}.
\begin{description}
\item[Text.]
  The \emph{text} segment is a labelled set of abstract allocations that
  used as targets for function pointers.  The definition of the abstract
  text segment is the same of the concrete case:  there is one text
  location for each function declared in the analyzed program and
  each location is labelled by the program point associated to the function
  declaration.
\item[Heap.]
  The \emph{heap} segment is a labelled set of allocations used to abstract
  all the possible heap-allocated objects. Each heap allocation has as
  attribute the program point of the statement that has caused the
  allocation which is also used as label to identify the allocation inside the
  segment.  This means that the abstract heap segment contains only one
  allocation for each allocating statement of the analyzed program.
\item[Global.]
  The \emph{global} segment is a sequence of allocations that represents
  the global variables of the analyzed program. The order of the
  allocations inside the \emph{abstract} global segment is chosen to reflect the layout
  of the \emph{concrete} global segment.
\end{description}
To represent the abstraction of the concrete stack we use three distinct
segments.
\begin{description}
\item[Stack top.]
  The \emph{stack top} segment represents the portion of the stack
  above the topmost link mark. As in the concrete case,
  the stack top is formalized as a sequence of blocks (not
  contained in any frame), followed by a sequence of allocations (not
  contained in any block.)
\item[Stack head.]
  The \emph{stack head} segment is a sequence of frames.
\item[Stack tail.]
  The \emph{stack tail} segment is a labelled set of frames where each
  frame is labelled by its `call site' attribute.  This means that
  the stack tail contains at most one frame for each of the possible call
  sites of the analyzed program.
\end{description}
Finally, we define an abstract memory \( a \in \absmemories \) as a labelled
set
\[
  a = \{ \textseg, \heap, \globalseg, \stacktail, \stackhead, \stacktop \}.
\]
As for the concrete memory, for convenience of notation
we write `\( a.X \)' to refer to the \( X \) segment of the
abstract memory \( a \); for example we write  `\( a.\textseg \)' to denote
the text segment of the abstract memory \( a \).
Now we present how the type of an abstract allocation determines the shape
of the sequence of its locations.

\begin{definition}
\definitionsummary{Abstract allocations}
Let \( t \in \types \).  If \( t \) is a scalar type or a function type, then
\[
  \absallocation(t) \defeq [ t ],
\]
If \( t \) is an
array of \( n \in \naturals \) elements of type \( t_0 \) then
\[
  \absallocation(t) \defeq
  \absallocation(t_0) \concat \absallocation(t_0) \concat \absallocation(t_0).
\]
If \( t \) is a structure type with fields:
\( t_0 \text{ field}_0; \cdots; t_n \text{ field}_n; \)
we define
\[
  \absallocation(t) \defeq
    \absallocation(t_0) \concat \cdots \concat \absallocation(t_n).
\]
\end{definition}

\begin{example}
Consider again \refcodesnippet{aggregate type definition}; this time we have
\begin{gather*}
  \absallocation(\Code{int[4]})
  = [\underbrace{\Code{int}, \Code{int}, \Code{int}}_{ 3 \Code{ times}}];
  \\
  \absallocation(\Code{struct A})
  = [\Code{int}, \Code{int}, \Code{int}, \Code{float}]; \\
  \absallocation(\Code{struct B}) =
    [ \Code{double},
      \Code{int}, \Code{int}, \Code{int},
      \Code{float}, \Code{char}].
\end{gather*}
\end{example}

Note that we approximate arrays using three parts.  In \refsection{pointer
arithmetics} we show how these parts can be used by the analysis.

\subsection{The Lattice Structure}
As in \refsection{our method}, we formalize the concrete domain as the
complete lattice generated by the powerset of the concrete memories \(
\memories \).  Our next step is to introduce the missing elements required
to complete the structure of complete lattice for the abstract domain.  The
bottom (\(\bot\)) and the top (\(\top\)) elements are defined ad-hoc.  Now
we introduce the two binary operations of meet (\( \memmeet \)) and join
(\( \memjoin \)) and the partial order (\( \memsubsumedby \)).  In our
analysis the operations of join and meet, as well as the query on the
partial order, always occur between abstractions having a \emph{similar
structure}; these are the cases that we consider ``interesting'' and on which
we define the operations. However, since the formalization requires
\emph{total} operations, we will extend the definition to
``non-interesting''
cases in a trivial way, that is when asked to compute the join or the meet,
we will simply answer \(\top\) and \(\bot\), respectively.  Note that this
is a specialization of the behaviour described in \refdefinition{Extended
abstract domain}. In this sense, when we say that two elements of \(
\absmemories \), say \( \langle f, P \rangle, \langle g, Q \rangle \),
share a similar structure we mean that \( f = g \).  To formalize the
concept of \emph{similar structure} we introduce the relation
`\(\mathord{\iscompatible}\)'.

\begin{definition}
\definitionsummary{Compatibility between abstract memories}
Let
\[
  \mathord{\iscompatible} \subseteq \absmemories \times \absmemories
\]
be defined as follows. Let \( A, B \in \absmemories \);
then we say that \( (A, B) \in \iscompatible \) when the following
conditions hold:
\begin{gather*}
  \shape(A.\textseg) = \shape(B.\textseg); \\
  \shape(A.\heap) = \shape(B.\heap); \\
  \shape(A.\globalseg) = \shape(B.\globalseg); \\
  \shape(A.\stacktop) = \shape(B.\stacktop); \\
  \shape(A.\stackhead) = \shape(B.\stackhead).
\end{gather*}
\end{definition}

Note that in the definition of the `\(\iscompatible\)' relation,
no constraints are specified on the shape of the stack tail segment.

\begin{definition}
\definitionsummary{Abstract domain partial order}
Let \( A \) and \( B \) be two labelled sets.\footnote{As said above
this definition is valid also if \( A \) and \( B \) are \emph{sequences},
as the \emph{sequence} is a particular case of \emph{labelled set}.}
Let
\[
    A \subsumedby B
  \quad \defiff \quad
    \domain(A) \subseteq \domain(B) \land
    \forall l \in \domain(A) \itc A(l) \subsumedby B(l).
\]
Let \( A, B \in \absmemories \).  We say that
\[
  A \memsubsumedby B
  \quad \defiff \quad
    (A, B) \in \iscompatible \land A \subsumedby B.
\]
\end{definition}

Note that this definition proceeds inductively on the structure of the
abstract memory. The base case of this induction are \emph{locations}.  On
locations, the definition of the partial order
`\(\mathord{\memsubsumedby}\)', of the operations `\(\mathord{\memjoin}\)'
and `\(\mathord{\memmeet}\)', depends on the particular abstract domain
adopted.

\begin{example}
With \emph{location address} we mean an information that allow to identify
a location inside a memory.  If the abstract memory is based on a
points-to domain, locations are formalized as \emph{sets of location
addresses} --- a set of location addresses is used to represent the set of the possibly pointed
locations.  In this case, the partial order on locations is simply the
relation of containment `\(\mathord{\subseteq}\)' between sets of location
addresses.
\end{example}

\begin{definition}
\definitionsummary{Abstract domain join operation}
Let \( A \) and \( B \) be labelled sets.
We define \( A \join B \) such that
\( \domain(A) \union \domain(B) = \domain(C) \) and, for all \( l \in
\locations \),
\[
  (A \join B)(l) \defeq
  \begin{cases}
    A(l), & \text{if } l \in \domain(A) \setminus \domain(B); \\
    B(l), & \text{if } l \in \domain(B) \setminus \domain(A); \\
    A(l) \join B(l),
          & \text{otherwise.}
  \end{cases}
\]
Let \( A, B \in \absmemories \).  We define
\[
  A \memjoin B \defeq
  \begin{cases}
    A \join B, & \text{if } (A, B) \in \iscompatible; \\
    \top,      & \text{otherwise.}
  \end{cases}
\]
\end{definition}

\begin{definition}
\definitionsummary{Abstract memory meet operation}
Let \( A \) and \( B \) be labelled sets. We define
\( A \meet B \) such that
\( \domain(A) \intersection \domain(B) = \domain(C) \) and, for all \( l
\in \locations \),
\[
  (A \meet B)(l) \defeq A(l) \meet B(l).
\]
Let \( A, B \in \absmemories \).  We define
\[
  A \memmeet B \defeq
  \begin{cases}
    A \meet B, & \text{if } (A, B) \in \iscompatible; \\
    \bot,      & \text{otherwise.}
  \end{cases}
\]
\end{definition}

In the computation of the meet operation it is possible to reach the bottom
on some of the locations.\footnote{Locations represents elements of the
underlying abstract domain. Computing the meet between two locations, it is
possible to reach the bottom of the abstract domain.} Depending on the
position of the locations inside the abstract memory, this bottom can be
propagated.  If the bottom is reached on a location contained in the stack
tail, then the bottom can be propagated to the frame that contains the
location: this is equivalent to removing the frame from the stack tail. If
the bottom location is in any other segment then the bottom can be extended
to the whole memory.  The reason of this will be clarified by the
definition of the semantics of the abstract memory.

\subsection{Concretization Function of the Abstract Memory}
This section presents the concretization function for the
abstract memory model \( \absmemories \).  The definition proceed
by structural induction on the definition of abstract memory.  The first step is to
find a mapping between the shape of the concrete memory and the shape of
the abstract memory. Note that at this point we are not interested in
dealing with the \emph{value} of the memory ---which is defined by the
value of the locations--- but only in describing a relation about the
\emph{shape}.  In other words, given an abstract element \( \langle f, P
\rangle \in \absmemories \) and a \( m \in \memories \),
we are now trying to identify the function \( f \in \funlocations \) is
defined on \( m \) (\refdefinition{Extended abstract domain}).
As already done for the definition of the operations of meet and join, we
first formulate a \emph{compatibility} relation to express the requirements
on the \emph{structure} of the concrete and abstract memories.

\begin{definition}
\definitionsummary{Compatibility between concrete and abstract memories}
Let
\[
  \mathord{ \iscompatible }
  \subseteq \memories \times \absmemories.
\]
Let \( A \in \absmemories \) and \( C \in \memories \); then we say that \(
(C, A) \in \iscompatible \) when hold the following conditions
\begin{gather*}
  \shape(C.\textseg) = \shape(A.\textseg); \tag{1} \\
  \forall l \in \domain(C.\heap) \itc
    C.\heap(l).\programpoint \in \domain(A.\heap); \tag{2} \\
  \shape(C.\globalseg) = \shape(A.\globalseg); \tag{3} \\
  \shape(A.\stacktop) = \shape(C.\stacktop); \tag{4} \\
  A.\stackhead.\size \leq C.\stackframes.\size; \tag{5} \\
  \forall i \in \{ 0, \cdots, A.\stackhead.\size - 1 \} \itc \\
  \quad
    \shape\bigl(A.\stackhead(i)\bigr)
      = \shape\bigl(C.\stackframes(i)\bigr); \tag{6} \\
  \forall i \in \{ A.\stackhead.\size, \cdots, C.\stackframes.\size - 1 \}
  \itc \\
  \quad C.\stackframes(i).\programpoint \in \domain(A.\stacktail). \tag{7}
\end{gather*}
\end{definition}

In words, a concrete memory \( C \in \memories \) and an abstract memory \(
C \in \absmemories \) are compatible when holds the following conditions.
\begin{enumerate}
\item
  The shapes of the text segments must be the equal. From the definition,
  both the concrete and the abstract segment contain an abstract allocation
  for every declared function.  Hence, as long as \( A \) and \( C \)
  refer to the same program, this is always true.
\item
  Recall that, within the concrete heap segment, allocations are identified
  by \emph{addresses}; whereas, in the abstract heap segment, allocations
  are identified by \emph{program points}.  For the heap segment we require
  that to each concrete heap allocation there corresponds an abstract heap
  allocation identified by the program point of the concrete allocation.
\item
  The shapes of the global segments must be equal.
  from the definition of shape, this implies that the global segments must
  contain the same number of allocations and that each concrete allocation
  corresponds to an abstract allocation with the same type. Again,
  as long as \( A \) and \( C \) refer to the same program this property is
  always true.
\item
  The \emph{stack top} segments must have the same shape; that is,
  the parts of the stack above the topmost link mark must have the same shape.
\item
  The stack frames segment of \( C \) does not contain less
  frames than the stack head segment of \( A \).
\item
  The shape of stack head segment of \( A \)
  must be a prefix of the shape of the \( \stackframes \) segment of \( C \).
\item
  The remaining part of the stack frames segment of \( C \) must be
  compatible with the stack tail segment of \( A \).  Recall that in the
  stack tail the frames are identified by their call site; thus, this means
  that to every frame of \( C.\stackframes \) corresponds in \(
  A.\stacktail \) a frame with the same call site.
\end{enumerate}

Given a concrete memory \( m \in \memories \) and an abstract memory \(
m^\sharp = \langle f, P \rangle \in \absmemories \), once we know that the
\( m \) is compatible with \( m^\sharp \), we ask how the locations of \( m
\) map onto the locations of \( m^\sharp \), that is, \emph{how} the
function \( \functiondef{ f(m, \cdot) }{ \locations }{ \abslocations } \)
is defined, as this is required
in order to complete the definition of the semantics of the abstract domain
(\refdefinition{Extended abstract domain semantics}).
Before going into the details we introduce the idea behind
the approach.  By looking at the definitions of the concrete and abstract
memories, note that these objects can be seen as \emph{trees} --- every
labelled set is a node with its elements as children.
If \emph{memories are trees}, then we can characterize \emph{locations} as
the \emph{leaves}. In other words, a location can be uniquely
identified within a memory by the \emph{path} that connects the root of the tree
to the corresponding leaf node.  Under these assumptions we can identify
\emph{concrete location addresses} as the paths inside the concrete memories
and the \emph{abstract location addresses} as the paths inside the abstract memories.
We
can now restate our initial problem as the problem of determining a mapping
from paths on a concrete tree to paths on an abstract tree.  To do this we
exploit the recursive structure of trees --- for each subtree \( S^\sharp
\) of \( m^\sharp \) we have to identify the set of subtrees \( S_0,
\cdots, S_n \) of \( m \) that are mapped into \( S^\sharp \);  being the
leaves the limit case of subtrees, we will end up having a
map from the leaves of \( m \) to the leaves of \( m^\sharp \). To
formalize this mapping we use triples of the form \[ \bigl< S^\sharp, \{
S_0, \cdots, S_n \}, M \bigr>, \] where \( S^\sharp \) is a subtree of \(
m^\sharp \), \( S_0, \cdots, S_n \) are the subtrees of \( m \)
mapped into \( S^\sharp \) and \( M \) is a map that defines how the
children of \( S_0, \cdots, S_n \) are mapped to the children of \(
S^\sharp \).

\begin{definition}
\definitionsummary{Concretization of allocations}
We define the \( \mathord{\absmap} \) function by structural induction on
the set \( \types \).  Let \( t \in \types \); then
\[
  \absmap(t) \defeq
  \begin{cases}
    \Bigl\{ \bigl< 0, \{ 0 \}, \emptyset \bigr> \Bigr\},
      \\ \qquad \text{if \( t \) is scalar or function type;}
    \\
    \Bigl\{
      \bigl< 0, \{ 0 \}, \absmap(t_0) \bigr>, \\
      \quad \bigl< 1, \{ 1, \cdots, n - 1 \}, \absmap(t_0) \bigr>, \\
      \quad \bigl< 2, \{ n \}, \absmap(t_0) \bigr>
    \Bigr\},
    \\ \qquad \text{if \( t \) is an array of size \( n \) of type } t_0;
    \\
    \Bigl\{\,
      \bigl< i, \{ i \}, \absmap(t_i) \bigr>
    \Bigm|
      i \in \{ 0, \cdots, n \}
    \,\Bigr\},
      \\ \qquad \text{if } t = \text{struct}
      : t_0 \text{ field}_0, \cdots, t_n \text{ field}_n;
  \end{cases}
\]
Let \( a_0, \cdots, a_n, a^\sharp \)
be allocations such that
\[
  \forall i \in \{ 0, \cdots, n \} \itc
    a^{ \sharp }.\type = a_i.\type.
\]
Then we define
\[
  \absmap\bigl( a^\sharp, \{ a_0, \cdots, a_n \} \bigr) \defeq
    \bigl< a^\sharp, \{ a_0, \cdots, a_n \}, \absmap(a_0.\type) \bigr>.
\]
\end{definition}

Recall that allocations are the base case of the definition of
\(\mathord{\shape}\): the shape of an allocation is its type attribute.
This means that the above condition on the types of the allocations is
equivalent to say that \( a_0, \cdots, a_n, a^\sharp \) must have the
same shape.

\begin{definition}
\definitionsummary{Concretization of labelled sets}
Let \( S_0, \cdots S_n, S^\sharp \)
be labelled sets such that
\[
  \forall i \in \{ 0, \cdots, n \} \itc \shape(S^\sharp) = \shape(S_i).
\]
Then we define
\( \absmap\bigl(S^\sharp, \{ S_0, \cdots, S_m \} \bigr) \) as the set
\[
    \Biggl<
      S^\sharp,
      \{ S_0, \cdots, S_m \},
      \biggl\{\,
        \absmap\Bigl( S^\sharp(l), \bigl\{ S_0(l), \cdots, S_n(l) \bigr\} \Bigr)
      \biggm|
        l \in \domain(S_0)
      \,\biggr\}
    \Biggr>.
\]
\end{definition}

Note that from the definiton of labelled set, if \( S_0, \cdots, S_n,
S^\sharp \) have the same shape, then they have also the same domain
(\refdefinition{Shape of a labelled set}); that is, the definition is well
formed.

\begin{definition}
\definitionsummary{Concretization of memories}
Let \( m \in \memories \) and \( m^\sharp \in \absmemories \) such that
\begin{gather*}
  m = \{ \textseg, \heap, \globalseg, \stackframes, \stacktop \}, \\
  m^\sharp =
    \{ \textseg^\sharp,
      \heap^\sharp,
      \globalseg^\sharp,
      \stacktail,
      \stackhead,
      \stacktop^\sharp \}.
\end{gather*}
If \( (m, m^\sharp) \in \iscompatible  \) we define
\( \absmap(m, m^\sharp) \) as the set
\begin{gather*}
    \Bigl\{
      \absmap\bigl(\textseg^\sharp, \{ \textseg \} \bigr),
      \absmap\bigl(\globalseg^\sharp, \{ \globalseg \} \bigr),
      \absmap\bigl(\stacktop^\sharp, \{ \stacktop \} \bigr)
    \Bigr\}
    \\
    \union
    \biggl\{\,
      \absmap\Bigl(a^\sharp, \bigl\{\, a \in \heap \bigm|
        a.\programpoint = a^\sharp.\programpoint\,\bigr\}\Bigr)
    \biggm|
      a^\sharp \in \heap^\sharp
    \,\biggr\}
    \\
    \union
    \biggl\{\,
      \absmap\Bigl(
        f^\sharp,
        \bigl\{\,
          \stackframes(i)
        \bigm|
          \stackframes(i).\callsite = f^\sharp.\callsite,
          \\ \qquad
          i \in \{ \stackhead.\size, \cdots, \stackframes.\size -1 \}
        \,\bigr\}
      \Bigr)
    \biggm|
      f^\sharp \in \stacktail
    \,\biggr\}
    \\
    \union
    \biggl\{\,
      \absmap\Bigl(\stacktail(i), \bigl\{ \stackframes(i) \bigr\}
        \Bigr)
    \biggm|
      i \in \{ 0, \cdots, \stackhead.\size - 1 \}
    \,\biggr\}.
\end{gather*}
\end{definition}

Note that the requirement of compatibility between the concrete memory \( m
\) and the abstraction \( m^\sharp \) ensures that the function
\(\mathord{\absmap}\) is well defined.
Once completed the definition of the function \( f \in \funlocations \),
the semantics of the abstraction can be completed following the idea
described in \refdefinition{Extended abstract domain semantics}.
Alternatively, using the approach informally presented in the introduction
(\refsection{general store based methods}), the concretization function can
be expressed in terms of approximation between locations, thus relying on
the definition of the concretization function for the elements of the
underlying abstract domain. Let \( m \in \memories \) and \( m^\sharp =
\langle f, P \rangle \in
\absmemories \) and let \( f \in \funlocations \) the location abstraction
function of \( m^\sharp \), then we say that
\( m \in \concretization(m^\sharp) \) when
\[
    \forall l \in \locations \itc
    f(m, l) \text{ is defined} \implies
    m[l] \in \concretization\Bigl( m^\sharp\bigl[f(m, l)\bigr] \Bigr);
\]
that for a points-to domain can also be written as
\[
  \forall l \in \locations \itc
  f(m, l) \text{ is defined} \implies
  f\bigl(m, \post(m, l)\bigr) \in
  \post\bigl(m^\sharp, f(m, l)\bigr).
\]

\subsubsection{Singular Locations}
The definition of \emph{singular location} introduced in
\refdefinition{Singular locations} is not applicable in a practical
implementation as it would require to explicitly check the existence of an
\( m \) in the concretization of \( m^\sharp \) with certain properties.
As a consequence we need a safe approximation of the set of
singular locations of an abstract memory. From the above definitions it can
be easily seen that every abstract location that represents the middle part
of an array of size not less that three is certainly non-singular.
The same holds also for stack tail segment:
each frame in this segment can represent \emph{more} concrete frames;
then, during the analysis we assume that all the locations contained
in the stack tail are non-singular.  Analogously for heap allocations;
it is impossible to tell for a given allocating statement
if it can be executed \emph{at most} one time;
in other words, it is impossible to tell if there exist a
\( m \in \concretization(m^\sharp) \) such that a given abstract heap
allocation abstracts \emph{more} concrete heap allocations.
As a consequence, we \emph{safely} assume that all
abstract heap allocations are non-singular.

\subsection{Abstract Operations}
Thus section presents some informal considerations about the remaining
operations required in order to complete the description of the execution
model.
We have already described the problem of formalizing operations
on the memory model in \refsection{the execution model and its operations}:
some operations are necessary to formulate the concrete
execution model \( \memories \); these are then generalized to the
concrete domain \( \partsof( \concretedomain ) \) and an approximation on
\( \abstractdomain \) is provided. Consider for instance the assignment
operation. Other operations are not required by the concrete execution
model, but are useful for the analysis; these operations are directly
formulated on the concrete domain \( \partsof( \concretedomain ) \) and, as
usual, an abstract counterpart is formulated on \( \abstractdomain \).
Consider
for instance the \emph{filter}, the \emph{merge} and \emph{meet} operations.
A more rigorous description of some of these is presented in
\cite{BagnaraHPZ07TR}.

\subsubsection{Notation}
Before proceeding we introduce some notation. Let \( A
\) be a non empty sequence. We write \( A = [H \mid T] \) to mean with \( H
\) the first element of \( A \), also called the \emph{head} element of \( A \);
and with \( T \) the remaining part of \( A \), also called the tail of the
sequence \( A \). We denote with `\([]\)' the empty sequence.

\subsubsection{The Mark Operation.}
This operation has the effect of closing the current block. In our memory
model we have modeled the stack top segment a sequence of blocks `\(
\text{Bs} \)' not contained in any frame, followed by a sequence of
allocations  `\( \text{As} \)' not contained in any block.  Let \( m
\in \memories \) be such that
\[
  m.\stacktop = [ \text{As}, \text{Bs} ].
\]
Then we have
\[
  \acsmark(m) = m_0 \in \memories
\]
such that
\[
  m_0.\stacktop = \bigl[ [], [\text{As} \mid \text{Bs} ] \bigr],
\]
while the rest of the memory is left unchanged.  In words, the allocations
`\(\text{As}\)' present in the stack top segment are moved in a block at
the head of the sequence of blocks `\(\text{Bs}\)'.  The abstract mark
operation is defined in the same way.

\subsubsection{The Link Operation}
This operation has the effect of creating a new frame on the stack.
Let \( m \in \memories \) be such that
\begin{gather*}
  m.\stacktop    = \bigl[ [], [ B \mid \text{Bs} ] \bigr], \\
  m.\stackframes = \text{Fs}.
\end{gather*}
Let
\[
  \acslink(m) = m_0 \in \memories,
\]
then we have
\begin{gather*}
  m_0.\stacktop    = \bigl[ [], [ B ] \bigr], \\
  m_0.\stackframes = [ \text{Bs} \mid \text{Fs} ],
\end{gather*}
and the rest of the memory is left unchanged.  The block denoted above as
\( B \) is intended to represent the arguments and the return value of the
function call that has triggered the link operation. To emulate the
arguments passing from the callee to the called context,
the allocations of the
block \( B \) are left in the stack top segment.
The abstract operation is formulated
similarly, the only difference is that the
`\(\stackhead\)' segment is used instead of the stack frames segment.
Let \( m^\sharp \in \absmemories \) be such that
\begin{gather*}
  m^\sharp.\stacktop  = \bigl[ [], [ B \mid \text{Bs} ] \bigr], \\
  m^\sharp.\stackhead = \text{Fs},
\end{gather*}
Let
\[
  \acsabslink(m^\sharp) = m^\sharp_0 \in \memories,
\]
then we have
\begin{gather*}
  m^\sharp_0.\stacktop  = \bigl[ [], [ B ] \bigr], \\
  m^\sharp_0.\stackhead = [ \text{Bs} \mid \text{Fs} ].
\end{gather*}

\subsubsection{The New Variable Operation}
This operation is required to populate the stack.
Ideally this operation can be split in two parts: first, the
creation of the new allocation; second the initialization of its locations.
Since the initialization can be treated a sequence of assignments,
here we consider
only the creation of the new locations.
Let \( m \in \memories \) be such that
\[
  m.\stacktop = [ \text{As}, \text{Bs} ].
\]
Let \( t \in \types \) be the type of the allocated object and let
\[
  \acsnewvar(m, t) = m_0 \in \memories.
\]
We have
\[
  m_0.\stacktop = \bigl[ [ A \mid \text{As} ], \text{Bs} \bigr],
\]
where (\refdefinition{Concrete allocations})
\[
  A = \allocation(t).
\]
Again, the abstract operation is defined in the same way, except that the
new allocation \( A \) is defined as (\refdefinition{Abstract allocations})
\(
  A = \absallocation(t).
\)

\subsubsection{The Unlink Operation}
This operation can be thought as the inverse of the link operation --- if
the link emulates the effects of a call statement then the unlink emulates
the effects of a return statement.
Let \( m \in \memories \) be such that
\begin{gather*}
  m.\stackframes = [ F \mid \text{Fs} ], \\
  m.\stacktop    = \bigl[ [], [ B ] \bigr].
\end{gather*}
The block \( B \) contains the arguments and the return value of the
called function that are returned to the caller context.
In particular we assume that the stack top contains only one block and that
the stack frames segment contains at least one frame --- in words, this
requires that every return statement must be preceded by a call statement.
Let
\[
  \acsunlink(m) = m_0 \in \memories.
\]
We have
\begin{gather*}
  m_0.\stackframes = \text{Fs}, \\
  m_0.\stacktop    = \bigl[ [], [ B \mid F ] \bigr],
\end{gather*}
while the rest of the of the memory is left unchanged.  Note that the
topmost frame \( F \) of the stack frames segment of \( m \) has
been moved in \( m_0 \) to the stack top segment and the block \( B \) has
been appended to it.  Basically, the abstract operation is defined in the
same way; the only difference is that instead of using the
`\(\stackframes\)' segment the `\(\stackhead\)' segment is used.

\subsubsection{The Unmark Operation}
This operation can be thought as the inverse of the mark operation --- if
the mark operation creates a new block gathering all the ungrouped
allocations of the stack top, then the unmark operation deletes these
allocations and replaces them with the allocations contained in the topmost
block.  Let \( m \in \memories \)  be such that
\[
  m.\stacktop = \bigl[ \text{As}, [ B \mid \text{Bs} ] \bigr].
\]
Let
\[
  \acsunmark(m) = m_0 \in \memories.
\]
We have
\[
  m.\stacktop = [ B, \text{Bs} ],
\]
while the rest of the memory is left unchanged.  Note that the sequence of
allocations `\(\text{As}\)' has been removed and in its place we now find
the allocations of the block \( B \).  The abstract operation is defined in
the same way.  It is worth stressing that the implementation of this
abstract operation probably
requires an additional step to
notify the remaining locations that the locations in `\(\text{As}\)'
no more exist; for instance, this is required for a pointer that was pointing to
one of the deallocated locations (\(\text{As}\)). In this case, depending
on the concrete execution model adopted, this pointer can be marked as
undefined.

In our model we use the following operations to set the degree of
context-sensitivity of the analysis and to approximate recursive function
calls. Both these operations have no effects on the concrete domain, that
is, for all \( m \in \memories \) we have \( \operation(m) = m \). In terms
of the approximation this means that for all \( m^\sharp \in \memories \) we
have that \( \concretization(m^\sharp) \subseteq \concretization\bigl(
\operation(m^\sharp) \bigr) \).

\subsubsection{The Stack Tail Push Operation}
This operation has the effect of moving the oldest frame of the stack head
segment (from now the `pushed frame') to the stack tail. Recall that the
stack tail segment is a labelled set of frames where each frame is
identified
by a call site and that the call site uniquely identifies the shape of the
frame. This means that for each call site the stack tail can contain only
one frame.  Thus, if it already contains a frame with the same
call site of the pushed frame then the pushed frame will be merged into the
corresponding stack tail frame. Otherwise, if no frames with the same call
site are already present, the frame will be simply added to the stack tail.
Let \( m^\sharp \in \absmemories \)  be such that
\begin{gather*}
  m^\sharp.\stackhead = \text{Fs} \concat [ F ], \\
  m^\sharp.\stacktail = \{ F_0, \cdots, F_n \},
\end{gather*}
where \( F \) denotes the last element of the non-empty stack head segment;
`\( \text{Fs} \)' denotes the remaining part of the
same sequence and \( n = m^\sharp.\stacktail.\size \in \naturals \).  Let
\[
  \acsabstailpush(m^\sharp) = m^\sharp_0 \in \absmemories.
\]
then we have
\begin{gather*}
  m^\sharp_0.\stackhead = \text{Fs}, \\
  m^\sharp_0.\stacktail =
    \begin{cases}
      \{ F_0 \sqcap F, \cdots, F_n \},
        & \text{if } F_0.\callsite = F.\callsite; \\
      \{ F_0, \cdots, F_n, F \},
        & \text{otherwise.}
    \end{cases}
\end{gather*}
Note that the stack tail segment is a labelled set, thus the order
indicated above, \( F_0, \cdots, F_n \), among its frames is completely
artificial and introduced for notational convenience --- writing \(
F_0.\callsite = F.\callsite \) we mean that there exists a frame in the stack
tail with the same call site of \( F \).

\subsubsection{The Stack Tail Pop Operation}
This operation is the inverse of the stack tail push --- it moves a frame
from the stack tail back into the stack head segment.
To do this we have to specify which frame to restore, that is the stack
tail pop operation requires a call site. Let
\[
  \functiondef
    { \acsabstailpop }
    { \absmemories \times \callsites }
    { \absmemories }
\]
Given \( c \in \callsites \) and \( m^\sharp \in \absmemories \), if the
stack tail segment of \( m^\sharp \) does not contain any frame labelled \(
c \) then the operation results in the \( \bot \) element.
Otherwise let
\begin{gather*}
  m^\sharp.\stackhead = \text{Fs}, \\
  m^\sharp.\stacktail = \{ F_0, \cdots, F_n \}
\end{gather*}
be such that \( F_0.\callsite = c \). Then calling
\[
  \acsabstailpop(m^\sharp) = m^\sharp_0 \in \absmemories,
\]
we have
\[
  m^\sharp_0.\stackhead = \text{Fs} \concat [ F ], \\
\]
while the rest of the memory, also the stack tail segment, remains
unchanged.

\subsection{Approximating the Stack}
The concept of \emph{stack tail} is introduced precisely to handle
\emph{recursion}. In presence of recursive function calls, the number of
frames on the concrete stack cannot be limited by any finite bound. Beyond
these theoretical considerations, just from a practical perspective it is
unfeasible to keep an arbitrary number of distinct abstract frames.  The
idea of our abstraction to address this problem is to represent `precisely'
the variables of the local environment, approximated by the \emph{stack
top} segment, and global variables, represented by the \emph{global}
segment. Also the topmost \( k \) frames of the concrete stack are
abstracted `precisely' by the \emph{stack head} segment.  However, we
approximate more roughly in the \emph{stack tail} segment, the content of
the concrete stack below the first \( k \) frames.  Frames in the stack
tail are identified by their \emph{call site}; this means that the concrete
frames labelled by the same call site \( c \) that are below the \( k \)-th
topmost frame, are all approximated by the same abstract frame, which is
contained in the stack tail and it is identified by \( c \).

\subsection{Pointer Arithmetic}
\label{section:pointer arithmetics}
This section presents a prototype for handling pointer arithmetic.
Complex approaches to this problem are already present in the literature;
for example, \emph{string cleanness} techniques associate an integer quantity to every possible
target of a pointer, to represent the distance between the
beginning of the pointed object and the pointed address. These integer
quantities are then approximated by the analysis using a some numerical
abstraction;
with the availability of \emph{relational} numerical domains,
these methods can be precise but costly \cite{Franchi06th}.
The method that we present now
is \emph{attribute independent} and it is completely handled by the points-to
domain; the presence of an external numeric domain is assumed only to query
for the value of integer expressions during the evaluation of the
pointer arithmetic.
Let \( m^\sharp \in \absmemories \) and consider the expression \( p + i
\) where
\begin{itemize}
\item
  the expression \( p \) is of pointer type and its abstract evaluation
  results in a location that is part of an array.  We assume to know the
  type of the elements of the array and the size of the array itself.
\item
  The expression \( i \) is of integer type and it represents the added
  offset.
\end{itemize}
To represent the possible errors that can arise from the concrete
evaluation of the expression \(p + i\), we use the set
\[
  \rtserrors \defeq \{ \Em, \Ep \},
\]
where with `\( \Em \)' we denote the \emph{array underflow error} and with
`\( \Ep \)' we denote the \emph{array overflow error}.
To formalize the concrete evaluation of a
pointer arithmetic expressions, let
\[
  \parfunctiondef
    { \pointerarith }
    { \memories \times \expressions \times \expressions }
    { \locations \union \rtserrors }
\]
be a partial function defined for every pair of expressions \( p, i \in
\expressions \) where \( p \) is of pointer type and \( i \) is of integer
type.\footnote{We assume that the two sets \( \locations \) and \(
\rtserrors \) have disjoint representations.}
Let
\[
  \parfunctiondef
    { \pointerarith }
    { \partsof(\memories) \times \expressions \times \expressions }
    { \partsof(\locations \union \rtserrors ) }
\]
be its extension to sets of concrete memories defined, for all
\( M \subseteq \memories \), as
\[
  \pointerarith(M, p, i) \defeq
  \bigunion
    \bigl\{\,
      \pointerarith(m, p, i)
    \bigm|
      m \in M
    \,\bigr\}.
\]
A rigorous definition of \( \pointerarith(m, p, i) \) would require a
rigorous definition concrete execution model \cite{BagnaraHPZ07TR}; an
informal presentation of the concrete semantics used here is later
discussed in \refsection{derivation of the tables}.
To denote the approximation for the concrete operation
\(\mathord{\pointerarith}\) we introduce the function\footnote{Also in this
case we assume that \( \abslocations \) and \( \rtserrors \) have disjoint
representations.}
\[
  \parfunctiondef
    { \pointerarith }
    { \absmemories \times \expressions \times \expressions }
    { \partsof(\abslocations \union \rtserrors) }.
\]
Generally, in an abstract memory description \( m^\sharp \), the
evaluation of a pointer expression results in a \emph{set} of abstract
locations.  It is however convenient to define the abstract semantics of
the \( \mathord{\pointerarith} \) function by working on one abstract
location at a time.
Thus, to ease the presentation we introduce the helper function
\[
  \functiondef
    { \pointerarith }
    { \absmemories \times \absLocations \times \expressions }
    { \partsof(\abslocations \union \rtserrors ) }
\]
where, given \( m^\sharp \in \absmemories \), \( l \in \absLocations \) and
the integer expression \( i \in \expressions \),
with \( \pointerarith(m^\sharp, l, i) \)
we represent the set of the possible abstract locations resulting
from the addition of the value of \( i \) to the location \( l \) in the memory
\( m^\sharp \).  Let again \( p, i \in \expressions \);
then we define
\[
  \pointerarith(m^\sharp, p, i) \defeq
  \bigunion
    \bigl\{\,
      \pointerarith(m^\sharp, l, i)
    \bigm|
      l \in \eval(m^\sharp, p)
    \,\bigr\}.
\]
To query the numerical domain about the value of the integer expression \(
i \) we assume the existence of a function
\[
  \functiondef
    { \evalInt }
    { \absmemories \times \expressions }
    { \partsof(\integers) }
\]
with the following semantics
\[
  \evalInt(m^\sharp, i) \defeq
  \Bigl\{\,
    z \in \integers
  \Bigm|
    \exists m \in \concretization(m^\sharp) \suchthat
    m\bigl[\eval(m, i)\bigr] = z
  \,\Bigr\}.
\]
In words, the function \( \mathord{\evalInt} \) returns the set of the
possible values that the integer expression \( i \) can assume in the
concrete memories \( m \) approximated by \( m^\sharp \).

The function \( \mathord{\pointerarith} \) is defined as follows.
We first introduce some notation. Let \( S \in \naturals \setminus \{ 0 \}
\) be the size of the array on which we are performing pointer
arithmetic.

\begin{center}
\begin{tabular}{lll}
\toprule
  Symbol        & Description & Concrete range \\
\hline
\( \Em \)   & Array underflow error.     & \( (-\infty, -1] \) \\
\( \Head \) & Array head location.       & \( [0, 1) \) \\
\( \Tail \) & Array tail location.       & \( [1, S) \) \\
\( \Off \)  & Array off-by-one location. & \( [S, S+1) \) \\
\( \Ep \)   & Array overflow error.      & \( [S+1, +\infty) \) \\
\bottomrule
\end{tabular}
\end{center}
The abstract memory model described in \refsection{extended memory model}
approximates array variables using three distinct abstract locations here
denoted with `\( \Head \)',
`\( \Tail \)' and `\( \Off \)'; we use the symbols `\(
\Ep \)' and `\( \Em \)' to denote the possible exceptional outcome of the
arithmetic operation due to the exceeding of the array bounds.
Let \( S \in \naturals \setminus \{ 0 \} \) be the size of the considered array;
we distinguish four possible cases: \( S = 1, S
= 2, S = 3, S \geq 4 \).
Each of these cases is described by one of the below tables.
In each of this tables, the first column contains a set of intervals of \(
\integers \) that forms a partition of \( \integers \) itself.
The first row of these tables represents instead the three possibility for
the abstract locations \( l \) supplied to the function
\( \mathord{\pointerarith} \).
Let \( D = D(S) \) be the table corresponding to the location \( l
\).  We denote as `\( D.\text{rows} \)'
the number of rows of the table \( D \).  For each \( n \in \{ 1, \ldots,
D.\text{rows} \} \) we denote as `\( D.\text{row}(n) \)' the \(n\)-th row of the
table \( D \).  Given a row \( R \) of \( D \) we denote as `\(
R.\text{interval} \)' the interval of \( \integers \) associated to \( R \),
which is located in the first column. With `\( R.\text{loc}(l) \)' we
denote the cell at the intersection of the row \( R \) and the column
associated to the location \( l \) --- the second column if \( l \)
represents the \emph{head} location \( \Head \) of the array, the third
column if \( l \) represents the \emph{tail} location \( \Tail \), or the
fourth column if \( l \) represents the \emph{off-by-one} location \( \Off
\).  With this notation, the function \( \mathord{\pointerarith} \) can be
defined as
\begin{gather*}
  \pointerarith(m^\sharp, l, i) \defeq
  \bigunion
    \bigl\{\,
      L(S, i)
    \bigm|
      i \in \{ 1, \cdots, D.\text{rows} \}
    \,\bigr\};
  \\
  \intertext{where}
  R(S, n) \defeq D(S).\text{row}(n);
  \\
  L(S, n, l) \defeq
  \begin{cases}
    R(S, n).\text{loc}(l),
      & \text{if } R(S, n).\text{interval} \intersection
        \evalInt(m^\sharp, i) \neq \emptyset; \\
    \emptyset, & \text{otherwise}.
  \end{cases}
\end{gather*}

\begin{table}[ht]
\centering
\subfloat{
\begin{pointerarithmetictable2}
$ S=1         $ & $ \Head $ & $ \Tail $ & $ \Off  $ \\ \hline
$ (-\infty, -2] $ & $ \Em   $ & $       $ & $ \Em   $ \\
$ -1          $ & $ \Em   $ & $       $ & $ \Head $ \\
$ 0           $ & $ \Head $ & $       $ & $ \Off  $ \\
$ 1           $ & $ \Off  $ & $       $ & $ \Ep   $ \\
$ [2, \infty)   $ & $ \Ep   $ & $       $ & $ \Ep   $ \\
\end{pointerarithmetictable2}
} \quad
\subfloat{
\begin{pointerarithmetictable2}
$ S=2         $ & $ \Head $ & $ \Tail $ & $ \Off  $ \\ \hline
$ (-\infty, -3] $ & $ \Em   $ & $ \Em   $ & $ \Em   $ \\
$ -2          $ & $ \Em   $ & $ \Em   $ & $ \Head $ \\
$ -1          $ & $ \Em   $ & $ \Head $ & $ \Tail $ \\
$ 0           $ & $ \Head $ & $ \Tail $ & $ \Off  $ \\
$ 1           $ & $ \Tail $ & $ \Off  $ & $ \Ep   $ \\
$ 2           $ & $ \Off  $ & $ \Ep   $ & $ \Ep   $ \\
$ [3, \infty)   $ & $ \Ep   $ & $ \Ep   $ & $ \Ep   $ \\
\end{pointerarithmetictable2}
}

\subfloat{
\begin{pointerarithmetictable2}
$ S = 3       $ & $ \Head $ & $ \Tail        $ & $ \Off  $ \\ \hline
$ (-\infty, -4] $ & $ \Em   $ & $ \Em          $ & $ \Em   $ \\
$ -3          $ & $ \Em   $ & $ \Em          $ & $ \Head $ \\
$ -2          $ & $ \Em   $ & $ \Em, \Head   $ & $ \Tail $ \\
$ -1          $ & $ \Em   $ & $ \Head, \Tail $ & $ \Tail $ \\
$ 0           $ & $ \Head $ & $ \Tail        $ & $ \Off  $ \\
$ 1           $ & $ \Tail $ & $ \Tail, \Off  $ & $ \Ep   $ \\
$ 2           $ & $ \Tail $ & $ \Off, \Ep    $ & $ \Ep   $ \\
$ 3           $ & $ \Off  $ & $ \Ep          $ & $ \Ep   $ \\
$ [4, \infty)   $ & $ \Ep   $ & $ \Ep          $ & $ \Ep   $ \\
\end{pointerarithmetictable2}
}
\quad
\subfloat{
\begin{pointerarithmetictable2}
$ S \geq 4      $ & $ \Head $ & $ \Tail             $ & $ \Off  $ \\
\hline
$ (-\infty, -S-1] $ & $ \Em   $ & $ \Em               $ & $ \Em   $ \\
$ -S            $ & $ \Em   $ & $ \Em               $ & $ \Head $ \\
$ 1-S           $ & $ \Em   $ & $ \Em, \Head        $ & $ \Tail $ \\
$ [2-S, -2]     $ & $ \Em   $ & $ \Em, \Head, \Tail $ & $ \Tail $ \\
$ -1            $ & $ \Em   $ & $ \Head, \Tail      $ & $ \Tail $ \\
$ 0             $ & $ \Head $ & $ \Tail             $ & $ \Off  $ \\
$ 1             $ & $ \Tail $ & $ \Tail, \Off       $ & $ \Ep   $ \\
$ [2, S-2]      $ & $ \Tail $ & $ \Tail, \Off, \Ep  $ & $ \Ep   $ \\
$ S-1           $ & $ \Tail $ & $ \Off, \Ep         $ & $ \Ep   $ \\
$ S             $ & $ \Off  $ & $ \Ep               $ & $ \Ep   $ \\
$ [S+1, \infty)   $ & $ \Ep   $ & $ \Ep               $ & $ \Ep   $ \\
\end{pointerarithmetictable2}
}
\end{table}

Since the C language provides various mechanism to create arrays whose size
is computed at run-time, we ought to consider the case of handling pointer
arithmetic on arrays of unknown size.\footnote{Or of \emph{partially}
unknown size.  For example, the analysis could be able to determine some
approximation of the value used to specify the size the
array during its allocation.}
To handle the case of arrays of unknown size we
compute a merge of the above cases, obtaining the following
table.

\begin{pointerarithmetictable}
$ S > 0       $ & $ \Head            $ & $ \Tail             $ & $ \Off              $ \\ \hline
$ (-\infty, -2] $ & $ \Em              $ & $ \Em, \Head, \Tail $ & $ \Em, \Head, \Tail $ \\
$ -1          $ & $ \Em              $ & $ \Head, \Tail      $ & $ \Head, \Tail      $ \\
$ 0           $ & $ \Head            $ & $ \Tail             $ & $ \Off              $ \\
$ 1           $ & $ \Tail, \Off      $ & $ \Tail, \Off       $ & $ \Ep               $ \\
$ [2, \infty)   $ & $ \Tail, \Off, \Ep $ & $ \Tail, \Off, \Ep  $ & $ \Ep               $ \\
\end{pointerarithmetictable}

\subsubsection{Examples}
The following examples illustrate the described method
applied to \refcodesnippet{pointer arithmetics}.
For convenience of notation we
represent the steps of the computation using a table with two columns: the first
column shows the program point currently executed and the second column
shows the abstract value of the variable \QuotedCode{first}; note indeed
that the value of the pointer variable \QuotedCode{last} is never changed
by the execution of the function \QuotedCode{foo}.  Since the
\QuotedCode{foo}' function contains a loop, the
abstract computation terminates when a fix-point is reached;
to separate the different iterations of the loop analysis we use horizontal
lines.
In the last row of the table we
will show the result of the merge of all the exit states of the loop.  For
simplicity of presentation we assume that the array \QuotedCode{a} declared
at \refline 2 contains at  least four elements.

\codecaption
{an example of a simple loop that depends on the pointer arithmetic
computation.}
{pointer arithmetics}
\begin{codesnippet}
const unsigned int N = ...;
int a[N];

void foo(const T* first, const T* last) {
  // PP0
  while (true) {
      // PP1
      if (first == last) break;
      // PP2
      ... = *first;
      // PP3
      ++first;
  }
  // PP5
}
\end{codesnippet}

\begin{example}
Consider the call \QuotedCode{foo(a, a + N)}.
In the concrete domain the expression \QuotedCode{a + N} evaluates to the
address one-past-the-end of the array \QuotedCode{a}, that in the abstract
domain corresponds to the off-by-one abstract location \( \Off \). During
all the execution of the \QuotedCode{foo} function we have \( \eval(
m^\sharp, \Code{last}) = \{ \Off \} \).
Instead, the expression \QuotedCode{a} evaluates to the
address of the begin of the array \QuotedCode{a}, that in the abstract
domain corresponds to the head abstract location \( \Head \).
Thus, at the entry point of \QuotedCode{foo}, the expression
\QuotedCode{first} evaluates to \( \Head \).
These are the steps of the execution
\begin{examplearithtbl}
PP & \( \eval( m^\sharp, \Code{first}) \) \\
\hline
  0 & \( \Head \) \\
\hline
  1 & \( \Head \) \quad (1st) \\
  5 & \( \bot \) \\
  2 & \( \Head \) \\
  3 & \( \Head \) \\
\hline
  1 & \( \Tail \) \quad (2nd) \\
  5 & \( \bot \) \\
  2 & \( \Tail \) \\
  3 & \( \Tail \) \\
\hline
  1 & \( \Tail, \Off \) \quad (3rd) Fixpoint \\
  5 & \( \Off \) \\
  2 & \( \Tail \) \\
  3 & \( \Tail \) \\
\hline
  5 & \( \Off \) \\
\end{examplearithtbl}
Note that the filter on the guard condition of the loop
\QuotedCode{first == last} at \refline 8, is able to split the points-to
information
\[
  \bigl\{
    \langle \Code{first}, \Code{a}.\Tail \rangle,
    \langle \Code{first}, \Code{a}.\Off \rangle
  \bigr\}
\]
into \( \bigl\{ \langle \Code{first}, \Code{a}.\Tail \rangle \bigr\} \) for
the \Code{else} branch --that represents the continuation of the loop-- and
into \( \bigl\{ \langle \Code{first}, \Code{a}.\Off \rangle \bigr\} \) for
the \Code{then} branch, that represents the execution paths that exit from
the loop.
In this case the analysis finds the fixpoint of the loop without signalling
any error due to the pointer arithmetic; that is, it is able to prove the
absence of errors in the execution of the loop.
\end{example}

\begin{example}
Consider the call \QuotedCode{foo(a, a)}.  During all the execution we have
\[
  \eval( m^\sharp, \Code{last}) = \{ \Head \}.
\]
These are the steps of the execution
\begin{examplearithtbl}
PP & \( \eval( m^\sharp, \Code{first}) \) \\
\hline
  0 & \( \Head \) \\
\hline
  1 & \( \Head \) \quad (1st) Fixpoint \\
  5 & \( \Head \) \\
  2 & \( \bot \) (unreachable) \\
\hline
  5 & \( \Head \) \\
\end{examplearithtbl}
In this case the analysis is able to prove that the execution exits
immediately from the loop without modifying the value of
\QuotedCode{first} and without any error.
\end{example}

\begin{example}
Consider the call \QuotedCode{foo(a + N, a + N)}.  This case is very
similar to the previous one.  During the execution we have \( \eval(
m^\sharp, \Code{last}) = \{ \Off \} \).  These are the steps of the
execution
\begin{examplearithtbl}
PP & \( \eval( m^\sharp, \Code{first}) \) \\
\hline
  0 & \( \Off \) \\
\hline
  1 & \( \Off \) \quad (1st) Fixpoint \\
  5 & \( \Off \) \\
  2 & \( \bot \) (unreachable) \\
\hline
  5 & \( \Off \) \\
\end{examplearithtbl}
Note that at the first iteration of the loop the filter is able to prove
that \QuotedCode{first} and \QuotedCode{last} are \emph{definitely} aliases.
Also in this case the analysis is able to prove that the execution exits
immediately from the loop without modifying the value of
\QuotedCode{first} and without any error.
\end{example}

\begin{example}
Consider the call \QuotedCode{foo(a + N, a)}.
During the execution we have
\[
  \eval( m^\sharp, \Code{last}) = \{ \Head \}.
\]
These are the steps of the execution
\begin{examplearithtbl}
PP & \( \eval( m^\sharp, \Code{first}) \) \\
\hline
  0 & \( \Off \) \\
\hline
  1 & \( \Off \) \quad (1st) \\
  5 & \( \bot \) \\
  2 & \( \Off \) \\
  3 & \( \Off \) (+ Dereference Warning) \\
\hline
  1 & \( \Ep, \bot \) \quad (2nd) Fixpoint \\
  \ldots \\
\hline
  5 & \( \bot \) \\
\end{examplearithtbl}
In this case the analysis is able to detect that in the first iteration of
the loop at program point 3 an off-by-one location is dereferenced.
Depending on the concrete execution model adopted, the analyzer may assume
that the concrete execution terminates or not.  In the last case the
analysis is able to prove that during the next iteration of the loop, the
pointer \QuotedCode{first} is incremented beyond the legal bounds of the
array.
\end{example}

\begin{example}
Consider the calls \QuotedCode{foo(a + 4, a + 6)}, \QuotedCode{foo(a + 5, a
+ 5)} and \QuotedCode{foo(a + 6, a + 4)}, which have the same abstraction.
Indeed the expressions \QuotedCode{a + 4}, \QuotedCode{a + 5},
\QuotedCode{a + 6} --and more generally the expressions `\( \Code{a} + i\)'
with \( i \in \{ 1, \cdots, \Code{N}-1 \} \)-- all evaluate in the abstract memory to
the tail location \( \Tail \) of the array \QuotedCode{a}.
These are the steps of the execution
\begin{examplearithtbl}
PP & \( \eval( m^\sharp, \Code{first}) \) \\
\hline
  0 & \( \Tail \) \\
\hline
  1 & \( \Tail \) \quad (1st) \\
  5 & \( \Tail \) \\
  2 & \( \Tail \) \\
  3 & \( \Tail \) \\
\hline
  1 & \( \Tail, \Off \quad \fixpoint \) \quad (2nd) Fixpoint \\
  5 & \( \Tail \) \\
  2 & \( \Tail, \Off \) \\
  3 & \( \Tail \) (+ Dereference warning)\\
\hline
  5 & \( \Tail \) \\
\end{examplearithtbl}
Note that \(\Code{a}.\Tail\) is not singular; thus, the filter at
the guard of the loop cannot remove the arc \( \langle \Code{first},
\Code{a}.\Tail \rangle \) from the else branch.
Then at program point 2 we still find \( \{ \Tail, \Off \} \).
The above table represents the case in which the execution model forbids to
dereference pointers to the off-by-one location of an array. In this case,
when the abstract execution reaches program point 3 in the last iteration of
the loop the analyzer filters away the off-by-one locations from the
possible targets of \QuotedCode{first} and raises a warning.  In this case
the analysis successfully detects the possibility of an error, indeed there
exist at least one concrete execution in which the off-by-one location is
dereferenced.
Otherwise, if the analyzer accepts as valid the dereference of the
off-by-one location at \refline 3 we would obtain
\begin{examplearithtbl}
  \ldots \\
  3 & \( \Tail, \Off \) \\
  1 & \( \Tail, \Off \quad \fixpoint \) (\(+ \Ep \)) \\
\end{examplearithtbl}
That is the analysis detects that the increment of \QuotedCode{first} at
\refline{12} can produce an error due to the exceeding of the array bounds.
\end{example}

Note that this model is symmetrical with respect to the direction of the
increasing indices --- the only difference is that the off-by-one
location cannot be dereferenced, while the head location \(\Head\) can.

\subsubsection{Derivation of the Rules}
\label{section:derivation of the tables}
This section provides the reader with
a justification of the presented rules for the handling of pointer arithmetic.
However, in this case the concepts are intuitive and the
additional burden required to introduce a rigorous model to
describe the rules does worth the effort. Therefore, we limit the
presentation to an informal justification of some of the cases with
the conviction that the remaining cases can be deduced similarly. 
Consider the case of an array whose elements are of scalar type \( t \in
\types \) which contains at least four elements, that is, \( S \geq 4 \).
Under these assumptions, the concrete allocation block generated by \( t \)
is
\[
  \allocation\bigl(t[S]\bigr) = [ l_0, l_1, \cdots, l_{S-1}, l_{S} ];
\]
where the last location of the sequence \( l_S \) represents the off-by-one
location of the array.
To this concrete allocation block corresponds the following abstract
allocation block
\[
  \absallocation\bigl(t[S]\bigr) = [ \Head, \Tail, \Off ].
\]
In this sense we can say that
\begin{gather*}
  \concretization(\Head)  = [ l_0 ], \\
  \concretization(\Tail)  = [ l_1, \cdots, l_{S-1} ], \\
  \concretization(\Off) = [ l_S ].
\end{gather*}
Let \( p \in \expressions \), \( A \in \abstractdomain \) and \( C \in
\concretization(A) \).
\begin{itemize}
\item
  Consider the case \( \eval(C, p) = \{ l_0 \} \). Since \( C \) is
  approximated by \( A \) and \( \concretization(\Head) = [ l_0 ] \) we
  have that \( \Head \in \eval(A, p) \).
  In the concrete model if we move below the location \(
  l_0 \) we cross the boundaries of the array triggering an undefined
  behaviour. In the
  abstract model we approximate this with \( \Em \) to mean the array
  underflow.  If we move above the location \( l_0 \) of \( n \in \naturals
  \) positions, with \( n \leq S \), we reach the concrete location \( l_n
  \).  In the abstract model, staring from the head abstract location
  \(\Head\) and adding \( n \), with \( n \in [1, S-1] \), we reach the tail
  location \( \Tail \); otherwise, for \( n = S \) the off-by-one location
  \( \Off \) is reached.  If we move above the location \( l_0 \) of \( n
  \in \naturals \) positions, with \( n > S \), we trespass the boundaries
  of the array producing an error, that we abstract with \( \Ep \).
  Summing up we have, for the concrete model
  \begin{center}
  \begin{tabular}{ll}
  \hline \hline
   Offset   & \( l_0 \) + Offset \\
  \hline
  \( (-\infty, 0) \)  & Error: array underflow. \\
  \( 0 \) & \( l_0 \) \\
  \( 1 \) & \( l_1 \) \\
  \( \cdots \) \\
  \( S-1 \) & \( l_{S-1} \) \\
  \( S \) & \( l_S \) \\
  \( [S+1, +\infty) \) & Error: array overflow.  \\
  \hline \hline
  \end{tabular}
  \end{center}
  and its abstraction is
  \begin{center}
  \begin{tabular}{ll}
  \hline \hline
   Offset     & \( \Head \) + Offset \\
  \hline
    \( (-\infty, 0) \)  & \( \Em \)   \\
    \( 0 \)           & \( \Head \) \\
    \( [1, S) \)      & \( \Tail \) \\
    \( S \)           & \( \Off \)  \\
    \( [S+1, +\infty) \) & \( \Ep \)  \\
  \hline \hline
  \end{tabular}
  \end{center}
\item
  In case we start from the off-by-one location \( l_S \), that is \(
  \eval(C, p) = \{ l_S \} \), in the abstract model we have \( \Off \in
  \eval(A, p) \). This case is quite symmetrical to the case of starting on
  the head location.
  \begin{center}
  \begin{tabular}{ll}
  \hline \hline
   Offset   & \( l_S \) + Offset \\
  \hline
  \( (-\infty, -S) \)  & Error: array underflow. \\
  \( -S \) & \( l_0 \) \\
  \( 1-S \) & \( l_1 \) \\
  \( \cdots \) \\
  \( -1 \) & \( l_{S-1} \) \\
  \( 0 \) & \( l_S \) \\
  \( [1, +\infty) \) & Error: array overflow.  \\
  \hline \hline
  \end{tabular}
  \end{center}
  and its abstraction is
  \begin{center}
  \begin{tabular}{ll}
  \hline \hline
   Offset     & \( \Off \) + Offset \\
  \hline
    \( (-\infty, -S) \)  & \( \Em \)   \\
    \( -S \)           & \( \Head \) \\
    \( [1-S, 0) \)      & \( \Tail \) \\
    \( 0 \)    & \( \Off \)  \\
    \( [1, +\infty) \) & \( \Ep \)  \\
  \hline \hline
  \end{tabular}
  \end{center}
\item
  We now consider all the cases \( \eval(C, p) = \{ l_n \} \) with \( n \in
  [1, S-1] \) as these cases have the same abstraction. All the concrete
  locations \( l_1, \cdots, l_{S-1} \) are indeed abstracted by the same
  abstract location \( \Tail \).
  The difference with respect to the two previous cases is that when we
  perform pointer arithmetic on the tail of an array we do not know on
  which concrete location we are working: there is indeed a \emph{set} of
  possible locations. This means for instance that if we move from the \(
  l_1 \) by an offset of 1 we reach \( l_2 \), which is still in the tail;
  but starting from \( l_{S-1} \) we obtain \( l_S \), which is in the
  off-by-one location \( \Off \).
  From this reasoning it can be easily derived the result presented in the
  following tables.
  \begin{center}
  \begin{tabular}{llll}
  \hline \hline
  Offset   & \( l_1 \) + Offset & \ldots & \( l_{S-1} \) + Offset \\
  \hline
  $ (-\infty, -S] $ & underflow   & underflow & underflow \\
  $ 1-S         $ & underflow   & \ldots   & $ l_0 $ \\
  $ 2-S         $ & underflow   & \ldots   & $ l_1 $ \\
  \ldots \\
  $ -2          $ & underflow   & \ldots   & $ l_{S-3} $ \\
  $ -1          $ & $ l_0 $     & \ldots   & $ l_{S-2} $ \\
  $ 0           $ & $ l_1 $     & \ldots   & $ l_{S-1} $ \\
  $ 1           $ & $ l_2 $     & \ldots   & $ l_S $ \\
  $ 2           $ & $ l_3 $     & \ldots   & overflow \\
  \ldots \\
  $ S-2         $ & $ l_{S-1} $ & \ldots   & overflow \\
  $ S-1         $ & $ l_S $     & overflow & overflow \\
  $ [S, \infty)   $ & overflow    & overflow & overflow \\
  \hline \hline
  \end{tabular}
  \end{center}
  and its abstraction is
  \begin{center}
  \begin{tabular}{ll}
  \hline \hline
   Offset     & \( \Off \) + Offset \\
  \hline
  $ (-\infty, -S]   $ & $ \Em               $ \\
  $ 1-S           $ & $ \Em, \Head        $ \\
  $ [2-S, -2]     $ & $ \Em, \Head, \Tail $ \\
  $ -1            $ & $ \Head, \Tail      $ \\
  $ 0             $ & $ \Tail             $ \\
  $ 1             $ & $ \Tail, \Off       $ \\
  $ [2, S-2]      $ & $ \Tail, \Off, \Ep  $ \\
  $ S-1           $ & $ \Off, \Ep         $ \\
  $ [S, \infty)     $ & $ \Ep               $ \\
  \hline \hline
  \end{tabular}
  \end{center}
\end{itemize}
Composing these three cases we obtain the complete table for the case \( S
\geq 4 \) for the abstract pointer arithmetic rules.

\subsection{Relational Operators}
\label{section:relational operators}
Just not cited above for simplicity of notation, we describe here one of
the possible extensions to the filter operation that in some sense is bound
to the handling of pointer arithmetic.  In particular now we want to
consider the use of relational operators ---the \QuotedCode{<=},
\QuotedCode{<} and their symmetric--- and their interaction with the
points-to problem.  We report here the statement of the C standard about
the use of relational operators between pointers
\cite[6.5.8.5]{ISO-C-1999}:
\begin{quote}
If the objects pointed to are members of the same aggregate object, pointers to
structure members declared later compare greater than pointers to members
declared earlier in the structure, and pointers to array elements with larger
subscript values compare greater than pointers to elements of the same array
with lower subscript values. [\ldots] If the expression P points to an element
of an array object and the expression Q points to the last element of the same
array object, the pointer expression Q+1 compares greater than P. In all other
cases, the behavior is undefined.
\end{quote}
Recalling the simplified model introduced in \refsection{our method}, we
need to extend the set of the possible operators \( \{ \equality,
\inequality \} \) to comprehend the additional operators of interest. Once
augmented the set \( \conditions \) with the new conditions we have to
define a proper concrete semantics for the new elements.  Formally, this
requires the definition of a partial order on the set of location addresses.
This partial order should satisfy the requirements of the C Standard
reported above. Using the terminology of the extended memory model
presented in \refsection{extended memory model} we can say that this order
is required to be defined only between the locations that belong to the
same allocation block.  In this model we have defined the concept of
allocation block as a sequence of locations and the order of the locations
within the allocation in such a way to reflect the actual memory layout.
Under these assumptions
it is possible to define the required partial order as the order
specified by the allocations; this way we are able to correctly
describe the semantics of the C Standard not only for pointers to arrays
but also for pointers to structure members.

Now, using the notation introduced in \refsection{our method},
assume to have already defined the needed \emph{strict partial order},
denoted as `\(\mathord{<}\)', on the set of locations
\(
  \mathord{<} \subseteq \locations \times \locations.
\)
Consider the following extension of the concrete execution model.  From
\refdefinition{Conditions} we extend the set of conditions
`\(\conditions\)' by adding to the set of the possible operators the
element `\(\less\)', as to represent the `less-than' operator of the C
language.
\[
  \conditions \defeq
    \{ \equality, \inequality, \less \}
    \times \expressions \times \expressions.
\]
We also need to extend \refdefinition{Value of conditions}, to comprehend
the newly added elements of `\(\conditions\)'. Let
\(
  \trueconditions \subseteq \concretedomain \times \conditions
\)
be extended, for all \( C \in \concretedomain \) and \( e, f \in
\expressions \), as
\[
    \bigl(C, (\less, e, f)\bigr) \in \trueconditions
    \quad \defiff \quad
    \eval(C, e) < \eval(C, f).
\]
Now we present
a possible extension of the abstract filter operation
(\refdefinition{Filter 3}) for handling the relational operator
`\(\less\)'.\footnote{For simplicity of exposition we treat explicitly only
the operator `\(\less\)' and we omit other relational
operators whose formalization can be deduced from the formalization of
`\(\less\)' by symmetry and by composition with the equality operator.}

\begin{definition}
\definitionsummary{Filter on the less-than operator}
Let
\[
  \functiondef
    { \filter }
    { \abstractdomain \times \conditions }
    { \abstractdomain }
\]
be defined as follows.
Let \( A \in \abstractdomain \) and \( e, f \in \expressions \). Let
\begin{gather*}
  E  \defeq
  \bigl\{\,
    l \in \eval(A, e)
  \bigm|
    \exists m \in \eval(A, f) \suchthat l < m
  \,\bigr\};
  \\
  F  \defeq
  \bigl\{\,
    m \in \eval(A, f)
  \bigm|
    \exists l \in \eval(A, e) \suchthat l < m
  \,\bigr\};
\end{gather*}
then
\[
    \filter\bigl(A, (\less, e, f)\bigr) \defeq
      \filter( A, E, e )
      \intersection
      \filter( A, F, f ).
\]
\end{definition}

But note that we have to consider separately the possible exceptional
outcomes due to the comparison between incompatible locations --- as
reported above, the C Standard states that the order \QuotedCode{<} is
defined only between addresses of the same object, or using our
nomenclature, between locations of the same allocation block;
in all other cases the behaviour is undefined.
Listings~\ref{codesnippet:pointer arithmetics and filter} and
\ref{codesnippet:pointer arithmetics and filter 2} are two examples of the
application of the filter on the relational operator `less-than'.

\subsubsection{Justification of the Definition}
Now we want to provide an
intuitive description of the motivations behind the presented definition
of the filter for the `less than' operator. Let again
\( (\less, e, f) \in \conditions \),
\( A \in \abstractdomain \) and let
\[
  C \in \concretization(A) \intersection
    \modelset\bigl((\less, e, f)\bigr)
    = \filter\bigl(\concretization(A), (\less, e, f)\bigr).
\]
We know, from our definition of the concrete semantics of the operator
`\(\less\)' that
\[
  C \in \modelset\bigl( (\less, e, f) \bigr) \implies
    \eval(C, e) < \eval(C, f).
\]
Basically,
since \( A \) is an abstraction of \( C \) then we have that the value of
\(f\) in \( C \) is approximated by the value of \( f \) in \( A \). The
same holds for the expression \( e \).  This means that the sets \( E \)
and \( F \) contain the value of \( e \) and \( f \) in \( C \),
respectively, then \( C \) is also approximated by \(
\filter\bigl(A, (\less, e, f)\bigr) \).
For completeness we also report a formal proof of the correctness of the
above definition.
First we prove
an analogue of \reflemma{Equality target} for the `less than' operator; then we extend the proof of
\reftheorem{Correctness of the filter} to the `\(\less\)' operator.

\begin{lemma}
\lemmasummary{Less-than target}
Let \( A \in \abstractdomain \) and \( e, f \in \expressions \); let
\begin{gather*}
  E \defeq
  \bigl\{\,
    l \in \eval(A, e)
  \bigm|
    \exists m \in \eval(A, f) \suchthat l < m
  \,\bigr\};
  \\
  F \defeq
  \bigl\{\,
    m \in \eval(A, f)
  \bigm|
    \exists l \in \eval(A, e) \suchthat l < m
  \,\bigr\};
\end{gather*}
then, for all \( C \in \filter\bigl( \concretization(A), c \bigr) \), holds
that
\[
  \eval(C, e) \subseteq E \land \eval(C, f) \subseteq F.
\]
\end{lemma}

\begin{proof}
Let
\( c = (\less, e, f ) \in \conditions \),
\( A \in \abstractdomain \) and let
\( C \in \concretization(A) \intersection \modelset(c) \).
Let \( E \) and \( F \) be defined as in the statement of this lemma.
Recall that from the concrete semantics of the operator `\(\less\)'
described above we have that \( C \in \modelset(c) \) implies that
\[
  \eval(C, e) < \eval(C, f).
\]
Then we have
\begin{prooftable}
TS  & \( \eval(C, e) \subseteq E \land \eval(C, f) \subseteq F \) \\
\hline
H0  & \( C \in \concretization(A) \) \\
H1  & \( \eval(C, e) = \{ l_0 \} \) \\
H2  & \( \eval(C, f) = \{ m_0 \} \) \\
H3  & \( l_0 < m_0 \) \\
H4  & \reflemma{Monotonicity of the eval function}, monotonicity of eval.\\
H5  & \refdefinition{concretization function}, the concretization function. \\
\hline
D0  & \( C \subseteq A \) & (H0, H5) \\
D1  & \( \eval(C, e) \subseteq \eval(A, e) \) & (D0, H4) \\
D2  & \( \eval(C, f) \subseteq \eval(A, f) \) & (D0, H4) \\
D3  & \( \exists l \in \eval(A, e) \suchthat l < m_0 \) & (H3, H1, D1) \\
D4  & \( \exists m \in \eval(A, f) \suchthat l_0 < m \) & (H3, H2, D2) \\
D5  & \( \eval(C, e) \subseteq E \) & (D4, H1) \\
D6  & \( \eval(C, f) \subseteq F \) & (D3, H2) \\
\proofleaf
    & \( \eval(C, e) \subseteq E \land \eval(C, f) \subseteq F \)
    & (D5, D6) \\
\end{prooftable}
\end{proof}

\begin{proof}
\summary{Correctness of the filter on the less-than}
Let \( A \in \abstractdomain \), let \( c = (\less, e, f) \in \conditions \) and let
\( C \in \concretedomain \).
Let \( E \) and \( F \) be defined as in \refdefinition{Filter on the
less-than operator}.
\begin{prooftable}
TS  & \( C \in \concretization\bigl( \filter(A, c) \bigr) \) \\
\hline
H0  & \( C \models c \) \\
H1  & \( C \in \concretization(A) \) \\
H2  & \reflemma{Correctness of filter 2}, correctness of filter 2. \\
H3  & \refdefinition{Filter on the less-than operator},
      filter on the less-than. \\
H4  & \refdefinition{concretization function}, concretization function. \\
H5  & \reflemma{Less-than target}, the less-than target. \\
\hline
D0  & \( \eval(C, e) \subseteq E \) & (H5, H1, H0) \\
D1  & \( \eval(C, f) \subseteq F \) & (H5, H1, H0) \\
D2  & \( C \in \concretization\bigl( \filter(A, E, e) \bigr) \)
    & (D0, H1, H2) \\
D3  & \( C \in \concretization\bigl( \filter(A, F, f) \bigr) \)
    & (D1, H1, H2) \\
D4  & \( C \subseteq \filter(A, E, e) \)  & (D2, H4) \\
D5  & \( C \subseteq \filter(A, F, f) \)  & (D3, H4) \\
D6  & \( \filter(A, c) = \filter(A, E, e) \intersection \filter(A, F, f) \)
    & (H3) \\
D7  & \( C \subseteq \filter(A, E, e) \intersection \filter(A, F, f) \)
    & (D5, D4) \\
D8  & \( C \subseteq \filter(A, c) \)
    & (D7, D6) \\
\proofleaf
    & \( C \in \concretization\bigl( \filter(A, c) \bigr) \)
    & (D8, H4) \\
\end{prooftable}
\end{proof}

Note that the structure of the proof for the correctness of the filter on
the less-that operator is very similar to the structure of the proof for
the equality case: actually the only difference is the
definition of the target sets \( E \) and \( F \).

\lstset{morecomment=[l]{//}}
\codecaption
{an example of analysis of a program involving pointer arithmetic and
filtering on relational operators.}
{pointer arithmetics and filter}
\begin{codesnippet}
int a[10], b[20], c[30], d[40], *p, *q;

if (...) {
  if (...) {
    if(...) { p = a; q = a + 5; }
    else    { p = c; q = d + 20; }
  } else    { p = b; q = b + 10; }

  // \( \eval(\indirection p) = \{ {a}.\Head, {b}.\Head, {c}.\Head \} \)
  // \( \eval(\indirection q) = \{ {a}.\Tail, {b}.\Tail, {d}.\Tail \} \)

  if (p < q) {
    // \( \eval(\indirection p) = \{ {a}.\Head, {b}.\Head \} \)
    // \( \eval(\indirection q) = \{ {a}.\Tail, {b}.\Tail \} \)
  } else { /* Unreachable */ }
else {
  if (...) {
    if (...) { p = a + 5; q = a + 10; }
    else     { p = c + 20; q = d; }
  } else     { p = b + 10; q = b + 20; }

  // \( \eval(\indirection p) = \{ {a}.\Tail, {b}.\Tail, {c}.\Tail \} \)
  // \( \eval(\indirection q) = \{ {a}.\Off, {b}.\Off, {c}.\Head \} \)

  if (p < q) {
    // \( \eval(\indirection p) = \{ {a}.\Tail, {b}.\Tail \} \)
    // \( \eval(\indirection q) = \{ {a}.\Off, {b}.\Off   \} \)
  } else { /* Unreachable */ }
}

// \( \eval(\indirection p) = \{ {a}.\Head, {a}.\Tail, {b}.\Head, {b}.\Tail \} \)
// \( \eval(\indirection q) = \{ {a}.\Tail, {a}.\Off, {b}.\Tail, {b}.\Off \} \)

if (p < q) { /* The same */}
else       { /* The same */}
\end{codesnippet}

\codecaption
{an example of analysis of a program involving pointer arithmetic and
filtering on relational operators.}
{pointer arithmetics and filter 2}
\begin{codesnippet}
struct T { int a, b, c; } t0, t1, t2, t3;
int *p, *q;

if (...) {
  if (...) {
    if(...) { p = &t0.a; q = &t0.b; }
    else    { p = &t2.a; q = &t3.b; }
  } else    { p = &t1.b; q = &t1.b; }

  // \( \eval(\indirection p) = \{ {t0.a}, {t1.b}, {t2.a} \} \)
  // \( \eval(\indirection q) = \{ {t0.b}, {t1.b}, {t3.b} \} \)

  if (p < q) {
    // \( \eval(\indirection p) = \{ {t0.a} \} \)
    // \( \eval(\indirection q) = \{ {t0.b} \} \)
  } else {
    // \( \eval(\indirection p) = \eval(\indirection q) =\{ {t1.b} \} \)
  }
} else {
  if (...) {
    if (...) { p = &t0.b; q = &t0.c; }
    else     { p = &t2.b; q = &t3.c; }
  } else     { p = &t1.a; q = &t1.c; }

  // \( \eval(\indirection p) = \{ {t0.b}, {t1.a}, {t2.b} \} \)
  // \( \eval(\indirection q) = \{ {t0.c}, {t1.b}, {t3.c} \} \)

  if (p < q) {
    // \( \eval(\indirection p) = \{ {t0.b}, {t1.a} \} \)
    // \( \eval(\indirection q) = \{ {t0.c}, {t1.b} \} \)
  } else { /* Unreachable */ }
}

// \( \eval(\indirection p) = \{ {t0.a}, {t0.b}, {t1.a}, {t1.b} \} \)
// \( \eval(\indirection q) = \{ {t0.b}, {t0.c}, {t1.b} \} \)

if (p < q) {
  // \( \eval(\indirection p) = \{ {t0.a}, {t0.b}, {t1.a} \} \)
  // \( \eval(\indirection q) = \{ {t0.b}, {t0.c}, {t1.b} \} \)
} else {
  // \( \eval(\indirection p) = \eval(\indirection q) = \{ t0.b, t1.b \} \)
}
\end{codesnippet}

\subsection{Special Locations}
One of the simplifications introduced in the model of
\refsection{our method}
is that all locations are treated in the same way.  In
particular, in the definition of the abstract evaluation function
(\refdefinition{eval function}) and of the assignment operation
(\refdefinition{Assignment evaluation}) there are no limitations on the
locations that can be dereferenced or modified. However, a realistic memory
model should provide a way to limit, on some locations, the possible
operations. For instance, consider a \emph{null pointer}. The C Standard
specifies that dereferencing a null pointer produces an undefined
behaviour.  From
\cite[6.5.3.2.4]{ISO-C-1999}
\begin{quote}
The unary * operator denotes indirection.  [\ldots] If an invalid value has
been assigned to the pointer, the behavior of the unary * operator is
undefined. [\ldots] Among the invalid values for dereferencing a pointer by
the unary * operator are a null pointer, an address inappropriately aligned
for the type of object pointed to, and the address of an object after the
end of its lifetime.  \end{quote} In other languages, like Java,
dereferencing a null reference throws an exception. Besides of the different
responses that each language exposes, it is quite common that a
language has its own set of configurations that are considered
\emph{exceptional} and treated in an ad-hoc way.
Consider for instance the case of \emph{uninitialized} variables; it
would be possible to formalize a concrete semantics where uninitialized
variables, or pointers pointing to a deallocated memory area,
cannot be \emph{evaluated} and then not \emph{copied}.
Though this kind of conformance is
uncommon in ``every-day'' programs, there exist application areas that require
these restrictions \cite[Rule 9.1]{MISRA-C-2004} \cite{JSF-2005}.
Note that the
general idea is to capture some classes of exceptional behaviours; though
the specific definition of what is exceptional can vary, also inside the
same language.  This section presents a possible extension of the
model presented in \refsection{our method}
that can be used to represent
the described concrete semantics.  We introduce two sets of locations.
\begin{itemize}
\item
  Let \( \noneval \subseteq \locations \) be the set of
  \emph{non-evaluable
  locations}. Informally, we say that
  trying to evaluate a non-evaluable location results in an error.
\item
  Let \( \nonderef \subseteq \locations \) be the set of
  \emph{non-dereference-able locations}. Informally, trying to apply the
  dereference operator to a location of this set will result in an error.
\end{itemize}
To represent the possible run-time errors we use the set
\[
  \rtserrors \defeq \{ \dereferror, \evalerror \}.
\]
The concrete behaviour can be
described by defining an extended version of the evaluation function
(\refdefinition{eval function}).  Let\footnote{Here we assume that the two
sets \( \locations \) and \( \rtserrors \) have disjoint representations.}
\[
  \functiondef
    { \evalp }
    { \concretedomain \times \expressions }
    { \locations \union \rtserrors }
\]
be the total function defined,
for all \( C \in \concretedomain \), \( l \in \locations \) and \( e \in
\expressions \), as
\begin{gather*}
    \evalp(C, l) \defeq
    \begin{cases}
      \evalerror, & \text{if } l \in \noneval; \\
      l,          & \text{otherwise;}
    \end{cases}
    \\
    T \defeq \evalp(C, e);
    \\
    \evalp(C, \indirection e) \defeq
    \begin{cases}
      \evalp(C, e),
        & \text{if } T \in \rtserrors; \\
      \dereferror,
        & \text{if } T \in \nonderef; \\
      \evalerror,
        & \text{if } \post(C, T) \in \noneval; \\
      \post(C, T),
        & \text{otherwise.}
    \end{cases}
\end{gather*}
Note that we have formalized the new evaluation function by tagging the
exceptional paths with the elements of the set `\( \rtserrors \)'.  An
implementation of the execution model here proposed will handle
these exceptional cases by signalling an error and terminating the
execution, by raising an exception and modifying the execution mode or
whatever else is considered appropriate.  This operation can be generalized
to sets as follows.  For every element in the result of the concrete
evaluation, we want to track the corresponding concrete memory description.
Also, we want to explicitly separate the exceptional and the normal
component.  Let
\[
  \functiondef
    { \evalp }
    { \partsof(\concretedomain) \times \expressions }
    { \partsof\bigl( \locations \times \concretedomain \bigr) \times
      \partsof\bigl( \rtserrors \times \concretedomain \bigr) }
\]
be a total function defined, for all \( C \in \concretedomain \) and \( e
\in \expressions \), as
\begin{align*}
  \evalp(C, e) \defeq
  \Bigl<
    & \bigl\{\, (l, D) \bigm| D \in C, \eval(D, e) = l \in \locations \,\bigr\}, \\
    & \bigl\{\, (x, D) \bigm| D \in C, \eval(D, e) = x \in \rtserrors \,\bigr\}
  \Bigr>.
\end{align*}
The abstract counterpart of the operation can thus be defined as
\[
  \functiondef
    { \absevalp }
    { \abstractdomain \times \expressions }
    { \partsof(\locations) \times \abstractdomain
      \times \partsof(\rtserrors) \times \abstractdomain }.
\]
Note that we are simplifying a little --- indeed we assume to approximate
elements of \( \partsof( \locations \times \concretedomain) \) with
elements of the product \( \partsof(\locations) \times \abstractdomain \)
and \( \partsof( \rtserrors \times \concretedomain) \) with elements of \(
\partsof(\rtserrors) \times \abstractdomain \); this is not completely
general, however is sufficient for our goals.  Given \( A \in
\abstractdomain \) and \( e \in \expressions \) we write
\[
  \absevalp(A, e) = \langle L, B, E, C \rangle;
\]
where \( L \subseteq \locations \), \( B, C \in \abstractdomain \) and
\( E \subseteq \rtserrors \) to mean that the abstract evaluation of the
expression \( e \) results in the set of abstract locations \( L \)
and the set of errors \( E \); \( B \) is an approximation of the
abstract memory that generates \( L \) and \( C \) is an approximation of the
abstract memory that generates \( E \).
The requirements for the soundness of the of the abstract operation are the
following. Let, for all \( A \in \abstractdomain \) and \( e \in
\expressions \),
\begin{gather*}
  \absevalp(A, e) = \langle L, B, E, C \rangle; \\
  \evalp\bigl( \concretization(A), e \bigr) = \langle R_0, R_1 \rangle;
\end{gather*}
then, to be sound, the abstract operation must satisfy the following requirements
\begin{gather*}
  \bigl\{\, l \bigm | (l, b) \in R_0 \,\bigr\} \subseteq L; \\
  \bigl\{\, b \bigm | (l, b) \in R_0 \,\bigr\} \subseteq \concretization(B); \\
  \bigl\{\, e \bigm | (e, c) \in R_1 \,\bigr\} \subseteq E; \\
  \bigl\{\, c \bigm | (e, c) \in R_1 \,\bigr\} \subseteq \concretization(C).
\end{gather*}
Let \( A \in \abstractdomain \) and \( l \in \locations \);
then, for the base case, let
\[
  \absevalp(A, l) \defeq
  \begin{cases}
    \bigl< \emptyset, \bot, \{ \evalerror \}, A \bigr>,
      & \text{if } l \in \noneval; \\
    \bigl< \{ l \}, A, \emptyset, \bot \bigr>,
      & \text{otherwise.} \\
  \end{cases}
\]
For the inductive case, let \( e \in \expressions \) and
\begin{gather*}
  \absevalp(A, e) = \langle L_0, A_0, E_0, B_0 \rangle, \\
  L_{ 0, X } = L_0 \intersection \nonderef; \\
  L_{ 0, N } = L_0 \setminus \nonderef; \\
  A_{ 0, N } = \filter(A_0, e, L_{ 0, N } ); \\
  A_{ 0, X } = \filter(A_0, e, L_{ 0, X } ); \\
  L_1        = \post(A_{ 0, N }, L_{ 0, N } ); \\
  L_{ 1, X } = L_1 \intersection \noneval; \\
  L_{ 1, N } = L_1 \setminus \noneval; \\
  A_{ 1, X } = \filter(A_{ 0, N }, \indirection e, L_{ 1, X } ); \\
  A_{ 1, N } = \filter(A_{ 0, N }, \indirection e, L_{ 1, N } ); \\
\end{gather*}
and
\begin{align*}
  E & \defeq E_0 \\
    & \union
    \begin{cases}
    \{ \dereferror \}, & \text{if } L_{ 0, X } \neq \emptyset; \\
    \emptyset,         & \text{otherwise;}
    \end{cases}
    \\
    &
    \union
    \begin{cases}
    \{ \evalerror \}, & \text{if } L_{ 1, X } \neq \emptyset; \\
    \emptyset,         & \text{otherwise.}
    \end{cases}
\end{align*}
Finally,
\[
  \absevalp(A, \indirection e) \defeq
    \langle
      L_{ 1, N },
      A_{ 1, N },
      E,
      B_0 \sqcup A_{ 0, X } \sqcup A_{ 1, X }
    \rangle.
\]
In words, to evaluate \( \indirection e \) we
\begin{enumerate}
\item evaluate \( e \),
\item filter away the non-dereference-able locations,
\item perform the actual dereference,
\item filter away the non-evaluable locations.
\end{enumerate}
We can have an error if the evaluation of \( e \) produces an error (\( B_0
\)), or if we obtain a non-dereference-able location (\( A_{ 0, X }\)) or
if in the last step we obtain a non-evaluable location (\(A_{ X, 1 }\)).
We have a location if all this steps are error free (\( A_{ 1, N }
\sqsubseteq A_{ 0, N } \sqsubseteq A_0 \sqsubseteq A \)).
Note that in the computation of \( A_{ 0, N } = \filter(A_0, e, L_{0,N}) \)
the filter cannot always remove all the non-dereference-able locations from
the result of \( \eval(A_0, e) \).
Note however that we compute the result of the dereference operator, \( L_1
= \post(A_{ 0, N, }, L_{ 0, N } ) \), on the set \( L_{0,N} \) that by
definition does not contain non-dereference-able locations.

Also the formulation of the assignment operator can have its own class of
special locations. For instance it is possible to define a set of
non-modifiable (read-only) locations. Finding a read-only location in the result of
the evaluation of the rhs, the analysis reacts by removing that
location and signalling an error.
For example, in our analyzer we have introduced two special locations.
\begin{itemize}
\item
  The \emph{null location} that represents the concrete \QuotedCode{NULL}
  address described by the C Standard. This location can be
  \emph{evaluated} but it cannot be
  \emph{dereferenced} nor \emph{modified}, i.e,
  \( \text{null} \in \nonderef \), \( \text{null} \not \in \noneval \)
  and it is read-only.\footnote{Limiting our
  view to the points-to analysis, non-dereference-able locations may be
  seen as locations that cannot be read. On the other side, the read-only
  locations proposed for the assignment operation cannot be written.
  In this sense the value of the null location can not be read or written:
  the null location can only be used as target for pointers.}
\item
  The \emph{undefined location} to be used as a target for all
  undefined pointers and for pointers pointing to deallocated memory.
  We have modelled the undefined location as a non-evaluable location.
\end{itemize}
\begin{example}
\codecaption
{an example of dereferentiation of a null pointer.}
{null pointer dereference}
\begin{codesnippet}
int *q, a;

int** f() {
  if (...)  return &q;
  else      return 0;
}

...
q = &a;
// \( \eval(*q) = \{ a \} \)
int **p = f();
// \( \eval(*p) = \{ q, null \} \)
int **p2 = p;
// Null can be evaluated, thus copied.
// \( \eval(*p2) = \{ q, null \} \)
if (...) {
  ... = *p;
  // Null cannot be dereferenced.
  // \( \eval(**p) = \{ a \} \)
} else {
  *p = ...;
  // Null cannot be written.
  // \( \eval(*p) = \{ q \} \)
}
\end{codesnippet}
Consider the code in \refcodesnippet{null pointer dereference}.
From the analysis point of view, the function \QuotedCode{f}
possibly returns null pointers, i.e.,
at \refline 12 we have \( \eval(\Code{*p}) = \{ \Code q,
\Code{null} \} \). The evaluation of the expression \( \indirection p \) at
\refline{21} and the evaluation of the expression \(\indirection
\indirection p \) at \refline{17} produces the following sequence of
steps\footnote{Recall that the syntax of the simplified language formalized in
\refsection{our method} is slightly different from the syntax of
the C language. Indeed we do not distinguish between \emph{expressions} and
\emph{lvalues}, then for example, the C-expression \QuotedCode{p} occurring
as the rhs of an assignment corresponds to \( \indirection p \) in our
language, the C-expression \QuotedCode{*p} as the rhs of
an assignment corresponds to \( \indirection \indirection p \),
while \QuotedCode{*p} as the lhs of an assignment remains the same.}
\begin{pointerarithmetictable}
i
& \( \eval(\indirection p, i) \)
& \( \eval(\indirection \indirection p, i)\) \\
\hline
2 & \( \emptyset \)     & \( \{ p \} \)       \\
1 & \( \{ p \} \)       & \( \{ \text{null}, q \} \) \\
0 & \( \{ \text{null}, q \} \) & \( \{ a, \dereferror \} \) \\
\end{pointerarithmetictable}
In this case, at \refline{17},
the analysis warns about the possibility of a dereference of a
null pointer and it continues the abstract execution assuming that
\QuotedCode{p} is not null.  Instead, at \refline{21}, the evaluation of
the expression \QuotedCode{*p} as the lhs of an assignment does not
raise any error and returns the set \( \{ \Code{null}, \Code{q} \} \).
However at this point the assignment operation detects that the program is
trying to modify the \Code{null} location and it triggers an error since we
have modeled the null location as read-only.
See \reffigure{null pointer dereference} for a graphical representation of
this example.
\begin{figure}
\begin{center}
\input{figures/12}
\caption
  {an example of analysis involving the special location \Code{null}.}
\label{figure:null pointer dereference}
\end{center}
\end{figure}
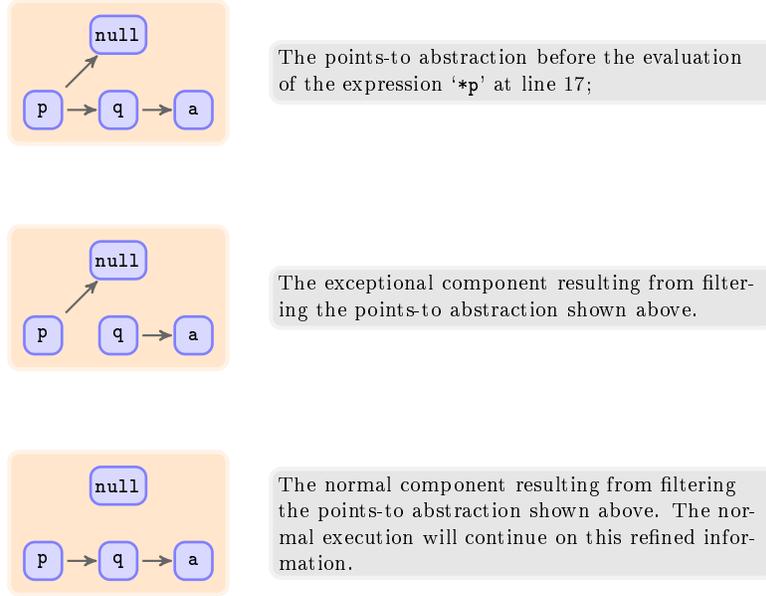
\end{example}

\begin{example}
\codecaption
{an example of the evaluation of an undefined pointer due to an
uninitialized variable.}
{undefined pointer}
\begin{codesnippet}
int *p, *q, **pp, a;

p = &a;
if (...) pp = &p;
// \( \eval(\indirection pp) = \{ p, \text{undef} \} \)
if (...) {
  ... = pp;
  // The undefined location cannot be evaluated.
  // \( \eval(\indirection pp) = \{ p \} \)
} else {
  pp = &q;
  // Assign \( pp \) without evaluating its value.
  // \( \eval(\indirection pp) = \{ q \} \)
}
\end{codesnippet}
Consider \refcodesnippet{undefined pointer}.  At \refline 5
the points-to information is \\ \( \eval(\Code{*pp}) = \{ p, \text{undef} \}
\) and the abstract evaluation of the expression \QuotedCode{*pp}
produces the following sequence of steps
\footnote{Again, using the formalization of the assignment
presented in \refsection{our method} the C-expression \QuotedCode{pp}
occurring as the rhs of an assignment
corresponds to \QuotedCode{*pp} in our formalization.}
\begin{pointerarithmetictable}
i & \( \evalp(\Code{*pp}, i) \) \\
\hline
2 & \( \{ pp \} \) \\
1 & \( \{ p, \evalerror \} \) \\
0 & \( \{ a \} \) \\
\end{pointerarithmetictable}
In the step \( i = 1 \) of the evaluation, the algorithm detects the presence of
the non-evaluable location `undef' and it proceeds by
removing it from the result of the evaluation and by filtering the
memory state against the condition \( (\inequality, \indirection pp, \text{undef}) \).
As result, the analysis is able to infer that after the
execution of \refline 7 holds that \( \eval(\Code{*pp}) = \{ \Code p \} \).
\reffigure{undefined pointer} shows a graphical representation of this
situation.
Instead at \refline{11}, the variable \QuotedCode{pp} is reassigned without
evaluating the undefined location, then without producing any error.
\begin{figure}
\begin{center}
\input{figures/11}
\caption
  {an example of analysis involving the special location \Code{undefined}.}
\label{figure:undefined pointer}
\end{center}
\end{figure}
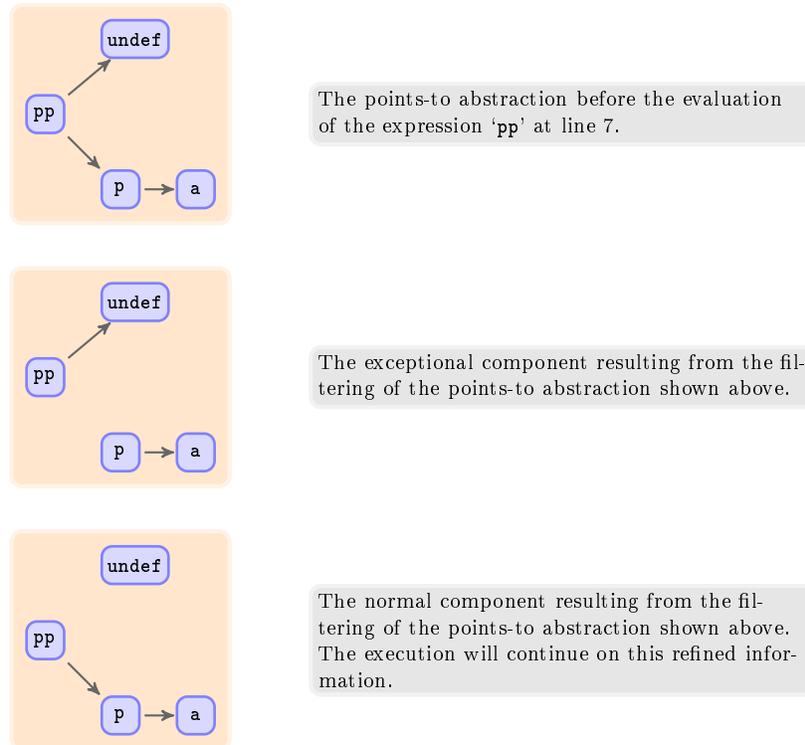
\end{example}

\begin{example}
\codecaption
{an example of the evaluation of an undefined pointer, this time due to a
memory deallocation.}
{undefined pointer 2}
\begin{codesnippet}
int **pp, *p, *q, a;

q = &a;
// \( \eval(*q) = \{ a \} \)

if (...) pp = &p;
else     pp = &q;
// \( \eval(*pp) = \{ q, p \} \)

if (...) {
  int x;
  p = &x;
  ...
} else {
  p = &a;
}
// \( \eval(*p) = \{ a, \text{undef} \} \)

... = *pp;
// Undef cannot be evaluated.
// However, the filter cannot improve the precision.
\end{codesnippet}
Consider the example in \refcodesnippet{undefined pointer 2}.  At
\refline{18} the points-to information \( m^\sharp \) is
\begin{gather*}
  \eval(\indirection pp) = \{ p, q \}, \\
  \eval(\indirection p) = \{ a, \text{undef} \}, \\
  \eval(\indirection q) = \{ a \}.
\end{gather*}
At this point the abstract evaluation of the expression \QuotedCode{*pp}
produces the following sequence of steps
\footnote{Again, using our formalization of the assignment operation the
C-expression \QuotedCode{*pp} corresponds to \( \indirection \indirection p
\).}
\begin{pointerarithmetictable}
i & \( \evalp(\indirection p, i) \) \\
\hline
2 & \( \{ pp \} \) \\
1 & \( \{ p, q \} \) \\
0 & \( \{ a, \evalerror \} \) \\
\end{pointerarithmetictable}
In the last step of the evaluation the algorithm detects the presence of
the non-evaluable location `undef' and it proceeds by
removing this location from the result of the evaluation.
However, in this case the filter is unable to divide the exceptional
from the normal component, as illustrated in \reffigure{undefined pointer 2}.
\begin{figure}
\begin{center}
\input{figures/14}
\caption
  {an representation of the situation of \refcodesnippet{undefined pointer 2}.}
\label{figure:undefined pointer 2}
\end{center}
\end{figure}
\end{example}

The idea of filtering away the exceptional component is formalized in
\cite{ConwayDNB08}. Removing from the abstract execution state those
exceptional configurations already signalled prevents that the same error
is propagated by the analysis from the first point to all the subsequent
program points with the result of soiling the results of the analysis.
Note that also other semantics are possible. For instance, it would be
possible to model the \emph{undefined location} as a non-dereference-able
location instead as of a non-evaluable location. Under this assumptions
uninitialized pointers and pointers pointing to deallocated memory can be
evaluated and thus copied, however it is still treated as an error their
dereference. In the above formalization we explicitly keep track of an
approximation of the exceptional execution paths; however, in many
situations this is too expensive and useless. In these
cases the implementation can simply skip the collection of the exceptional
states and gather only the signalled memory errors.

\subsection{Logical Operators}
The model described in \refsection{our method}
presents a very simplified definition of \emph{boolean condition};
for example, it does not consider \emph{logical operators}: \emph{and}
(\Code{\&\&}), \emph{or} (\Code{||}) and the \emph{not} (\Code{!}).
The first step necessary in order to handle these operators,
is to extend the set of \emph{conditions}.

\begin{definition}
\definitionsummary{Extended conditions}
Let `\( \extconditions \)' be the set defined as the language generated by the
grammar
\[
  e ::= c \mid (\Cnot e_0) \mid (e_0 \Cor e_1) \mid (e_0 \Cand e_1)
\]
where \( c \in \conditions \) is an atomic condition
and \( e_0, e_1 \in \extconditions \) are two extended conditions.
\end{definition}

The next step is to define the \emph{value} of the new conditions.

\begin{definition}
\definitionsummary{Concrete semantics of the extended conditions}
Let \( C \in \concretedomain \) and \( c_0, c_1 \in \extconditions \); then
\begin{gather*}
  C \models (\Cnot c_0)     \defiff C \not \models c_0; \\
  C \models (c_0 \Cand c_1) \defiff C \models c_0 \land C \models c_1; \\
  C \models (c_0 \Cor c_1)  \defiff C \models c_0 \lor C \models c_1.
\end{gather*}
\end{definition}

The definition of concrete filter do not need to be updated as it is
expressed in terms of the \emph{value} of the conditions.
Finally, we update the definition of the abstract filter as to handle
the new conditions.

\begin{definition}
\definitionsummary{Extended filter}
Let
\begin{gather*}
  \functiondef
    { \filter }
    { \abstractdomain \times \extconditions }
    { \abstractdomain \times \abstractdomain };
  \\
  \functiondef
    { \filter }
    { \abstractdomain \times \abstractdomain \times \extconditions }
    { \abstractdomain \times \abstractdomain };
\end{gather*}
be defined, for all \( A, B \in \abstractdomain \) and \( e,f \in \expressions
\), as
\begin{gather*}
  \filter\bigl(A, B, (\equality, e, f) \bigr) \defeq
    \Bigl<
      \filter\bigl(A, (\equality, e, f) \bigr),
      \filter\bigl(B, (\inequality, e, f) \bigr)
    \Bigr>;
 \\
  \filter\bigl(A, B, (\inequality, e, f) \bigr) \defeq
    \Bigl<
      \filter\bigl(A, (\inequality, e, f) \bigr),
      \filter\bigl(B, (\equality, e, f) \bigr)
    \Bigr>;
\end{gather*}
for all \( c_0, c_1 \in \extconditions \), as
\begin{gather*}
  \filter(A, B, \Cnot c_0)       \defeq \filter(B, A, c_0); \\
  \filter(A, B, c_0 \Cor c_1)  \defeq
  \filter\Bigl(A, B, \Cnot \bigl((\Cnot c_0) \Cand (\Cnot c_1)\bigr) \Bigr); \\
  \filter(A, B, c_0 \Cand c_1) \defeq
  \bigl< A_0 \sqcap A_1, B_0 \sqcup (A_0 \sqcap B_1) \bigr>;
\end{gather*}
where
\begin{gather*}
  \filter(A, B, c_0) = \langle A_0, B_0 \rangle; \\
  \filter(A, B, c_1) = \langle A_1, B_1 \rangle.
\end{gather*}
Finally, for all \( A \in \abstractdomain \) and \( c \in \extconditions
\), we define
\[
    \filter(A, c) \defeq \filter(A, A, c).
\]
\end{definition}

In this formalization of the filter operation, \( \filter(A,
c) \) returns a pair of abstract memories: the first component is an
approximation of the states of \( A \) in which the condition \( c \) is
\emph{true}; the second is an approximation of the states of \( A \) in which
the condition \( c \) is \emph{false}.
In this definition we have mentioned only the equality and inequality
operator; however, it can be easily extended to comprehend relational
operators (\refsection{relational operators}).

Note that as shown by \refsection{completeness}, the formulation
of the filter (\refdefinition{Filter 3}) is not optimal and
\refexample{iterating filter 1} shows that iterating the application
of the filter it is possible to improve the precision.
In \reffigure{iterating filter 2} we show that also having a
filter that is optimal on the atomic conditions,
iterating the application of the filter
can improve the precision.
Note indeed that on the atomic
conditions \QuotedCode{****p4 == \&a} and \QuotedCode{***q3 == \&b},
the filter operates optimally (\refdefinition{Filter 3}).

\begin{figure}
\begin{center}
\input{figures/13}
\caption
  {an example of that shows that iterating the filter on more conditions
  can improve the precision of the approximation.}
\label{figure:iterating filter 2}
\end{center}
\end{figure}

%% file: figures/12.tex
\input{figures/common_style}
\begin{tikzpicture}
  \def\drawThePicture[#1]{
    \abstractLocation p (0, 0)
    \abstractLocation q (1, 0)
    \abstractLocation a (2, 0)
    \node[abstract location style] (n) at (1, 1) {\Code{null}};
    \figureBackground
        (p.south -| p.west) (n.north -| a.east)
    \node [description label] at (2, 0.5) {#1};
    \path \edgePointsTo(q,a);
  }
  \begin{scope}
    \drawThePicture[
      {
        The points-to abstraction before the evaluation of
        the expression \QuotedCode{*p} at \refline{17};
      }]
    \path \edgePointsTo(p,q)
          \edgePointsTo(p,n);
  \end{scope}
  \begin{scope}[yshift=-3cm]
    \drawThePicture[
      {
        The exceptional component resulting from filtering the
        points-to abstraction shown above.
      }]
    \path \edgePointsTo(p,n);
  \end{scope}
  \begin{scope}[yshift=-6cm]
    \drawThePicture[
      {
        The normal component resulting from filtering the
        points-to abstraction shown above. The normal execution will continue on
        this refined information.
      }]
    \path \edgePointsTo(p,q);
  \end{scope}
\end{tikzpicture}

%% file: figures/11.tex
\input{figures/common_style}
\begin{tikzpicture}
  \def\drawThePicture[#1]{
    \abstractLocation pp (0, 1)
    \abstractLocation p (1, 0)
    \abstractLocation a (2, 0)
    \path (1,2) node (q) {};
    \node[abstract location style, anchor=west] (u) at (p.west |- q) {\Code{undef}};
    \figureBackground
        (p.south -| pp.west) (u.north -| a.east)
    \node [description label] at (2.5, 1) {#1};
    \path \edgePointsTo(p,a);
  }
  \begin{scope}
    \drawThePicture[
      {
        The points-to abstraction before the evaluation of the expression
        \QuotedCode{pp} at \refline{7}.
      }]
    \path
      \edgePointsTo(pp,p)
      \edgePointsTo(pp,u)
      ;
  \end{scope}
  \begin{scope}[yshift=-3.5cm]
    \drawThePicture[
      {
        The exceptional component resulting from the filtering of the
        points-to abstraction shown above.
      }]
    \path \edgePointsTo(pp,u);
  \end{scope}
  \begin{scope}[yshift=-7cm]
    \drawThePicture[
      {
        The normal component resulting from the filtering of the
        points-to abstraction shown above. The execution will continue on
        this refined information.
      }]
    \path \edgePointsTo(pp,p);
  \end{scope}
\end{tikzpicture}

%% file: figures/14.tex
\input{figures/common_style}
\begin{tikzpicture}
  \def\drawThePicture[#1]{
    \abstractLocation pp (0, 1)
    \abstractLocation p (1, 0)
    \abstractLocation q (1, 2)
    \abstractLocation a (2, 2)
    \node[abstract location style, anchor=west] (u) at (a.west |- p) {\Code{undef}};
    \figureBackground
        (p.south -| pp.west) (a.north -| u.east)
    \node [description label] at (2.5, 1) {#1};
    \path \edgePointsTo(q,a);
  }
  \begin{scope}
    \drawThePicture[
      {
        The points-to abstraction \( m^\sharp \) before the evaluation of
        the expression \QuotedCode{*pp} at \refline{19};
      }]
    \path
      \edgePointsToL(pp,q)
      \edgePointsToR(pp,p)
      \edgePointsTo(p,u)
      \edgePointsTo(p,a)
      ;
  \end{scope}
  \begin{scope}[yshift=-4cm]
    \drawThePicture[
      {
        A concrete memory description model \( m_0 \) of the condition
        \[
          c = ( \inequality, \indirection \indirection pp, \text{undef})
        \]
        approximated by \( m^\sharp \).
      }]
    \path
      \edgePointsToL(pp,q)
      \edgePointsTo(p,u)
      ;
  \end{scope}
  \begin{scope}[yshift=-8cm]
    \drawThePicture[
      {
        Another concrete memory description model \( m_1 \) of the condition
        \( c \) approximated by \( m^\sharp \).  Note however that
        \[
          \abstraction\bigl( \{ m_0, m_1 \} \bigr) = m^\sharp,
        \]
        that is the filter cannot remove any arc.
      }]
    \path
      \edgePointsToR(pp,p)
      \edgePointsTo(p,a)
      ;
  \end{scope}
\end{tikzpicture}

%% file: figures/13.tex
\input{figures/common_style}
\begin{tikzpicture}
  \def\drawThePicture[#1,#2]{
    \abstractLocation p4  (0, 1)
    \abstractLocation p3  (1, 0)
    \abstractLocation p2  (2, 0)
    \abstractLocation p1  (3, 0)
    \abstractLocation a   (4, 0)
    \abstractLocation q3  (1, 2)
    \abstractLocation q2  (2, 2)
    \abstractLocation q1  (3, 2)
    \abstractLocation b   (4, 2)
    #2
    \figureBackground
        (p3.south -| p4.west) (b.north -| b.east)
    \node [description label] at (4, 1) {#1};
    \path
      \edgePointsToR(p4,p3)
      \edgePointsTo(p3,p2)
      \edgePointsTo(p2,p1)
      \edgePointsTo(p1,a)
      \edgePointsTo(q3,q2)
      \edgePointsTo(q2,q1)
      \edgePointsTo(q1,b)
      ;
  }
  \begin{scope}
    \drawThePicture[
      {
        A representation of the initial points-to information.
      },{}]
    \path
      \edgePointsToL(p4,q3)
      \edgePointsTo(q3,p2)
      \edgePointsTo(p2,q1)
      \edgePointsTo(q1,a)
      ;
  \end{scope}
  \begin{scope}[yshift=-3.5cm]
    \drawThePicture[
      {
        Filtering the points-to information against the expression
        \QuotedCode{**q3}, and the target set \( \{ \Code b \} \).
        The arc \( (\Code{q1}, \Code{a}) \) is removed.
      },{
        \selAbstractLocation q3 (1, 2)
        \selAbstractLocation b (4, 2)
      }]
    \path
      \edgePointsToL(p4,q3)
      \edgePointsTo(q3,p2)
      \edgePointsTo(p2,q1)
      (q1) edge [points--to style, dashed, red] (a)
      ;
  \end{scope}
  \begin{scope}[yshift=-7cm]
    \drawThePicture[
      {
        Filtering the points-to information against the expression
        \QuotedCode{***q4}, and the target set \( \{ \Code a \} \).
        The arc \( (\Code{p2}, \Code{q1}) \) is removed.
      },{
        \selAbstractLocation p4 (0, 1)
        \selAbstractLocation a (4, 0)
      }]
    \path
      \edgePointsToL(p4,q3)
      \edgePointsTo(q3,p2)
      (p2) edge [points--to style, dashed, red] (q1)
      ;
  \end{scope}
  \begin{scope}[yshift=-10.5cm]
    \drawThePicture[
      {
        Filtering the points-to information against the expression
        \QuotedCode{**q3}, and the target set \( \{ \Code b \} \).
        The arc \( (\Code{q3}, \Code{p2}) \) is removed.
      },{
        \selAbstractLocation q3 (1, 2)
        \selAbstractLocation b (4, 2)
      }]
    \path
      \edgePointsToL(p4,q3)
      (q3) edge [points--to style, dashed, red] (p2)
      ;
  \end{scope}
  \begin{scope}[yshift=-14cm]
    \drawThePicture[
      {
        Filtering the points-to information against the expression
        \QuotedCode{***q4}, and the target set \( \{ \Code a \} \).
        The arc \( (\Code{p4}, \Code{q3}) \) is removed.
      },{
        \selAbstractLocation p4 (0, 1)
        \selAbstractLocation a (4, 0)
      }]
    \path
      (p4) edge [points--to style, bend left, dashed, red] (q3)
      ;
  \end{scope}
\end{tikzpicture}

%% file: tex/conclusions.tex
\section{Conclusions and Future Developments}
\label{section:conclusion}
Alias analysis is an important step in the process of static analysis of
programs.  Compiler oriented applications are the most common clients of
alias information. However, compilers stress the focus on \emph{fast}
analyses, whereas \emph{verifier} oriented applications require
\emph{precise} but slower techniques.  The present work, trying to address
verifier needs, discusses one of the most common method used to model the
aliasing problem: the \emph{points-to} representation.  Known results are
presented within a formal model; a novel operation of \emph{filter} is
described and finally a formal proof of correctness of the presented method
is reported.

A working prototype of the method has been implemented as part of
the ECLAIR system, which targets the analysis of mainstream languages by
building upon CLAIR, the `Combined Language and Abstract Interpretation
Resource', which was initially developed and used in a teaching context
(see \url{http://www.cs.unipr.it/clair/}).

However, many tasks have to be completed.  Some of the features of the C
language are still missing.  One of the questions not answered is how it is
possible to exploit the knowledge of the architecture/compiler target of
the analysis process.  For instance, the precise handling of \emph{unions}
and \emph{casts} requires the knowledge of the relative size of basic types, the
alignment issues and all the details that relate to the memory layout.

The memory model described in \refsection{extended memory model} and
implemented makes strong hypotheses about the correctness of the type
information.  For example, the described abstract memory does not allow to
precisely track pointers of type \Code{char*} resulting from casts of
pointer to objects of other types.  Though the literature contains some
proposals of how to avoid the necessity of relying on type informations
\cite{WilsonL95} and how to analyze union and casts \cite{Mine06a}, it is
unclear whether these can be applied to our situation.  On the other hand,
the memory model does not require any special information about the type of
variables. For instance, our analysis is able `out of the box' to track
pointer casted and assigned to integer.  Architecture-specific information
is also required in order to resolve the many \emph{implementation}-defined
behaviours present in the C Standard.  When the behaviour of the analyzed
programs depends on these rules of the language, the analyzer, if not
provided with additional information, can only warn and proceed with
a conservative approximation of the execution that very often in few steps
degenerates to the \emph{top} approximation.

Consider for instance \emph{off-by-one} locations.  Currently, the memory
model reserves an explicit abstract location address to represent
off-by-one locations only at the end of arrays; this means that scalar
variables do not have a corresponding off-by-one location. Hence, the
current implementation forbids pointer arithmetics on the address of a
scalar object, also when the increment is equal to 1, though the C Standard
allows it \cite[6.5.6.7]{ISO-C-1999}.  Moreover, in the presented
formulation, the handling of pointer arithmetic on arrays assumes that the
off-by-one address never overlaps with another valid location, though this
is allowed by the standard \cite[5.6.9.6]{ISO-C-1999}.

To increase the precision of the provided alias analysis it would be
possible to couple the points-to analysis with a \emph{shape analysis} that
would produce a more precise approximation of recursive data structures
\cite{Deutsch94}.

For the implementation it will be necessary to realize a complete
experimental evaluation of the proposed technique in order to produce
quantitative data for the comparison with other approaches.

%% file: tex/acknowledgments.tex
I would like to express my gratitude to those who had supported and
assisted me in this work.

I wish to thank my advisor Roberto Bagnara and my co-advisor Enea
Zaffanella for the countless suggestions they gave me about the alias
problem and about the difficult task of writing a scientific paper.
My understanding of the static analysis problem and of the abstract
interpretation theory is due to their teachings. They helped
me in the initial phase of the formalization of the problem and they
indicated to me the path to follow for the proof of correctness.
They also showed me the possible approaches for the interaction with
other domains.  Finally, I must thank them for the careful proofreading
of this and earlier versions of this paper.

I must also thank Paolo Bolzoni, as I developed my initial
understanding of the alias problem by reading with him the MSc
Thesis of Maryam Emami [Ema93].  Moreover, the research described in this
paper begun with a joint work with Paolo and several ideas are the
results of this early effort.